\documentclass[prb,twocolumn,showpacs,amsmath,amssymb,superscriptaddress,floatfix]{revtex4-1}
\usepackage{amsmath}
\usepackage{amssymb}
\usepackage{amsthm}
\usepackage{amsfonts}
\usepackage{dsfont}
\usepackage{bigints}
\usepackage{tikz}

\usetikzlibrary{matrix,fit}
\usepackage{easybmat}
\usepackage{multirow,bigdelim}
\usepackage{listings}
\usepackage{enumerate}
\usepackage{latexsym}
\usepackage{mathdots}
\usepackage{amsthm}

\usepackage{psfrag}

\usepackage{bm}
\usepackage{graphicx}
\usepackage{subfigure}
\usepackage{color}

\newcommand{\mm}{\text{-}1}

\newcommand{\spa}{\text{ }}

\newcommand{\beq}{\begin{equation}}
\newcommand{\eneq}{\end{equation}}

\newcommand*{\bordr}{\multicolumn{1}{c|}{}}
\newcommand*{\bordl}{\multicolumn{1}{|c}{}}

\newcommand{\braket}[2]{\left\langle #1 | #2 \right\rangle}
\newcommand{\bra}[1]{\left\langle#1\right|}
\newcommand{\ket}[1]{\left|#1\right\rangle}

\newcommand{\bM}{\boldsymbol{M}}
\newcommand{\bh}{\boldsymbol{l}}
\newcommand{\bv}{\boldsymbol{r}}
\newcommand{\bP}{\boldsymbol{P}}
\newcommand{\bQ}{\boldsymbol{Q}}
\newcommand{\bJ}{\boldsymbol{J}}
\newcommand{\ba}{\boldsymbol{\alpha}}
\newcommand{\bb}{\boldsymbol{\beta}}

\newcommand{\bLam}{\boldsymbol{\Lambda}}

\newcommand{\bl}{\boldsymbol{\lambda}}
\newcommand{\bD}{\boldsymbol{\Delta}}

\newcommand{\bO}{\boldsymbol{O}}

\newcommand{\bS}{\boldsymbol{S}}

\newcommand{\ml}{\mathcal{L}}
\newcommand{\mr}{\mathcal{R}}
\newcommand{\ma}{\mathcal{A}}
\newcommand{\mb}{\mathcal{B}}
\newcommand{\ms}{\mathcal{S}}

\newcommand{\wc}{\Upsilon}
\newcommand{\wvr}{W_R}
\newcommand{\wvl}{W_L}
\newcommand{\wbr}{{\beta^r_F}}
\newcommand{\wbl}{{\beta^l_F}}

\newcommand{\wJ}{J_{\textrm{unit}}}
\newcommand{\wF}{F_{\textrm{unit}}}
\newcommand{\wmr}{\mr_{\textrm{unit}}}
\newcommand{\wml}{\ml_{\textrm{unit}}}
\newcommand{\wLam}{\Lambda_{\textrm{unit}}}
\newcommand{\wS}{S_{\textrm{unit}}}
\newcommand{\wP}{P_{\textrm{unit}}}
\newcommand{\wQ}{Q_{\textrm{unit}}}
\newcommand{\wwP}{\widetilde{P}_{\textrm{unit}}}
\newcommand{\wwQ}{\widetilde{Q}_{\textrm{unit}}}

\newcommand{\rS}{S_{\textrm{rest}}}
\newcommand{\rJ}{J_{\textrm{rest}}}
\newcommand{\rLam}{\Lambda_{\textrm{rest}}}
\newcommand{\rP}{P_{\textrm{rest}}}
\newcommand{\rQ}{Q_{\textrm{rest}}}
\newcommand{\wrP}{\widetilde{P}_{\textrm{rest}}}
\newcommand{\wrQ}{\widetilde{Q}_{\textrm{rest}}}

\newcommand{\bR}{{\bf r}}
\newcommand{\bL}{{\bf l}}
\newcommand{\prodal}[2]{\underset{#1}{\overset{#2}{\prod}}}
\newcommand{\sumal}[2]{\underset{#1}{\overset{#2}{\sum}}}

\newcommand*\hexbrace[2]{%
  \underset{#2}{\underbrace{\rule{#1}{0pt}}}}

\newtheorem{theorem}{Theorem}[section]

\newtheorem{lemma}[theorem]{Lemma}

\def\eg{{\it e.g.}\ }

\input{epsf}

\newcommand{\nn}{\nonumber}
\newcommand{\rhored}{\rho_{{\rm red}}}

\begin{document}

\tolerance 10000

\newcommand{\vk}{{\bf k}}

\def\ltriangle{\mbox{\begin{picture}(7,10)
\put(1,0){\line(1,0){5}}
\put(6,0){\line(-1,2){5}}
\put(1,0){\line(0,1){10}}
\end{picture}
}}

\def\utriangle{\mbox{\begin{picture}(7,0)
\put(1,7.5){\line(1,0){5}}
\put(6,-2.5){\line(-1,2){5}}
\put(6,-2.5){\line(0,1){10}}
\end{picture}
}}

\newcommand{\twopartdef}[4]
{
	\left\{
		\begin{array}{ll}
			#1 & \mbox{\textrm{if} } #2 \\
			#3 & \mbox{\textrm{if} } #4
		\end{array}
	\right.
}

\newcommand{\threepartdef}[6]
{
	\left\{
		\begin{array}{ll}
			#1 & \mbox{\textrm{if} } #2 \\
			#3 & \mbox{\textrm{if} } #4 \\
			#5 & \mbox{\textrm{if} } #6
		\end{array}
	\right.
}

\newcommand{\threepartdeftwo}[5]
{
	\left\{
		\begin{array}{ll}
			#1 & \mbox{\textrm{if} } #2 \\
			#3 & \mbox{\textrm{if} } #4 \\
			#5 & \mbox{\textrm{otherwise} }
		\end{array}
	\right.
}

\newcommand{\fourpartdef}[7]
{
	\left\{
		\begin{array}{ll}
			#1 & \mbox{\textrm{if} } #2 \\
			#3 & \mbox{\textrm{if} } #4 \\
			#5 & \mbox{\textrm{if} } #6 \\
			#7 & \mbox{\textrm{otherwise}}
		\end{array}
	\right.
}


\title{Entanglement of Exact Excited States of AKLT Models: Exact Results, Many-Body Scars and the Violation of Strong ETH}

\author{Sanjay Moudgalya}
\affiliation{Department of Physics, Princeton University, Princeton, NJ 08544, USA}
\author{Nicolas Regnault}
\affiliation{Laboratoire Pierre Aigrain, Ecole normale sup\'erieure, PSL University, Sorbonne Universit\'e, Universit\'e Paris Diderot, Sorbonne Paris Cit\'e, CNRS, 24 rue Lhomond, 75005 Paris France}
\author{B. Andrei Bernevig}
\affiliation{Department of Physics, Princeton University, Princeton, NJ 08544, USA}
\affiliation{
Dahlem Center for Complex Quantum Systems and Fachbereich Physik,
Freie Universitat Berlin, Arnimallee 14, 14195 Berlin, Germany
        }
\affiliation{
Max Planck Institute of Microstructure Physics, 
06120 Halle, Germany
}
\affiliation{Donostia International Physics Center, P. Manuel de Lardizabal 4, 20018 Donostia-San Sebasti\'{a}́n, Spain}
\affiliation{Sorbonne Universit\'{e}s, UPMC Univ Paris 06, UMR 7589, LPTHE, F-75005, Paris, France}
\affiliation{Laboratoire Pierre Aigrain, Ecole normale sup\'erieure, PSL University, Sorbonne Universit\'e, Universit\'e Paris Diderot, Sorbonne Paris Cit\'e, CNRS, 24 rue Lhomond, 75005 Paris France}

\date{\today}

\begin{abstract}
We obtain multiple exact results on the entanglement of the exact excited states of non-integrable models we introduced in arXiv:1708.05021. We first discuss a general formalism to analytically compute the entanglement spectra of exact excited states using Matrix Product States and Matrix Product Operators and illustrate the method by reproducing a general result on single-mode excitations. We then apply this technique to analytically obtain the entanglement spectra of the infinite tower of states of the spin-$S$ AKLT models in the zero and finite energy density limits. We show that in the zero density limit, the entanglement spectra of the tower of states are multiple shifted copies of the ground state entanglement spectrum in the thermodynamic limit. We show that such a resemblance is destroyed at any non-zero energy density. Furthermore, the entanglement entropy $\mathcal{S}$ of the states of the tower that are in the bulk of the spectrum is sub-thermal $\mathcal{S} \propto \log L$, as opposed to a volume-law $\mathcal{S} \propto L$, thus indicating a violation of the strong Eigenstate Thermalization Hypothesis (ETH). These states are examples of what are now called many-body scars. Finally, we analytically study the finite-size effects and symmetry-protected degeneracies in the entanglement spectra of the excited states, extending the existing theory. 
\end{abstract}

\maketitle

\section{Introduction}
Non-integrable translation invariant models have been of great interest recently. Such models have very few conserved quantities and show various interesting dynamical phenomena, including thermalization\cite{rigol2008thermalization} and, upon the introduction of disorder or quasiperiodicity, many-body localization.\cite{basko2006metal,nandkishore2014many, khemani2017two} 
Since dynamics depends on the properties of all the eigenstates, highly-excited states of non-integrable models have been extensively studied in various models in one and two dimensions. \cite{rigol2008thermalization, pal2010many, kjall2014many, vznidarivc2008many, luitz2015many, bardarson2012unbounded,iyer2013many, rigol2009breakdown,geraedts2017characterizing, chandran2014many, fendley2016strong}
Particularly, the eigenstates in the bulk of the spectrum of several non-integrable models are expected to satisfy the Eigenstate Thermalization Hypothesis (ETH),\cite{deutsch1991quantum, srednicki1994chaos, rigol2008thermalization} with the notable exception of Many-Body Localizalized (MBL) systems.\cite{huse2013localization, pal2010many, rigol2009breakdown}   While several analytical results on the entanglement structure of highly excited states in generic models have been obtained,\cite{huang2017universal, lu2017renyi, huang2017eigenstate, liu2018quantum} exactly solvable examples are desired.  
The entanglement structure of low-energy excitations in integrable and non-integrable models has been studied analytically and numerically in detail,\cite{haegeman2013elementary, pizorn2012universality, haegeman2012variational, haegeman2013post, haegeman2017diagonalizing, zauner2015transfer, zauner2018topological, thomale2015entanglement} particularly using the language of Matrix Product States (MPS).\cite{schollwock2011density,orus2014practical}
Similar to the ground states of gapped Hamiltonians,\cite{verstraete2006matrix} low-energy excited states of gapped Hamiltonians are in principle also captured by this MPS framework.\cite{haegeman2013elementary}
However, even within single-mode excitations, the lack of explicit examples has hindered a study of their entanglement in more detail; for example the general nature of finite-size corrections to the entanglement spectra is unknown.
Beyond low-energy excitations, the structure of excited states has been studied in the MBL regime, where all the eigenstates exhibit area-law entanglement,\cite{huse2013localization} and consequently have an efficient MPS representation.\cite{friesdorf2015many,khemani2016obtaining, yu2017finding}
In the thermal regime, however, very little is analytically known about the kind of excited states that can exist in the bulk of the spectrum of generic non-integrable models.\cite{shiraishi2017systematic, mondaini2017comment, shiraishi2017reply, turner2018weak, schecter2018many, vafek2017entanglement} 
For example, can certain highly excited states of thermal non-integrable models have an exact or approximate matrix product structure with a finite or low bond dimension in the thermodynamic limit?
Recently, a tower of exact excited states were analytically obtained by us in a family of non-integrable models, the spin-$S$ AKLT models.\cite{self} 
The entanglement of the ground states of the spin-$S$ AKLT models and their generalizations has been extensively studied in the literature. \cite{fan2004entanglement, korepin2010entanglement, katsura2007exact, santos2011negativity, santos2012entanglement, santos2012entanglement2, santos2012entanglement3, santos2016negativity, katsura2008entanglement, xu2008block, xu2008entanglement}
Being the first few known examples of exact eigenstates of non-integrable models, we propose to use the excited states of these models to test conjectures on eigenstates that exist in the literature. 
We recover the general entanglement spectra of single-mode excitations, earlier obtained on general grounds.\cite{pizorn2012universality, haegeman2013elementary}
We also derive the entanglement spectrum of an entire tower of exact states, thus generalizing the single-mode results to these set of states.
The tower of states have an interesting entanglement structure in that the zero energy density states entanglement spectra is composed of shifted copies of the ground state entanglement spectrum. This structure generalizes the earlier result obtained on the entanglement spectra of SMA excitations.
We find that the finite energy density states in the tower have a sub-thermal entanglement entropy scaling in spite of the fact that they appear to be in the bulk of the spectrum.\cite{self} More precisely, the entanglement entropy $\ms$ for these states scales as $\ms \propto \log L$ where $L$ is the subsystem size. This indicates a violation of the strong ETH,\cite{kim2014testing, garrison2018does} which states that \emph{all} the eigenstates in the bulk of the spectrum of a non-integrable model in a given quantum number sector are thermal, i.e. their entanglement entropy scales with the volume of the subsystem ($\ms \propto L$). 
This paper is organized as follows.
We begin by reviewing the tools we use to compute the entanglement spectrum, i.e. Matrix Product States (MPS) and their properties in Sec.~\ref{sec:MPS}, and Matrix Product Operators (MPO) in Sec.~\ref{sec:MPO}. In these sections, we provide some examples for the AKLT models.
Readers familiar with these approaches can directly proceed to Sec.~\ref{sec:MPOMPS}, where we discuss the structure and properties of states that are created by the action of an operator (MPO) on the ground state (MPS). 
From Sec.~\ref{sec:SMA}, we move on to the main results and derive the entanglement spectra of single-mode excitations, focusing on the AKLT Arovas states and spin-2$S$ magnons.
In Secs.~\ref{sec:beyondSMA} and \ref{sec:Towerofstates}, we consider states beyond single-mode excitations. We compute the entanglement spectrum of the tower of states in spin-$S$ AKLT models, where we work in the zero energy density and finite energy density regimes separately.
Further, in Sec.~\ref{sec:ETH}, we discuss the violation of the Eigenstate Thermalization Hypothesis and then show numerical results away from the AKLT point.
In Secs.~\ref{sec:SPTorderMPO} and \ref{sec:finitesize}, we review symmetries and their effects on the entanglement spectra of the ground states, and discuss symmetry-protected exact degeneracies and finite-size splittings in the entanglement spectra of the excited states. 
We close with conclusions and outlook in Sec.~\ref{sec:conclusions}.

\section{Matrix Product States}\label{sec:MPS}

In this section we provide a basic introduction to the Matrix Product States (MPS) and their properties. We invite readers not familiar with MPS to read numerous reviews and lecture notes in the literature.\cite{schollwock2011density, verstraete2008matrix, orus2014practical, perez2006matrix} 
\subsection{Definition and properties}\label{sec:MPSproperties}
We consider a spin-$S$ chain with $L$ sites. A simple many-body basis for the system is made of the product states $\ket{m_1 m_2 \dots m_L}$ where $m_i = -S, -S + 1, \dots, S-1, S$ is the projection along the $z$-axis of the spin at site $i$. Any wavefunction of the many-body Hilbert space can be decomposed as
\begin{equation}
    \ket{\psi} = \sum_{\{m_1,m_2,\dots,m_L\}}{c_{m_1 m_2 \dots m_L} \ket{m_1 m_2 \dots m_L}}.   
\end{equation}
In all generality, the coefficients $c_{m_1, m_2 \dots m_L}$ can always be written as an MPS,\cite{verstraete2006matrix} i.e, 
\begin{equation}
c_{m_1, m_2 \dots m_L} = [{b^l_A}^T A_1^{[m_1]} A_2^{[m_2]} \dots A_L^{[m_L]} b^r_A].
\label{coeffMPS}
\end{equation}
The state $\ket{\psi}$ then reads 
\begin{equation}
    \ket{\psi} = \sumal{\{m_1 m_2 \dots m_L\}}{}{[{b^l_A}^T A_1^{[m_1]} \dots A_L^{[m_L]} b^r_A] \ket{m_1 \dots m_L}}.
\label{generalOBCMPS}
\end{equation}
In Eqs.~(\ref{coeffMPS}) and (\ref{generalOBCMPS}), $A_1^{[m_1]}, \dots ,A_{L-1}^{[m_{L-1}]}$ and $A_L^{[m_L]}$ are $\chi \times \chi$ matrices over an auxiliary space. $\chi$ is the bond-dimension of the MPS and the corresponding indices are the ancilla. $b^l_A$ and $b^r_A$ are $\chi$-dimensional left and right boundary vectors that determine the boundary conditions for the wavefunction. The $\{[m_i]\}$ are called the physical indices and can take $d = 2S + 1$ values ($d$ is the physical dimension, i.e. the dimension of the local physical Hilbert space on site $i$). In a compact notation, we can think of the $A_i$'s as $d \times \chi \times \chi$ tensors.  

An MPS representation is particularly powerful if the matrices $A^{[m_i]}_i$ are site-independent, i.e. $A^{[m_i]}_i = A^{[m_i]}$. Typically, translation invariant systems admit such a site independent MPS. Many computations involving an MPS can then be simplified once we introduce the transfer matrix
\begin{equation}
    E = \sum_m{{A^{[m]}}^\ast \otimes A^{[m]}} 
\label{MPSTransfer}
\end{equation}
where $\ast$ denotes complex conjugation and the $\otimes$ is over the ancilla. The transfer matrix is thus a $\chi \times \chi \times \chi \times \chi$ tensor that can also be viewed as $\chi^2 \times \chi^2$ matrix by grouping the left and right ancilla of the two MPS copies together. The simplification provided by the MPS description can be illustrated by computing the norm $\braket{\psi}{\psi}$ of the state $\ket{\psi}$, 
\begin{equation}
    \braket{\psi}{\psi} = {b^l}_E^T E^L {b^r_E}
\end{equation}
where $b^l_E$ and $b^r_E$ are the left and right boundary vectors of the transfer matrix defined as
\begin{eqnarray}
    {b^l_E} &=& ({b^l_A}^\ast \otimes {b^l_A}) \nn \\
    {b^r_E} &=& ({b^r_A}^\ast \otimes {b^r_A}).
\end{eqnarray}
An MPS representation is said to be in a left (right) canonical form if the largest left (right) eigenvalue of the transfer matrix $E$ is unique, is equal to 1 (this can always be obtained by rescaling the $B$'s) and most importantly the corresponding left (right) eigenvector is the identity $\chi \times \chi$ matrix.\cite{perez2006matrix} Thus, for a right canonical MPS, 
\begin{equation}
    \sum_{\gamma, \epsilon}{E_{\alpha \beta, \gamma \epsilon}} \delta_{\gamma, \epsilon} = \delta_{\alpha\beta}
\end{equation}
where $\delta$ denotes the Kronecker delta function.
However, in general, an MPS cannot be in both a left and right canonical form simultaneously. 

Another useful construction with an MPS is the generalized transfer matrix $E_{\hat{O}}$
\begin{equation}
    E_{\hat{O}} = \sum_{n,m}{{A^{[m]}}^{\ast} \otimes O_{mn} A^{[n]}}.
\label{genTransfer}
\end{equation}
Here $\hat{O}$ is any single-site operator with matrix elements $\bra{m}\hat{O}\ket{n} = O_{mn}$. $E_{\hat{O}}$ is useful when computing the expectation value of an operator $\hat{O}$ acting on a site $i$, where
\begin{equation}
    \bra{\psi}\hat{O}_i\ket{\psi} = {b^l_E}^T E^{i-1} E_{\hat{O}} E^{L-i} b^r_E.
\label{onesiteOpexpect}
\end{equation}
Similarly, assuming $i < j$, the two point function associated with $\hat{O}$ reads
\begin{equation}
    \bra{\psi}\hat{O}_i \hat{O}_j\ket{\psi} = {b^l_E}^T E^{i-1} E_{\hat{O}} E^{j-i-1} E_{\hat{O}} E^{L - j}  b^r_E. 
\label{twositeOpexpect}
\end{equation}
Using Eqs.~(\ref{onesiteOpexpect}) and (\ref{twositeOpexpect}) for large $L$, the correlation length $\xi$ of the MPS defined using 
\begin{equation}
    \langle \hat{O}_i \hat{O}_j \rangle - \langle \hat{O}_i \rangle \langle \hat{O}_j \rangle \sim \exp(-\frac{|i-j|}{\xi})
\end{equation}
is given by 
\begin{equation}
    \xi = - \frac{1}{\log |\epsilon_2|}
\label{MPScorrelationlength}
\end{equation}
where $\epsilon_2$ is the second largest eigenvalue of the transfer matrix.\cite{orus2014practical} Note that $-1/\log |\epsilon_2|$ is an upper bound for $\xi$ that is saturated unless $\hat{O}$ has a special structure. Thus, if the spectrum of the transfer matrix is gapless, the state has an infinite correlation length. Note that a finite correlation length for an MPS in a canonical form guarantees that the wavefunction is normalized in the thermodynamic limit. 
\subsection{Entanglement spectrum and MPS}\label{sec:ESMPS}
The MPS representation of any wavefunction encodes the entanglement structure of the wavefunction. For any state $\ket{\psi}$ with a number $L$ of spin-$S$', a bipartition into two contiguous regions $\ma$ and $\mb$ with an $L_\ma$ number of spins in region $\ma$ and an $L_\mb$ number of spins in region $\mb$ ($L_\ma + L_\mb = L$) is defined as
\begin{equation}
    \ket{\psi} = \sum_{\alpha = 1}^\chi{\ket{\psi_\ma}_\alpha \otimes \ket{\psi_\mb}_\alpha}
\label{generalbipartition}
\end{equation}
where $\ket{\psi_\ma}_\alpha$ and $\ket{\psi_\mb}_\alpha$ are many-body states belonging to the physical Hilbert spaces of subsystems $\ma$ and $\mb$ respectively. Using the MPS representation of $\ket{\psi}$ Eq.~(\ref{generalOBCMPS}), if the region $\ma$ is defined as the set of sites $\{1,2,\dots,L_\ma\}$ and the region $\mb$ as $\{L_\ma + 1, L_\ma + 2, \dots, L\}$, the bipartition can be written using
\begin{eqnarray}
    \ket{\psi_\ma}_\alpha &=& \sum_{\{m_i\}, i\in \ma}{[{b^l_A}^T\prodal{l\in\ma}{}{A^{[m_l]}_{l}}]_\alpha\ket{\{m_i\}}} \nn \\
    \ket{\psi_\mb}_\alpha &=& \sum_{\{m_i\}, i\in \mb}{[\prodal{l\in\mb}{}{A^{[m_l]}_{l}}b^r_A]_\alpha\ket{\{m_i\}}}.
\end{eqnarray}
Note that $\{\ket{\psi_\ma}_\alpha\}$ and $\{\ket{\psi_\mb}_\alpha\}$ form complete but not necessarily orthonormal bases on the subsystems $\ma$ and $\mb$ respectively. The reduced density matrix with respect to such a bipartition is constructed as $\rho_\ma = {\rm Tr}_\mb \ket{\psi}\bra{\psi}$. The eigenvalue spectrum of $-\log \rho_\ma$ is the entanglement spectrum and $S \equiv -{\rm Tr}_\ma\,(\rho_\ma \log \rho_\ma)$ is the von Neumann entanglement entropy. An alternate way to obtain $\rho_\ma$ that is useful for MPS is through the definition of Gram matrices $\ml$ and $\mr$,
\begin{equation}
    \ml_{\alpha\beta} = {}_\alpha\braket{\psi_\ma}{\psi_\ma}_\beta,\;\; \mr_{\alpha\beta} = {}_\alpha\braket{\psi_\mb}{\psi_\mb}_\beta.
\end{equation}
Up to a overall normalization factor, the reduced density matrix can be expressed in terms of these Gram matrices as\cite{cirac2011entanglement}
\begin{equation}
    \rho_\ma = \sqrt{\ml}\mr^T\sqrt{\ml},
\label{actualrhored}
\end{equation}
where $\sqrt{\ml}$ is well-defined since Gram matrices are positive semi-definite. 
The Gram matrices $\ml$ and $\mr$ can be expressed in terms of the the MPS transfer matrix $E$ of Eq.~(\ref{MPSTransfer}) as
\begin{equation}
\ml = (E^T)^{L_\ma} b^l_E \;\;\; \mr = E^{L_\mb} b^r_E.
\label{MPSLRdefn}
\end{equation}
In Eq.~(\ref{MPSLRdefn}), $E$ is viewed as a $\chi \times \chi \times \chi \times \chi$ tensor, $b^l_E$ and $b^r_E$ as $\chi \times \chi$ matrices. Consequently, $\ml$ and $\mr$ are $\chi \times \chi$ matrices. Note that $\rho_\ma$ in Eq.~(\ref{actualrhored}) has the same spectrum as the matrix
\begin{equation}
    \rhored = \ml \mr^T.
\label{rhored}
\end{equation}
Since we are only interested in the spectrum of $\rho_\ma$ in this article, we refer to $\rhored$ to be the reduced density matrix of the system even though it is not guaranteed to be Hermitian. Assuming that the eigenvalue of unit magnitude of the transfer matrix is non-degenerate (i.e. $\log |\epsilon_2| \neq 0$), if $L_\ma$ and $L_\mb$ are large, ${(E^T)}^{L_\ma}$ and ${E}^{L_\mb}$ project onto $e_L$ and $e_R$, the left and right eigenvectors corresponding to the largest eigenvalue of $E$.
Thus, 
\begin{equation}
\ml = e_L (e_L^T b^l_E)\;\;\; \mr = e_R (e_R^T b^r_E).
\end{equation}
The density matrix thus reads, up to an overall constant (equal to $(e_L^T b^l_E) (e_R^T b^r_E)$),
\begin{equation}
    \rhored = e_L e_R^T.
\label{gsdensitymatrix}
\end{equation}

One should note that the construction of an MPS for a given state is not unique. Indeed, MPS matrices and boundary vectors redefined as 
\begin{eqnarray}
    \widetilde{A^{[m]}} = G A^{[m]} G^{\mm} \nn \\
    \widetilde{b^l_M} = {G^{\mm}}^T b^l_M \;\;\; \widetilde{b^r_M} = G b^r_M
\end{eqnarray}
represent the same wavefunction. 
When constructed in a canonical form, the bipartition Eq.~(\ref{generalbipartition}) is the same as a Schmidt decomposition\cite{perez2006matrix} of the state $\ket{\psi}$ with respect to subregions $\ma$ and $\mb$, defined as
\begin{equation}
    \ket{\psi} = \sum_{\alpha = 1}^{\chi_s}{\lambda_\alpha\ket{\psi^s_\ma}_\alpha\ket{\psi^s_\mb}_\alpha}
\end{equation}
where $\{\ket{\psi^s_\ma}_\alpha\}$ and $\{\ket{\psi^s_\mb}_\alpha\}$ are sets of orthonormal vectors on the subsystems $\ma$ and $\mb$ respectively and $\{\lambda_\alpha\}$ are referred to as the Schmidt values and $\chi_s$ is the number of non-zero Schmidt values (Schmidt rank). The bond dimension $\chi$ of the MPS constructed in the canonical form is the Schmidt rank $\chi_s$ of the wavefunction $\ket{\psi}$. Thus we refer to $\chi_s$ as the optimum bond dimension for an MPS representation of state $\ket{\psi}$. The entanglement entropy then satisfies
\begin{equation}
    \ms = -\sum_{\alpha=1}^{\chi_s}{\lambda_\alpha^2 \log \lambda_\alpha^2} \leq \log \chi_s
\label{entropybound}
\end{equation}
The entanglement entropy of an MPS about a given cut is thus upper-bounded by $\log\chi_s$. Since the Schmidt decomposition is the optimal bipartition of the system, $\chi \geq \chi_s$ and hence $\ms \leq \log\chi$.  

\section{MPS and the AKLT models}\label{sec:AKLTandMPS}
In this section, we provide a few examples of MPS based on the AKLT models.
\subsection{Ground state of the spin-1 AKLT model}
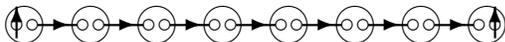
\begin{figure}[ht!]
\centering
\setlength{\unitlength}{1pt}
\begin{picture}(180,10)(0,0)
\linethickness{0.8pt}
\multiput(0,10)(25,0){8}{\circle{4}} 
\multiput(6,10)(25,0){8}{\circle{4}} 
\multiput(3,10)(25,0){8}{\circle{14}}
\thicklines
\multiput(8,10)(25,0){7}{\line(1,0){15}} 
\multiput(20,10)(25,0){7}{\vector(1,0){0}} 
\multiput(0,5)(181,0){2}{\vector(0,1){12}}
\end{picture}
\caption{Ground State of the spin-1 AKLT model with Open Boundary Conditions. Big and small circles represent physical spin-1 and spin-1/2 Schwinger bosons respectively. The lines repesent singlets between spin-1/2. The two edge spin-1/2's are free.}
\label{fig:groundstate}
\end{figure}
We first focus on the ground state of the spin-1 AKLT model with Open Boundary Conditions (OBC),\cite{aklt1987rigorous} one of the first examples of an MPS.\cite{klumper1993matrix} The state with $L$ spin-1's can be thought to be composed of two spin-1/2 Schwinger bosons, each in a singlet configuration with the spin-1/2 Schwinger boson of the left and right nearest neighbor spin-1s. Thus there are \emph{dangling} spin-1/2's on each edge of the chain. A cartoon picture of this state is shown in Fig.~\ref{fig:groundstate}. For a more detailed discussion of the model, we refer the reader to Ref.~[\onlinecite{self}].

The two spin-1/2 Schwinger bosons within a spin-1 (see Fig.~\ref{fig:groundstate}) form a virtual Hilbert space that corresponds to the auxillary space of the MPS. The normalized wavefunction can be written as a matrix product state with physical dimension $d = 3$ (the Hilbert space dimension of the physical spin-1) and a bond dimension $\chi = 2$ (the Hilbert space dimension of the spin-1/2 Schwinger boson).\cite{schollwock2011density} The derivation of the MPS representation for this state is shown in App.~\ref{sec:AKLTMPS}. The $d$ normalized $\chi \times \chi$ matrices for the AKLT ground state are (see Eq.~(\ref{spinSAKLTMPSform}))
\begin{eqnarray}
    &A^{[1]} = \sqrt{\frac{2}{3}}
    \begin{pmatrix}
        0 & 1 \\
        0 & 0
    \end{pmatrix} \;\;
    A^{[0]} = \frac{1}{\sqrt{3}}
    \begin{pmatrix}
        -1 & 0 \\
        0 & 1
    \end{pmatrix} \;\; \nonumber \\
    &A^{[\mm]} = \sqrt{\frac{2}{3}}
    \begin{pmatrix}
        0 & 0 \\
        -1 & 0
    \end{pmatrix}
\label{AKLTMPS}
\end{eqnarray}
corresponding to $S_z = 1, 0, \mm$ of the physical spin-1 respectively. 

Using the matrices of Eq.~(\ref{AKLTMPS}), the AKLT ground state transfer matrix can be computed to be 
 \begin{equation}
     E = 
     \begin{pmatrix}
        \frac{1}{3} & 0 & 0 & \frac{2}{3} \\
        0 & -\frac{1}{3} & 0 & 0 \\
        0 & 0 & -\frac{1}{3} & 0 \\
        \frac{2}{3} & 0 & 0 & \frac{1}{3}
    \end{pmatrix}
\label{AKLTTransfer}
 \end{equation}
where the left and right indices of the transfer matrix are grouped together. The eigenvalues of this transfer matrix are $(1, -\frac{1}{3}, -\frac{1}{3}, -\frac{1}{3})$. Since the largest eigenvalue is non-degenerate, using Eq.~(\ref{MPScorrelationlength}) the AKLT groundstate is a finitely-correlated state with correlation length $\xi = 1/\log(3)$.  The boundary vectors of Eq.~(\ref{generalOBCMPS}) for the AKLT ground state correspond to the free spin-1/2's on the left and right edges of an open spin-1 chain, shown in Fig.~\ref{fig:groundstate}. With both edge spins set to $S_z = +1/2$ the boundary vectors are (see Eq.~(\ref{spinSAKLTMPSform}))
 \begin{equation}
     b^l_A =
     \begin{pmatrix}
        1 \\
        0
    \end{pmatrix}\;\;
    b^r_A = 
    \begin{pmatrix}
        0 \\
        1
    \end{pmatrix}.\;\;
\label{AKLTboundaryvectors}
 \end{equation}
The Gram matrices $\ml$ and $\mr$ for the AKLT ground state are the left and right eigenvectors of $E$ corresponding to eigenvalue 1, $\ml = \mr = \frac{1}{\sqrt{2}}\mathds{1}_{2 \times 2}$. Using Eq.~(\ref{rhored}) the reduced density matrix is $\rhored = \frac{1}{2}\mathds{1}_{2 \times 2}$ and the entanglement entropy is $\ms = \log 2$, corresponding to a free spin-1/2 dangling spin.
\subsection{Ground state of the spin-$S$ AKLT model}
 \begin{figure}[ht!]
    \setlength{\unitlength}{1pt}
    \begin{picture}(180,10)(0,0)
    \linethickness{0.8pt}
    \multiput(0,10)(25,0){8}{\circle{4}} 
    \multiput(6,10)(25,0){8}{\circle{4}} 
    \multiput(0,4)(25,0){8}{\circle{4}} 
    \multiput(6,4)(25,0){8}{\circle{4}} 
    \multiput(3,7)(25,0){8}{\circle{16}}
    \thicklines
    \multiput(8,10)(25,0){7}{\line(1,0){15}} 
    \multiput(20,10)(25,0){7}{\vector(1,0){0}}
    \multiput(8,4)(25,0){7}{\line(1,0){15}} 
    \multiput(0,5)(181,0){2}{\vector(0,1){12}}
    \multiput(0,-1)(181,0){2}{\vector(0,1){12}}
    \multiput(20,4)(25,0){7}{\vector(1,0){0}}
\end{picture}
    \caption{Spin-2 AKLT model ground state with 2 singlets between nearest neighbors. The four edge spin-1/2s are free. Spin-$S$ AKLT has $S$ singlets.}
    \label{fig:spinSgroundstate}
\end{figure}
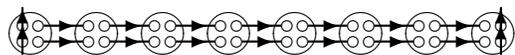
In a spin-$S$ chain, each of the physical spin-$S$ can be thought of as composed of $2S$ spin-1/2 Schwinger bosons, or equivalently, two spin-$(S/2)$ bosons.\cite{self} The ground state of the spin-$S$ AKLT model then has $S$ singlets between the $2S$ Schwinger bosons ($S$ on each site) on neighboring sites, as shown for $S = 2$ in Fig.~\ref{fig:spinSgroundstate}. It can also be interpreted as having a ``spin-$(S/2)$ singlet" between the spin-$(S/2)$'s of neighboring sites. Here, a spin-$(S/2)$ singlet is the state formed by two spin-$(S/2)$ with a total spin $J = 0$, $J_z = 0$. In the case of $S = 1$, this coincides with a usual spin-1/2 singlet. Consequently, with OBC, there are two free spin-$(S/2)$'s that set the boundary conditions of the wavefunction (see Fig.~\ref{fig:spinSgroundstate}).\cite{self} 

An MPS representation for the spin-$S$ AKLT ground state can be developed in close analogy to the spin-1 AKLT ground state (see App.~\ref{sec:AKLTMPS}). Here as well, the virtual Hilbert space of the spin-$S/2$ bosons corresponds to the auxiliary space. Thus, the MPS physical dimension is $d = 2 S + 1$ (because of spin-$S$ physical spins) and the bond dimension is $\chi = S + 1$ (because of the spin-$S/2$ virtual spins). Using Eq.~(\ref{spinSAKLTMPSform}), the $\chi \times \chi$ MPS matrices of the spin-$S$ AKLT ground states have the form
\begin{equation}
    A^{[m]}_{\alpha\beta} = \kappa_{m\alpha\beta} \delta_{\alpha-\beta,m}
\label{spinSAKLTMPS}
\end{equation}
where $\kappa_{m\alpha\beta}$ is a constant given in Eqs.~(\ref{spinSAKLTMPSform}) and (\ref{akltform}).

Analogous to Eq.~(\ref{AKLTboundaryvectors}), the boundary vectors of the MPS corresponding to boundary conditions with both the edge spin-$(S/2)$'s with $S_z = +S/2$ are  $\chi$-dimensional vectors with components
\begin{equation}
    ({b^l_A})_\alpha = \delta_{\alpha,1} \;\; ({b^r_A})_\alpha = \delta_{\alpha,\chi}.
\label{spinSboundary}
\end{equation}
Indeed one can verify that the spin-$S$ AKLT ground state of Eq.~(\ref{spinSAKLTMPS}) is finitely correlated, and the left and right eigenvectors corresponding to the largest eigenvalue 1 are both $\ml = \mr = \mathds{1}_{\chi \times \chi}$. Thus the reduced density matrix reads
\begin{equation}
    \rhored = \frac{1}{S+1}\mathds{1}_{(S+1)\times(S+1)}
\label{spinSAKLTrhored}
\end{equation}
and the entanglement entropy is $\ms = \log(S+1)$.

\subsection{Ferromagnetic states}
As discussed in detail in Ref.~[\onlinecite{self}], the ferromagnetic state is one of the highest excited states of all of the spin-$S$ AKLT models. Because of the $SU(2)$ symmetry of the AKLT models, these states appear in multiplets of $2S + 1$, of different $S_z$. In the highest weight state of the multiplet, all the physical spin-$S$ have $S_z = S$.\cite{self} Since this is a product wavefunction, an injective MPS has a bond dimension $\chi = 1$ and the matrices are scalars satisfying
\begin{equation}
    A^{[m]} = \delta_{m, S}.
\label{FMMPS}
\end{equation}
The boundary vectors are just 1 and this trivial MPS leads to a trivial transfer matrix, which is a scalar 1. Thus $\rho = 1$ and $\ms = 0$.

\section{Matrix Product Operators}\label{sec:MPO}
In this section, we briefly review Matrix Product Operators (MPO) and provide some examples relevant to the AKLT models. A comprehensive discussion of MPOs can be found in existing literature.\cite{pirvu2010matrix,schollwock2011density,mcculloch2007density,verstraete2008matrix,verstraete2004matrix}
\subsection{Definition and properties}\label{sec:MPOdefn}
Since the exact excited states derived in Ref.~[\onlinecite{self}] are expressed in terms of operators on the ground state or the highest excited state, it is crucial to understand how to apply these operators on an MPS. An MPO representation of an operator $\mathcal{O}$ is defined as
\begin{eqnarray}
    &\hat{\mathcal{O}} = \sum\limits_{\{s_n\}, \{t_n\}}{[{b_M^l}^T M_1^{[s_1 t_1]} M_2^{[s_2 t_2]} \dots M_L^{[s_L t_L]} b_M^r]}\nn \\
    &\ket{\{s_n\}}\bra{\{t_n\}}.
\label{generalOBCMPO}
\end{eqnarray}
In Eq.~(\ref{generalOBCMPO}), the operator $\hat{\mathcal{O}}$ is written in terms of $L$ $\chi_m \times \chi_m$ matrices with elements expressed as $d \times d$ matrices acting on the physical indices. $\chi_m$ is referred to as the bond dimension of the MPO and the corresponding vector space is the auxiliary space. $\hat{\mathcal{O}}$ can compactly be represented as a $\chi_m \times \chi_m \times d \times d$ tensor $M_i$ with two physical indices ($\{[s_i], [t_i]\}$ and two auxiliary indices. $b_M^l$ and $b_M^r$ are the boundary vectors of the MPO in the operator auxiliary space. 

Similar to an MPS, the construction of an MPO for a given operator is not unique. We now describe a method to construct an MPO for an operator $\hat{\mathcal{O}}$. The particular MPO construction we describe here relies on a generalized version of a Finite State Automation (FSA).\cite{crosswhite2008finite,motruk2016density, schollwock2011density} An FSA is a system with a finite set of ``states" and a set of rules for transition between the states at each iteration. In such a setup, each state maps to a unique state after an iteration. When the states of the FSA are viewed as basis elements of a vector space, each state is denoted as a vector and the transition between the states is described by a square matrix. For example, we consider an FSA with two states $\ket{R}$ and $\ket{F}$, that are denoted as
\begin{equation}
\ket{R} =
\begin{pmatrix}
1 \\
0
\end{pmatrix}\;\;\;
\ket{F} = 
\begin{pmatrix}
0 \\
1
\end{pmatrix}.
\label{fsavector}
\end{equation}
If at each iteration, $\ket{R}$ and $\ket{F}$ are interchanged, the transition matrix $T$ is
\begin{equation}
T = 
\begin{pmatrix}
0 & 1 \\
1 & 0
\end{pmatrix}.
\end{equation}
In principle, these transition matrices could vary from an iteration to the next. 

To exemplify the construction of an MPO, we start with a simple example:
\begin{equation}
    \hat{\mathcal{O}} = \sum_{j=1}^L{e^{i k j} \hat{C_j}}
\end{equation}
where $e^{i k j}\hat{C_j}$ can be written in the physical Hilbert space as
\begin{equation}
    e^{i k j} \hat{C_j} \equiv \underbrace{e^{ik}\mathds{1} \otimes \dots \otimes e^{i k}\mathds{1}}_{j - 1\, {\rm times}} \otimes \,e^{i k}\hat{C} \otimes \underbrace{\mathds{1} \otimes \dots \otimes \mathds{1}}_{L - j\, {\rm times}},
\label{eikjCj}
\end{equation}
such that the index $j$ does not explicitly appear in any of the operators. 
Consider an FSA that iterates $L$ times and constructs the operator $\hat{\mathcal{O}}$ by appending a physical operator (either $\mathds{1}$ or $\hat{C}$) at each iteration to a string of operators. If $\ket{S_n}$ is the state of the FSA at the $n$-th iteration, the appended physical operator is the matrix element $\bra{S_n}T_n\ket{S_{n+1}}$ where $T_n$ is the transition matrix at the $n$-th iteration.
For example, an FSA that constructs $e^{i k j}C_j$ of Eq.~(\ref{eikjCj}) starts in a state $\ket{R}$. It remains in the state $\ket{R}$ for $j - 1$ iterations with a transition matrix
\begin{equation}
    T_R =
    \begin{pmatrix}
    e^{i k}\mathds{1} & 0 \\
    0 & 0
    \end{pmatrix}
\end{equation}
appending an $\mathds{1}$ at each step. At the $j$-th iteration, the FSA transitions to $\ket{F}$ (different from $\ket{R}$) with a transition matrix $T_j$
\begin{equation}
    T_j =
    \begin{pmatrix}
    0 & e^{ik} \hat{C} \\
    0 & 0
    \end{pmatrix},
\end{equation}
thus appending the operator $\hat{C}$ on site $j$ and remains in $\ket{F}$ in the rest of $L - j$ iterations with transition matrix
\begin{equation}
    T_F =
    \begin{pmatrix}
    0 & 0 \\
    0 & \mathds{1}
    \end{pmatrix}.
\end{equation}
$\hat{\mathcal{O}}$ is then the sum of operators obtained using an FSA for all $j$. The sum over operators can be efficiently represented by generalizing an FSA to allow for superpositions of FSA states with operators as coefficients.
For example, we allow for FSA states such as $e^{ik}\mathds{1}\ket{R} + e^{ik}\hat{C}\ket{F}$.
The transition matrix in such a generalized FSA is an arbitrary square matrix with operators as matrix elements.
Indeed, fixing the initial and final states of the FSA to be $\ket{R}$ and $\ket{F}$, we can construct the operator $\mathcal{O}$ with a transition matrix $M_j$ on site $j$ with elements:
\begin{equation}
    M_j =
    \begin{pmatrix}
    e^{ik}\mathds{1} & e^{ik}\hat{C} \\
    0 & \mathds{1}
    \end{pmatrix}.
\label{mpoeasyexample}
\end{equation}
Writing the entire process of the generalized FSA, $\bra{F}\prod_{j = 1}^L{M_j}\ket{R}$, we obtain exactly the representation of $\hat{\mathcal{O}}$ as an MPO of the form Eq.~(\ref{generalOBCMPO}), where the auxiliary space is the vector space spanned by states of the generalized FSA. Note that since $M_j$ does not depend on the site index $j$, we can omit this index. The left and right boundary vectors $b^l_M$ and $b^r_M$ are the vector representations of the FSA states $\ket{R}$ and $\ket{F}$ respectively (Eq.~(\ref{fsavector})),
\begin{equation}
    b^l_M =
    \begin{pmatrix}
        1 \\
        0
    \end{pmatrix}\;\;\;
    b^r_M =
    \begin{pmatrix}
        0 \\
        1
    \end{pmatrix}.
\end{equation}

The MPO representations of more general operators can be computed similarly with the introduction of intermediate states of the generalized FSA.
For example, in the construction of the MPO for the operator
\begin{equation}
    \hat{\mathcal{O}} = \sum_j{e^{i k j} \left(\hat{W}_j \hat{X}_{j+1}\right)},
\label{MPOexample}
\end{equation}
one introduces an intermediate state $\ket{I_1}$ of the generalized FSA, such that the transition matrix elements at any step read $\bra{R} T \ket{I_1} = e^{ik} \hat{W}$ and $\bra{I_1} T \ket{F} = \hat{X}$. 
The MPO for $\hat{\mathcal{O}}$ in the auxiliary dimension thus reads
\begin{equation}
    M = 
    \begin{pmatrix}
        e^{ik}\mathds{1} & e^{ik}\hat{W}  & 0 \\
        0 & 0 & \hat{X} \\
        0 & 0 & \mathds{1}
    \end{pmatrix}.
\label{MPOstructure}
\end{equation}
The bond dimension of the MPO $\chi_m$ is the number of states of the generalized FSA generating it. Since the initial state of the FSA is $\ket{R}$ and the final state is $\ket{F}$, the components of the left and right boundary vectors of an MPO are always
\begin{equation}
    ({b^l_M})_\alpha = \delta_{\alpha, 1}\;\;\;
    ({b^r_M})_\alpha = \delta_{\alpha, \chi_m}
\label{MPOBoundary}
\end{equation}

Since the flow of an FSA is uni-directional, the MPO is always an upper triangular matrix in the auxiliary indices.
For a translation invariant MPO, any element on the MPO diagonal appears in the operator as multiple direct products of the same operator. For example, the MPO
\begin{equation}
    M_{\hat{O}} =
    \begin{pmatrix}
        \hat{W} & \hat{C} \\
        0 & \hat{X}
    \end{pmatrix}
\label{mpolongrangeexample}
\end{equation}
represents an operator $\hat{O}$ defined on a lattice of length $L$ that reads 
\begin{equation}
    \hat{O} = \left(\prod_{i = 1}^{L-1}{\hat{W}_i}\right) \hat{C}_L + \hat{C}_1 \left(\prod_{i = 2}^{L}{\hat{X}_i}\right) + \dots,
\end{equation}
which is not a strict local operator unless $\hat{W}$ and $\hat{X}$ are proportional to $\mathds{1}$. 
Thus, for an operator that is a sum of strictly local terms, the only diagonal element that can appear in the MPO is $\mathds{1}$, up to an overall constant (such as $e^{ik}$).
Moreover, if the diagonal element in an MPO corresponding to an intermediate state is $\mathds{1}$, the operator $\hat{\mathcal{O}}$  includes a non-local term, i.e. a long range coupling between sites.
For example, for the MPO 
\begin{equation}
    M_{\hat{O}} =
    \begin{pmatrix}
        \mathds{1} & \hat{W} & 0 \\
        0 & \mathds{1} & \hat{X} \\
        0 & 0 & \mathds{1}
    \end{pmatrix},
\label{mpolongrangeexample2}
\end{equation}
the operator $\hat{O}$ reads
\begin{equation}
    \hat{O} = \sum_{i = 1}^{L-1}{\sum_{j = i + 1}^{L}{\hat{W}_i \hat{X}_j}}.
\end{equation}
Thus, for operators that are the sum of non-trivial operators with a finite support, the only non-vanishing diagonal elements correspond to the auxiliary states $\ket{R}$ and $\ket{F}$.
\subsection{The AKLT model and MPOs}
We now introduce the MPOs for some of the operators required to build exact excited states of the AKLT model. These will be useful for the study of the entanglement of these excited states, introduced in Refs.~[\onlinecite{arovas1989two}] and [\onlinecite{self}]. Whereas the Arovas A and Arovas B states discussed therein were for exact eigenstates only for periodic boundary conditions, here we assume open boundary conditions. The motivation for this assumption is twofold. First, analytic calculations using MPS and MPOs are greatly simplified with open boundaries. Second, we are interested in the thermodynamic limit or large systems where the properties of the system are essentially independent of boundary conditions.

We start with the spin-1 AKLT model.
The Arovas A state was introduced in Ref.~[\onlinecite{arovas1989two}]. The closed-form expression for the state, up to an overall normalization factor, reads 
\begin{equation}
    \ket{A} = \left[\sum_{j = 1}^{L-1}{(-1)^j \vec{S}_{j}\cdot \vec{S}_{j+1}}\right]\ket{G}
\label{ArovasA}
\end{equation}
where $\ket{G}$ is the ground state of the spin-1 AKLT model and we have assumed open boundary conditions. The operator that appears in the Arovas A state can be written as
\begin{eqnarray}
    \hat{\mathcal{O}}_A &=& \sum_{j}{(-1)^j \vec{S}_j\cdot\vec{S}_{j+1}} \nn \\
    &=& \sum_j{(-1)^j \left(\frac{S^+_j S^-_{j+1} + S^-_j S^+_{j+1}}{2} + S^z_j S^z_{j+1}\right)}.
\label{firstarovas}
\end{eqnarray}
By analogy to the MPO of Eq.~(\ref{MPOstructure}) corresponding to the operator Eq.~(\ref{MPOexample}), the MPO for $\hat{\mathcal{O}}_A$ (in the case of open boundary conditions) reads (also see Eq.~(\ref{ArovasAmpoapp}))
\begin{equation}
    M_A = 
    \begin{pmatrix}
    -\mathds{1} & -\frac{S^+}{\sqrt{2}} & -\frac{S^-}{\sqrt{2}} & -S^z & 0 \\
    0 & 0 & 0 & 0 & \frac{S^-}{\sqrt{2}} \\
    0 & 0 & 0 & 0 & \frac{S^+}{\sqrt{2}} \\
    0 & 0 & 0 & 0 & S^z \\
    0 & 0 & 0 & 0 & \mathds{1}
    \end{pmatrix},
\label{firstarovasmpo}
\end{equation}
where the negative signs appear due to the $(-1)^j$ in Eq.~(\ref{firstarovas}).
Similarly, the Arovas B state, introduced in Ref.~[\onlinecite{arovas1989two}] is another exact excited state of the AKLT model.\cite{self} As mentioned in Ref.~[\onlinecite{self}], its closed-form expression, up to an overall normalization factor, can be written as
\begin{equation}
    \ket{B} = \hat{\mathcal{O}}_B\ket{G}
\label{ArovasB}
\end{equation}
with
\begin{eqnarray}
    \hat{\mathcal{O}}_B &=& \sum_{j = 2}^{L-1}{(-1)^j \{\vec{S}_{j-1}\cdot\vec{S}_j, \vec{S}_j\cdot\vec{S}_{j+1}\}} \nn \\
\label{ArovasBop}
\end{eqnarray}
where we have assumed open boundary conditions.
As shown in Eq.~(\ref{ArovasBmpoapp}) in App.~\ref{sec:ArovasMPO}, the MPO for $\hat{\mathcal{O}}_B$ can be compactly expressed as
\begin{equation}
    M_B = 
    \begin{pmatrix}
        -\mathds{1} & -\overline{\boldsymbol{S}} & 0 & 0 \\
        0 & 0 & \boldsymbol{T} & 0 \\
        0 & 0 & 0 & \boldsymbol{S} \\
        0 & 0 & 0 & \mathds{1}
    \end{pmatrix}
\label{secondarovasmpo}
\end{equation}
where
\begin{eqnarray}
&\overline{\boldsymbol{S}} = 
\begin{pmatrix}
\frac{S^+}{\sqrt{2}} & \frac{S^-}{\sqrt{2}} & S^z
\end{pmatrix} \nn \\
&\boldsymbol{S} = 
\begin{pmatrix}
\frac{S^-}{\sqrt{2}} & \frac{S^+}{\sqrt{2}} & S^z
\end{pmatrix}^T \nn \\
&\boldsymbol{T} = 
\begin{pmatrix}
    \frac{\{S^-, S^+\}}{2} & (S^-)^2 & \frac{\{S^-,  S^z\}}{\sqrt{2}} \\
    (S^+)^2 & \frac{\{S^+, S^-\}}{2} & \frac{\{S^+, S^z\}}{\sqrt{2}} \\
    \frac{\{S^z, S^+\}}{\sqrt{2}} & \frac{\{S^z, S^-\}}{\sqrt{2}} & 2S^z S^z
\end{pmatrix}.
\label{SbarST}
\end{eqnarray}
The bond dimension of the MPO $M_B$ is thus $\chi_m = 8$.

Another set of excited states for spin-$S$ AKLT models was obtained in Ref.~[\onlinecite{self}], i.e. the spin-$2S$ magnons. The closed-form expression for the spin-$2S$ magnon state in the spin-$S$ AKLT models, up to an overall normalization factor, reads
\begin{equation}
    \ket{SS_2} = \sum_{j = 1}^{L}{(-1)^j (S^+_j)^{2S}}\ket{SG},
\label{Spin2S}
\end{equation}
where $\ket{SG}$ is the ground state of the spin-$S$ AKLT model. Unlike the two previous states, $\ket{SS_2}$ is an exact excited state irrespective of the boundary conditions.\cite{self} 
The spin-$2S$ magnon creation operator thus reads
\begin{equation}
    \hat{\mathcal{O}}_{SS_2} = \sum_j{(-1)^j (S_j^+)^{2S}}.
\label{spin2Smagnon}
\end{equation}
Since $\hat{\mathcal{O}}_{SS_2}$ is a sum of single-site operators, by analogy to Eqs.~(\ref{eikjCj}) and (\ref{mpoeasyexample}), its MPO has $\chi_m = 2$ and reads
\begin{equation}
    M_{SS_2} = 
    \begin{pmatrix}
    -\mathds{1} & -({S^+})^{2S} \\
    0 & \mathds{1}
    \end{pmatrix}
\label{spin2SMPO}
\end{equation}

Following the spin-$2S$ magnon in Eq.~(\ref{spin2Smagnon}), a tower of states from the ground state to a highest excited state was introduced for spin-$S$ AKLT models in Ref.~[\onlinecite{self}]. The states in the tower are comprised of multiple spin-$2S$ magnons, and are all exact excited states for open and periodic boundary conditions. The closed-form expression for the $N$-th state of the tower of states for the spin-$S$ AKLT model reads
\begin{eqnarray}
    \ket{SS_{2N}} &=& (\hat{\mathcal{O}}_{SS_2})^N \ket{SG}.
\label{towerofstates}
\end{eqnarray}
When written naively, the MPO for the operator $(\hat{\mathcal{O}}_{SS_2})^N$ has a bond dimension $2^N$, since it is a direct product of $N$ copies of the MPO $M_{SS_2}$ on the auxiliary space. However, a more efficient MPO can be constructed for $(\hat{\mathcal{O}}_{SS_2})^N$.
For example, consider $N = 2$.
$(\hat{\mathcal{O}}_{SS_2})^2$ can be written as (up to an overall factor)
\begin{equation}
    (\hat{\mathcal{O}}_{SS_2})^2 = \sum_{i \leq j}{(-1)^{i + j} (S_i^+)^{2S} (S_j^+)^{2S}}.
\label{SS2sq}
\end{equation}
Since $(S_j^+)^{4S} = 0$, Eq.~(\ref{SS2sq}) can be written as
\begin{equation}
    (\hat{\mathcal{O}}_{SS_2})^2 = \sum_{i}{(-1)^{i} (S_i^+)^{2S}\sum_{i < j}{(-1)^j(S_j^+)^{2S}}}.
\label{SS2sqcopy}
\end{equation}
From Eq.~(\ref{SS2sqcopy}) it is evident that the MPO $M_{SS_{4}}$ for $(\hat{\mathcal{O}}_{SS_2})^2$ can be viewed as two copies of the generalized FSA generating $M_{SS_2}$, where the final state of the first generalized FSA is the initial state for the second generalized FSA. 
The MPO thus reads 
\begin{equation}
    M_{SS_{4}} = 
    \begin{pmatrix}
        -\mathds{1} & -(S^+)^{2S} & 0 \\
        0 & \mathds{1} & (S^+)^{2S} \\
        0 & 0 & -\mathds{1} \\
    \end{pmatrix}
\end{equation}
The appearance of three $\pm \mathds{1}$ on the diagonal of $M_{SS_4}$ reflects the non-locality of the operator $(\hat{\mathcal{O}}_{SS_2})^2$.
The same strategy can be applied to construct the MPO $M_{SS_{2N}}$ corresponding to the operator $(\hat{\mathcal{O}}_{SS_2})^N$.
For general $N$, the MPO reads
\begin{equation}
        M_{SS_{2N}} = 
    \begin{pmatrix}
        -\mathds{1} & -(S^+)^{2S} & 0 & \dots & 0\\
        0 & \mathds{1} & (S^+)^{2S} & \ddots & \vdots \\
        \vdots & \ddots & \ddots & \ddots & 0 \\
        \vdots & \ddots & \ddots & (-1)^N \mathds{1} & (-1)^N (S^+)^{2S} \\
        0 & \dots & \dots & 0 & (-1)^{N+1}\mathds{1} \\
    \end{pmatrix}.
\label{TowerofstatesMPO}
\end{equation}
The bond dimension of the MPO $M_{SS_{2N}}$ is thus $\chi_m = N + 1$.

\section{MPO $\boldsymbol{\times}$ MPS}\label{sec:MPOMPS}
The exact states that we are interested in are obtained by acting local operators on the ground states.\cite{self} In this section we study some of the properties of an MPS formed by acting an MPO (operator) on an MPS with a finite correlation length (ground state). Similar approaches (\eg tangent space methods) have been used to study low energy excitations of gapped Hamiltonians.\cite{haegeman2013post, haegeman2013elementary, vanderstraeten2015excitations, zauner2015transfer, haegeman2012variational, pirvu2012matrix} 
\subsection{Definition and properties}\label{sec:MPOMPSproperties}
A state defined by the action of an MPO on an MPS (we assume both to be site-independent) has a natural MPS description, 
\begin{equation}
    B^{[m]} = \sum_n{M^{[mn]} \otimes A^{[n]}}.
\label{generalMPOMPS}
\end{equation}
where the tensor product $\otimes$ acts on the ancilla. We refer to $B$ as an MPO$\times$MPS to distinguish it from the MPS $A$, which we assume to have a finite correlation length. $B$ has a bond dimension of 
\begin{equation}
\wc = \chi_m\chi,
\end{equation}
where $\chi_m$ and $\chi$ are the bond dimensions of the MPO and MPS respectively. Note that $\wc$ need not be the optimum bond dimension of $B$ (i.e. Schmidt rank of the state $B$ represents), though it is typically the case when $M$ and $A$ have optimum bond dimensions.
%
%
The transfer matrix of $B$ reads
\begin{eqnarray}
  F &=& \sum_m{{B^{[m]}}^\ast \otimes B^{[m]}} \nn \\
  &=& \sum_{m,n,l}{{A^{[m]}}^\ast \otimes {M^{[nm]}}^\ast \otimes M^{[nl]} \otimes A^{[l]}}
\label{MPOMPStransfer}
\end{eqnarray}
where $\otimes$ acts on the ancilla. $F$ is thus a $\wc \times \wc \times \wc \times \wc$ tensor that can also be viewed as $\wc^2 \times \wc^2$ matrix by grouping both the left and right ancilla. $F$ can also be written as
\begin{equation}
    F = \sum_{m,l}{{A^{[m]}}^\ast \otimes \mathcal{M}^{[ml]} \otimes A^{[l]}}.
\end{equation}
where 
\begin{eqnarray}
    \mathcal{M}^{[ml]} &\equiv& \sum_n{{M^{[nm]}}^\ast \otimes M^{[nl]}} \nn \\
    &=& \sum_n{{M^\dagger}^{[mn]} \otimes M^{[nl]}},
\end{eqnarray}
where $\dagger$ acts on the physical indices on the MPO.
From Eqs.~(\ref{generalMPOMPS}) and (\ref{MPOMPStransfer}), the boundary vectors of an MPO$\times$MPS and its transfer matrix are given by
\begin{eqnarray}
    &b^l_B = b^l_M \otimes b^l_A \;\;\; b^r_B = b^r_M \otimes b^r_A \nn \\
    &{b^l_F} = ({b^l}_B^\ast \otimes {b^l}_B) \;\;\;  {b^r_F} = ({b^r}_B^\ast \otimes {b^r}_B).
\label{MPOMPSboundary}
\end{eqnarray}
Since $M$ is always upper triangular in the auxiliary indices (as discussed in Sec.~\ref{sec:MPO}), $\mathcal{M}$ is a $\chi_m^2 \times \chi_m^2$ matrix with a nested upper triangular structure in the ancilla, with elements as $d \times d$ matrices, where $d$ is the physical Hilbert space dimension. For example, if we consider the MPO of Eq.~(\ref{mpoeasyexample}), $\mathcal{M}$ reads
\begin{equation}
    \mathcal{M} =
    \begin{pmatrix}
        \mathds{1} & \hat{C} & \hat{C}^\dagger & \hat{C}^\dagger \hat{C} \\
        0 & e^{-ik}\mathds{1} & 0 & e^{-ik}\hat{C}^\dagger \\
        0 & 0 & e^{ik}\mathds{1} & e^{ik}\hat{C} \\
        0 & 0 & 0 & \mathds{1}
    \end{pmatrix}.
\label{MMexample}
\end{equation}
In Eq.~(\ref{MPOMPStransfer}), the matrix elements of $F$ can also be viewed as a $\chi_m^2 \times \chi_m^2$ matrix with matrix elements
\begin{equation}
    F_{\mu\nu} = \sum_{m,l}{{A^{[m]}}^\ast \otimes \mathcal{M}^{[ml]}_{\mu\nu} A^{[l]}}.
\label{MPOMPStransfermatrixele}
\end{equation}
$F_{\mu\nu}$ is indeed the generalized transfer matrix (see Eq.~(\ref{genTransfer}) in Sec.~\ref{sec:MPSproperties}) of the operator $\mathcal{M}_{\mu\nu}$. Thus, $F$ is also a nested upper triangular matrix with elements $\chi^2 \times \chi^2$ generalized transfer matrices of the elements of $\mathcal{M}$ with the original MPS $A$. For $\mathcal{M}$ of Eq.~(\ref{MMexample}), we obtain 
\begin{equation}
    F =
    \begin{pmatrix}
        E & E_{\hat{C}} & E_{\hat{C}^\dagger} & E_{\hat{C}^\dagger \hat{C}} \\
        0 & e^{-ik} E & 0 & e^{-ik} E_{\hat{C}^\dagger} \\
        0 & 0 & e^{ik} E & e^{ik} E_{\hat{C}} \\
        0 & 0 & 0 & E
    \end{pmatrix}
\label{genTransferstructure}
\end{equation}
where $E$ is the transfer matrix of the MPS $A$ and $E_{\hat{C}}$, $E_{\hat{C}^\dagger}$ and $E_{\hat{C}^\dagger \hat{C}}$ are the generalized transfer matrices (defined in Eq.~(\ref{genTransfer})) of operators $\hat{C}$, $\hat{C}^\dagger$ and $\hat{C}^\dagger \hat{C}$ respectively. Furthermore, since the MPO boundary conditions are always of the form of Eq.~(\ref{MPOBoundary}), using Eq.~(\ref{MPOMPSboundary}) the boundary vectors for the transfer matrix $F$ read
\begin{equation}
    b^r_F 
    =
    \begin{pmatrix}
        0 \\
        0 \\
        0 \\
        b^r_E
    \end{pmatrix}\;\;\;
    b^l_F
    =
    \begin{pmatrix}
        b^l_E \\
        0 \\
        0 \\
        0
    \end{pmatrix}.
\label{Transferboundary}
\end{equation}

As illustrated in the previous section using Eqs.~(\ref{mpolongrangeexample}) and (\ref{mpolongrangeexample2}), non-vanishing diagonal elements of the MPO can only be of the form $e^{i\theta}\mathds{1}$. Consequently, the diagonal elements of $F$ are always of the form $e^{i\theta} E$, as can be observed in the example in Eq.~(\ref{genTransferstructure}). The generalized eigenvalues and structure of the Jordan normal form of block upper triangular matrices such as $F$ is discussed in App.~\ref{sec:Triangular}. As evident from Eqs.~(\ref{eigvalsblocktriangular}) and (\ref{uppertriangleexample}), the block upper triangular structure of $F$ dictates that its generalized eigenvalues are those of $e^{i\theta} E$ blocks on the diagonal. The eigenvalue of unit magnitude of the transfer matrix $F$ is thus not unique in general, and an MPO $\times$ MPS typically does not have exponentially decaying correlations even if the MPS has.
Moreover, the transfer matrix $F$ need not be diagonalizable. In general, it would have a Jordan normal form consisting of Jordan blocks corresponding to various degenerate generalized eigenvalues. The Jordan decomposition of  $F$ reads %
\begin{equation}
    F = P J P^{\text{-}1}
\label{En}
\end{equation}
where $J$ is the Jordan normal form of $F$, the columns of $P$ are the right generalized eigenvectors of $F$ and the rows of $P^{\text{-}1}$ are the left generalized eigenvectors of $F$ (same as right generalized eigenvectors of $F^T$). $J$ is composed of several Jordan blocks of various sizes, and has the form 
\begin{equation}
    J = \bigoplus_{i \in \Lambda}{J_i}
\end{equation}
where $\Lambda$ is a set of indices that label the Jordan blocks, $J_i$ is a Jordan block of size $|J_i|$ of an eigenvalue $\lambda_i$ and $\sum_{i \in \Lambda}{|J_i|} = \wc^2$. That is, up to a shuffling of rows and columns, 
\begin{equation}
    J_i = 
    \begin{pmatrix}
    \lambda_i & 1 & 0 & \dots  & \dots & 0\\
    0 & \lambda_i & 1 & \ddots & \ddots & \vdots \\
    \vdots & \ddots & \ddots & \ddots & \ddots & \vdots \\
    \vdots & \ddots & \ddots & \ddots & \ddots & 0\\
    \vdots & \ddots & \ddots & \ddots & \lambda_i & 1 \\
    0 & \dots & \dots & \dots & 0 & \lambda_i \\
    \end{pmatrix}_{|J_i| \times |J_i|}
\label{jbexample}
\end{equation}
For a diagonalizable matrix, $|J_i| = 1$ for all $i \in \Lambda$. 
\subsection{Entanglement spectra of MPO $\times$ MPS}\label{sec:ESMPOMPS}
In this section, we outline the computation of the entanglement spectrum for an MPO $\times$ MPS state, i.e., for an MPS with a non-diagonalizable transfer matrix. 
Since the MPO $\times$ MPS is also an MPS, Eqs.~(\ref{generalbipartition}) to (\ref{rhored}) of Sec.~\ref{sec:ESMPS} are valid here as well.
Analogous to Eq.~(\ref{MPSLRdefn}), here we obtain
\begin{equation}
\ml = (F^T)^{L_\ma} b^l_F \;\;\; \mr = F^{L_\mb} b^r_F.
\label{MPSLRdefnmpomps}
\end{equation}
In the following, we will mostly be interested in the limit $n \equiv L_\ma = L_\mb \rightarrow \infty$, i.e. the thermodynamic limit with an equal bipartition.  
Since $F^n = P J^n P^{\mm}$, $J^n = \bigoplus_{i \in \Lambda}{J_i^n}$, and 
\begin{equation}
    J_i^n = 
    \begin{pmatrix}
    \lambda_i^n & \binom{n}{1} \lambda_i^{n-1} & \binom{n}{2} \lambda_i^{n-2} & \dots & \binom{n}{|J_i|-1}\lambda_i^{n-|J_i|+1}\\
    0 & \lambda_i^n & \binom{n}{1} \lambda_i^{n-1} & \ddots & \vdots \\
    \vdots & \ddots & \ddots & \ddots & \binom{n}{2} \lambda_i^{n-2}\\
    \vdots & \ddots & \ddots & \lambda_i^n & \binom{n}{1} \lambda_i^{n-1} \\
    0 & \dots & \dots & 0 & \lambda_i^n \\
    \end{pmatrix}_{|J_i| \times |J_i|},
\label{Jn}
\end{equation}
all the Jordan blocks $J_i$ corresponding to $|\lambda_i| < 1$, vanish in the thermodynamic ($n \rightarrow \infty$) limit.
We can thus truncate $J$ to a subspace with generalized eigenvalues of magnitude $1$, by including a projector $Q$ onto that subspace. This subspace could involve several Jordan blocks, each of possibly different dimension. We define
\begin{eqnarray}
    \wJ &=& Q J Q \nn \\
      &=& \bigoplus_{i \in \Lambda_{\textrm{unit}}}{J_i}.
\label{Jordantrunc}
\end{eqnarray}    
where $\Lambda_{\textrm{unit}}$ is a set defined such that $|\lambda_i| = 1$ for $i \in \Lambda_{\textrm{unit}}$, and the dimension of $\wJ$ is $|\wJ|$, where
\begin{equation}
    |\wJ| = \sum_{i \in \Lambda_{\textrm{unit}}}{|J_i|}.
\end{equation}
Since we are interested in the limit $n \rightarrow \infty$, instead of $F$, we use a truncated transfer matrix $\wF$ defined as 
\begin{equation}
    \wF \equiv P \wJ P^{\mm},
\end{equation}
such that
\begin{equation}
    \wF^n = F^n\;\;\; \textrm{as } n \rightarrow \infty.
\end{equation}
Since $Q^2 = Q$, using Eq.~(\ref{Jordantrunc}), the expression for $\wF$ can be written as
\begin{eqnarray}
    \wF &=& P Q (Q J Q) Q {P^{\mm}} \nn \\
    &\equiv& V_R \wJ V_L^T,
\label{Eexp}
\end{eqnarray}
where we have used Eq.~(\ref{Jordantrunc}) and have defined
\begin{eqnarray}
    V_R &\equiv& P Q \nn \\
    V_L^T &\equiv& Q P^{\mm}.
\label{vrvlproj}
\end{eqnarray}
Since $V_R$ consists of the columns of $P$ (right generalized eigenvectors of $F$) corresponding to the generalized eigenvalues in $J$ and $V_L^T$ consists of the rows of $P^{\mm}$ (left generalized eigenvectors of $F$), $V_R$ and $V_L$ have the forms
\begin{eqnarray}
    V_R &=&
    \begin{pmatrix}
    r_1 & r_2 & \dots & r_{|\wJ|}
    \end{pmatrix} \nonumber \\
    V_L &=&
    \begin{pmatrix}
    l_1 & l_2 & \dots & l_{|\wJ|}
    \end{pmatrix},
\label{vlvrform}
\end{eqnarray}
where $\{r_i\}$ (resp. $\{l_i\}$) are the $\wc^2$-dimensional right (resp. left) generalized eigenvectors of $F$ corresponding the generalized eigenvalues of magnitude $1$. 
Using Eqs.~(\ref{Eexp}) and (\ref{vrvlproj}), the truncated Gram matrices read
\begin{eqnarray}
    \wmr &=& V_R (\wJ)^n V_L^T b^r_F \nn \\
    \wml &=& V_L (\wJ^{T})^n V_R^T b^l_F.
\label{lrexp}
\end{eqnarray}
We split Eq.~(\ref{lrexp}) into two parts. We first define the $|\wJ|$-dimensional ``modified" boundary vectors that are independent of $n$ as
\begin{eqnarray}
    &&\wbr \equiv V_L^T b^r_F \nn \\
    &&\wbl \equiv V_R^T b^l_F.
\label{tildeboundary}
\end{eqnarray}
The $n$-dependent parts of $\wml$ and $\wmr$ are then encoded in the $(\wc)^2 \times |\wJ|$ dimensional matrices 
\begin{eqnarray}
\wvr &\equiv& V_R (\wJ)^n \nn \\
\wvl &\equiv& V_L (\wJ^{T})^n.
\label{tildedefn}
\end{eqnarray}
Since $\ml$ and $\mr$ are viewed as $\wc \times \wc$ matrices in Eq.~(\ref{rhored}), it is natural to view the columns of $\widetilde{\ml}$ and $\widetilde{\mr}$ as $\wc \times \wc$ matrices in Eq.~(\ref{tildedefn}).
Consequently, we can directly view the columns of $V_L$ and $V_R$ (defined in Eq.~(\ref{vlvrform})) as $\wc \times \wc$ matrices.
To obtain a direct relation between the generalized eigenvectors of $F$ and the projected Gram matrices $\wml$ and $\wmr$ (defined in Eq.~(\ref{lrexp})), we need to determine how $\wvl$ and $\wvr$ depend on the generalized eigenvectors. 
Suppose the components of $\wvr$ and $\wvl$ have the following forms
\begin{eqnarray}
     \wvr &\equiv&
    \begin{pmatrix}
    R_1 & R_2 & \dots & R_{|\wJ|}
    \end{pmatrix} \nn \\
    \wvl &\equiv&
    \begin{pmatrix}
    L_1 & L_2 & \dots & L_{|\wJ|}
    \end{pmatrix},
\label{tildcomp}
\end{eqnarray}
where $\{R_i\}$ and $\{L_i\}$ are $\wc \times \wc$ matrices.
$\wmr$ and $\wml$ are $n$-independent superpositions of the matrices $\{R_i\}$ and $\{L_i\}$. Their expressions read 
\begin{eqnarray}
    \wmr &=& \sum_{i = 1}^{|\wJ|}{R_i (\wbr)_i } \nn \\
  \wml &=& \sum_{i = 1}^{|\wJ|}{L_i (\wbl)_i }.
\label{LR}
\end{eqnarray}
To relate $\{R_i\}$ and $\{L_i\}$ to $\{r_i\}$ and $\{l_i\}$, we need to consider the Jordan block structure of $\wJ$. 
If $\wJ$ consists of a single Jordan block of generalized eigenvalue $\lambda$, dimension $|\wJ|$, and of the form of Eq.~(\ref{jbexample}); using Eqs.~(\ref{Jn}) and (\ref{tildedefn}), we directly obtain 
\begin{eqnarray}
  R_i &=& \sum_{j = 0}^{i - 1}{\binom{n}{j}r_{i - j} \lambda^{n-j}} \nonumber \\
  L_{i} &=& \sum_{j = 0}^{|\wJ| - i}{\binom{n}{j}l_{i + j} \lambda^{n - j}},
\label{LRform}
\end{eqnarray}
where $\{r_i\}$ and $\{l_i\}$ are viewed as $\wc \times \wc$ matrices.
For $\wJ$ composed of several Jordan blocks, $\{J_i\}$, (e.g. in Eq.~(\ref{Jordantrunc})), Eq.~(\ref{LRform}) holds for each Jordan block separately. 
We first consider a subset of right and left generalized eigenvectors of $\wF$, $\{r^{(J_k)}_i\} \subset \{r_i\}$ and $\{l^{(J_k)}_i\} \subset \{l_i\}$ that are associated with the Jordan block $J_k$ of dimension $|J_k|$ and generalized eigenvalue $\lambda_k$, $|\lambda_k| = 1$. Here, we assume that $r^{(J_k)}_1$ (resp. $l^{(J_k)}_1$) is the right (resp. left) eigenvector and $r^{(J_k)}_i$ (resp. $l^{(J_k)}_i$) is the $(i - 1)$-th right (resp. left) generalized eigenvector. We then define $\{R^{(J_k)}_i\} \subset \{R_i\}$ and $\{L^{(J_k)}_i\} \subset \{L_i\}$ that are related to $\{r^{(J_k)}\}$ and $\{l^{(J_k)}\}$ as 
\begin{eqnarray}
  R^{(J_k)}_i &=& \sum_{j = 0}^{i - 1}{\binom{n}{j}r^{(J_k)}_{i - j} \lambda_k^{n-j}} \nonumber \\
  L^{(J_k)}_{i} &=& \sum_{j = 0}^{|J_k| - i}{\binom{n}{j}l^{(J_k)}_{i + j} \lambda_k^{n - j}}.
\label{LRformjdep}
\end{eqnarray}
This is the analogue of Eq.~(\ref{LRform}) for a single Jordan block $J_k$.
Using Eqs.~(\ref{LR}) and (\ref{LRformjdep}),  $\wmr$ and $\wml$ are of the form
\begin{eqnarray}
    \wmr &=& \sum_{i = 1}^{|\wJ|}{f_R(i,n, \wbr) r_i} \nn \\
    \wml &=& \sum_{i = 1}^{|\wJ|}{f_L(i,n, \wbl) l_i} 
\label{nsuperpositions}
\end{eqnarray}
where $\{f_R(i,n, \wbr)\}$ and $\{f_L(i,n, \wbl)\}$ are \emph{scalar} coefficients that depend on $n$ through  Eq.~(\ref{LRformjdep}) and on the boundary condition dependent vectors $\wbr$ and $\wbl$ respectively.
Since $\wml$ and $\wmr$ are the same as $\ml$ and $\mr$ in the thermodynamic limit,  using Eq.~(\ref{nsuperpositions}), the unnormalized and usually non-Hermitian matrix $\rhored$ of Eq.~(\ref{rhored}) that has the same spectrum as the reduced density matrix reads
\begin{equation}
    \rhored = \sum_{i, j = 1}^{|\wJ|}{f_L(i, n, \wbl) f_R(j, n, \wbr) l_i r_j^T}.
\label{rhoredorder}
\end{equation}
This calculation has been illustrated in App.~\ref{sec:exactexample} with an example from the AKLT model.
In the limit of large $n$, $\rhored$ can be computed using Eq.~(\ref{rhoredorder}) order by order in $n$.
Such a calculation will be discussed with concrete examples from the AKLT models in the next three sections.
\section{Single-mode Excitations}\label{sec:SMA}
As an example, to illustrate the results of the previous section, we first consider single-mode excitations.
A single-mode excitation is defined as an excited eigenstate created by a local operator acting on the ground state.
It is known that such wavefunctions are efficient variational ansatzes for low energy excitations of gapped Hamiltonians.\cite{haegeman2013elementary} 
Such excitations, dubbed as Single-Mode Approximation (SMA) or the Feynman-Bijl ansatz, have also been used as trial wavefunctions for low energy excitations in a variety of models.\cite{arovas1989two, arovas1988extended,yi2002single, arovas2008simplex, haegeman2013elementary, haegeman2012variational,girvin1986magneto}
\subsection{Structure of the transfer matrix}
The SMA state obtained by a local operator $\hat{O}$ can be written as
\begin{eqnarray}
    \ket{\overline{O}_k} &=& \sum_j{e^{i k j} \hat{O}_j \ket{G}} \nn \\
    &\equiv& \overline{O}_k \ket{G}.
\label{smastate}
\end{eqnarray}
where $\hat{O}_j$ denotes the operator $\hat{O}$ in the vicinity of site $j$ of the spin chain (if not purely onsite), $\ket{G}$ is the ground state of the system and $k$ is the momentum of the SMA state.
In the spin-1 AKLT model, the three low-lying exact states shown in Eqs.~(\ref{ArovasA}), (\ref{ArovasB}) and (\ref{Spin2S})  have the form of Eq.~(\ref{smastate}) with $k = \pi$, i.e., the SMA generates an exact eigenstate.\cite{arovas1989two, self}  
In the language of matrix product states, SMA states can be represented as an MPO$\times$MPS, where the MPO represents the operator $\overline{O}_k$, and the MPS is the matrix product representation of the ground state $\ket{G}$.  
As discussed in Sec.~\ref{sec:MPO}, the MPO of a translation invariant local operator $\overline{O}_k$ defined in Eq.~(\ref{smastate}) can be constructed such that it is upper triangular with only two non-vanishing diagonal elements, $e^{ik}\mathds{1}$ and $\mathds{1}$. 
This structure can also be observed in the MPOs of the creation operators of the excited states of the AKLT model, shown in Eqs.~(\ref{firstarovasmpo}), (\ref{secondarovasmpo}) and (\ref{spin2SMPO}).
For the single-mode approximation, the transfer matrix $F$ of $\ket{\overline{O}_k}$ thus has four non-vanishing blocks on the diagonal and its generalized eigenvalues are those of the submatrices on the diagonal (see App.~\ref{blockuppereig}).  
Since all the SMA states of the AKLT model are at momentum $\pi$, we set $k = \pi$ in the following. The same analysis holds for any $k \neq 0$.
We illustrate the entanglement spectrum calculation for the simplest case, where $F$ has the form of Eq.~(\ref{genTransferstructure}), corresponding to an MPO with bond dimension $\chi_m = 2$, the one in Eq.~(\ref{mpoeasyexample}) and $k = \pi$,  
\begin{equation}
    F =
    \begin{pmatrix}
        E & E_{\hat{C}} & E_{\hat{C}^\dagger} & E_{\hat{C}^\dagger \hat{C}} \\
        0 & -E & 0 & -E_{\hat{C}^\dagger} \\
        0 & 0 & -E & -E_{\hat{C}} \\
        0 & 0 & 0 & E
    \end{pmatrix}.
\label{nestedtriangulartransfer}
\end{equation}%
The transfer matrix boundary vectors then have the form of Eq.~(\ref{Transferboundary})
\begin{equation}
    b^l_F
    =
    \begin{pmatrix}
        b^l_E \\
        0 \\
        0 \\
        0
    \end{pmatrix}\;\;\;
    b^r_F 
    =
    \begin{pmatrix}
        0 \\
        0 \\
        0 \\
        b^r_E
    \end{pmatrix}.
\label{smaboundary}
\end{equation}
\subsection{Derivation of $\boldsymbol{\rhored}$}
The structure of generalized eigenvalues and generalized eigenvectors of block upper triangular matrices of the form of $F$ in Eq.~(\ref{nestedtriangulartransfer}) is explained in App.~\ref{sec:Triangular}, and the Jordan normal form of the generalized eigenvalues of unit magnitude is in derived in App.~\ref{sec:JBexamplesma}.
The generalized eigenvalues of $F$ of Eq.~(\ref{nestedtriangulartransfer}) with a unit magnitude are $\{+1, -1, -1, +1\}$, the largest eigenvalues of the submatrices $E$ (the transfer matrices of the ground state MPS).
The $+1$ generalized eigenvalues in $F$ form a Jordan block as long as a certain condition holds (see Eq.~(\ref{JBconditionsma})), which is satisfied for a typical operator $\overline{O}_k$.
Since the off-diagonal block between the subspaces of the two $-E$ blocks is $0$ (as seen in Eq.~(\ref{nestedtriangulartransfer})), the two $-1$ generalized eigenvalues in $F$ do not form a Jordan block.
Thus, for a typical operator $\overline{O}_k$, the Jordan normal form $\wJ$ of the truncated transfer matrix $\wF$ (defined in Eq.~(\ref{Eexp})) is the one in Eq.~(\ref{JBsma}). It can be decomposed into three Jordan blocks as
\begin{equation}
    \wJ = J_0 \oplus J_{\mm} \oplus J_1, 
\end{equation}
where the blocks read
\begin{equation}
    J_0 =
    \begin{pmatrix}
        1 & 1 \\
        0 & 1
    \end{pmatrix}\;\;
    J_{\mm} = (-1) \;\;
    J_1 = (-1).
\label{jordanblocksdecomp}
\end{equation}
Following the convention of Eq.~(\ref{vlvrform}), we assume that $V_R$ and $V_L$ have the forms
\begin{eqnarray}
&&V_R = (r_1 \spa r_2 \spa r_3 \spa r_4) \nn \\
&&V_L = (l_1 \spa l_2 \spa l_3 \spa l_4)
\label{vlvrsmaform}
\end{eqnarray}
Since the $+1$ generalized eigenvalues are due to the top and bottom blocks of $F$, $r_1$ (resp. $l_1$) and $r_4$ (resp. $l_4$) are the right (resp. left) generalized eigenvectors corresponding to $J_3$.
Similarly, $r_2$ (resp. $l_2$) and $r_3$ (resp. $l_3$) correspond to the right (resp. left) generalized eigenvectors of $J_{\mm}$ and $J_1$ respectively.
Thus, the generalized eigenvectors associated with the Jordan blocks can be defined as
\begin{eqnarray}
    &&r^{(J_0)}_1 = r_1 \;\;\; r^{(J_0)}_2 = r_4 \;\;\; r^{(J_{\mm})}_1 = r_2 \;\;\; r^{(J_1)}_1 = r_3 \nn \\
    &&l^{(J_0)}_1 = l_1 \;\;\; l^{(J_0)}_2 = l_4 \;\;\; l^{(J_{\mm})}_1 = l_2 \;\;\; l^{(J_1)}_1 = l_3.
\label{jordanvectors}
\end{eqnarray}
Equivalently, we could also write the truncated Jordan normal form of $F$ as 
\begin{equation}
    \wJ =
    \begin{pmatrix}
    1 & 0 & 0 & 1 \\
    0 & \text{-}1 & 0 & 0 \\
    0 & 0 & \text{-}1 & 0 \\
    0 & 0 & 0 & 1
    \end{pmatrix}.
\label{smajordanblock}
\end{equation}
Since the columns of $V_R$ and $V_L$ are right and left generalized eigenvectors of $F$ corresponding to generalized eigenvalues of unit magnitude, they read (see Eqs.~(\ref{smarightcolumnapp}) and (\ref{smaleftcolumnapp})) 
\begin{equation}
    r_1 = 
    \begin{pmatrix}
    c_1 e_R \\
    0 \\
    0 \\
    0
    \end{pmatrix}\;
    r_2 = 
    \begin{pmatrix}
    \ast \\
    c_2 e_R \\
    0 \\
    0
    \end{pmatrix}\;
    r_3 = 
    \begin{pmatrix}
    \ast \\
    \ast \\
    c_3 e_R \\
    0
    \end{pmatrix}\;
    r_4 = 
    \begin{pmatrix}
    \ast \\
    \ast \\
    \ast \\
    c_4 e_R
    \end{pmatrix}\;
\label{smarightcolumn}
\end{equation}
and 
\begin{equation}
    l_1 = 
    \begin{pmatrix}
    \frac{e_L}{c_1} \\
    \ast \\
    \ast \\
    \ast
    \end{pmatrix}\;
    l_2 = 
    \begin{pmatrix}
    0 \\
    \frac{e_L}{c_2} \\
    \ast \\
    \ast
    \end{pmatrix}\;
    l_3 = 
    \begin{pmatrix}
    0 \\
    0 \\
    \frac{e_L}{c_3} \\
    \ast
    \end{pmatrix}\;
    l_4 = 
    \begin{pmatrix}
    0 \\
    0 \\
    0 \\
    \frac{e_L}{c_4}
    \end{pmatrix}\;
\label{smaleftcolumn}
\end{equation}
where $e_R$ and $e_L$ are the $\chi^2$-dimensional left and right eigenvectors of the $E$ corresponding to the eigenvalue $1$ and the $c_j$'s are some constants. The constant $c_j$ can be set freely if $r_j$ and $l_j$ are eigenvectors (not generalized eigenvectors) of $F$.   
However, in the calculation of $\wvr$ and $\wvl$ (defined in Eq.~(\ref{tildedefn})), the generalized eigenvectors $\{r_i\}$ and $\{l_i\}$ of Eqs.~(\ref{smarightcolumn}) and (\ref{smaleftcolumn}) are viewed as $\wc \times \wc$ matrices. They read 
\begin{eqnarray}
    &r_1 =
    \begin{pmatrix}
    c_1 e_R & 0 \\
    0 & 0
    \end{pmatrix}\;
    r_2 = 
    \begin{pmatrix}
    \ast & 0 \\
    c_2 e_R & 0 
    \end{pmatrix}\nn \\
    &r_3 = 
    \begin{pmatrix}
    \ast & c_3 e_R \\
    \ast & 0
    \end{pmatrix}\;
    r_4 = 
    \begin{pmatrix}
    \ast & \ast \\
    \ast & c_4 e_R
    \end{pmatrix}
\label{righteigv}
\end{eqnarray}
and
\begin{eqnarray}
    &l_1 =
    \begin{pmatrix}
     e_L/c_1 & \ast \\
    \ast & \ast
    \end{pmatrix}\;
    l_2 = 
    \begin{pmatrix}
    0 & \ast \\
    e_L/c_2 & \ast 
    \end{pmatrix}\nn \\
    &l_3 = 
    \begin{pmatrix}
    0 & e_L/c_3 \\
    0 & \ast
    \end{pmatrix}\;
    l_4 = 
    \begin{pmatrix}
    0 & 0 \\
    0 & e_L/c_4
    \end{pmatrix}
\label{lefteigv}
\end{eqnarray}
where $e_R$ and $e_L$ are the right and left eigenvectors of the transfer matrix $E$, now viewed as $\chi \times \chi$ matrices.
Using Eqs.~(\ref{jordanvectors}) and (\ref{LRformjdep}) (or directly Eqs.~(\ref{smajordanblock}) and (\ref{vlvrsmaform})), $\wvr$ and $\wvl$ (whose components are defined in Eq.~(\ref{LRform})) read
\begin{eqnarray}
    \wvr &=&
    \begin{pmatrix}
    r_1  &\spa (\text{-}1)^n r_2 &\spa (\text{-}1)^n r_3 &\spa n r_1 + r_4
    \end{pmatrix} \nn \\
    \wvl &=&
    \begin{pmatrix}
    l_1 + n l_4 &\spa (\text{-}1)^n l_2 &\spa (\text{-}1)^n l_3 &\spa l_4
    \end{pmatrix}.
\end{eqnarray}
Using Eq.~(\ref{LR}), we know that $\wmr$ and $\wml$ read
\begin{eqnarray}
    \wmr &=& r_1 \wbr_1 + (-1)^n r_2 \wbr_2 + (-1)^n r_3 \wbr_3 \nn \\
    &&+ (n r_1 + r_4) \wbr_4 \nn \\
    \wml &=& (l_1 + n l_4) \wbl_1 + (-1)^n l_2 \wbl_2 + (-1)^n l_3 \wbl_3 \nn \\
    &&+ l_4 \wbl_4,
\label{wmrwml}
\end{eqnarray}
where $\{r_i\}$ (resp. $\{l_i\}$) are $\wc \times \wc$ matrices defined in Eq.~(\ref{righteigv}) (resp. Eq.~(\ref{lefteigv})) respectively, and $\wbr_i$ (resp. $\wbl_i$) is the $i$-th component of the right (resp. left) modified boundary vector. 

Since we are mainly interested in the $n \rightarrow \infty$ limit, we obtain $\rhored$ order by order in $n$. Using Eq.~(\ref{rhoredorder}), to order $n^2$, the $\rhored$ which has the same spectrum as the reduced density matrix (up to a global normalization factor), is given by the product of $\mathcal{O}(n)$ terms from both $\wml$ and $\wmr$ in Eq.~(\ref{wmrwml}):
\begin{equation}
    \rhored = n^2 \wbl_1 \wbr_4 l_4 r_1^T + \mathcal{O}(n).
\label{rhoordern2}
\end{equation}
However, from Eqs.~(\ref{righteigv}) and (\ref{lefteigv}), since $l_4 r_1^T = 0$, $\rhored$ is a zero matrix at order $n^2$. 
If we define $b_{i, j} \equiv \wbl_i \wbr_j$, to the next order $n$, $\rhored$ reads 
\begin{eqnarray}
    &\rhored = n ( b_{1, 4} (l_1 r_1^T + l_4 r_4^T) + b_{1, 1} l_4 r_1^T + b_{44} l_4 r_1^T + \nonumber \\
    &+ b_{2, 4} (\text{-}1)^n l_2 r_1^T  +b_{3, 4} (\text{-}1)^n l_3 r_1^T +  b_{1, 2} (\text{-}1)^n l_4 r_2^T + \nonumber \\
    & + b_{1, 3} (\text{-}1)^n l_4 r_3^T)  + \mathcal{O}(1).
\label{rhoordern}
\end{eqnarray}
Computing $\rhored$ in Eq.~(\ref{rhoordern}) using Eqs.~(\ref{righteigv}) and (\ref{lefteigv}), we obtain  
\begin{equation}
    \rhored = n b_{14}
    \begin{pmatrix}
    e_L e_R^T & 0 \\
    \ast & e_L e_R^T
    \end{pmatrix}
    + \mathcal{O}(1)
\label{rhoredlowertriangular}
\end{equation}
Using Eq.~(\ref{gsdensitymatrix}), we know that $e_L e_R^T$ is nothing but the reduced density matrix of the ground state.
Since the $\rhored$ in Eq.~(\ref{rhoredlowertriangular}) is block lower triangular, its eigenvalues are those of its diagonal blocks. 
Thus, the entanglement spectrum, given by the spectrum of $\rhored$, of an MPO$\times$MPS for a single-mode excitation is two degenerate copies of the MPS entanglement spectrum, in the thermodynamic limit (as $n \rightarrow \infty$). 
We then immediately deduce that the entanglement entropy is given by 
\begin{equation}
    \ms = \ms_G + \log 2
\label{SMAentropy}
\end{equation}
The extra $\log 2$ entropy has an alternate interpretation as the Shannon entropy due to the SMA quasiparticle being either in part $\ma$ or part $\mb$ of the system. Thus, we have provided a proof that in the thermodynamic limit, single-mode excitations have an entanglement spectrum that is two copies of the ground state entanglement spectrum. Alternate derivations of the same result were obtained in Refs.~[\onlinecite{pizorn2012universality}] and [\onlinecite{haegeman2013elementary}].
We now move to exact examples obtained in the AKLT models.\cite{self} The Arovas A and B states, and the spin-2 magnon of the spin-1 AKLT model, the Arovas B states and the spin-$2S$ magnon of the spin-$S$ AKLT model are all examples of single-mode excitations. While the Arovas states are exact eigenstates only for periodic boundary conditions, it is reasonable to believe that they are exact eigenstates for open boundary conditions too in the thermodynamic limit. Thus, we expect their entanglement spectra to be two degenerate copies of the ground state entanglement spectra in the thermodynamic limit. While the entanglement spectra in the thermodynamic limit are the same for all the single-mode excitations of the AKLT models, they differ in the nature of their finite-size corrections. We will discuss these differences in Sec.~\ref{sec:finitesize}.
\section{Beyond Single-Mode Excitations}\label{sec:beyondSMA}
We now move on to the computation of the entanglement entropy of states that are obtained by the application of multiple local operators on the ground state. 
Unlike the single-mode approximation, the number of operators acted on the ground state does not uniquely specify entanglement spectrum. 
We thus focus on a concrete example in the 1D AKLT models, the tower of states of Eq.~(\ref{towerofstates}).\cite{self}
We first focus on the state with two magnons ($N = 2$) and then generalize the result to arbitrary $N$ in the next section.
\subsection{Jordan decomposition of the transfer matrix}
For $N = 2$, the MPO $M_{SS_{4}}$ in Eq.~(\ref{TowerofstatesMPO}) has a bond dimension $\chi_m = 3$ and it reads
\begin{equation}
    M_{SS_{4}} =
    \begin{pmatrix}
        -\mathds{1} & -(S^+)^{2S} & 0 \\
        0 & \mathds{1} & (S^+)^{2S} \\
        0 & 0 & -\mathds{1} 
    \end{pmatrix}.
\label{N=2MPO}
\end{equation}
Consequently, using Eq.~(\ref{MPOMPStransfer}) and shorthand notations for the generalized transfer matrices as
\begin{equation}
  E_{+} \equiv E_{(S^+)^{2S}} \;\;\; E_{-} \equiv E_{(S^-)^{2S}} \;\;\; E_{-+} \equiv E_{(S^-)^{2S}( S^+)^{2S}},
\end{equation}
the transfer matrix $F$ can be written as a $9 \times 9$ matrix:
\begin{widetext}
\begin{equation}
    F =
    \begin{pmatrix}
        E & E_+ & 0 & E_- & E_{-+} & 0 & 0 & 0 & 0 \\
        0 & -E & -E_+ & 0 & -E_- & E_{-+} & 0 & 0 & 0 \\ 
        0 & 0 & E & 0 & 0 & E_- & 0 & 0 & 0 \\
        0 & 0 & 0 & -E & E_+ & 0 & E_- & E_{-+} & 0 \\
        0 & 0 & 0 & 0 & E & E_{+} & 0 & -E_- & E_{-+} \\
        0 & 0 & 0 & 0 & 0 & -E & 0 & 0 & E_- \\
        0 & 0 & 0 & 0 & 0 & 0 & E & E_+ & 0 \\
        0 & 0 & 0 & 0 & 0 & 0 & 0 & -E & -E_+ \\
        0 & 0 & 0 & 0 & 0 & 0 & 0 & 0 & E \\
    \end{pmatrix}.
\label{N=2Transfer}
\end{equation}
\end{widetext}
The generalized eigenvalues of $F$ that have magnitude $1$ are due to the $\pm E$ blocks on the diagonals of $F$.
Thus, $F$ has nine generalized eigenvalues of magnitude $1$, five $(+1)$'s and four $(-1)$'s.

In App.~\ref{sec:JBexampletower}, we have derived the Jordan block structure of $F$ of Eq.~(\ref{N=2Transfer}).
There, we used the property (see Eq.~(\ref{AKLTspec}))
\begin{equation}
    E_{+} e_R  = E_{-} e_R = 0\;\;\; e_L^T E_{+} = e_L^T E_{-} = 0 \;\;\; e_L^T E_{-+} e_R \neq 0
\label{towerjbcond}
\end{equation}
where $e_L$ and $e_R$ are the left and right eigenvectors of $E$ corresponding to the eigenvalue $+1$, to show that the largest generalized eigenvalues of any two diagonal blocks in $F$ belong to the same Jordan block if they are related by an off-diagonal block $E_{-+}$ in $F$. 
Thus, for $F$, the truncated Jordan normal form $\wJ$ of the generalized eigenvalues of largest magnitude reads (see Eq.~(\ref{JBtower}))
\begin{equation}
\wJ =
    \begin{pmatrix}
        1 & 0 & 0 & 0 & 1 & 0 & 0 & 0 & 0 \\
        0 & -1 & 0 & 0 & 0 & 1 & 0 & 0 & 0 \\ 
        0 & 0 & 1 & 0 & 0 & 0 & 0 & 0 & 0 \\
        0 & 0 & 0 & -1 & 0 & 0 & 0 & 1 & 0 \\
        0 & 0 & 0 & 0 & 1 & 0 & 0 & 0 & 1 \\
        0 & 0 & 0 & 0 & 0 & -1 & 0 & 0 & 0 \\
        0 & 0 & 0 & 0 & 0 & 0 & 1 & 0 & 0 \\
        0 & 0 & 0 & 0 & 0 & 0 & 0 & -1 & 0 \\
        0 & 0 & 0 & 0 & 0 & 0 & 0 & 0 & 1 \\
    \end{pmatrix}
\label{N=2jordanblock}
\end{equation}
The forms of the right and left generalized eigenvectors corresponding to the generalized eigenvalues in $\wJ$ are determined by Eqs.~(\ref{eigenvectorforms}) in App.~\ref{sec:Triangular}.
For example, the left and right generalized eigenvectors corresponding to the fourth eigenvalue ($-1$) on the diagonal of $\wJ$ in Eq.~(\ref{N=2jordanblock}) read
\begin{equation}
r_4 = 
    \begin{pmatrix}
        \ast \\
        \ast \\
        \ast \\
        c_{1,2} e_R \\
        0 \\
        0 \\
        0 \\
        0 \\
        0
    \end{pmatrix}\;\;
l_4 =
    \begin{pmatrix}
        0 \\
        0 \\
        0 \\
        \frac{e_L}{c_{1,2}} \\
        \ast \\
        \ast \\
        \ast \\
        \ast \\
        \ast
    \end{pmatrix}
\end{equation}
where $e_R$ and $e_L$ are the left and right generalized eigenvectors of $E$ and $c_{1,2}$ is some constant. 
When viewed as $3 \times 3$ matrices, these read
\begin{equation}
r_{2,1} \equiv r_4 = 
    \begin{pmatrix}
        \ast & c_{1,2} e_R & 0 \\
        \ast & 0 & 0 \\
        \ast & 0 & 0
    \end{pmatrix}\;\;
l_{2,1} \equiv l_4 =
    \begin{pmatrix}
        0 & \frac{e_L}{c_{1,2}} & \ast \\
        0 & \ast & \ast \\
        0 & \ast & \ast \\
    \end{pmatrix}
\label{r4l4N=2}
\end{equation}
where we have defined 
\begin{equation}
    r_{\alpha,\beta} \equiv r_{3 (\alpha -1) + \beta}\;\;\; l_{\alpha, \beta} \equiv l_{3(\alpha -1) + \beta}   
\end{equation}
to be the generalized eigenvectors of $F$ corresponding to the generalized eigenvalue of magnitude $1$ and $e_R$ and $e_L$ are viewed as $\chi \times \chi$ matrices.
Thus, in general, the expression for the $3 \times 3$ $r_{\alpha, \beta}$ (resp. $l_{\alpha, \beta}$) is obtained by filling in irrelevant elements ``$\ast$"'s column-wise from top-to-bottom (resp. bottom-to-top) starting from the top-left (resp. bottom-right) corner until the $(\alpha, \beta)$-th element, which is set to $c_{\alpha, \beta} e_R$ (resp. $e_L/c_{\alpha, \beta}$). 
Using the structure of $\wJ$ in Eq.~(\ref{N=2jordanblock}), we observe that five Jordan blocks $J_m$, $-2 \leq m \leq 2$ are formed, that have generalized eigenvalues $(-1)^m$ and consist of generalized eigenvectors $r_{\alpha, \alpha + m}$ and $l_{\alpha, \alpha + m}$ with $1 - \textrm{min}(0,m) \leq \alpha \leq 3 - \textrm{max}(0, m)$. 

\subsection{General properties of $\boldsymbol{\wmr}$ and $\boldsymbol{\wml}$}\label{sec:wmrwmlproperties}
We now proceed to derive some general properties of $\wmr$ and $\wml$ that are helpful in the calculation of $\rhored$ (see Eq.~(\ref{rhoredorder})).
Since $\rhored$ is a sum products of the form $l_{\alpha, \beta} r_{\gamma, \delta}^T$ (see Eq.~(\ref{rhoredorder})), using the forms of the generalized eigenvectors $l_{\alpha, \beta}$ and $r_{\alpha, \beta}$ (for example Eq.~(\ref{r4l4N=2})), we note the following properties:
\begin{eqnarray}
    &&l_{\alpha,\beta} r_{\gamma,\delta}^T = 0\;\;\; \textrm{if}\;\;\; \beta > \delta,
\label{N=2observevanish} \\
    &&l_{\alpha,\beta} r_{\gamma,\beta}^T = \twopartdef{\ltriangle}{\alpha > \gamma}{\ltriangle + \Delta(\alpha, e_L e_R^T)}{\alpha = \gamma},
\label{N=2productstructure}
\end{eqnarray}
where $\ltriangle$ represents a strictly lower-triangular matrix and $\Delta(\alpha, x)$ is a diagonal matrix with the $\alpha$-th element on the diagonal equal to $x$. As we will see in the next section, these properties are valid for any number of magnons $N$.
To compute $\rhored$ order by order in the length of the subsystem $n$, we need to determine the factor of $n$ that appears in front of the product $l_{\alpha, \beta} r_{\gamma, \delta}^T$ in $\rhored$.
We first obtain the factors of $n$ that accompany each of $r_{\alpha,\beta}$ and $l_{\alpha, \beta}$ in $\wmr$ and $\wml$ respectively.
Using Eqs.~(\ref{LRformjdep}) and (\ref{LR}), when $N = 2$ the expression for $\wmr$ reads
\begin{eqnarray}
  &\hspace{-3mm}\wmr = \left(\binom{n}{2} r_{1,1} + n r_{2,2} + r_{3,3}\right)\wbr_9 + (n r_{1,1} + r_{2,2}) \wbr_5 + r_{1,1} \wbr_1 \nn \\
  &+ (-1)^n [(n r_{1,2} + r_{2,3}) \wbr_8 + r_{1,2} \wbr_4] \nn \\
  &+ (-1)^n [(n r_{2,1} + r_{3,2}) \wbr_6 + r_{2,1} \wbr_2] \nn \\
  &+ r_{1,3} \wbr_7 \nn \\
  &+ r_{3,1} \wbr_3,
\label{N=2rtild}
\end{eqnarray}
where terms on the same line come from the same Jordan block $J_m$.
Similarly, the expression for $\wml$ for $N = 2$ reads 
\begin{eqnarray}
&\hspace{-3mm}\wml= \left(\binom{n}{2} l_{3,3} + n l_{2,2} + l_{1,1}\right)\wbl_1 + (n l_{3,3} + l_{2,2}) \wbl_5 + l_{3,3} \wbl_9 \nn \\
 &+ (-1)^n [(n l_{2,3} + l_{1,2}) \wbl_4 + l_{2,3} \wbl_8] \nn \\
  &+ (-1)^n [(n l_{3,2} + l_{2,1}) \wbl_2 + l_{3,2} \wbl_6] \nn \\
  &+ l_{1,3} \wbl_7 \nn \\
  &+ l_{3,1} \wbl_3,
\label{N=2ltild}
\end{eqnarray}
The structure of Eqs.~(\ref{N=2rtild}) and (\ref{N=2ltild}) exemplify properties of $\mr$ and $\ml$ that are valid for any value of $N$: 
\begin{enumerate}
\item The largest combinatorial factors $C^R_{\alpha,\beta}$ and $C^L_{\alpha,\beta}$ that multiply the right and left generalized eigenvectors $r_{\alpha, \beta}$ and $l_{\alpha, \beta}$ in $\wmr$ and $\wml$ respectively read (as a consequence of Eqs.~(\ref{LR}) and (\ref{LRformjdep}))
\begin{equation}
    C^R_{\alpha, \beta} = \binom{n}{N - \textrm{max}(\alpha, \beta) +1} 
\label{N=2CR}
\end{equation}
\begin{equation}
    C^L_{\alpha, \beta} = \binom{n}{\textrm{min}(\alpha, \beta) - 1} 
\label{N=2CL}
\end{equation}
For example, the largest combinatorial factors to multiply $r_{1,1}$ and $l_{3,3}$ in Eqs.~(\ref{N=2rtild}) and (\ref{N=2ltild}) are $C^R_{1,1} = \binom{n}{2}$ and $C^L_{3,3} = \binom{n}{2}$ respectively. 
\item The dominant term (with the largest factor of $n$) involving generalized eigenvectors of any given Jordan block are all multiplied by the same boundary vector component in the expression for $\wml$ and $\wmr$. This is derived using Eqs.~(\ref{LR}) and (\ref{LRformjdep}). 
For example, $r_{1,1}$, $r_{2,2}$ and $r_{3,3}$ (resp. $l_{1,1}$, $l_{2,2}$ and $l_{3,3}$) are all associated with the same Jordan block ($J_0$), and the largest factors of $n$ that multiply them are $\binom{n}{2}\wbr_9$, $n \wbr_9$ and $\wbr_9$ (resp. $\wbl_1$, $n \wbl_1$, and $\binom{n}{2}\wbl_1$). That is, the dominant terms involving these generalized eigenvectors are all multiplied by the same boundary vector component $\wbr_9$ (resp. $\wbl_1$) in $\wmr$ (resp. $\wml$). 
\item All the terms in Eq.~(\ref{N=2productstructure}) associated with a given Jordan block are multiplied by $\lambda^n$, where $\lambda$ is the eigenvalue associated with the Jordan block involved (here either $(+1)$ or $(-1)$). This is seen in Eq.~(\ref{LRformjdep}). 
\end{enumerate}
Using $C^L_{\alpha, \beta}$ and $C^R_{\alpha, \beta}$ of Eqs.~(\ref{N=2CL}) and (\ref{N=2CR}) respectively, one can directly compute $\rhored$ (defined in Eq.~(\ref{rhored})) order by order in $n$.
Note that 
\begin{equation}
    C^L_{\alpha, \beta} C^R_{\gamma, \delta} \sim \mathcal{O}\left(n^{N + \textrm{min}(\alpha,\beta) - \textrm{max}(\gamma, \delta)}
    \right).
\label{lrorder}
\end{equation}
Using Eq.~(\ref{lrorder}), we note that any term of order strictly greater than $n^N$ requires $\textrm{min}(\alpha, \beta) > \textrm{max}(\gamma, \delta)$, which necessarily implies $\beta > \delta$.
Since all products $l_{\alpha,\beta} r_{\gamma,\delta}^T$ vanish (using Eq.~(\ref{N=2observevanish})), the dominant non-vanishing terms appear at order $n^N$ or smaller.
Directly from Eq.~(\ref{lrorder}),  if $\beta < \delta$, $\beta < \gamma$, $\alpha < \gamma$ or $\alpha < \delta$,  the product $C^L_{\alpha,\beta} C^R_{\gamma,\delta}$ necessarily has a smaller order than $n^N$.
Thus, at order $n^N$, we obtain products that satisfy $\alpha \geq \gamma$, $\alpha \geq \delta$, $\beta \geq \delta$ and $\beta \geq \gamma$.
The products with $\beta > \delta$ vanish (using Eq.~(\ref{N=2observevanish})) and products with $\alpha > \gamma$ give rise to lower triangular terms (using Eq.~(\ref{N=2productstructure})); and they do not contribute to the eigenvalues of $\rhored$ when no upper triangular terms are present. We thus deduce that the products that determine the spectrum of $\rhored$ (and hence the entanglement spectrum) at leading order in $n$ satisfy $\beta = \delta$, $\alpha = \gamma$, $\alpha \geq \delta$ and $\beta \geq \gamma$; and consequently, $\alpha = \beta = \gamma = \delta$.
Furthermore, since all the $r_{\alpha, \alpha}$'s and $l_{\alpha, \alpha}$'s belong to the largest Jordan block with eigenvalue $+1$, all the products $l_{\alpha, \alpha} r_{\alpha, \alpha}^T$ are multiplied with the same modified boundary vector components. 
Indeed, these arguments can be verified using the exact form of $\rhored$ at order $n^2$ using $\wml$ and $\wmr$ in Eqs.~(\ref{N=2rtild}) and (\ref{N=2ltild}):
\begin{eqnarray}
    &\rhored = \left(\binom{n}{2} l_{1,1} r_{1,1}^T  + n^2 l_{2,2}r_{2,2}^T + \binom{n}{2} l_{3,3}r_{3,3}^T \right) b_{1,9} + \nn \\
    &n^2 [l_{3,2} r_{1,2}^T b_{2,8} + (-1)^n (l_{3,2} r_{2,2}^T b_{2,9} + l_{2,2} r_{1,2}^T b_{1,8})] \nn \\
    &+\binom{n}{2}[(l_{3,3} r_{1,3}^T b_{1,7} + l_{3,1} r_{1,1}^T b_{3,9} \nn \\
    &+(-1)^n (l_{3,3}r_{2,3}^T b_{1,8} + l_{2,1}r_{1,1}^T b_{2,9})],
\label{N=2expand}
\end{eqnarray}
where $b_{i, j} \equiv \wbl_i \wbr_j$.
Thus, at order $n^2$, using Eqs.~(\ref{N=2expand}) and (\ref{N=2productstructure}), $\rhored$ reads
\begin{eqnarray}
    \rhored &=&
    b_{1,9}\begin{pmatrix}
        \binom{n}{2} e_L e_R^T & 0 & 0 \\
        \ast & n^2 e_L e_R^T & 0 \\
        \ast & \ast & \binom{n}{2} e_L e_R^T \\
    \end{pmatrix} + \mathcal{O}(n) \nn \\
    &\approx& 
    n^2 b_{1,9}\begin{pmatrix}
        \frac{1}{2} e_L e_R^T & 0 & 0 \\
        \ast & e_L e_R^T & 0 \\
        \ast & \ast & \frac{1}{2} e_L e_R^T \\
    \end{pmatrix} + \mathcal{O}(n),
\end{eqnarray}
where we have used $\binom{n}{2} \approx \frac{n^2}{2}$, an approximation that is exact as $n \rightarrow \infty$.
The entanglement spectrum of two magnons on the ground state is thus three copies of the ground state entanglement spectrum. The three copies are however, separated into one non-degenerate and two degenerate copies.
%

%

%

%

%
%
\section{Tower of States}\label{sec:Towerofstates}
We now move on to the calculation of the entanglement spectra for the AKLT tower of states with $N > 2$ magnons on the ground state.
The expression for the MPO $M_{SS_{2N}}$ for the tower of states operator has a bond dimension $\chi_m = N + 1$ and is shown in Eq.~(\ref{TowerofstatesMPO}).
Several results in this section are a straightforward generalization of results in the previous section.
\subsection{Jordan decomposition of the transfer matrix}
Analogous to Eq.~(\ref{N=2Transfer}), the transfer matrix $F$ for arbitrary $N$ can be written as a $(N+1) \times (N+1)$ block upper triangular matrix, with $\chi \times \chi$ blocks.
Thus, the generalized eigenvectors of $F$ for a general $N$ have inherited a structure as those in Eq.~(\ref{r4l4N=2}).
The right and left generalized eigenvectors $r_{\alpha,\beta} \equiv r_{(N+1) (\alpha - 1) + \beta}$ and $l_{\alpha,\beta} \equiv l_{(N+1) (\alpha - 1) + \beta}$ have the forms (when viewed as $(N+1) \times (N+1)$ matrices),
\begin{eqnarray}
    r_{\alpha,\beta} = 
    \begin{pmatrix}
        \ast & \cdots & \cdots & \ast & 0 & \cdots & 0 \\
        \vdots &  \ddots & \ddots & \vdots & \vdots & \ddots & \vdots \\
        \ast & \cdots & \cdots & \ast & \vdots & \ddots & \vdots \\
        \ast & \cdots & \ast & c_{\alpha, \beta} e_R & 0 & \cdots & 0 \\
        \vdots & \ddots & \vdots & 0 & \cdots & \cdots & 0 \\
        \vdots & \ddots &  \vdots & \vdots & \ddots & \ddots & \vdots \\
        \ast & \cdots & \ast & 0 & \cdots & \cdots & 0 \\
    \end{pmatrix} \nn \\
    l_{\alpha,\beta} = 
    \begin{pmatrix}
        0 & \cdots & \cdots & 0 & \ast & \cdots & \ast \\
        \vdots &  \ddots & \ddots & \vdots & \vdots & \ddots & \vdots \\
        0 & \cdots & \cdots & 0 & \vdots & \ddots & \vdots \\
        0 & \cdots & 0 & \frac{e_L}{c_{\alpha, \beta}} & \ast & \cdots & \ast \\
        \vdots & \ddots & \vdots & \ast & \cdots & \cdots & \ast \\
        \vdots & \ddots &  \vdots & \vdots & \ddots & \ddots & \vdots \\
        0 & \cdots & 0 & \ast & \cdots & \cdots & \ast \\
    \end{pmatrix},
\label{generalNeig}
\end{eqnarray}
where the $(\alpha, \beta)$-th element in $r_{\alpha,\beta}$ and $l_{\alpha,\beta}$ are proportional to $e_R$ and $e_L$ respectively.
Since the off-diagonal blocks of $F$ have the same structure as those in Eq.~(\ref{N=2Transfer}) (because the structures of the MPOs $M_{SS_2}$ and $M_{SS_{2N}}$ are the same), the Jordan normal form is similar to the $N = 2$ case.
That is, we obtain $(2N + 1)$ Jordan blocks $J_m$, $-N \leq m \leq N$, that correspond to an eigenvalue $(-1)^m$ and consist of generalized eigenvectors $r_{\alpha, \alpha + m}$ and $l_{\alpha, \alpha + m}$ with $1 - \textrm{min}(0,m) \leq \alpha \leq N+1 - \textrm{max}(0,m)$.
As pointed out in Sec.~\ref{sec:wmrwmlproperties}, the properties observed there are valid for all $N$.
Thus, using $C^R_{\alpha, \beta}$ and $C^L_{\alpha, \beta}$, $\rhored$ can be constructed order by order in $n$.
However, for arbitrary $N$, we can study two types of limits (i) $n \rightarrow \infty$, $N$ finite, and (ii) $n \rightarrow \infty$, $N \rightarrow \infty$, $N/n \rightarrow \textrm{const.} > 0$.
Since $n = L/2$, $N$ is the number of magnons in the state $\ket{SS_{2N}}$, and the state has an energy $E = 2N$,\cite{self} the energy density of the state we are studying is $E/L = N/n$. 
The limits (i) and (ii) thus correspond to zero and finite energy density excitations respectively.
\subsection{Zero density excitations}\label{sec:zerodensity}
In the limit where $N$ is finite as $n \rightarrow \infty$, we can use the approximation
\begin{equation}
    \binom{n}{N} \approx \frac{n^N}{N!},
\label{binomlowkapprox}
\end{equation}
which is asymptotically exact. Thus, the product of combinatorial factors can be classified by order in $n$.
Since the structure of the generalized eigenvectors $l_{\alpha, \beta}$ and $r_{\alpha, \beta}$ in Eq.~(\ref{generalNeig}) are the same as the $N = 2$ case in the previous section, properties Eqs.~(\ref{N=2observevanish}) and (\ref{N=2productstructure}) are valid here. 
Using the arguments following Eq.~(\ref{lrorder}) in Sec.~\ref{sec:wmrwmlproperties}, the first non-vanishing term appears at order $n^N$, and the expression for $\rhored$ reads
\begin{eqnarray}
    \hspace{-2mm}\rhored &=& b_{1, (N+1)^2} \sum_{\alpha = 0}^{N}{\binom{n}{\alpha} \binom{n}{N-\alpha} l_{\alpha, \alpha} r_{\alpha, \alpha}^T} + \ltriangle + \mathcal{O}(n^{N-1}) \nn \\
    \hspace{-5mm}&\approx& n^N b_{1, (N+1)^2} \sum_{\alpha = 0}^{N}{\frac{1}{\alpha! (N-\alpha)!} l_{\alpha, \alpha} r_{\alpha, \alpha}^T} + \ltriangle + \mathcal{O}(n^{N-1}) \nn \\
\label{rhogeneralN}
\end{eqnarray}
where $\ltriangle$ represents strictly lower triangular matrices.
Using Eq.~(\ref{N=2productstructure}), to leading order in $n$, we obtain the unnormalized density matrix:
\begin{equation}
    \rhored = 
    n^{N}b_{1, (N+1)^2} \begin{pmatrix}
    \frac{e_L e_R^T}{N!0!} & 0 \dots & \dots & 0 \\
    \ast & \frac{e_L e_R^T}{(N-1)!1!} & \ddots & \ddots & \vdots \\
    \vdots & \ddots & \ddots & \ddots & \vdots \\
    \vdots & \ddots & \ddots & \frac{e_L e_R^T}{1!(N-1)!} & 0 \\
    \ast & \dots & \dots & \ast & \frac{e_L e_R^T}{0!N!}
    \end{pmatrix}
\label{rhocopies}
\end{equation}
where $e_L e_R^T$ is the ground state reduced density matrix. 
Since $e_L e_R^T$ for the spin-$S$ AKLT model has $(S+1)$ degenerate levels (see Eq.~(\ref{spinSAKLTrhored})), after normalizing $\rhored$, the entanglement spectrum has $(N+1)$ copies of $(S+1)$ degenerate levels, and each $(S+1)$-multiplet reads
\begin{equation}
    \lambda_\alpha = \frac{1}{2^N (S+1)}\binom{N}{\alpha}\;\;\; 0 \leq \alpha \leq N.
\label{ESlowk}
\end{equation}
The trace of $\rhored$ is indeed 1, 
\begin{eqnarray}
    {\rm Tr}\left[\rhored\right] &=& (S+1) \sum_{\alpha = 0}^N{\lambda_\alpha} \nn \\
    &=& \frac{1}{2^N}\sum_{\alpha = 0}^N{\binom{N}{\alpha}} \nn \\
    &=& 1.
\end{eqnarray}
The entanglement entropy is thus
\begin{eqnarray}
      \ms &=& - {\rm Tr}\left[ \rhored \log \rhored \right] \nonumber \\
        &=& -(S+1) \sum_{\alpha = 0}^N{\lambda_\alpha \log \lambda_\alpha} \nonumber \\
        &=& \ms_G +  N\log 2 - \frac{1}{2^N}\sum_{\alpha = 0}^{N}{\binom{N}{\alpha} \log \binom{N}{\alpha}} \label{entropylowk} \\
        &\sim& \ms_G + \frac{1}{2}\log\left(\frac{\pi N}{2}\right)\;\;\;\textrm{for large $N$}
\label{entropylargeN}
\end{eqnarray}
where $\ms_G = \log(S+1)$, the entanglement entropy of the spin-$S$ AKLT ground state.
Eq.~(\ref{entropylargeN}) is derived from Eq.~(\ref{entropylowk}) in App.~\ref{sec:towerentropy} using a saddle point approximation.
For $N = 1$, using Eq.~(\ref{entropylowk}), we recover the Single-Mode Approximation result of Eq.~(\ref{SMAentropy}).
Furthermore, note that $\mathcal{O}(n^{N-1})$ and lower order corrections to $\rhored$ in Eq.~(\ref{rhocopies}) are typically not lower triangular matrices. Thus, the replica structure of $\rhored$ breaks at any finite $n$, giving a particular structure to the finite-size corrections. We discuss the nature of these finite-size corrections in Sec.~\ref{sec:finitesize}. 
\subsection{Finite density excitations}\label{sec:finitedensity}
We now proceed to the case where the excited state has a finite energy density, corresponding to a finite density of magnons on the ground state.
That is, 
\begin{equation}
    E/L = N/n > 0.
\label{finitedensitycondition}
\end{equation}
For a large enough $N$, approximation Eq.~(\ref{binomlowkapprox}) breaks down.
Nevertheless, since the MPO for $\ket{SS_{2N}}$ and the MPS for the ground state of the spin-$S$ AKLT model have bond dimensions of $\chi_m = (N+1)$ and $\chi = (S+1)$ respectively, the MPO $\times$ MPS for $\ket{SS_{2N}}$ has a bond dimension $\chi \chi_m = (S+1)(N+1)$, i.e. it grows linearly in $N$. 
Consequently, using Eq.~(\ref{entropybound}), the entanglement entropy of $\ket{SS_{2N}}$ is bounded by 
\begin{equation}
    \ms \leq \log(\chi\chi_m) = \log[(S+1)(N+1)].
\label{entropyscaling}
\end{equation}
Using Eqs.~(\ref{entropylargeN}) and (\ref{entropyscaling}), we would be tempted to find a stronger bound or an asymptotic expression for the entanglement entropy in the finite density limit. Indeed, we expect this entanglement entropy to have the form
\begin{equation}
    \ms \sim P \log N
\label{prefactor}
\end{equation}
where $P$ is some constant.
Without the approximation of Eq.~(\ref{binomlowkapprox}), terms that are weighted by the combinatorial factor $\binom{n}{a}\binom{n}{k -a}$ do not necessarily suppress the terms that appear with a factor $\binom{n}{a}\binom{n}{k-a-b}$, where $k$, $a$ and $b$ are some positive integers. 
This invalidates an expansion in orders of $n$ such as Eq.~(\ref{rhogeneralN}).
Consequently the lower triangular structure of $\rhored$ (see Eq.~(\ref{rhocopies})) breaks down.
Hence, it is not clear if the expression for the entanglement entropy of Eq.~(\ref{entropylargeN}) survives in the finite density regime.
A detailed discussion of this is given in App.~\ref{sec:breakdown}.
\begin{figure*}[ht]
 \begin{tabular}{cc}
\includegraphics[scale = 0.43]{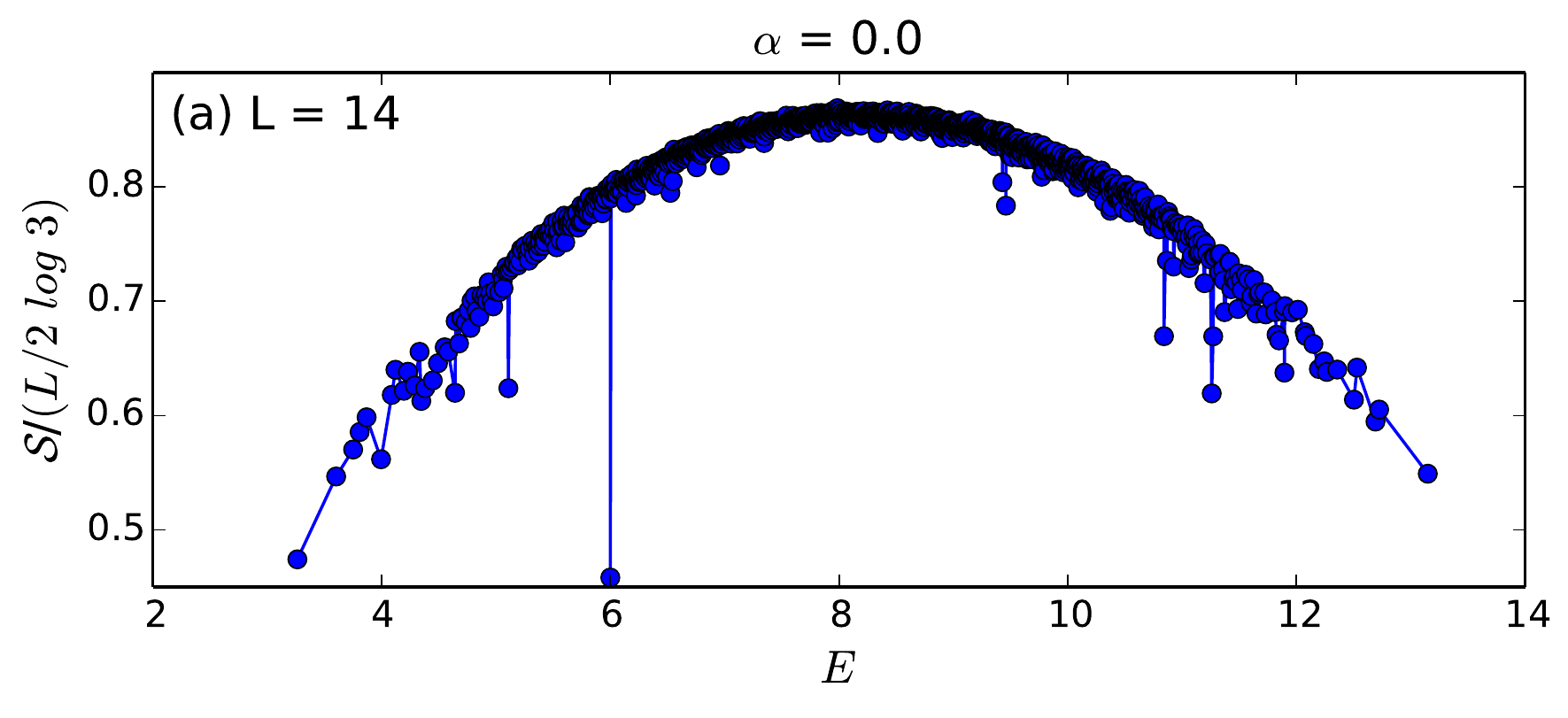}&\includegraphics[scale = 0.43]{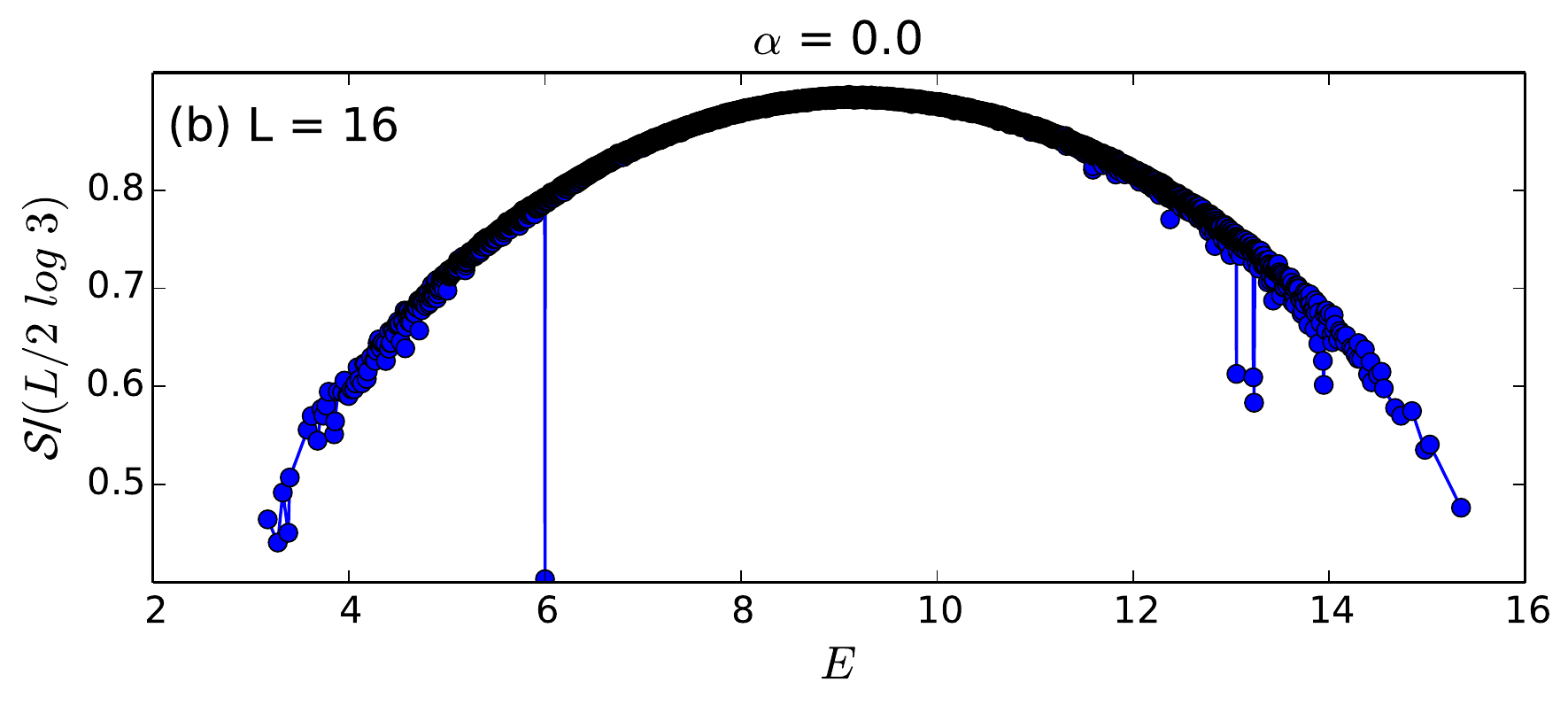}\\
\includegraphics[scale = 0.43]{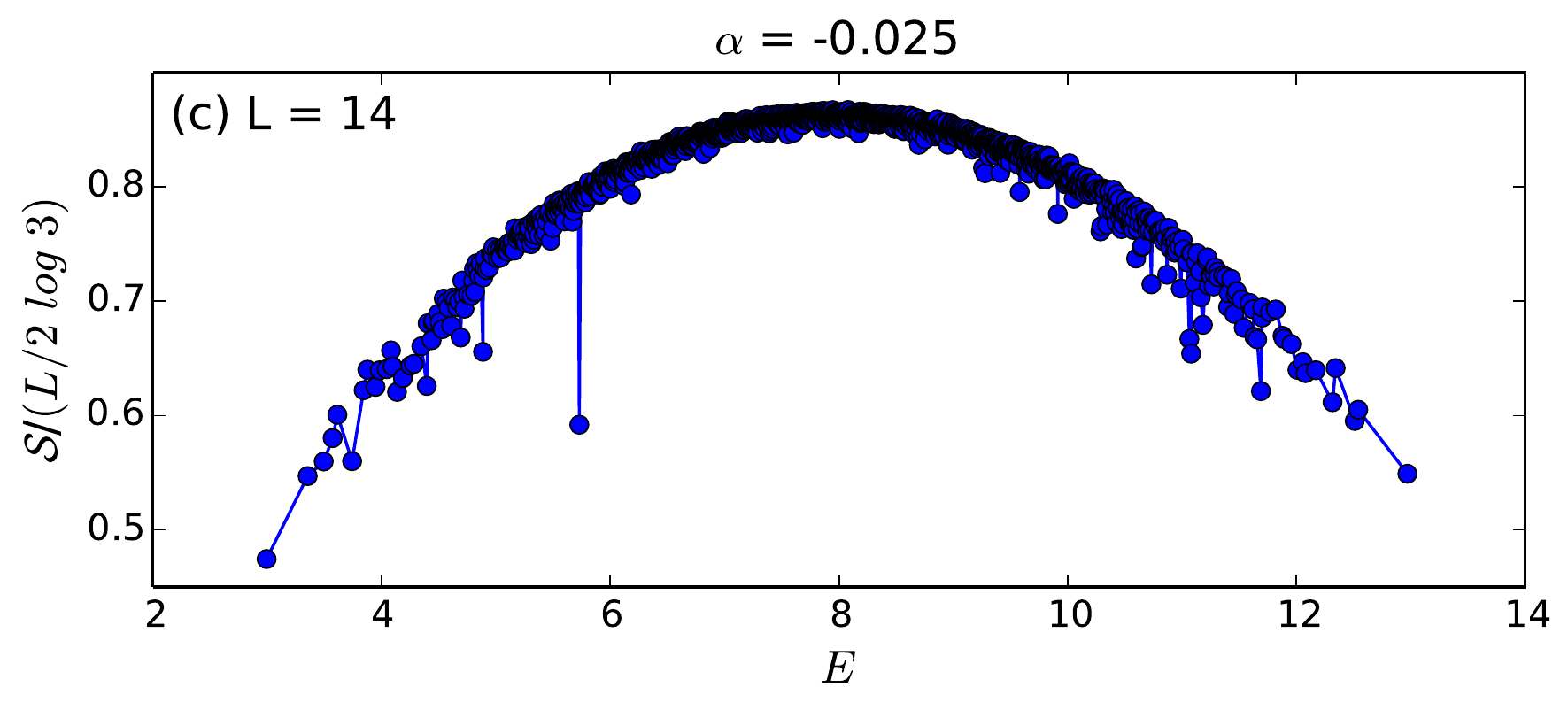}&\includegraphics[scale = 0.43]{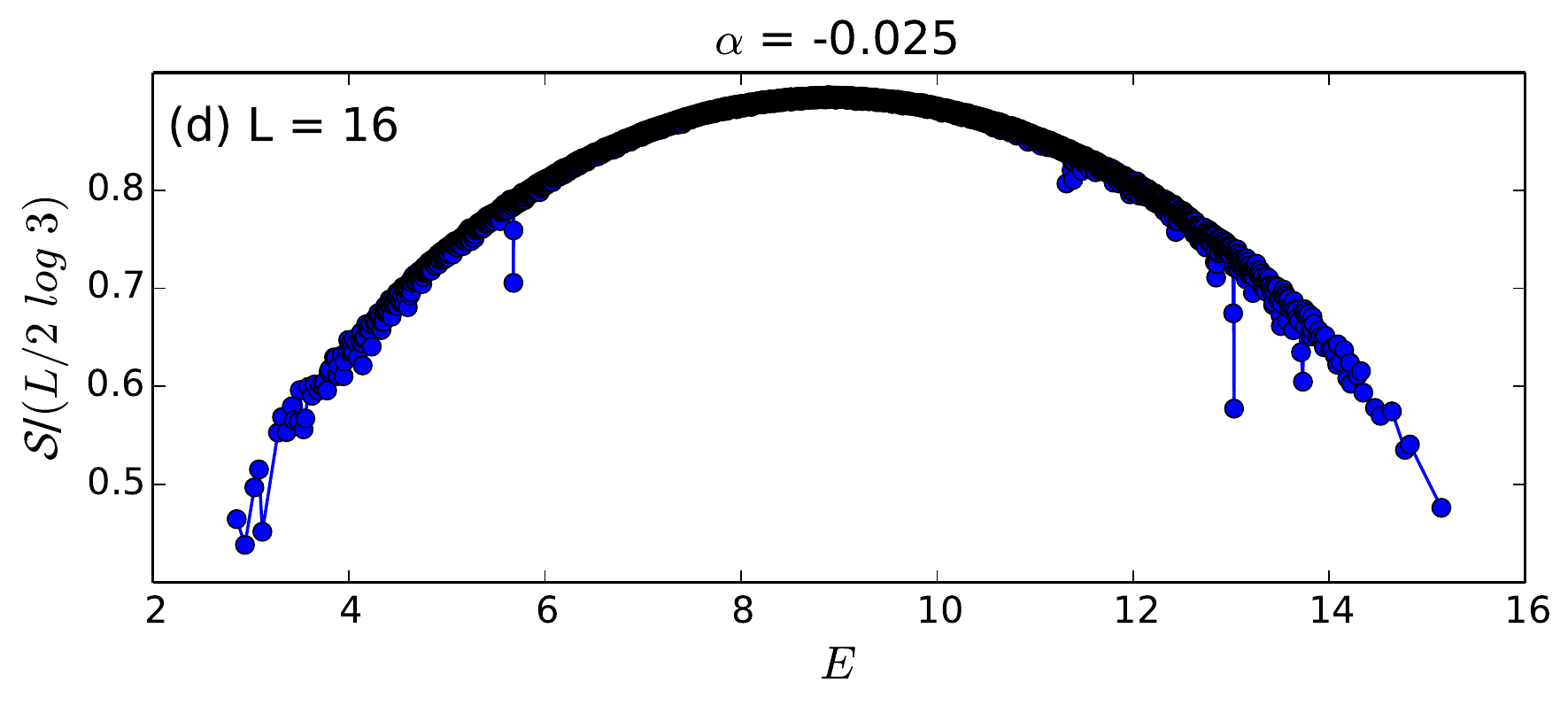}\\
\includegraphics[scale = 0.43]{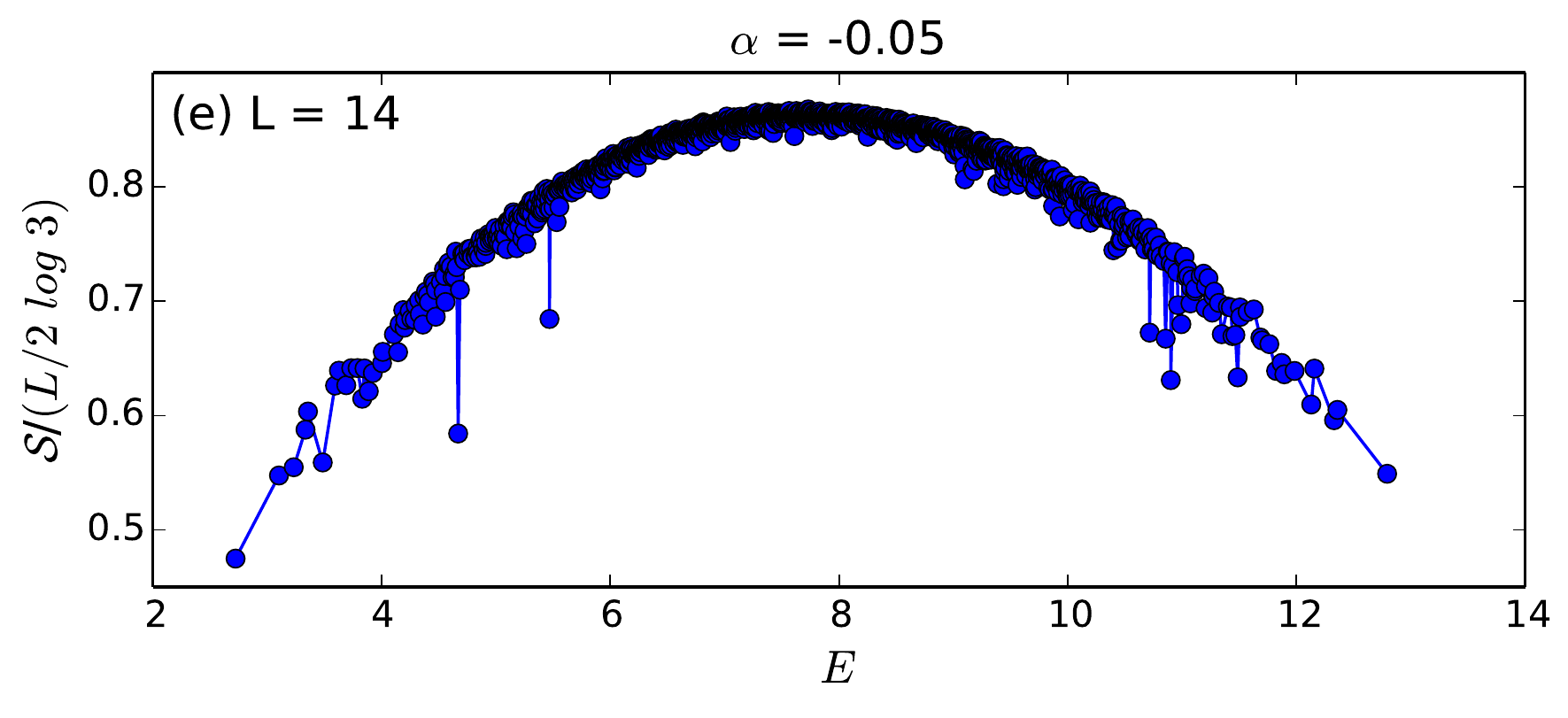}&\includegraphics[scale = 0.43]{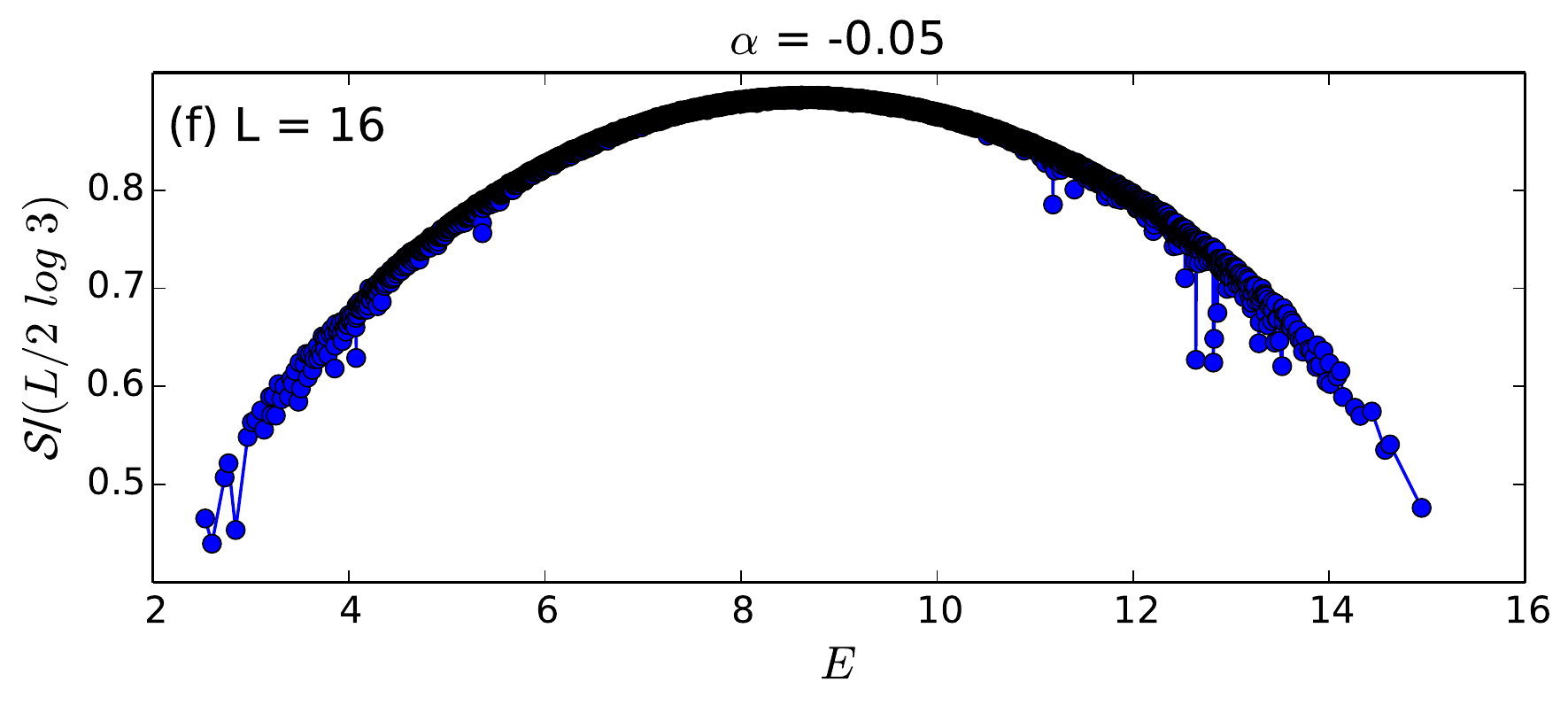}\\
\end{tabular}
\caption{The normalized entanglement entropy $\ms/((L/2) \log 3$) for the eigenstates of the Hamiltonian with energy $E$ Eq.~(\ref{hamilpert}) in the quantum number sector $(s, S_z, k, I, P_z) = (6, 0, \pi, -1, +1)$. Panels (a) and (b) show the entropy at the AKLT point. This sector has a single exact state $\ket{S_6}$ that belongs to the tower of states, which exhibits a sharp dip at $E = 6$. Panels (c) and (d) show the entropy for $\alpha = -0.025$, where remnants of the low entropy states are seen. Panels (e) and (f) show the entropy in the same sector for $\alpha = -0.05$.   \label{fig:awayaklt}}
\end{figure*}

\section{Implications for the Eigenstate Thermalization Hypothesis (ETH)}\label{sec:ETH}
In Ref.~[\onlinecite{self}] we conjectured and provided numerical evidence that in the thermodynamic limit some states of the tower of states are in the bulk of the spectrum, i.e. in a region of finite density of states of their own quantum number sector.
Furthermore, we showed that the AKLT model is non-integrable, i.e. it exhibits Gaussian Orthogonal Ensemble (GOE) level statistics. 
According to the Eigenstate Thermalization Hypothesis (ETH), typical states in the bulk of the spectrum look thermal.\cite{srednicki1994chaos, rigol2008thermalization,d2016quantum} That is, the entanglement entropy of any such states exhibits a volume law scaling, $\ms \propto L$.
A strong form of the ETH conjuctures that \emph{all} states in a region of finite density of states of the same quantum number sector look thermal.\cite{kim2014testing, garrison2018does}
In the spin-$S$ AKLT tower of states, for a state with a finite density of magnons, using Eq.~(\ref{entropyscaling}), 
\begin{equation}
    \ms \propto \log L.
\end{equation}
The $\log L$ scaling of the entanglement entropy in Eq.~(\ref{entropyscaling}) is thus a clear violation of the strong ETH.  
The atypical behavior of the tower of states is illustrated in Fig.~\ref{fig:awayaklt}. In Figs.~\ref{fig:awayaklt}a and \ref{fig:awayaklt}b, we plot the entanglement entropy of all the states in a given quantum number sector for two system sizes $L = 14$ and $L = 16$ at the AKLT point. 
The dip of the entanglement entropy at energy $E = 6$ corresponds to the state $\ket{S_6}$, which clearly violates the trends of entanglement entropy within its own quantum number sector. The dip persists for $L = 16$, the largest system size accessible to exact diagonalization. 
These states are thus the first examples of what are now known as ``quantum many-body scars".\cite{turner2018weak, turner2018quantum, schecter2018many, bernien2017probing} 
\begin{figure*}[ht]
 \begin{tabular}{cc}
\includegraphics[scale = 0.43]{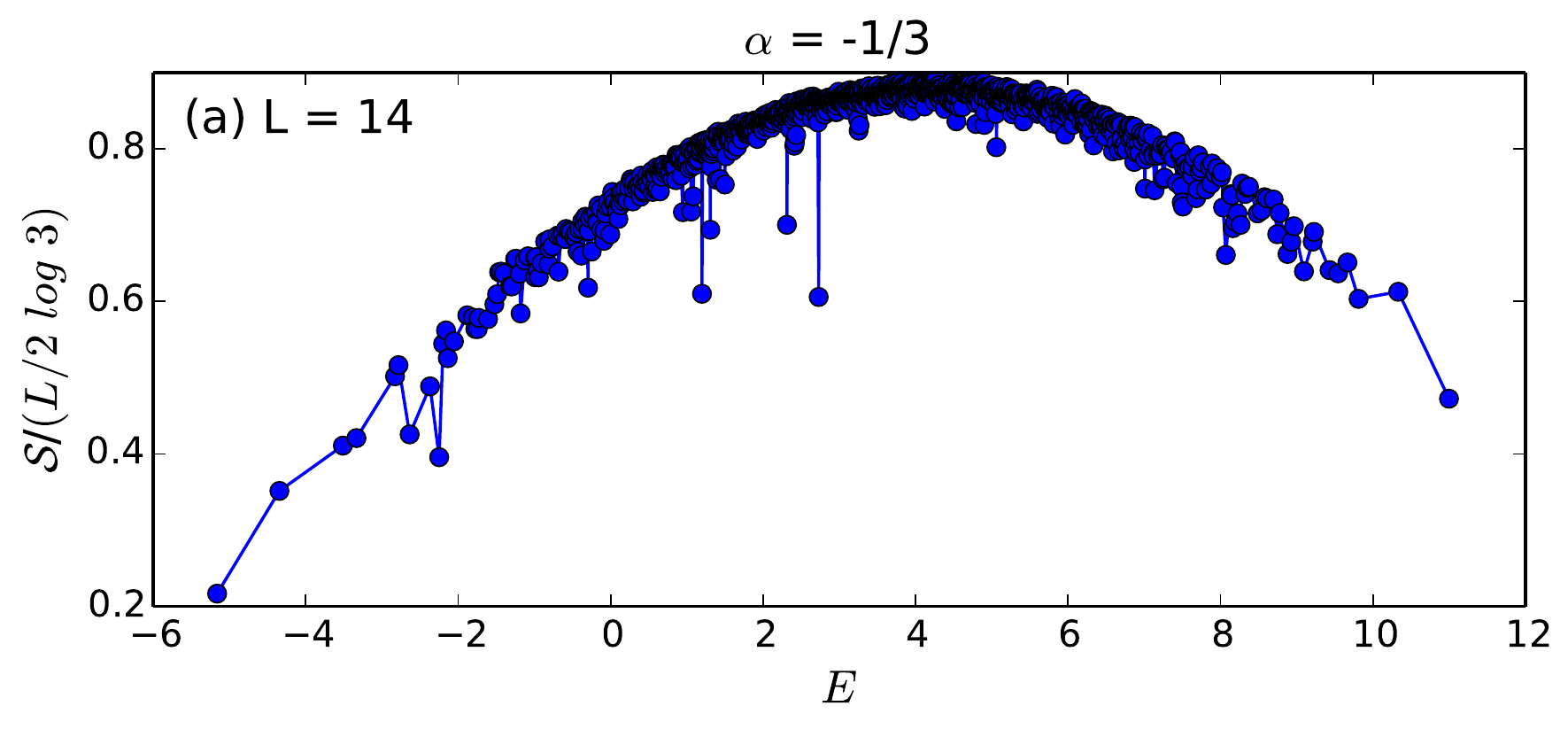}&\includegraphics[scale = 0.43]{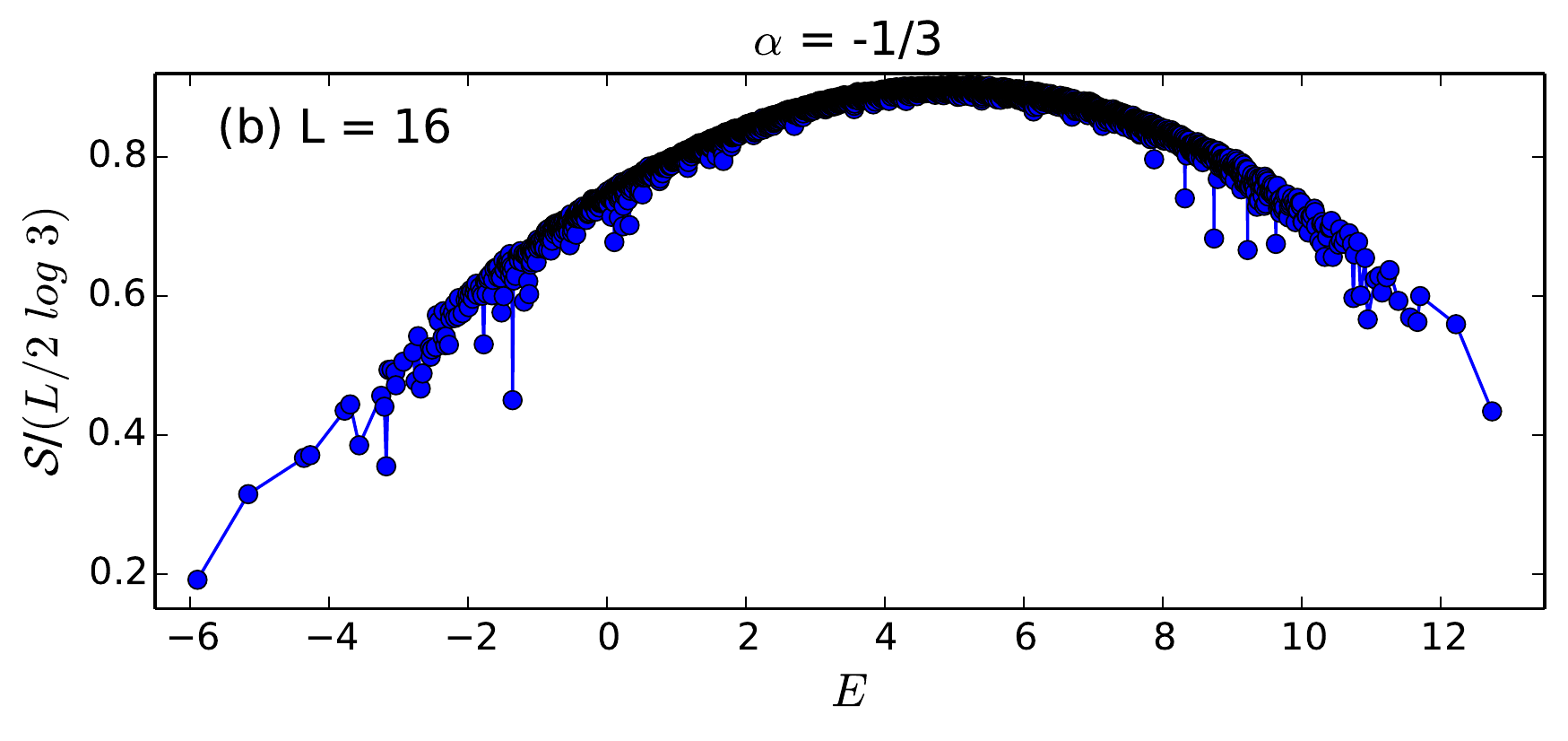}\\
\end{tabular}
\caption{Apparent atypical eigenstates in the spin-1 pure Heisenberg model (i.e. $\alpha = -1/3$ in Eq.~(\ref{hamilpert})) in the quantum number sector $(s, S_z, k, I, P_z) = (0, 0, 0, +1, +1)$ for system sizes (a) $L = 14$ and (b) $L = 16$. We use the same conventions as Fig.~\ref{fig:awayaklt} for the axis labels. Looking at the evolution between $L = 14$ and $L = 16$ might suggest that those atypical states are finite-size artifacts. \label{fig:heisenbergpt}}
\end{figure*}

One might wonder if such a violation of ETH is generic in nature, i.e., if these states have a sub-thermal entanglement entropy even when the Hamiltonian is perturbed away from the AKLT point.
We explore this using the Hamiltonian
\begin{equation}
    H_{\alpha} = \sum_{i = 1}^{L}{\left(\frac{1}{3} + \frac{1}{2}\vec{S}_i\cdot\vec{S}_{i+1} + \left(\frac{1}{6} + \frac{\alpha}{2}\right) (\vec{S}_i\cdot\vec{S}_{i+1})^2\right)}
\label{hamilpert}
\end{equation}
where $\alpha = 0$ corresponds to the Hamiltonian of the AKLT model.
We find that the dip in the entanglement entropy is stable up to a value of $\alpha = -0.025$, as shown in Figs.~\ref{fig:awayaklt}c and \ref{fig:awayaklt}d. 
However, we cannot exclude that the range of $\alpha$ where we observe this low entanglement in the bulk of spectrum, will go to zero in the thermodynamic limit (as observed for $\alpha = -0.05$ in Figs.~\ref{fig:awayaklt}e and \ref{fig:awayaklt}f).
Finally, we draw attention to the existence of apparently atypical states in the (non-integrable) spin-1 Heisenberg model, shown in Figs.~\ref{fig:heisenbergpt}a and \ref{fig:heisenbergpt}b that could be an artifact of the finite system size. 
Since the number of states that belong to the tower of states grows only polynomially in $L$, the set of ETH violating states has a measure zero.
Thus, the existence of these sub-thermal states do not preclude the weak ETH, which states \emph{almost all} eigenstates in a region of finite density of states look thermal.

\section{Degeneracies in the entanglement Spectra of Excited States}\label{sec:SPTorderMPO}

We now move on to describe the constraints that the AKLT Hamiltonian symmetries on the entanglement spectra of the exact excited states.
\subsection{Symmetries of MPS and symmetry protected topological phases}\label{sec:SPTorder}
We first briefly review the action of symmetries on an MPS, the concept of Symmetry Protected Topological (SPT) phases in 1D, and their connections to degeneracies in the entanglement spectrum.\cite{pollmann2010entanglement, pollmann2012symmetry, pollmann2012detection, chen2011classification} A state $\ket{\psi}$ that is invariant under any symmetry (that has a local action on an MPS) admits an MPS representation that transforms under the particular symmetry as\cite{perez2008string, cirac2011entanglement,pollmann2010entanglement,pollmann2012detection} 
\begin{equation}
    u({A^{[m]}}) = e^{i\theta} U A^{[m]} U^\dagger,
\label{sptsym}
\end{equation}
where $u$ is the symmetry operator that transforms the MPS, $U$ is a unitary matrix that acts on the ancilla, and $e^{i \theta}$ is an arbitrary phase.  We now discuss various useful symmetries that are relevant to the AKLT models.

Since the inversion symmetry flips the chain of length $L$ (and hence the MPS representation of the state) by interchanging sites $i$ and $L - i$, the site-independent MPS of the transformed state is the same as the MPS of the original state read from right to left in Eq.~(\ref{generalOBCMPS}). Consequently, the site-independent MPS transforms under inversion as\cite{pollmann2010entanglement} 
\begin{equation}
    u_I(A^{[m]}) = ({A^{[m]}})^T = e^{i\theta_I} U_I^\dagger A^{[m]} U_I
\label{inversiontransform}
\end{equation}
In Ref.~[\onlinecite{pollmann2010entanglement}], it was shown that for an MPS $A$ in the canonical form, the $U_I$ matrices should satisfy $U_I U_I^\ast = \pm \mathds{1}$. 
As shown in App.~\ref{sec:InvSym}, each level of the entanglement spectrum has a degeneracy that should be a multiple of two. The origin of the degeneracy can be traced back to the existence of symmetry protected edge modes at the ends of the chain and the SPT phase.

Time-reversal, by virtue of being an anti-unitary operation, acts on the MPS as 
\begin{equation}
    u_T(A^{[m]}) = \sum_n{\mathcal{T}_{mn} A^{[n]}} = e^{i\theta_T} U_T^\dagger A^{[m]} U_T 
\label{timetransform}
\end{equation}
where $\mathcal{T}_{mn} = \left(e^{i\pi S^y_p}\right)_{mn}\mathcal{K}$, where $\mathcal{K}$ is the complex conjugation operator and $S^y_p$ acts on the physical index. The two classes of $U_T$ matrices are again $U_T U_T^\ast = \pm \mathds{1}$, with $U_T U_T^\ast = -\mathds{1}$ indicating an SPT phase.\cite{pollmann2010entanglement}

In the case of $\mathbb{Z}_2 \times \mathbb{Z}_2$ spin-rotation symmetries ($\pi$-rotations about $x$ and $z$ axes), the MPS transforms under the symmetries as
\begin{equation}
    u_\sigma(A^{[m]}) = \sum_n{{\mathcal{R}_\sigma}_{mn} A^{[n]}} = e^{i\theta_\sigma} U_\sigma^\dagger A^{[m]} U_\sigma 
\label{rotationtransform}
\end{equation}
where ${\mathcal{R}_\sigma}_{mn} = \left(e^{i \pi S^\sigma_p}\right)_{mn}$, $\sigma = x,z$ and $S^\sigma_p$ acts on the physical index. The two classes of $U_\sigma$ are the ones that satisfy $U_x U_z U_x^\dagger U_z^\dagger = \pm \mathds{1}$, where $U_x U_x^\ast = U_z U_z^\ast = \mathds{1}$. Thus the classes of matrices can be written as $(U_x U_z) (U_x U_z)^\ast = \pm \mathds{1}$.

In each of the cases above, we refer to the transformations with positive and negative signs as linear and projective transformations respectively. Since the conditions of SPT order for the symmetry groups are of the form $U U^\ast = -\mathds{1}$, where $U$ is unitary, $U$ should be $\chi \times \chi$ anti-symmetric matrix.  If $\chi$ is odd, 0 is an eigenvalue of $U$, contradicting the fact that $U$ is unitary. Thus, protected degeneracies cannot exist due to the symmetries we have discussed if the bond dimension of the MPS representation in the canonical form is odd.

The spin-1 AKLT ground state MPS Eq.~(\ref{AKLTMPS}) satisfies Eqs.~(\ref{inversiontransform}), (\ref{timetransform}) and (\ref{rotationtransform}) with $U_I = U_T = i\sigma_y$, $U_x = \sigma_x$ and $U_z = \sigma_z$. Thus the entanglement spectrum of the spin-1 AKLT ground state is degenerate. This analysis can be extended straight forwardly to a spin-$S$ AKLT model groundstate. Since even $S$ AKLT ground states have an odd bond dimension, they do not have SPT order nor a doubly degenerate entanglement spectrum. For odd $S$, the operators $U_I = U_T = e^{i \pi S^y_a}$,  $U_x = e^{i \pi S^x_a}$ and $U_z = e^{i\pi S^z_a}$, where $S^\sigma_a$ ($\sigma = x,y,z$) are spin-$S/2$ operators that act on the ancilla, satisfy Eqs.~(\ref{inversiontransform}), (\ref{timetransform}) and (\ref{rotationtransform}) respectively. Since these matrices satisfy 
\begin{equation}
    U_I U_I^\ast = U_T U_T^\ast = (U_x U_z) (U_x U_z)^\ast = (-1)^S\mathds{1},
\label{gslinearproj}
\end{equation}
all odd-$S$ AKLT chains have SPT order and a doubly degenerate entanglement spectrum whereas even-$S$ chains do not.

\subsection{Symmetries of MPO}\label{sec:symmetryMPO}
For any Hamiltonian that is invariant under certain symmetries, each of eigenstates are labelled by quantum numbers corresponding to a maximal set of commuting symmetries.
As shown in the previous section, the AKLT ground states are invariant under inversion, time-reversal, and $\mathbb{Z}_2 \times \mathbb{Z}_2$ rotation symmetries. 
However, some of the excited states we consider are not invariant under the said symmetries. 
For example, the tower of states we have consider have $S_z \neq 0$, and are not invariant under time-reversal or $\mathbb{Z}_2 \times \mathbb{Z}_2$ symmetries but they are invariant under inversion symmetry.
When an excited state is invariant under a certain symmetry, it can trivially be expressed in terms of an operator invariant under the same symmetry acting on the ground state.
Thus, analogous to Eq.~(\ref{sptsym}), under a symmetry $u$, the MPO of such an operator should transform as
\begin{equation}
    u(M^{[m n]}) = e^{i\theta} \Sigma^\dagger M^{[m n]} \Sigma.
\label{mposymmetry}
\end{equation}
where $u$ acts on the physical indices of the MPO and $\Sigma$ on the ancilla.
We first discuss the symmetries that we discussed with regard to MPS in Sec.~\ref{sec:SPTorder}, i.e., inversion, time-reversal and $\mathbb{Z}_2 \times \mathbb{Z}_2$ rotation.
The actions of these symmetries on an MPO are similar to the actions on the MPS. 
Inversion symmetry interchanges the operators acting on sites $i$ and $L - i$. 
Thus, similar to Eq.~(\ref{inversiontransform}), we obtain
\begin{equation}
    u_I(M^{[m n]}) = ({M^{[m n]}})^T = e^{i\theta_I} \Sigma_I^\dagger M^{[m n]} \Sigma_I
\label{mpoinversiontransform}
\end{equation}
Time-reversal and $\mathbb{Z}_2 \times \mathbb{Z}_2$ rotation symmetries act on the physical indices of the MPO via conjugation as 
\begin{equation}
    u_T(M^{[m n]}) = \sum_{l,k}{\mathcal{T}_{ml} M^{[l k]} \mathcal{T}^{\dagger}_{k n}} = e^{i\theta_T} \Sigma_T^\dagger M^{[m n]} \Sigma_T 
\label{mpotimetransform}
\end{equation}
and 
\begin{equation}
    u_\sigma(M^{[m n]}) = \sum_{l,k}{{\mathcal{R}_\sigma}_{ml} M^{[l k]} {\mathcal{R}^\dagger_\sigma}_{k n}} = e^{i\theta_\sigma} \Sigma_\sigma^\dagger M^{[m n]} \Sigma_\sigma 
\label{mporotationtransform}
\end{equation}
where $\mathcal{T}_{mn} = \left(e^{i\pi S^y_p}\right)_{m,n} \mathcal{K}$, ${\mathcal{R}_\sigma}_{mn} = \left(e^{i \pi S^\sigma_p}\right)_{mn}$, $\sigma = x,z$ act on the physical index of the MPO. 
In each of these cases, the auxiliary indices of the MPO transform under the $\Sigma_I$, $\Sigma_T$, $\Sigma_x$, $\Sigma_z$ matrices under the various symmetries respectively.
Similar to the case of an MPS, we could have MPOs that transform in two distinct ways $\Sigma_I \Sigma_I^\ast = \pm \mathds{1}$, $\Sigma_T \Sigma_T^\ast = \pm \mathds{1}$ and $(\Sigma_x \Sigma_z) (\Sigma_x \Sigma_z)^\ast = \pm\mathds{1}$. We refer to the transformation with the positive and negative signs as linear and projective MPO transformations respectively. 
Thus, under physical symmetries, if an MPS transforms on the ancilla under $U$, and an MPO transforms under $\Sigma$, the MPO $\times$ MPS transforms on the ancilla under $U \otimes \Sigma$.
As a consequence, if an MPO transforms projectively (resp. linearly), the MPO $\times$ MPS transforms in a different (resp. the same) way as the MPS. 
For example,  the MPO corresponding to the Arovas A operator (see Eq.~(\ref{ArovasA})) transforms linearly under inversion, time-reversal and $\mathbb{Z}_2 \times \mathbb{Z}_2$ rotation symmetries, and the transformation matrices are shown in Eqs.~(\ref{firstarovasinv}), (\ref{firstarovastr}) and (\ref{firstarovasrot}) respectively. The Arovas B operator (see Eq.~(\ref{secondarovasmpo})) transforms projectively under inversion, linearly under time-reversal and rotation symmetries, and the transformation matrices are shown in Eqs.~(\ref{secondarovasinv}), (\ref{secondarovastr}) and (\ref{secondarovasrot}) respectively. The tower of states operator transforms projectively and linearly under inversion symmetry for odd and even $N$ respectively, with the transformation matrices shown in Eq.~(\ref{towerinv}). The transformation matrices are shown in App.~\ref{sec:symtransform}. Note that we do not claim any topological protection of these states. Indeed, they have a degenerate largest eigenvalue of the transfer matrix, leading to long-range correlations that do not decay exponentially. 
We discuss the implications of these transformations to the excited state entanglement spectrum in the next section using concrete examples from the AKLT models.
%

%



%

\section{Finite-size Effects in the Entanglement Spectra of Excited states}\label{sec:finitesize}
We proceed to describe the finite-size effects and symmetry-protected degeneracies in the entanglement spectra of the exact excited states of the AKLT models. Since the exact entanglement spectra depend on the configuration of the free boundary spins, we \emph{freeze them to their highest weight states}. Such a boundary configuration is inversion symmetric, although it violates time-reversal and $\mathbb{Z}_2 \times \mathbb{Z}_2$ rotation symmetries (on the edges only).
\subsection{Spin-$\boldsymbol{S}$ AKLT ground states}
As described in Sec.~\ref{sec:SPTorder}, the entanglement spectrum of the spin-$S$ AKLT ground state consists of $(S+1)$ degenerate levels in the thermodynamic limit.
Generically, such a degeneracy between $(S+1)$ levels is broken for a finite system. 
However, as shown in the thermodynamic limit in Ref.~[\onlinecite{pollmann2010entanglement}] and for a finite system in App.~\ref{sec:InvSym}, the entanglement spectrum is always doubly degenerate when symmetries act projectively.
Thus, for odd $S$, since inversion, time-reversal and $\mathbb{Z}_2 \times \mathbb{Z}_2$ act projectively (see Eq.~(\ref{gslinearproj})), the entanglement spectrum consists of $(S+1)/2$ exactly degenerate doublets. 
For even-$S$, the entanglement spectrum need not consist of degenerate levels for generic configurations of boundary spins, though some levels can be degenerate for particular choices of the boundary spins.
While the exact form of the splitting between the entanglement spectrum levels depends on the configuration of the boundary spins, we find that it is exponentially small in the system size.

\subsection{Spin-1 AKLT tower of states}
We first describe the entanglement spectrum of the spin-2 magnon state of spin-1 AKLT model, $\ket{S_2}$.
In Sec.~\ref{sec:SMA}, we have seen that the entanglement spectrum of such a state consists of two copies of the ground state entanglement spectrum.
For a finite $n$, using an explicit computation of $\rhored$ using the methods described in Sec.~\ref{sec:ESMPOMPS} and illustrated in App.~\ref{sec:exactexample}, with MPS boundary vectors of Eq.~(\ref{AKLTboundaryvectors}), the four normalized eigenvalues of $\rhored$ read (see Eq.~(\ref{rhoredexact})) 
\begin{equation}
    2 \times \left(\frac{n}{4n - 3}, \frac{3 - 2n}{6 - 8n}\right),
\label{spin2magnones}
\end{equation}
where $2 \times$ indicates two copies.
In Eq.~(\ref{spin2magnones}), we have ignored exponentially small finite-size splitting to obtain a closed form expression.
The two degenerate copies of the entanglement spectrum thus split into two doublets that have an $\mathcal{O}(1/n)$ (power-law) splitting.
Similarly, the six eigenvalues of $\rhored$ for $\ket{S_{4}}$ read 
\begin{equation}
    2 \times \left(\frac{4 n^2 - 22 n + 27}{32n^2 - 88n + 54}, \frac{2n^2 - 5n}{16 n^2 - 44n + 27}, \frac{4n^2 - 6n}{16n^2 - 44n + 27} \right) 
\end{equation}
This is consistent with the $n \rightarrow \infty$ behavior calculated in Sec.~\ref{sec:zerodensity}, i.e. the entanglement spectrum is composed of three copies of the ground state split into three doublets, two of which are degenerate in the thermodynamic limit at half the entanglement energy of the other.
The doublets that are degenerate in the thermodynamic limit have an $\mathcal{O}(1/n)$ finite-size splitting between them.
More generically, we observe the following pattern in the entanglement spectrum of $\ket{S_{2N}}$.
The $(N+1)$ copies of the ground state split into $(N+1)$ doublets, some of which are separated by $\mathcal{O}(1/N)$ in the thermodynamic limit. 
The pairs of doublets that are degenerate in the thermodynamic limit have a power-law finite size splitting of $\mathcal{O}(1/n)$.
A schematic plot of the entanglement spectra of the tower of states is shown in Fig.~\ref{fig:spin1estower}.
We now distinguish between exact degeneracies and exponential finite-size splittings. 
As shown in Sec.~\ref{sec:symmetryMPO}, the MPO for the tower of states transforms projectively (resp. linearly) under inversion symmetry if $N$ is odd (resp. even).
Since the spin-1 AKLT ground state transforms projectively under inversion, the MPO $\times$ MPS transforms projectively (resp. linearly) under inversion symmetry if $N$ is even (resp. odd). 
While the proof for double degeneracy due to projective representations in Ref.~[\onlinecite{pollmann2010entanglement}] relied on the uniqueness of the largest eigenvalue of the transfer matrix of the MPS, in App.~\ref{sec:InvSym} we show the existence of the degeneracy in the mid-cut entanglement spectrum for a finite system irrespective of the structure of the transfer matrix.
Consequently, we observe exact degeneracies of the doublets for even $N$ and exponential finite-size splittings within the doublets for odd $N$. This effect is schematically shown in Fig.~\ref{fig:spin1estower}.
The exponential splitting happens for generic symmetry-preserving configurations of the boundary spins, though certain configurations of boundary spins lead to ``accidental" degeneracies in the entanglement spectrum.

\begin{figure*}[ht]
\hspace{-0.9cm}
\begin{tabular}{cc}
\includegraphics[scale=0.32]{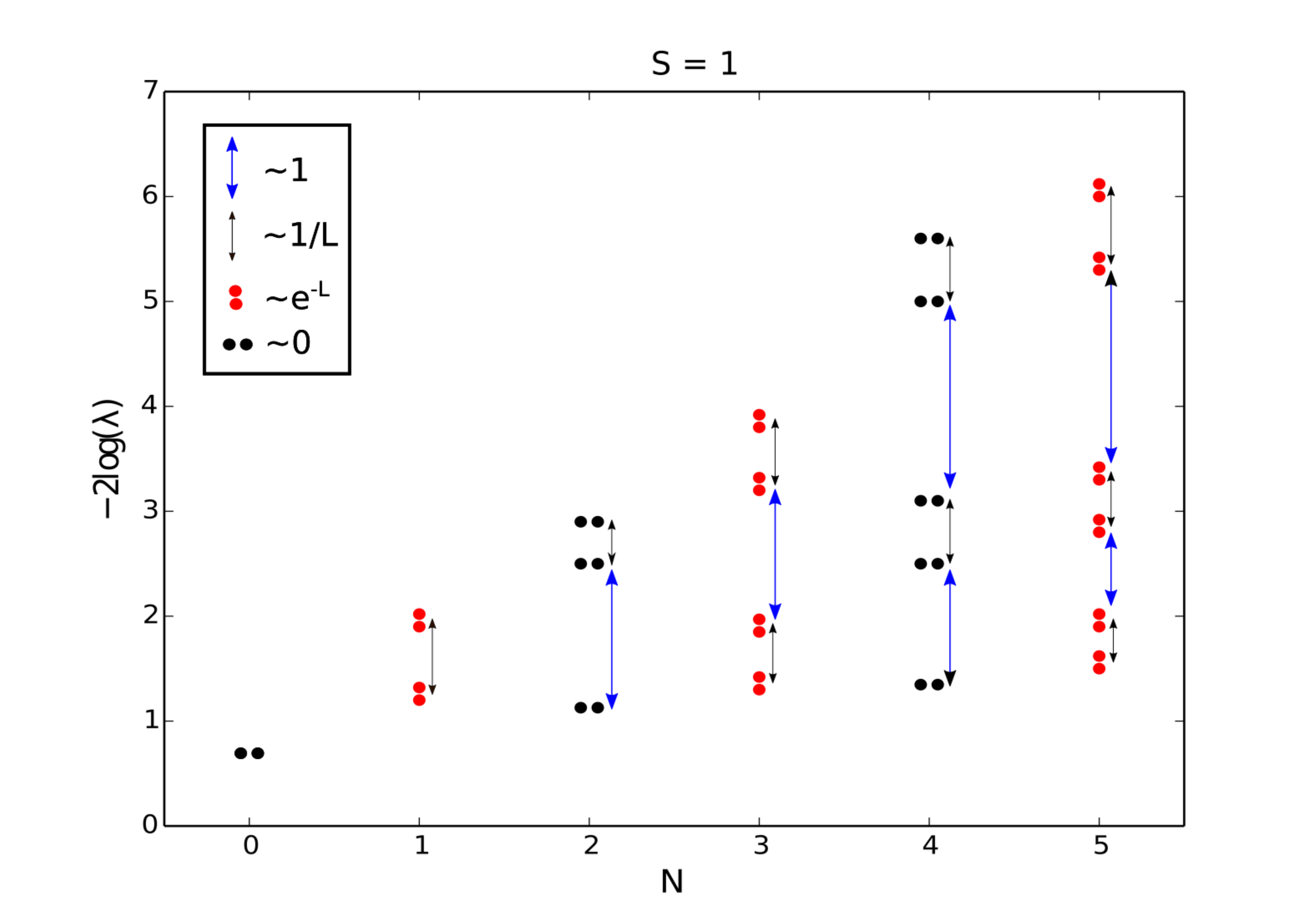}&\includegraphics[scale=0.32]{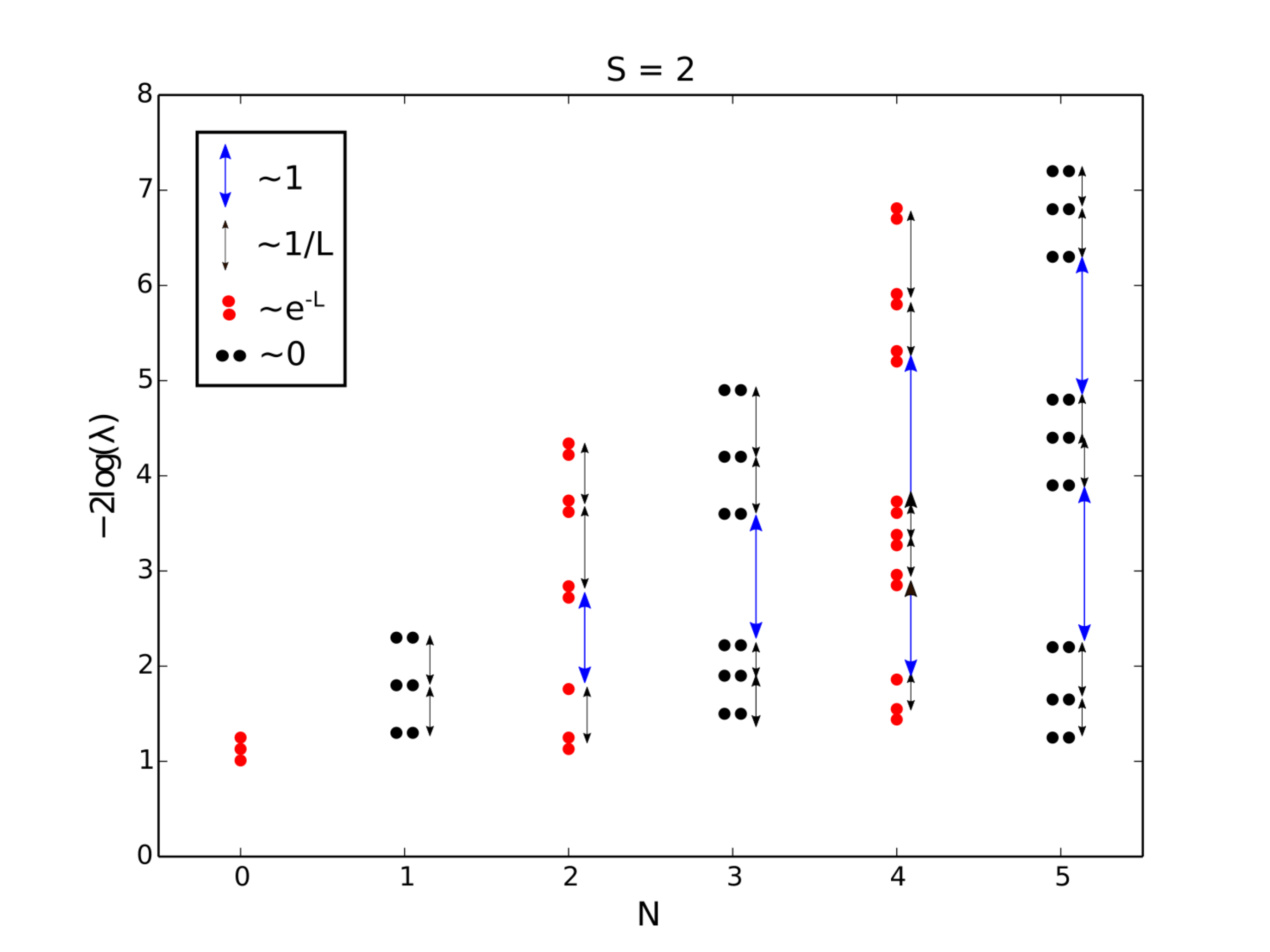}
\end{tabular}
\caption{(Color online) Schematic depiction of the entanglement spectra of spin-1 AKLT tower of states $\{\ket{S_{2N}}\}$ (left) and spin-2 AKLT tower of states $\{\ket{2S_{2N}}\}$ (right). The almost-degenerate levels shown in red have an exponential finite-size splitting whereas the black doublets are exactly degenerate. Power-law finite-size splittings are depicted by black two-headed arrows and constants by blue two-headed arrows. \label{fig:spin1estower}}
\end{figure*}

\subsection{Spin-$\boldsymbol{S}$ AKLT tower of states}
Similar to the spin-1 AKLT tower of states, we compute the exact entanglement spectra for the spin-$S$ tower of states of Ref.~[\onlinecite{self}].
We start with $S = 2$. The spectrum of $\rhored$ for the state $\ket{2S_2}$ (obtained via a direct computation similar to the one described in App.~\ref{sec:exactexample}) has six eigenvalues that read
\begin{equation}
    2 \times \left(\frac{9n + 28}{84 + 54 n}, \frac{9n + 4}{84 + 54n}, \frac{9n + 10}{84 + 54n}\right).
\end{equation}
Similar to the spin-1 case, we note that the two copies of the ground state entanglement spectrum split into three doublets that are separated by an $\mathcal{O}(1/n)$ finite-size splitting.
For the state $\ket{2S_4}$, the eigenvalues of $\rhored$ read (ignoring exponentially small splitting)
\begin{eqnarray}
    &2 \times \left(\frac{27 n^2 + 117 n - 80}{6(54n^2 + 117 n - 40)}, \frac{27 n^2 - 27 n - 104}{6(54n^2 + 117 n - 40)}\right)\nn \\
    &2 \times \left(\frac{27n^2 + 9n - 128}{6(54n^2 + 117 n - 40)}, \frac{(9n + 28)(9n+10)}{9(54n^2 + 117 n - 40)}\right) \nn \\
    &1 \times \left(\frac{(9n + 4)^2}{9(54n^2 + 117 n - 40)}\right)
\end{eqnarray}
Thus, we find that the nine levels due to the three copies of the ground state entanglement spectrum split into four doublets and one singlet. 
Two of the copies of the ground state entanglement spectrum are degenerate in the thermodynamic limit, and at a finite size, these six entanglement levels split into three doublets that have an $\mathcal{O}(1/n)$ splitting.
We numerically observe that a similar pattern holds true for arbitrary $S$. 
For the state $\ket{SS_{2N}}$, the $(N+1)$ copies of the ground state entanglement spectrum (that consists of $(S+1)$ levels) splits into doublets and singlets.
If $S$ is odd, we obtain $(S+1)/2$ doublets and if $S$ is even, we obtain $S/2$ doublets and one singlet. The doublets and singlets that are degenerate in the thermodynamic limit have an $\mathcal{O}(1/n)$ finite-size splitting.
Furthermore, as shown in Sec.~\ref{sec:symmetryMPO}, the MPO for the tower of states transforms projectively (resp. linearly) under inversion symmetry if $N$ is odd (resp. even).
Consequently, using Eqs.~(\ref{gslinearproj}) and (\ref{towerlinearproj}), the MPO $\times$ MPS transforms projectively (resp. linearly) under inversion symmetry if $(N+S)$ is odd (resp. even).
Indeed, similar to the spin-1 AKLT tower of states, we find exactly degenerate doublets in the entanglement spectrum for arbitrary symmetry-preserving boundary conditions when the MPO $\times$ MPS transforms projectively (i.e. when $(N+S)$ is odd). 
If $(N + S)$ is even, we find that for generic symmetry-preserving boundary conditions, we obtain an exponential finite-size splitting between the doublets that are degenerate in the thermodynamic limit.

\subsection{Spin-1 Arovas states}
For the spin-1 Arovas A state, via a direct computation similar to the example in App.~\ref{sec:exactexample}, we find that the eight eigenvalues of $\rhored$ read 
\begin{eqnarray}
    &2 \times \left(\frac{n + 3 + 2 \sqrt{2(1+n)}}{4n + 14}, \frac{n + 3 - 2 \sqrt{2(1+n)}}{4n + 14}\right) \nn \\
    &4 \times \left(\frac{1}{8n + 28}\right).
\end{eqnarray}
Thus, similar to the spin-2 magnon, we obtain two copies of the ground state entanglement spectrum that splits into two doublets that have an $\mathcal{O}(1/\sqrt{n})$ splitting between them. 
In addition, we obtain 4 entanglement levels that are of $\mathcal{O}(1/n)$.
As mentioned in Sec.~\ref{sec:symmetryMPO},  the Arovas A MPO transforms linearly under inversion, time-reversal and $\mathbb{Z}_2 \times \mathbb{Z}_2$ symmetries. Consequently, the MPO $\times$ MPS transforms projectively and all the doublets are exactly degenerate for a finite system.
While we were not able to obtain a closed-form expression for the entanglement spectra of the spin-1 and spin-2 Arovas B states,\cite{self} we numerically observe similar phenomenology as the Arovas A and the spin-2$S$ magnon entanglement spectra, although the magnitude of the finite-size splittings ($\mathcal{O}(1/\sqrt{n})$ versus $\mathcal{O}(1/n)$) are not clear. 

\section{Conclusion}\label{sec:conclusions}
We have computed the entanglement spectra of the exact excited states of the AKLT models that were derived in Ref.~[\onlinecite{self}].
To achieve this, we expressed the states as MPO $\times$ MPS' and developed a general formalism to compute the entanglement spectra of states using the Jordan normal form of the MPO $\times$ MPS transfer matrix. 
We first exemplified our method by reproducing existing results on single-mode excitations: we show that their entanglement spectra in the thermodynamic limit consist of two copies of the ground state entanglement spectrum.
The low-lying exact excited states of the AKLT model such as the Arovas states and the spin-$2S$ magnon states for the spin-$S$ AKLT chain fall into this category. 
For single-mode excitations, our method is exactly equivalent to the tangent-space and related methods developed to numerically as well as analytically probe low-energy excited states in the MPS formalism.\cite{haegeman2013post, vanderstraeten2015excitations, vanderstraeten2015scattering,haegeman2017diagonalizing,haegeman2012variational, draxler2013particles, zauner2015transfer, haegeman2013elementary} 
We note that our method can be applied to obtain results on the entanglement spectra of single-mode excitations in the Fractional Quantum Hall Effect.\cite{girvin1986magneto, fqheinprep}
We then generalized our method to states with multiple magnons, that are beyond single-mode excitations (``double tangent space"\cite{haegeman2014geometry} and beyond).
This allowed us to obtain the exact expression for the entanglement spectra for the spin-$S$ AKLT tower of states for a zero density of magnons in the thermodynamic limit.
We showed that the entanglement spectrum of the $N$-th state of the tower consists of $(N + 1)$ copies of the ground state entanglement spectrum, not all degenerate.
Apart from the specific Jordan block structure derived for the special AKLT tower of states, our method to obtain the entanglement spectrum was completely general.
In particular, it applies to states of the form $\hat{O}^N \ket{\psi}$, where $\hat{O}$ is any translation invariant operator and $\ket{\psi}$ is a state that admits a site-independent MPS representation.
Moreover, since the entanglement entropy of the AKLT tower of states in Eq.~(\ref{entropylowk}) has a similar form as the entanglement entropy in equal-momentum quasiparticle excited states of free-field theories and certain integrable models,\cite{castro2018entanglement, castro2018entanglement2, alba2009entanglement, molter2014bound} it is likely that our formulae for the entanglement spectrum and entropy holds in more general integrable and non-integrable models for equal-momentum quasiparticle excited states in the zero density limit. 
We defer the exploration of equal and unequal momentum quasiparticle excited states using our formalism in a generic setting to future work. 
For the AKLT tower of states, we also showed that the replica structure of the entanglement spectra of the tower of states persists in the thermodynamic limit only for states at a zero energy density, conforming with folklore that only low energy excitations resemble the ground state.
An interesting problem is to prove this on general grounds for excited states in integrable and/or non-integrable models.
Moreover, since the exact excited states of the AKLT model have non-injective matrix-product expressions with finite bond dimensions, perhaps one could obtain a class of non-injective matrix-product states that describe excited states, similar to a classification of matrix-product ground states.\cite{sanz2009matrix, tu2010exact}

We also studied finite-size effects in the entanglement spectra of these states and showed a universal power-law splitting between the different copies of the ground state. 
We identified exact degeneracies and exponential splittings based on projective versus linear transformations of the MPO $\times$ MPS at a finite size.
While protected exact degeneracies in the entanglement spectrum of excited states are reminiscent of SPT phases for the ground states, it is unclear if these have a topological origin in the excited states, given that excited states do not have a protecting gap.
We emphasized that the states of the tower have an entanglement entropy that scales as $\ms \propto \log L$, which is incompatible with strong ETH, if these states indeed exist in the bulk of the energy spectrum.\cite{self} 
Further, we showed that the violation of ETH seems to persist for SU(2) symmetric spin-1 Hamiltonians slightly away from the AKLT point, and we pointed out numerically apparent low-entropy states in the pure Heisenberg model, far away from the AKLT point.
However, a systematic numerical study of these low-entropy states away from the AKLT point is necessitated, with and without breaking the SU(2) symmetry.
These special states, first obtained in Ref.~[\onlinecite{self}], provide analytically tractable examples of ``quantum many-body scars", described in Refs.~[\onlinecite{turner2018weak}] and [\onlinecite{turner2018quantum}].
While such anomalous eigenstates are known to exist in single-particle chaotic systems, very few examples are known in many-body quantum systems.\cite{berry1989quantum}
An interesting problem is to determine if these anomalous states play any interesting role in the dynamics of the AKLT models.\cite{turner2018weak, schecter2018many, bernien2017probing}

\section*{Acknowledgements}

We thank Yang-Le Wu for an initial collaboration, Stephan Rachel for collaboration on related work and David Huse, Shivaji Sondhi and Michael Zaletel for useful discussions. BAB wishes to thank Ecole Normale Superieure, UPMC Paris, and Donostia International Physics Center for their generous sabbatical hosting during some of the stages of this work. BAB acknowledges support for the analytic work Department of Energy de-sc0016239, Simons Investigator Award, the Packard Foundation, and the Schmidt Fund for Innovative Research.  The computational part of the Princeton work was performed under NSF EAGER grant DMR-1643312, ONR-N00014-14-1-0330,  NSF-MRSEC DMR-1420541.

\appendix

\section{Matrix Product States for spin-$S$ AKLT ground states}\label{sec:AKLTMPS}
In this section, we derive the Matrix Product State (MPS) representations and the structure of the transfer matrix for the spin-$S$ AKLT ground states with Open Boundary Conditions (OBC). We follow the derivation in Ref.~[\onlinecite{schollwock2011density}].
Similar expressions can be obtained by alternate methods in the literature. \cite{totsuka1995matrix, fannes1989exact, fannes1992finitely, klumper1993matrix, karimipour2008matrix, kolezhuk1997matrix} 
\subsection{MPS}
\begin{figure}[ht!]
\centering
\setlength{\unitlength}{1pt}
\begin{picture}(180,10)(0,0)
\linethickness{0.8pt}
\multiput(80,10)(35,0){2}{\circle{4}} 
\multiput(86,10)(35,0){2}{\circle{4}} 
\multiput(83,10)(35,0){2}{\circle{14}}
\put(85,-5){$v_i$}
\put(75,-5){$u_i$}
\put(105,-5){$u_{i+1}$}
\put(125,-5){$v_{i+1}$}
\end{picture}
\caption{Labelling in the MPS construction of the spin-$S$ AKLT ground state. Big and small circles represent physical spin-$S$ and virtual spin-$S/2$ degrees of freedom respectively.}
\label{fig:mpsconstruction}
\end{figure}
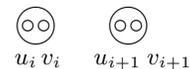
As mentioned in Sec.~\ref{sec:AKLTandMPS}, each spin-$S$ can be viewed as two symmetrized spin-$S/2$ bosons. The AKLT ground state is then a  product of spin-$S/2$ singlets, i.e. the $J = 0$ state formed by two spin-$S/2$ on nearest neighbor spin-$S$ (see Fig.~\ref{fig:spinSgroundstate}). We use the labels $u_i$ and $v_i$ to denote the $S_z$ values of the left and right spin-$S/2$ on site $i$ respectively (see Fig.~\ref{fig:mpsconstruction}).
Thus, the spin-$S/2$ singlet state $\ket{\spa 0 \spa 0\spa}^{\frac{S}{2}}_{i, i+1}$ formed between the spin-$S/2$'s $v_i$ and $u_{i+1}$ can be written in the $S_z$ basis of spin-$S/2$ (denoted by $\ket{v_i, u_{i+1}}_{\frac{S}{2}}$) as
\begin{eqnarray}
    \ket{0 \spa 0}^{S/2}_{i, i+1} &=& \sumal{\alpha = -\frac{S}{2}}{\frac{S}{2}}{{}_{\frac{S}{2}}\braket{\alpha, -\alpha}{\spa 0 \spa 0\spa}\ket{\alpha, -\alpha}_{i, i+1}} \nn \\
    &\equiv& \sumal{v_i, u_{i+1}}{}{\Theta^{\frac{S}{2}}_{v_i, u_{i+1}}\ket{v_i, u_{i+1}}}_{\frac{S}{2}}
\end{eqnarray}
where ${}_{s}\braket{s_1, s_2}{\spa J\spa J_z}$ is the Clebsch-Gordan coefficient for two spin-$s$ with $S_z = s_1$ and $S_z = s_2$ to form a state with total spin $J$ and $S_z = J_z = s_1 + s_2$. The matrix $\Theta$ thus assumes the form
\begin{equation}
    \Theta^{\frac{S}{2}}_{\alpha\beta} = {}_{\frac{S}{2}}\braket{\alpha,\beta}{\spa0\spa0\spa} \delta_{\alpha,-\beta}
\label{Sigmamatrix}
\end{equation}
where the indices $-S/2 \leq \alpha, \beta \leq S/2$.
For example, for $S = 1$ (the spin-1 AKLT ground state), we know that $v_i, u_{i+1} = \uparrow, \downarrow$ and the singlet $\ket{\spa 0 \spa 0\spa}^{\frac{1}{2}}_{i,i+1}$ can be written as
\begin{equation}
    \ket{\spa 0 \spa 0\spa}^{\frac{1}{2}}_{i,i+1} = \frac{\ket{v_i = \uparrow, u_{i+1} = \downarrow} - \ket{v_i = \downarrow, u_{i+1} = \uparrow}}{\sqrt{2}}.
\end{equation}
For $S = 1$, the matrix $\Theta^{\frac{1}{2}}$ thus reads
\begin{equation}
    \Theta^{\frac{1}{2}} = 
    \begin{pmatrix}
        0 & \frac{1}{\sqrt{2}} \\
        -\frac{1}{\sqrt{2}} & 0
    \end{pmatrix}.
\label{spin1thetamatrix}
\end{equation}
In terms of these matrices, the spin-$S$ AKLT ground state $\ket{SG}_{\frac{S}{2}}$ in the spin-$\frac{S}{2}$ basis with OBC and the edge spins both having $S_z = S/2$ (denoted by $\ket{S/2}_1$ and $\ket{S/2}_L$) reads
\begin{eqnarray}
    \ket{SG}_{\frac{S}{2}} &=& \ket{S/2}_1\prodal{i = 1}{L - 1}{\ket{\spa 0\spa 0\spa}^{\frac{S}{2}}_{i, i+1}}\ket{S/2}_L \nn \\
    &=& \sum_{\{u_i, v_i\}}{\delta_{u_1,\frac{S}{2}}\Theta^{\frac{S}{2}}_{v_1 u_2}\dots\Theta^{\frac{S}{2}}_{v_{L-1} u_L}\delta_{v_L,\frac{S}{2}}\ket{\{u_i,v_i\}}_{\frac{S}{2}}} \nn \\[-2mm]
\end{eqnarray}
where $\ket{\{u_i,v_i\}}_{\frac{S}{2}} = \ket{u_1, v_1, \dots, u_L, v_L}_{\frac{S}{2}}$. The ground state can be written in the onsite spin-$S$ basis using a projector $P^{(S, \frac{S}{2})}_i$ to symmetrize the two spin-$S/2$ on each site, where the projector reads
\begin{equation}
    P^{(S, \frac{S}{2})}_i = \sum_{m_i}{\sum_{u_i v_i}{M^{[m_i]}_{u_i,v_i} \ket{m_i}_S{}_{\frac{S}{2}}\bra{u_i,v_i}}}
\end{equation}
where $\ket{m_i}_S$ denotes the spin-$S$ state on site $i$ with $S_z = m_i$. The tensor $M$ assumes the form
\begin{equation}
    M^{[m]}_{\alpha \beta} = \braket{S\spa m}{\alpha, \beta}_{\frac{S}{2}} \delta_{m, \alpha + \beta}.
\label{Mmatrix}
\end{equation}
For example, for $S = 1$, the projector $P_i^{(1, \frac{1}{2})}$ reads
\begin{eqnarray}
    &&P_i^{(1, \frac{1}{2})} = \ket{m_i = 1}_1 {}_{\frac{1}{2}}\bra{u_i = \uparrow, v_i =  \uparrow} \nn \\
    &&+ \ket{m_i = 0}_1\frac{ {}_{\frac{1}{2}}\bra{u_i = \uparrow, v_i =  \downarrow}
    + {}_{\frac{1}{2}}\bra{u_i = \downarrow, v_i =  \uparrow}}{\sqrt{2}} \nn \\
    &&+ \ket{m_i = -1}_1 {}_{\frac{1}{2}}\bra{u_i = \downarrow, v_i =  \downarrow}.
\end{eqnarray}
The matrices $M$ for $S = 1$ thus read
\begin{eqnarray}
    &M^{[1]} = 
    \begin{pmatrix}
        1 & 0 \\
        0 & 0
    \end{pmatrix} \;\;
    M^{[0]} = 
    \begin{pmatrix}
        0 & \frac{1}{\sqrt{2}} \\
        \frac{1}{\sqrt{2}} & 0
    \end{pmatrix} \;\; \nonumber \\
    &M^{[\mm]} =
    \begin{pmatrix}
        0 & 0 \\
        0 & 1
    \end{pmatrix}.
\label{spin1Mmatrices}
\end{eqnarray}

The projector on the full state, $\mathcal{P}^{(S, \frac{S}{2})} = \prod_i{P_i^{(S, \frac{S}{2})}}$ is then
\begin{eqnarray}
    \mathcal{P}^{(S, \frac{S}{2})} &=& \sum_{\{m_i\}}{\sum_{\{u_i v_i\}}{M^{[m_1]}_{u_1 v_1}\dots M^{[m_L]}_{u_L v_L}\ket{\{m_i\}}_S{}_{\frac{S}{2}}\bra{\{u_i,v_i\}}}} \nn \\[-2mm] 
\end{eqnarray}
where $\ket{\{m_i\}}_S = \ket{m_1,m_2,\dots,m_L}_S$. The ground state in the spin-$S$ basis $\ket{SG}_S = \mathcal{P}\ket{SG}_{\frac{S}{2}}$ reads
\begin{eqnarray}
    &\ket{SG}_S = \sumal{\{m_i\}}{}{\sumal{\{u_i, v_i\}}{}{\left(\delta_{u_1,S/2} M^{[m_1]}_{u_1 v_1}\Theta^{\frac{S}{2}}_{v_1 u_2} M^{[m_2]}_{u_2 v_2}\dots\right.}}\nn \\
    &\left.\dots\Theta^{\frac{S}{2}}_{v_{L-1}u_L}M^{[m_L]}_{u_L v_L}\delta_{v_L,S/2}\right)\ket{\{m_i\}}_S \nn \\
    &= \sumal{\{m_i, u_i\}}{}{{b^l_A}_{u_1} A^{[m_1]}_{u_1 u_2} \dots A^{[m_L]}_{u_{L-1} u_L} b^r_{u_L}}\ket{\{m_i\}}_S
\label{AKLTgeneralMPS}
\end{eqnarray}
where 
\begin{eqnarray}
    &&A^{[m]}_{u_i u_{i+1}} = \sumal{v_i = -\frac{S}{2}}{\frac{S}{2}}{M^{[m]}_{u_i v_i}\Theta^{\frac{S}{2}}_{v_i u_{i+1}}} \nn \\
    &&(b^l_A)_{u_1} = \delta_{S/2,u_1} \nn \\
    &&(b^r_A)_{u_L} = \sumal{v_L = -\frac{S}{2}}{\frac{S}{2}}{{(\Theta^{\frac{S}{2}})}^{\mm}_{u_L v_{L}}\delta_{v_L,S/2}} = {(\Theta^{\frac{S}{2}})}^{\mm}_{u_L,S/2}.\hspace{5mm} 
\label{Matricesandboundary}
\end{eqnarray}
Eq.~(\ref{AKLTgeneralMPS}) is the MPS representation of Eq.~(\ref{generalOBCMPS}) for the AKLT ground state. The matrices and boundary vectors of the MPS are defined in Eq.~(\ref{Matricesandboundary}).
The MPS tensors $A$ can be brought to a canonical form by ensuring that the largest eigenvalue of the transfer matrix Eq.~(\ref{genTransfer}) is $1$.
For example, using Eqs.~(\ref{spin1thetamatrix}) and (\ref{spin1Mmatrices}), the spin-1 AKLT matrices after normalization read
\begin{eqnarray}
    &A^{[1]} = \sqrt{\frac{2}{3}}
    \begin{pmatrix}
        0 & 1 \\
        0 & 0
    \end{pmatrix} \;\;
    A^{[0]} = \frac{1}{\sqrt{3}}
    \begin{pmatrix}
        -1 & 0 \\
        0 & 1
    \end{pmatrix} \;\; \nonumber \\
    &A^{[\mm]} = \sqrt{\frac{2}{3}}
    \begin{pmatrix}
        0 & 0 \\
        -1 & 0
    \end{pmatrix}.
\end{eqnarray}
The boundary vectors, up to an overall factor, read
 \begin{equation}
     b^l_A =
     \begin{pmatrix}
        1 \\
        0
    \end{pmatrix}\;\;
    b^r_A = 
    \begin{pmatrix}
        0 \\
        1
    \end{pmatrix}.\;\;
 \end{equation}

To further study the structure of the matrix $A$, it is convenient to re-label the indices of $\Theta^{\frac{S}{2}}$ and $M$ in Eqs.~(\ref{Sigmamatrix}) and (\ref{Mmatrix}) respectively to matrix indices as
\begin{eqnarray}
    \widetilde{\Theta^{\frac{S}{2}}}_{cd} &\equiv& \Theta_{\frac{S}{2} + 1 - c, \frac{S}{2} + 1 - d} \nn \\
    \widetilde{M}^{[m]}_{c d} &\equiv& M^{[m]}_{\frac{S}{2} + 1 - c, \frac{S}{2} + 1 - d}
\end{eqnarray}
such that the matrix indices satisfy $ 1 \leq c, d \leq S + 1$. The matrices $\widetilde{\Theta^{\frac{S}{2}}}$ and $\widetilde{M}$ then read
\begin{eqnarray}
    &\widetilde{\Theta^{\frac{S}{2}}}_{c d} = {}_{\frac{S}{2}}\braket{\frac{S}{2} + 1 - c,\spa \frac{S}{2} + 1 - d}{\spa0\spa0\spa} \nn \\
    & \widetilde{M}^{[m]}_{c d} = \braket{S\spa m}{\frac{S}{2}+1-c,\spa\frac{S}{2} + 1 - d}_{\frac{S}{2}}.
\label{Msigma}
\end{eqnarray}
From Eqs.~(\ref{Matricesandboundary}) and (\ref{Msigma}), the MPS tensor $A^{[m]}$ and the boundary vectors $b^l_A$ and $b^r_A$ can be computed to be
\begin{eqnarray}
    A^{[m]}_{c d} &= \braket{S\spa m}{\frac{S}{2}+1-c,\spa m - (\frac{S}{2} + 1 - d)}_{\frac{S}{2}} \label{spinSAKLTMPSform}\\
    &\times {}_{\frac{S}{2}}\braket{\frac{S}{2} + 1 - c,-(\frac{S}{2} + 1 - d)}{\spa 0\spa0\spa}\delta_{c - d, m} \nn \\
    &(b^l_A)_c = \delta_{1,c}, \;\; (b^r_A)_c = \delta_{S + 1,c}.
\label{spinSAKLTboundaryform}
\end{eqnarray}
where $c$ and $d$ are matrix indices. Thus there are $(2S + 1)$ $(S+1) \times (S+1)$ MPS matrices for the spin-$S$ AKLT model.   
\subsection{Transfer matrix}
We now derive the structure of the spin-$S$ AKLT transfer matrix. Denoting the expression for the MPS Eq.~(\ref{spinSAKLTMPSform}) (after rescaling the matrices such that the MPS is canonical, i.e. the transfer matrix has a largest eigenvalue $1$)  as
\begin{equation}
    A^{[m]}_{c d} = \kappa_{mcd} \delta_{c-d,m},
\label{akltform}
\end{equation}
the corresponding transfer matrix (Eq.~(\ref{AKLTTransfer})) reads
\begin{equation}
    E_{cd,ef} = \sum_{m = -S}^S{\kappa^\ast_{mcd}\kappa_{mef} \delta_{c-d,m} \delta_{e-f,m}}.
\label{AKLTTransferfull}
\end{equation}
We can group the indices $c,e$ (left ancilla) into a single index $x$ and the indices $d,f$ (right ancilla) into $y$, as
\begin{equation}
    x = (c - 1)(S + 1) + e, \;\;\; y = (d - 1)(S+1) + f
\end{equation}
where $1 \leq x,y \leq (S+1)^2$.
In terms of $x$ and $y$, the transfer matrix reads
\begin{equation}
    E_{xy} = \sum_{m = -S}^S{\gamma_{mxy} \delta_{x, y+ m(S + 2)}}
\label{AKLTblocktransfermatrix}
\end{equation}
where $\gamma_{mxy} = \kappa^\ast_{mcd} \kappa_{mef}$.
Using Eq.~(\ref{spinSAKLTMPSform}), $\kappa_{mcd} = \kappa_{mdc}$, and thus $A^{[m]}$ is symmetric under the exchange of ancilla. Hence the transfer matrix $E_{xy}$ is also symmetric.
For example, the spin-1 AKLT transfer matrix (after grouping the ancilla) reads
 \begin{equation}
     E = 
     \begin{pmatrix}
        \frac{1}{3} & 0 & 0 & \frac{2}{3} \\
        0 & -\frac{1}{3} & 0 & 0 \\
        0 & 0 & -\frac{1}{3} & 0 \\
        \frac{2}{3} & 0 & 0 & \frac{1}{3}
    \end{pmatrix}.
\label{AKLTTransferagain}
 \end{equation}

Moreover, since $E_{xy}$ is non-zero in Eq.~(\ref{AKLTblocktransfermatrix}) only when $x \mod (S+2) = y \mod (S+2)$, the transfer matrix is block-diagonal with blocks $E_p$ formed by the following set of indices: 
\begin{equation}
    \{x,y\spa |\spa x \mod (S+2) =  y \mod (S+2) = p+1 \}.
\label{indices}
\end{equation} 
That is, the transfer matrix $E$ in Eq.~(\ref{AKLTblocktransfermatrix}) has a direct sum structure 
\begin{equation}
    E = E_0 \oplus E_1 \oplus \dots \oplus E_{S+1},
\label{Transferblockdiagonal}
\end{equation}
where $E_p$ is a block with dimension $\left\lfloor{\frac{S^2 + 2S - p}{S + 2}}\right\rfloor + 1$. Consequently, $E_0$ is the largest block, with a dimension $(S + 1)$.
For example, in the transfer matrix of Eq.~(\ref{AKLTTransferagain}), the blocks $E_0$, $E_1$ and $E_2$ read
\begin{equation}
    E_0 = 
    \begin{pmatrix}
        \frac{1}{3} & \frac{2}{3} \\
        \frac{2}{3} & \frac{1}{3}
    \end{pmatrix},\;\;
    E_1 = 
    \begin{pmatrix}
        -\frac{1}{3}
    \end{pmatrix},\;\;
    E_2 = 
    \begin{pmatrix}
        -\frac{1}{3}
    \end{pmatrix}. 
\end{equation}
This block-diagonal structure of the transfer matrix imposes a constraint on the structure of its generalized eigenvectors.
In particular, for the largest block $E_0$, the eigenvalue equation for the transfer matrix (without the ancilla combined) of Eq.~(\ref{AKLTTransferfull}) reads
\begin{equation}
    \sum_{df}{E_{cd,ef} v_{\alpha df}\delta_{df}} = \lambda_\alpha v_{\alpha ce}\delta_{ce}. 
\end{equation}
Thus, the eigenvectors of the $E$ corresponding to the block $E_0$ are diagonal when viewed as $\chi \times \chi$ matrices.
In particular, since the MPS is in the canonical form, $\mathds{1}_{\chi \times \chi}$ is an eigenvector of $E$ corresponding to the eigenvalue of unit magnitude.
Thus, the largest eigenvalue belongs to the block $E_0$ with eigenvalue $1$. 

\section{MPO of the Arovas operators}\label{sec:ArovasMPO}
To represent the Arovas A and B MPOs compactly, we first define the notation
\begin{eqnarray}
&\overline{\boldsymbol{S}} = 
\begin{pmatrix}
\frac{S^+}{\sqrt{2}} & \frac{S^-}{\sqrt{2}} & S^z
\end{pmatrix} \nn \\
&\boldsymbol{S} = 
\begin{pmatrix}
\frac{S^-}{\sqrt{2}} & \frac{S^+}{\sqrt{2}} & S^z
\end{pmatrix}^T. 
\label{SbarSTapp}
\end{eqnarray}
Using Eq.~(\ref{SbarSTapp}), we first obtain 
\begin{equation}
    \vec{S}_j\cdot\vec{S}_{j+1} = \overline{\boldsymbol{S}}_j\boldsymbol{S}_j.
\label{S.S}
\end{equation}
Consequently, the MPO for the Arovas A operator of Eq.~(\ref{firstarovas}),
\begin{equation}
    \mathcal{O}_A = \sumal{j = 1}{L-1}{(-1)^j \overline{\boldsymbol{S}}_j\boldsymbol{S}_{j+1}}   
\end{equation}
reads
\begin{equation}
    M_A = 
    \begin{pmatrix}
        -\mathds{1} & -\overline{\boldsymbol{S}} & 0 \\
        0 & 0 & \boldsymbol{S} \\
        0 & 0 & \mathds{1},
    \end{pmatrix}
\label{ArovasAmpoapp}
\end{equation}
where $0$ denotes zero matrices of appropriate dimensions.
Using Eq.~(\ref{S.S}), the Arovas B operator of Eq.~(\ref{ArovasBop}) can be written as
\begin{eqnarray}
    &\hat{\mathcal{O}}_B = \sumal{j = 2}{L-1}{(-1)^j \{\overline{\boldsymbol{S}}_{j-1} \boldsymbol{S}_j, \overline{\boldsymbol{S}}_j \boldsymbol{S}_{j+1}\}} \nn \\
    &= \sumal{j = 2}{L-1}{(-1)^j\left( \overline{\boldsymbol{S}}_{j-1} \left(\boldsymbol{S}_j\otimes\overline{\boldsymbol{S}}_j\right) \boldsymbol{S}_{j+1}\right.} \nn \\
    &+{\left. \overline{\boldsymbol{S}}_{j-1} \left(\overline{\boldsymbol{S}}_{j} \otimes \boldsymbol{S}_j \right) \boldsymbol{S}_{j+1}\right)} \nn \\
    &= \sumal{j = 2}{L-1}{(-1)^j \{\overline{\boldsymbol{S}}_{j-1} \boldsymbol{T}_j \boldsymbol{S}_{j+1}\}},
\label{ArovasBopapp}
\end{eqnarray}
where
\begin{eqnarray}
    \boldsymbol{T} &\equiv& \boldsymbol{S} \otimes \overline{\boldsymbol{S}} + \overline{\boldsymbol{S}} \otimes \boldsymbol{S} \nn \\
    &=& \begin{pmatrix}
    \frac{\{S^-, S^+\}}{2} & (S^-)^2 & \frac{\{S^-,  S^z\}}{\sqrt{2}} \\
    (S^+)^2 & \frac{\{S^+, S^-\}}{2} & \frac{\{S^+, S^z\}}{\sqrt{2}} \\
    \frac{\{S^z, S^+\}}{\sqrt{2}} & \frac{\{S^z, S^-\}}{\sqrt{2}} & 2S^z S^z
\end{pmatrix}.
\end{eqnarray}
Using Eq.~(\ref{ArovasBopapp}), the MPO for the Arovas $B$ operator reads
\begin{equation}
    M_B = 
    \begin{pmatrix}
        -\mathds{1} & -\overline{\boldsymbol{S}} & 0 & 0 \\
        0 & 0 & \boldsymbol{T} & 0 \\
        0 & 0 & 0 & \boldsymbol{S} \\
        0 & 0 & 0 & \mathds{1}
    \end{pmatrix},
\label{ArovasBmpoapp}
\end{equation}
where $0$ denotes zero matrices of appropriate dimensions.

\section{Exact entanglement spectrum for the spin-2 magnon of the spin-1 AKLT model}\label{sec:exactexample}
In this section, we explicitly work out the exact expression for the entanglement spectrum of the spin-2 magnon in the spin-1 model, the simplest excited state. The MPS bond dimension $\chi$, the MPO bond dimension $\chi_m$ and the MPO $\times$ MPS bond dimension $\wc$ are
\begin{equation}
    \chi = 2,\;\;\;\chi_m = 2, \;\;\; \wc = 4.
\end{equation}
Substituting $C = (S^+)^2$ and $k = \pi$ in Eq.~(\ref{genTransferstructure}), the transfer matrix $F$ reads 
\begin{equation}
    F =
    \begin{pmatrix}
        E & E_{+} & E_{-} & E_{-+} \\
        0 & -E & 0 & -E_{-} \\
        0 & 0 & -E & -E_{+} \\
        0 & 0 & 0 & E
    \end{pmatrix},
\label{spin2exacttransfer}
\end{equation}%
\newline
where
\begin{equation}
  E_{+} \equiv E_{(S^+)^{2}} \;\;\; E_{-} \equiv E_{(S^-)^{2}} \;\;\; E_{-+} \equiv E_{(S^-)^2 (S^+)^{2}},
\end{equation}
shown in Eq.~(\ref{spin1transfer}). We refer to the blocks of $F$ as the MPS blocks.
The Jordan decomposition of $F$ reads
\begin{equation}
    F = P J P^{\mm},
\end{equation}
where $J$ (obtained using symbolic calculations) reads  \newline
\begin{equation}
    J = \left(
    \begin{array}{cccc|cccc|cccc|cccc}
        1 & 0 & 0 & 0 & 0 & 0 & 0 & 0 & 0 & 0 & 0 & 0 & 1 & 0 & 0 & 0 \\
        0 & \text{-}\frac{1}{3} & 0 & 0 & 0 & 0 & 0 & 0 & 0 & 0 & 0 & 0 & 0 & 1 & 0 & 0 \\
        0 & 0 & \text{-}\frac{1}{3} & 0 & 0 & 0 & 0 & 0 & 0 & 0 & 0 & 0 & 0 & 0 & 1 & 0 \\
        0 & 0 & 0 & \text{-}\frac{1}{3} & 0 & 0 & 0 & 0 & 0 & 0 & 0 & 0 & 0 & 0 & 0 & 1 \\\hline
        0 & 0 & 0 & 0 & \mm & 0 & 0 & 0 & 0 & 0 & 0 & 0 & 0 & 0 & 0 & 0 \\
        0 & 0 & 0 & 0 & 0 & \frac{1}{3} & 0 & 0 & 0 & 0 & 0 & 0 & 0 & 0 & 0 & 0 \\
        0 & 0 & 0 & 0 & 0 & 0 & \frac{1}{3} & 0 & 0 & 0 & 0 & 0 & 0 & 0 & 0 & 0 \\
        0 & 0 & 0 & 0 & 0 & 0 & 0 & \frac{1}{3} & 0 & 0 & 0 & 0 & 0 & 0 & 0 & 0 \\\hline
        0 & 0 & 0 & 0 & 0 & 0 & 0 & 0 & \mm & 0 & 0 & 0 & 0 & 0 & 0 & 0 \\
        0 & 0 & 0 & 0 & 0 & 0 & 0 & 0 & 0 & \frac{1}{3} & 0 & 0 & 0 & 0 & 0 & 0 \\
        0 & 0 & 0 & 0 & 0 & 0 & 0 & 0 & 0 & 0 & \frac{1}{3} & 0 & 0 & 0 & 0 & 0 \\
        0 & 0 & 0 & 0 & 0 & 0 & 0 & 0 & 0 & 0 & 0 & \frac{1}{3} & 0 & 0 & 0 & 0 \\\hline
        0 & 0 & 0 & 0 & 0 & 0 & 0 & 0 & 0 & 0 & 0 & 0 & 1 & 0 & 0 & 0 \\
        0 & 0 & 0 & 0 & 0 & 0 & 0 & 0 & 0 & 0 & 0 & 0 & 0 & \text{-}\frac{1}{3} & 0 & 0 \\
        0 & 0 & 0 & 0 & 0 & 0 & 0 & 0 & 0 & 0 & 0 & 0 & 0 & 0 & \text{-}\frac{1}{3} & 0 \\
        0 & 0 & 0 & 0 & 0 & 0 & 0 & 0 & 0 & 0 & 0 & 0 & 0 & 0 & 0 & \text{-}\frac{1}{3}
    \end{array} \right), \label{exactjordan}
\end{equation}
and $P$, $P^{\mm}$ read
\begin{widetext}
\begin{equation}
    P = \left(
    \begin{array}{cccc|cccc|cccc|cccc}
        1 & 0 & 0 & \mm & 0 & 0 & 0 & 0 & 0 & 0 & 0 & 0 & \text{-}\frac{3}{2} & 0 & 0 & \text{-}\frac{3}{2} \\
        0 & 1 & 0 & 0 & 0 & 0 & \mm & 0 & 0 & 0 & 0 & 0 & 0 & 0 & 0 & 0 \\
        0 & 0 & 1 & 0 & 0 & 0 & 0 & 0 & 0 & \mm & 0 & 0 & 0 & 0 & 0 & 0 \\
        1 & 0 & 0 & 1 & 0 & 0 & 0 & 0 & 0 & 0 & 0 & 0 & 0 & 0 & 0 & 0 \\\hline
        0 & 0 & 0 & 0 & 1 & 0 & 0 & \mm & 0 & 0 & 0 & 0 & 0 & 0 & 0 & 0 \\
        0 & 0 & 0 & 0 & 0 & 1 & 0 & 0 & 0 & 0 & 0 & 0 & 0 & 0 & 0 & 0 \\
        0 & 0 & 0 & 0 & 0 & 0 & 1 & 0 & 0 & 0 & 0 & 0 & 0 & \text{-}\frac{3}{2} & 0 & 0 \\
        0 & 0 & 0 & 0 & 1 & 0 & 0 & 1 & 0 & 0 & 0 & 0 & 0 & 0 & 0 & 0 \\\hline
        0 & 0 & 0 & 0 & 0 & 0 & 0 & 0 & 1 & 0 & 0 & \mm & 0 & 0 & 0 & 0 \\
        0 & 0 & 0 & 0 & 0 & 0 & 0 & 0 & 0 & 1 & 0 & 0 & 0 & 0 & \text{-}\frac{3}{2} & 0 \\
        0 & 0 & 0 & 0 & 0 & 0 & 0 & 0 & 0 & 0 & 1 & 0 & 0 & 0 & 0 & 0 \\
        0 & 0 & 0 & 0 & 0 & 0 & 0 & 0 & 1 & 0 & 0 & 1 & 0 & 0 & 0 & 0 \\\hline
        0 & 0 & 0 & 0 & 0 & 0 & 0 & 0 & 0 & 0 & 0 & 0 & 3 & 0 & 0 & 3 \\
        0 & 0 & 0 & 0 & 0 & 0 & 0 & 0 & 0 & 0 & 0 & 0 & 0 & \frac{3}{2} & 0 & 0 \\
        0 & 0 & 0 & 0 & 0 & 0 & 0 & 0 & 0 & 0 & 0 & 0 & 0 & 0 & \frac{3}{2} & 0 \\
        0 & 0 & 0 & 0 & 0 & 0 & 0 & 0 & 0 & 0 & 0 & 0 & 3 & 0 & 0 & \text{-}3
    \end{array} \right) \; P^{\mm} = \left(
    \begin{array}{cccc|cccc|cccc|cccc}
        \frac{1}{2} & 0 & 0 & \frac{1}{2} & 0 & 0 & 0 & 0 & 0 & 0 & 0 & 0 & \frac{1}{4} & 0 & 0 & 0 \\
        0 & 1 & 0 & 0 & 0 & 0 & 0 & 0 & 1 & 0 & 0 & 0 & 0 & 1 & 0 & 0 \\
        0 & 0 & 1 & 0 & 0 & 0 & 0 & 0 & 0 & 1 & 0 & 0 & 0 & 0 & 1 & 0 \\
        \text{-}\frac{1}{2} & 0 & 0 & \frac{1}{2} & 0 & 0 & 0 & 0 & 0 & 0 & 0 & 0 & \text{-}\frac{1}{4} & 0 & 0 & 0 \\\hline
        0 & 0 & 0 & 0 & \frac{1}{2} & 0 & 0 & \frac{1}{2} & 0 & 0 & 0 & 0 & 0 & 0 & 0 & 0 \\
        0 & 0 & 0 & 0 & 0 & 1 & 0 & 0 & 0 & 0 & 0 & 0 & 0 & 0 & 0 & 0 \\
        0 & 0 & 0 & 0 & 0 & 0 & 1 & 0 & 0 & 0 & 0 & 0 & 0 & 1 & 0 & 0 \\
        0 & 0 & 0 & 0 & \text{-}\frac{1}{2} & 0 & 0 & \frac{1}{2} & 0 & 0 & 0 & 0 & 0 & 0 & 0 & 0 \\\hline
        0 & 0 & 0 & 0 & 0 & 0 & 0 & 0 & \frac{1}{2} & 0 & 0 & \frac{1}{2} & 0 & 0 & 0 & 0 \\
        0 & 0 & 0 & 0 & 0 & 0 & 0 & 0 & 0 & 1 & 0 & 0 & 0 & 0 & 1 & 0 \\
        0 & 0 & 0 & 0 & 0 & 0 & 0 & 0 & 0 & 0 & 1 & 0 & 0 & 0 & 0 & 0 \\
        0 & 0 & 0 & 0 & 0 & 0 & 0 & 0 & \text{-}\frac{1}{2} & 0 & 0 & \frac{1}{2} & 0 & 0 & 0 & 0 \\\hline
        0 & 0 & 0 & 0 & 0 & 0 & 0 & 0 & 0 & 0 & 0 & 0 & \frac{1}{6} & 0 & 0 & \frac{1}{6} \\
        0 & 0 & 0 & 0 & 0 & 0 & 0 & 0 & 0 & 0 & 0 & 0 & 0 & \frac{2}{3} & 0 & 0 \\
        0 & 0 & 0 & 0 & 0 & 0 & 0 & 0 & 0 & 0 & 0 & 0 & 0 & 0 & \frac{2}{3} & 0 \\
        0 & 0 & 0 & 0 & 0 & 0 & 0 & 0 & 0 & 0 & 0 & 0 & \frac{1}{6} & 0 & 0 & \text{-}\frac{1}{6}
    \end{array} \right) \hspace{8mm} \label{exactP},
\end{equation}
\end{widetext}
where in $J$, $P$ and $P^{\mm}$ the lines demarcate the MPS blocks. 
Using Eq.~(\ref{exactjordan}), the truncated Jordan block $\wJ$ (defined in Eq.~(\ref{Jordantrunc}))) reads
\begin{equation}
    \wJ =
    \begin{pmatrix}
        1 & 0 & 0 & 1 \\
        0 & \mm & 0 & 0 \\
        0 & 0 & \mm & 0 \\
        0 & 0 & 0 & 1
    \end{pmatrix}.
\label{exactJordantrunc}
\end{equation}
Using $P$ and $P^{\mm}$ in Eq.~(\ref{exactP}), $V_L$ and $V_R$ define in Eq.~(\ref{vrvlproj}) read
\begin{equation}
    V_R =
    \begin{pmatrix}
        1 & 0 & 0 & \text{-}{\frac{3}{2}} \\
        0 & 0 & 0 & 0 \\
        0 & 0 & 0 & 0 \\
        1 & 0 & 0 & 0 \\\hline
        0 & 1 & 0 & 0 \\
        0 & 0 & 0 & 0 \\
        0 & 0 & 0 & 0 \\
        0 & 1 & 0 & 0 \\\hline
        0 & 0 & 1 & 0 \\
        0 & 0 & 0 & 0 \\
        0 & 0 & 0 & 0 \\
        0 & 0 & 1 & 0 \\\hline
        0 & 0 & 0 & 3 \\
        0 & 0 & 0 & 0 \\
        0 & 0 & 0 & 0 \\
        0 & 0 & 0 & 3 \\
    \end{pmatrix}\;
    V_L =
    \begin{pmatrix}
        \frac{1}{2} & 0 & 0 & 0 \\
        0 & 0 & 0 & 0 \\
        0 & 0 & 0 & 0 \\
        \frac{1}{2} & 0 & 0 & 0 \\\hline
        0 & \frac{1}{2} & 0 & 0 \\
        0 & 0 & 0 & 0 \\
        0 & 0 & 0 & 0 \\
        0 & \frac{1}{2} & 0 & 0 \\\hline
        0 & 0 & \frac{1}{2} & 0 \\
        0 & 0 & 0 & 0 \\
        0 & 0 & 0 & 0 \\
        0 & 0 & \frac{1}{2} & 0 \\\hline
        \frac{1}{4} & 0 & 0 & \frac{1}{6} \\
        0 & 0 & 0 & 0 \\
        0 & 0 & 0 & 0 \\
        0 & 0 & 0 & \frac{1}{6} \\
    \end{pmatrix}\label{VRVLexact}
\end{equation}
For simplicity, we assume that the boundary spin-1/2 are in the $S_z = +1/2$ configuration. Consequently, the boundary vectors read (see Eq.~(\ref{AKLTboundaryvectors}))
\begin{equation}
    b^l_A = 
    \begin{pmatrix}
        1 \\
        0
    \end{pmatrix}\;
    b^r_A =
    \begin{pmatrix}
        0 \\
        1
    \end{pmatrix}.
\end{equation}
Consequently, using Eqs.~(\ref{MPOMPSboundary}) and (\ref{MPOBoundary}), we obtain the 16-dimensional boundary vectors of the transfer matrix whose components read
\begin{equation}
    (b^l_F)_i = \delta_{i, 1}, \;\;\;
    (b^r_F)_i = \delta_{i, 16}.
\label{boundaryexact}
\end{equation}
Using
\begin{equation}
    \wJ^n =
    \begin{pmatrix}
        1 & 0 & 0 & n \\
        0 & (\mm)^n & 0 & 0 \\
        0 & 0 & (\mm)^n & 0 \\
        0 & 0 & 0 & 1
    \end{pmatrix},
\end{equation}
$V_R$ and $V_L$ from Eq.~(\ref{VRVLexact}) and the boundary vectors from Eq.~(\ref{boundaryexact}), $\wmr$ and $\wml$ in Eq.~(\ref{lrexp}) (when viewed as $\wc^2$-dimensional vectors) read
\begin{equation}
    \wmr = 
    \begin{pmatrix}
        \frac{n}{6} - \frac{1}{4} \\
        0 \\
        0 \\
        \frac{n}{6}\\\hline
        0 \\
        0 \\
        0 \\
        0 \\\hline
        0 \\
        0 \\
        0 \\
        0 \\\hline
        \frac{1}{2} \\
        0 \\
        0 \\
        \frac{1}{2}
    \end{pmatrix}\;
    \wml =
    \begin{pmatrix}
        \frac{1}{2} \\
        0 \\
        0 \\
        \frac{1}{2} \\\hline
        0 \\
        0 \\
        0 \\
        0 \\\hline
        0 \\
        0 \\
        0 \\
        0 \\\hline
        \frac{n}{6}\\
        0 \\
        0 \\
        -\frac{1}{4} + \frac{n}{6}
    \end{pmatrix},
\end{equation}
where the lines demarcate the MPS blocks. 
$\wmr$ and $\wml$ can be viewed as $\wc \times \wc$ matrices, where the MPS blocks are reshaped separately. %
That is, the reshaped $\wc \times \wc$ matrices $\wmr$ and $\wml$ read
\begin{equation}
    \wmr = 
    \left(
    \begin{array}{cc|cc}
        \frac{n}{6} - \frac{1}{4} & 0 & 0 & 0 \\
        0 & \frac{n}{6} & 0 & 0 \\\hline
        0 & 0 & \frac{1}{2} & 0 \\
        0 & 0 & 0 & \frac{1}{2}
    \end{array}\right)\;
    \wml =
    \left(
    \begin{array}{cc|cc}
        \frac{1}{2}  & 0 & 0 & 0 \\
        0 & \frac{1}{2} & 0 & 0 \\\hline
        0 & 0 & \frac{n}{6}  & 0 \\
        0 & 0 & 0 & \frac{n}{6} - \frac{1}{4}
    \end{array}\right).
\label{ltildrtildexact}
\end{equation}
The (normalized) density matrix $\rhored$ (defined in Eq.~(\ref{rhored})) then reads
\begin{equation}
    \rhored = \wml \wmr^T = 
    \begin{pmatrix}
        \frac{2n - 3}{8n - 6} & 0 & 0 & 0 \\
        0 & \frac{n}{4n - 3} & 0 & 0 \\
        0 & 0 & \frac{n}{4n - 3} & 0 \\
        0 & 0 & 0 & \frac{2n - 3}{8n - 6}
    \end{pmatrix}.
\label{rhoredexact}
\end{equation}
We now illustrate the same derivation of $\wml$ and $\wmr$ using the procedure shown in Eqs.~(\ref{tildeboundary}) to (\ref{nsuperpositions}).
The columns of $V_R$ and $V_L$  (after reshaping the MPS and MPO spaces separately) are $\wc \times \wc$ matrices that read: 
\begin{eqnarray}
    &&r_1 =
    \left(
    \begin{array}{cc|cc}
        1 & 0 & 0 & 0 \\
        0 & 1 & 0 & 0 \\\hline
        0 & 0 & 0 & 0 \\
        0 & 0 & 0 & 0
    \end{array}\right)\;
    r_2 = \left(
    \begin{array}{cc|cc}
        0 & 0 & 1 & 0 \\
        0 & 0 & 0 & 1 \\\hline
        0 & 0 & 0 & 0 \\
        0 & 0 & 0 & 0
    \end{array}\right)\nn \\
    &&r_3 = 
    \left(
    \begin{array}{cc|cc}
        0 & 0 & 0 & 0 \\
        0 & 0 & 0 & 0 \\\hline
        1 & 0 & 0 & 0 \\
        0 & 1 & 0 & 0
    \end{array}\right)\;
    r_4 =
    \left(
    \begin{array}{cc|cc}
        \text{-}\frac{3}{2} & 0 & 0 & 0 \\
        0 & 0 & 0 & 0 \\\hline
        0 & 0 & 3 & 0 \\
        0 & 0 & 0 & 3
    \end{array}\right)\nn \\
    &&l_1 = 
    \left(
    \begin{array}{cc|cc}
        \frac{1}{2} & 0 & 0 & 0 \\
        0 & \frac{1}{2} & 0 & 0 \\\hline
        0 & 0 & 0 & 0 \\
        0 & 0 & 0 & \frac{1}{4}
    \end{array}\right)\;
    l_2 =
    \left(
    \begin{array}{cc|cc}
        0 & 0 & \frac{1}{2} & 0 \\
        0 & 0 & 0 & \frac{1}{2} \\\hline
        0 & 0 & 0 & 0 \\
        0 & 0 & 0 & 0
    \end{array}\right)\nn \\
    &&l_3 = 
    \left(
    \begin{array}{cc|cc}
        0 & 0 & 0 & 0 \\
        0 & 0 & 0 & 0 \\\hline
        \frac{1}{2} & 0 & 0 & 0 \\
        0 & \frac{1}{2} & 0 & 0
    \end{array}\right)\;
    l_4 =
    \left(
    \begin{array}{cc|cc}
        0 & 0 & 0 & 0 \\
        0 & 0 & 0 & 0 \\\hline
        0 & 0 & \frac{1}{6} & 0 \\
        0 & 0 & 0 & \frac{1}{6}
    \end{array}\right). \label{VRVLexact2}
\end{eqnarray}
The components of $W_R$ and $W_L$ (defined in Eqs.~(\ref{tildedefn}) and (\ref{tildcomp})) are computed using Eq.~(\ref{LRformjdep}). $\wJ$ can be written as
\begin{equation}
    \wJ = J_0 \oplus J_{\mm} \oplus J_1, 
\end{equation}
where the blocks read
\begin{equation}
    J_0 =
    \begin{pmatrix}
        1 & 1 \\
        0 & 1
    \end{pmatrix}\;\;
    J_{\mm} = (-1) \;\;
    J_1 = (-1).
\end{equation}
The sizes $\{|J_k|\}$ and generalized eigenvalues $\{\lambda_k\}$ associated with the Jordan blocks $\{J_k\}$ are
\begin{eqnarray}
    &&|J_{\mm}| = 1\;\;\; |J_0| = 2\;\;\; |J_1| = 1 \nn \\ 
    &&\lambda_{\mm} = \mm \;\;\; \lambda_0 = 1\;\;\; \lambda_{1} = \mm,
\label{params}
\end{eqnarray}
and the corresponding generalized eigenvectors associated with the Jordan blocks are
\begin{eqnarray}
    &&r^{(J_0)}_1 = r_1 \;\;\; r^{(J_0)}_2 = r_4 \;\;\; r^{(J_{\mm})}_1 = r_2 \;\;\; r^{(J_1)}_1 = r_3 \nn \\
    &&l^{(J_0)}_1 = l_1 \;\;\; l^{(J_0)}_2 = l_4 \;\;\; l^{(J_{\mm})}_1 = l_2 \;\;\; l^{(J_1)}_1 = l_3.
\label{vecmap}
\end{eqnarray}
Using Eqs.~(\ref{LRformjdep}), (\ref{params}) and (\ref{vecmap}), we obtain
\begin{eqnarray}
    R_1 &=& r_1\;\;\; L_1 = l_1 + n l_4 \nn \\
    R_2 &=& (\mm)^n r_2 \;\;\; L_2 = (\mm)^n l_2 \nn \\
    R_3 &=& (\mm)^n r_3 \;\;\; L_3 = (\mm)^n l_3 \nn \\
    R_4 &=& n r_1 + r_4 \;\;\; L_4 = l_4.
    \label{compexpand}
\end{eqnarray}
Using Eqs.~(\ref{compexpand}) and (\ref{VRVLexact2}), we obtain
\begin{eqnarray}
    &&R_1 = 
    \left(
    \begin{array}{cc|cc}
        1 & 0 & 0 & 0 \\
        0 & 1 & 0 & 0 \\\hline
        0 & 0 & 0 & 0 \\
        0 & 0 & 0 & 0
    \end{array}\right)\;
    R_2 = \left(
    \begin{array}{cc|cc}
        0 & 0 & (\mm)^n & 0 \\
        0 & 0 & 0 & (\mm)^n \\\hline
        0 & 0 & 0 & 0 \\
        0 & 0 & 0 & 0
    \end{array}\right)\nn \\
    &&R_3 = \left(
    \begin{array}{cc|cc}
        0 & 0 & 0 & 0 \\
        0 & 0 & 0 & 0 \\\hline
        (\mm)^n & 0 & 0 & 0 \\
        0 & (\mm)^n & 0 & 0
    \end{array}\right)\;
    R_4 = \left(
    \begin{array}{cc|cc}
        \text{-}\frac{3}{2} + n & 0 & 0 & 0 \\
        0 & n & 0 & 0 \\\hline
        0 & 0 & 3 & 0 \\
        0 & 0 & 0 & 3
    \end{array}\right)\nn \\
    &&L_1 = \left(
    \begin{array}{cc|cc}
        \frac{1}{2} & 0 & 0 & 0 \\
        0 & \frac{1}{2} & 0 & 0 \\\hline
        0 & 0 & \frac{1}{4} + \frac{n}{6} & 0 \\
        0 & 0 & 0 & \frac{n}{6}
    \end{array}\right)\;
    L_2 = \left(
    \begin{array}{cc|cc}
        0 & 0 & \frac{(\mm)^n}{2} & 0 \\
        0 & 0 & 0 & \frac{(\mm)^n}{2} \\\hline
        0 & 0 & 0 & 0 \\
        0 & 0 & 0 & 0
    \end{array}\right)\nn \\
    &&L_3 = \left(
    \begin{array}{cc|cc}
        0 & 0 & 0 & 0 \\
        0 & 0 & 0 & 0 \\\hline
        \frac{(\mm)^n}{2} & 0 & 0 & 0 \\
        0 & \frac{(\mm)^n}{2} & 0 & 0
    \end{array}\right)\;
    L_4 = \left(
    \begin{array}{cc|cc}
        0 & 0 & 0 & 0 \\
        0 & 0 & 0 & 0 \\\hline
        0 & 0 & \frac{1}{6} & 0 \\
        0 & 0 & 0 & \frac{1}{6}
    \end{array}\right). \label{WRWLexact}
\end{eqnarray}
Using Eqs.~(\ref{VRVLexact2}), (\ref{boundaryexact}) and (\ref{tildeboundary}), the modified boundary vectors read
\begin{equation}
    \beta^r_F =
    \begin{pmatrix}
        0 \\
        0 \\\hline
        0 \\
        \frac{1}{6}
    \end{pmatrix}\;
    \beta^l_F =
    \begin{pmatrix}
        1 \\
        0 \\\hline
        0 \\
        \text{-}\frac{3}{2}
    \end{pmatrix}.\label{modifiedboundaryexact}
\end{equation}
Consequently, using Eqs.~(\ref{WRWLexact}), (\ref{modifiedboundaryexact}) and (\ref{LR}), $\wmr$ and $\wml$ read
\begin{eqnarray}
    \wmr &=& R_1 (\beta^r_F)_1 + R_2 (\beta^r_F)_2 + R_3 (\beta^r_F)_3 + R_4 (\beta^r_F)_4 \nn \\
    &=& \frac{R_4}{6} \nn \\
    \wml &=& L_1 (\beta^l_F)_1 + L_2 (\beta^l_F)_2 + L_3 (\beta^l_F)_3 + L_4 (\beta^l_F)_4 \nn \\
    &=& L_1 - \frac{3 L_4}{2}, 
\end{eqnarray}
which are precisely the matrices in Eq.~(\ref{ltildrtildexact}).
Note that in all our examples in the text, the form of the $\beta^r_F$ and $\beta^l_F$ do not matter to the entanglement spectrum in the limit $n \rightarrow \infty$.

\section{Jordan normal form of block upper triangular matrices}\label{sec:Triangular}
In this section we describe a procedure to determine the structure of generalized eigenvalues, eigenvectors and Jordan normal forms of particular block upper triangular matrices that arise in the analysis of the MPO $\times$ MPS states in the text. The systematic construction of Jordan normal forms for general matrices has been discussed in existing literature.\cite{fletcher1983algorithmic, bartlett2013jordan}
In this section we consider a block upper triangular matrix of the form
\begin{equation}
    \bM =
    \begin{pmatrix}
        M_{11} & M_{12} & M_{13} & \dots & M_{1D} \\
        0 & M_{22} & M_{23} & \ddots & M_{2D} \\
        \vdots & \ddots & \ddots & \ddots & \vdots \\
        \vdots & \ddots & \ddots & M_{D-1,D-1} & M_{D-1,D} \\
        0 & \dots & \dots & 0 & M_{DD}
    \end{pmatrix}
\label{uppertriangleexample}
\end{equation}
where diagonal submatrices $M_{ii}$'s are $\chi \times \chi$ diagonalizable matrices that have at most a single non-degenerate eigenvalue of magnitude $1$. 
We assume $d$ of the diagonal submatrices have an eigenvalue of magnitude $1$, and they are written as $\{M_{\sigma(i),\sigma(i)},\ 1 \leq i \leq d\}$, where %
\begin{eqnarray}
    &\sigma : \{1, \dots, d\} \rightarrow \{1,\dots, D\} \nn \\
    &\sigma(i) = j\; \implies \textrm{$M_{jj}$ is the $i$'th block with} \nn \\
    &\textrm{eigenvalue of magnitude $1$}.
\label{sigmadef}
\end{eqnarray}
Furthermore, we restrict ourselves to determining the Jordan block structure of generalized eigenvalues of unit magnitude and the structure of the corresponding generalized eigenvectors. 
\subsection{Generalized eigenvalues}\label{blockuppereig}
We first derive the generalized eigenvalues of $\bM$ using its characteristic equation. Note that for any $\lambda$, 
\begin{equation}
    {\rm det} (\bM - \lambda \mathds{1}_{D\chi}) = \prod_{i = 1}^D{{\rm det}(M_{ii} - \lambda \mathds{1}_{\chi}}).
\label{eigvalsblocktriangular}
\end{equation}
Thus, the generalized eigenvalues of $\bM$ are the eigenvalues of its submatrices on the diagonal.
However, as we will see, an eigenvector of $\bM$ corresponding to an eigenvalue $\bl_{\ba}$ need not exist, particularly due to the upper triangular structure of $\bM$. In such a case, $\bM$ is not diagonalizable, $\bl_{\ba}$ is called a generalized eigenvalue, and corresponding generalized eigenvector exists. 
In general,  a Jordan decomposition of $\bM$ of the form
\begin{equation}
    \bM = \bP \bJ \bP^{\mm}
\label{jordandecomp}
\end{equation}
always exists, where $\bJ$ is the Jordan normal form of $\bM$, the columns of $\bP$ are the right generalized eigenvectors of $\bM$ and the rows of $\bP^{\mm}$ are its left generalized eigenvectors. Since $\bP^{\mm} \bP = \mathds{1}_{D\chi}$, the conventional form for the generalized eigenvectors of $\bM$ is 
\begin{equation}
    \bh_{\ba}^T \bv_{\bb} = \delta_{\ba \bb}
\label{geneigvecnorm}
\end{equation}
where $\bh_{\ba}$ and $\bv_{\bb}$ are left and right generalized eigenvectors of $\bM$, the rows and columns of $\bP^{\mm}$ and $\bP$ respectively.
We now derive the form of $\bh_{\ba}$ and $\bv_{\bb}$ when $\bM$ has the form of Eq.~(\ref{uppertriangleexample}). 
The Jordan normal form $\bJ$ of $\bM$ is related to $\bM$ by means of a similarity transformation, that is,
\begin{equation}
    \bJ = \bP^{\mm} \bM \bP.
\label{jordanopp}
\end{equation}
Thus, we can construct $\bJ$, $\bP$ and $\bP^{\mm}$ by sequentially performing similarity transformations on $\bM$ to reduce it to a Jordan normal form. A similarity transformation on a matrix $B$ using a matrix $A$ is defined as the transformation
\begin{equation}
    B \rightarrow A^{\mm} B A.
\label{similarity}
\end{equation}
Before we show the explicit construction of the Jordan normal form, we summarize the three main steps that we use to proceed:
\begin{enumerate}[(I)]
    \item A similarity transformation of $\bM$ using a block-diagonal matrix $\bD$. The resultant matrix is $\bLam^{(1,2)}$,
    \begin{equation}
        \bLam^{(1,2)} = \bD^{\mm} \bM \bD.
    \label{stepasimil}
    \end{equation}
    $\bLam^{(1,2)}$ has the form 
    \begin{equation}
    \bLam^{(1,2)} = 
        \begin{pmatrix}
        \Lambda_{11} & \Lambda^{(1,2)}_{12} & \Lambda^{(1,2)}_{13} & \dots & \Lambda^{(1,2)}_{1D} \\
        0 & \Lambda_{22} & \Lambda^{(1,2)}_{23} & \ddots & \Lambda^{(1,2)}_{2D} \\
        \vdots & \ddots & \ddots & \ddots & \vdots \\
        \vdots & \ddots & \ddots & \Lambda_{D-1,D-1} & \Lambda^{(1,2)}_{D-1,D} \\
        0 & \dots & \dots & 0 & \Lambda_{DD}
    \end{pmatrix},
    \label{Lam0struct}
    \end{equation}
    where $\Lambda_{ii}$ is the eigenvalue matrix of $M_{ii}$.
    \item  A similarity transformation is then applied to $\bLam^{(1,2)}$ using a carefully chosen block-upper triangular matrix $\bO$, such that
    \begin{equation}
        \bLam = \bO^{\mm} \bLam^{(1,2)} \bO,
    \label{stepbsimil}
    \end{equation}
    where $\bLam$ can be written as 
    \begin{equation}
        \bLam = 
        \begin{pmatrix}
            \Lambda_{11} & \Lambda_{12} & \Lambda_{13} & \dots & \Lambda_{1D} \\
        0 & \Lambda_{22} & \Lambda_{23} & \ddots & \Lambda_{2D} \\
        \vdots & \ddots & \ddots & \ddots & \vdots \\
        \vdots & \ddots & \ddots & \Lambda_{D-1,D-1} & \Lambda_{D-1,D} \\
        0 & \dots & \dots & 0 & \Lambda_{DD}
        \end{pmatrix}
    \label{blamstruct}
    \end{equation}
    where 
    \begin{equation}
        \left(\Lambda_{ij}\right)_{\alpha\beta} \neq 0 \implies \left(\Lambda_{ii}\right)_{\alpha\alpha} = \left(\Lambda_{jj}\right)_{\beta\beta},\;\;\; i < j.
    \label{reqdproperty}
    \end{equation}
    $\bO$ in Eq.~(\ref{stepbsimil}) has the form
    \begin{equation}
        \bO = \prodal{j = 2}{D}{\left(\prodal{i = j - 1}{1}{\bO_{ij}}\right)},
    \label{bOstruct}
    \end{equation}
    where $\bO_{ij}$ and $\bO^{\mm}_{ij}$ respectively read
    \begin{eqnarray}
&&\bO_{ij}= 
\begin{array}{cccccc}
\left(
  \begin{BMAT}[3.5pt]{cccccc}{cccccc}
    \mathds{1} & 0 & \cdots & \cdots & \cdots & 0 \\
        0 & \ddots & \ddots & \ddots & \ddots & \vdots \\
        \vdots & \ddots & \ddots & O_{ij} & \ddots & \vdots \\
        \vdots & \ddots & \ddots & \ddots & \ddots & \vdots \\
        \vdots & \ddots & \ddots & \ddots & \ddots & 0 \\        
        0 & \cdots & \cdots & \cdots & 0 & \mathds{1} \\
  \end{BMAT} 
\right)
& 
\begin{array}{l}
  \\ [-40mm]\hspace{-2mm}\rdelim\}{5}{0mm}[i\textrm{ rows}] \\
\end{array} \\[-1ex]\hspace{-11mm}
\hexbrace{26mm}{j\textrm{ columns}}
\end{array} \nn \\
&&\bO^{\mm}_{ij}= 
\begin{array}{cccccc}
\left(
  \begin{BMAT}[3.5pt]{cccccc}{cccccc}
    \mathds{1} & 0 & \cdots & \cdots & \cdots & 0 \\
        0 & \ddots & \ddots & \ddots & \ddots & \vdots \\
        \vdots & \ddots & \ddots & -O_{ij} & \ddots & \vdots \\
        \vdots & \ddots & \ddots & \ddots & \ddots & \vdots \\
        \vdots & \ddots & \ddots & \ddots & \ddots & 0 \\        
        0 & \cdots & \cdots & \cdots & 0 & \mathds{1} \\
  \end{BMAT} 
\right)
& 
\begin{array}{l}
  \\ [-40mm]\hspace{-2mm}\rdelim\}{5}{0mm}[i\textrm{ rows}] \\
\end{array} \\[-1ex]\hspace{-11mm}
\hexbrace{26mm}{j\textrm{ columns}}
\end{array}.
\label{similarity2}
\end{eqnarray}

    \item A similarity transformation $\bS$ of the form is applied to $\bLam$ to obtain the Jordan normal form $\bJ$, such that
    \begin{equation}
        \bJ = \bS^{\mm} \bLam \bS.
    \label{stepcsimil}
    \end{equation}
\end{enumerate}
\subsection{Step (I)}
We first transform $\bM$ to an upper triangular matrix (from a block upper triangular matrix) by a similarity transformation using the block-diagonal matrix $\bD$, defined as
\begin{equation}
    \bD = 
    \begin{pmatrix}
        \Delta_{11} & 0 & \cdots & \cdots & 0 \\
        0 & \Delta_{22} & \ddots & \ddots & 0 \\
        \vdots & \ddots & \ddots & \ddots & \vdots \\
        \vdots & \ddots & 0 & \Delta_{D-1, D-1} & 0 \\
        0 & \cdots & \cdots & 0 & \Delta_{DD}
    \end{pmatrix}, 
\label{Deltadef}
\end{equation}
where
\begin{equation}
    M_{jj} = \Delta_{jj} \Lambda_{jj} \Delta^{\mm}_{jj}, \;\;\; 1 \leq j \leq D,
\label{Mcond}
\end{equation}
where $\Lambda_{jj}$'s are diagonal matrices consisting of the eigenvalues of $M_{jj}$'s.
Consequently, the upper triangular matrix $\bLam^{(1,2)}$ is of the form of Eq.~(\ref{Lam0struct}), where
\begin{equation}
    \Lambda^{(1,2)}_{ij} \equiv \Delta^{\mm}_{ii} M_{ij} \Delta_{jj},\;\;\; 1 \leq i < j \leq D.
\label{uppertriang}
\end{equation}
Since the $\Lambda_{jj}$'s are diagonal matrices, $\bLam^{(1,2)}$ is an upper triangular matrix.
\subsection{Step (II)}\label{subsec:stepII}
We first prove a useful lemma.
\begin{lemma}\label{usefullemma}
An equation of the form 
\begin{equation}
    Y = C + \Theta_1 X - X \Theta_2, 
\label{finaleqnform}
\end{equation}
where $\Theta_1$ and $\Theta_2$ are diagonal matrices with $\left(\Theta_1\right)_{\alpha\alpha} = \theta_{1\alpha}$ and $\left(\Theta_2\right)_{\alpha\alpha} = \theta_{2\alpha}$, admits solutions to $X$ and $Y$ that read
\begin{eqnarray}
    &&X_{\alpha \beta} = \twopartdef{\frac{C_{\alpha\beta}}{\theta_{2\beta} - \theta_{1\alpha}}}{\theta_{1\alpha} \neq \theta_{2\beta}}{0}{\theta_{1\alpha} = \theta_{2\beta}} \nn \\
    &&Y_{\alpha \beta} = \twopartdef{0}{\theta_{1\alpha} \neq \theta_{2\beta}}{C_{\alpha\beta}}{\theta_{1\alpha} = \theta_{2\beta}}.
\label{usefulresult}
\end{eqnarray}
\end{lemma}
\begin{proof}
Writing the components of Eq.~(\ref{finaleqnform}), 
\begin{eqnarray}
    \theta_{1\alpha} X_{\alpha \beta} + C_{\alpha \beta} = Y_{\alpha\beta} +X_{\alpha \beta} \theta_{2\beta}  \nn \\
    \label{ineqj}
\end{eqnarray}
where $\theta_{1\alpha}$ and $\theta_{2\alpha}$ are the diagonal entries of $\Theta_{1}$ and $\Theta_{2}$ (here eigenvalues of $M_{11}$ and $M_{22}$ respectively).
As long as $\theta_{2\beta} \neq \theta_{1\alpha}$, a solution of Eq.~(\ref{ineqj}) is obtained using
\begin{eqnarray}
    X_{\alpha \beta} &=& \frac{C_{\alpha\beta}}{\theta_{2\beta} - \theta_{1\alpha}} \nn \\
    Y_{\alpha\beta} &=& 0.
\label{asolution}
\end{eqnarray}
While Eq.~(\ref{asolution}) is not the unique solution to Eq.~(\ref{ineqj}), as we illustrate later in this section, this particular solution chosen so that the $\bLam$ matrix we obtain in step (II) satisfies Eq.~(\ref{reqdproperty}). 
However, if $\theta_{1\alpha} = \theta_{2\beta}$, again using Eq.~(\ref{ineqj}), we obtain as a solution
\begin{eqnarray}
    X_{\alpha \beta} &=& 0 \nn \\
    Y_{\alpha \beta} &=& C_{\alpha\beta}.
\end{eqnarray}
\end{proof}

\subsubsection{D = 2 case}
We first illustrate the similarity transformation of $\bLam^{(1,2)}$ to $\bLam$ when $D = 2$.
Here the matrix $\bLam^{(1,2)}$ reads
\begin{equation}
    \bLam^{(1,2)} \equiv \bD^{\mm} \bM \bD =
    \begin{pmatrix}
        \Lambda_{11} & \Lambda^{(1,2)}_{12} \\
        0 & \Lambda_{22}
    \end{pmatrix},
\label{NdefnD2}
\end{equation}
where 
\begin{equation}
    \Lambda^{(1,2)}_{12} = \Delta^{\mm}_{11} M_{12} \Delta_{22}.
\end{equation}
To obtain the Jordan normal form, we further apply a similarity transformation using $\bO_{12}$ defined as %
\begin{equation}
    \bO_{12} \equiv 
    \begin{pmatrix}
        \mathds{1} & O_{12} \\
        0 & \mathds{1}
    \end{pmatrix}.
\label{O12struct}
\end{equation}
The resulting matrix $\bLam$ reads
\begin{eqnarray}
    \bLam &\equiv& \bO^{\mm}_{12} \bLam^{(1,2)} \bO_{12} \nn \\
    &=&
    \begin{pmatrix}
        \Lambda_{11} & \Lambda_{12} \\
        0 & \Lambda_{22}
    \end{pmatrix},
\label{lamstructD2}
\end{eqnarray}
where
\begin{equation}
    \Lambda_{12} = \Lambda^{(1,2)}_{12} + \Lambda_{11} O_{12} - O_{12} \Lambda_{22}.
\label{offdiageqnsimp}
\end{equation}
Eq.~(\ref{offdiageqnsimp}) is of the form of Eq.~(\ref{finaleqnform}) with
\begin{eqnarray}
    &C = \Lambda^{(1,2)}_{12} = \Delta_{11}^{\mm} M_{12} \Delta_{22},\;\;\;\Theta_1 = \Lambda_{11}, \;\;\; \Theta_2 = \Lambda_{22},  \nn \\
    &X = O_{12}, \;\;\; Y = \Lambda_{12},
\label{D2subs}
\end{eqnarray}
where we need to solve for $X$ and $Y$.
Thus, using Lemma~\ref{usefullemma} and Eq.~(\ref{usefulresult}) we obtain a solution to $\Lambda_{12}$ that satisfies 
\begin{equation}
    \left(\Lambda_{12}\right)_{\alpha\beta} \neq 0 \;\;\;\textrm{only if}\;\;\left(\Lambda_{11}\right)_{\alpha\alpha} = \left(\Lambda_{22}\right)_{\beta\beta}.
\label{nonvanishcond}
\end{equation}
Thus, $\bLam$ satisfies the property of Eq.~(\ref{reqdproperty}).

\subsubsection{General D case}
To make our derivation simpler, we first define the matrices
\begin{equation}
    \bLam^{(i, j)} = 
        \begin{pmatrix}
        \Lambda_{11} & \Lambda_{12} & \cdots & \cdots & \Lambda_{1,j-1}  & \Lambda^{(i,j)}_{1j} &  \cdots & \Lambda^{(i,j)}_{1D} \\
        0 & \Lambda_{22} & \ddots & \ddots & \ddots & \vdots & \ddots  & \vdots \\
        \vdots & \ddots & \ddots & \ddots & \ddots & \Lambda^{(i,j)}_{i-1,j} &   \ddots & \vdots \\
        \vdots & \ddots & \ddots & \ddots  & \ddots & \Lambda_{ij} &   \ddots & \vdots \\
        \vdots & \ddots & \ddots & \ddots & \ddots & \vdots & \ddots &  \vdots \\
        \vdots & \ddots & \ddots & \ddots & \ddots & \Lambda_{jj} &  \ddots & \vdots \\
        \vdots & \ddots & \ddots & \ddots & \ddots & \ddots & \ddots & \Lambda^{(i,j)}_{D-1,D} \\
        0 & \dots & \dots & \dots & \dots & \dots & 0 & \Lambda_{DD}
        \end{pmatrix}
    \label{bLammijstruct}
\end{equation}
where $\Lambda_{mn}$'s are matrices that satisfy the property of Eq.~(\ref{reqdproperty}).
To show that a $\bLam$ of the form of Eq.~(\ref{blamstruct}) whose off-diagonal blocks satisfy the property of Eq.~(\ref{reqdproperty}), we proceed via induction on $D$ and assume that an intermediate matrix has the form
\begin{equation}
    \bLam^{(D-1,D)} = 
        \begin{pmatrix}
            \Lambda_{11} & \Lambda_{12} & \cdots & \Lambda_{1, D-1} & \Lambda^{(D-1, D)}_{1D} \\
        0 & \Lambda_{22} & \Lambda_{23} & \ddots & \Lambda^{(D-1, D)}_{2D} \\
        \vdots & \ddots & \ddots & \ddots & \vdots \\
        \vdots & \ddots & \ddots & \Lambda_{D-1,D-1} & \Lambda^{(D-1, D)}_{D-1,D} \\
        0 & \dots & \dots & 0 & \Lambda_{DD}
        \end{pmatrix}
    \label{bLamkstruct}
\end{equation}
where $\Lambda_{ij}$, $1 \leq i < j \leq D - 1$ satisfy the property of Eq.~(\ref{reqdproperty}). 
We apply a similarity transformation to $\bLam^{(D-1, D)}$ using $\bO_{D-1, D}$ that has the structure shown in Eq.~(\ref{similarity2}): 
\begin{equation}
    \bO_{D-1, D} = 
    \begin{pmatrix}
            \mathds{1} & 0 & \cdots & 0 & 0 \\
        0 & \mathds{1} & \ddots & \ddots & \vdots \\
        \vdots & \ddots & \ddots & \ddots & \vdots \\
        \vdots & \ddots & \ddots & \mathds{1} & O_{D-1, D} \\
        0 & \dots & \dots & 0 & \mathds{1}
        \end{pmatrix}.
\end{equation}
The resulting matrix $\bLam^{(D-2, D)}$ reads
\begin{eqnarray}
    &\bLam^{(D-2, D)} \equiv \bO^{\mm}_{D-1, D} \bLam^{(D-1, D)} \bO_{D-1, D} \nn \\
    &=
        \begin{pmatrix}
        \Lambda_{11} & \Lambda_{12} & \cdots & \Lambda_{1, D-1} & \Lambda^{(D-2, D)}_{1D} \\
        0 & \Lambda_{22} & \Lambda_{23} & \ddots & \vdots \\
        \vdots & \ddots & \ddots & \ddots & \Lambda^{(D-2, D)}_{D-2, D} \\
        \vdots & \ddots & \ddots & \Lambda_{D-1,D-1} & \Lambda_{D-1,D} \\
        0 & \dots & \dots & 0 & \Lambda_{DD}
        \end{pmatrix},
\end{eqnarray}
where 
\begin{equation}
    \Lambda_{D-1, D} = \Lambda^{(D-1, D)}_{D-1, D} + \Lambda_{D-1, D-1} O_{D-1, D} - O_{D-1, D} \Lambda_{DD},
\label{offdiageqngen}
\end{equation}
and $\Lambda^{(D-2, D)}_{lm}$'s are matrices irrelevant to the current discussion.
Note that this similarity transformation using $\bO_{D-1, D}$ only affects the blocks in the $D$-th column (i.e. the $\Lambda_{ij}$'s are not modified). 
Eq.~(\ref{offdiageqngen}) is of the form of Eq.~(\ref{finaleqnform}) where
\begin{eqnarray}
    &C = \Lambda^{(D-1, D)}_{D-1, D},\;\;\;\Theta_1 = \Lambda_{D-1, D-1}, \;\;\; \Theta_2 = \Lambda_{DD},  \nn \\
    &X = O_{D-1, D}, \;\;\; Y = \Lambda_{D-1, D}.
\label{gensubs}
\end{eqnarray}
Thus, using Lemma~\ref{usefullemma} and Eq.~(\ref{usefulresult}), Eq.~(\ref{offdiageqngen}) has a solution for $\Lambda_{D-1, D}$ that satisfies the property of Eq.~(\ref{reqdproperty}).
We then apply another induction hypothesis on the last column and assume that an intermediate matrix $\bLam^{(l, D)}$ has the structure
\begin{equation}
    \bLam^{(l, D)} = 
        \begin{pmatrix}
        \Lambda_{11} & \Lambda_{12} & \cdots & \cdots & \Lambda_{1, D-1} & \Lambda^{(l,D)}_{1D} \\
        0 & \Lambda_{22} & \Lambda_{23} &  \ddots & \ddots & \vdots \\
        \vdots & \ddots & \ddots & \ddots & \ddots &  \Lambda^{(l,D)}_{l-1, D} \\
        \vdots & \ddots & \ddots & \ddots & \ddots & \Lambda^{(l,D)}_{lD} \\
        \vdots & \ddots & \ddots & \ddots & \ddots & \Lambda_{l+1, D} \\
        \vdots & \ddots & \ddots & \ddots & \Lambda_{D-1,D-1} & \vdots \\
        0 & \dots & \dots & \dots & 0 &  \Lambda_{DD}
        \end{pmatrix}
    \label{bLammstruct}
\end{equation}
where $\Lambda_{ij}$'s satisfy the property of Eq.~(\ref{reqdproperty}). 
Applying a similarity transformation using $\bO_{l D}$, we obtain a resulting matrix $\bLam^{(m+1)}$ that reads
\begin{eqnarray}
    &\bLam^{(l-1, D)} \equiv \bO^{\mm}_{lD} \bLam^{(l, D)} \bO_{l D} \nn \\
    &=
        \begin{pmatrix}
        \Lambda_{11} & \Lambda_{12} & \cdots & \cdots & \Lambda_{1, D-1} & \Lambda^{(l-1,D)}_{1D} \\
        0 & \Lambda_{22} & \Lambda_{23} &  \ddots & \ddots & \vdots \\
        \vdots & \ddots & \ddots & \ddots & \ddots &  \Lambda^{(l-1, D)}_{l-1, D} \\
        \vdots & \ddots & \ddots & \ddots & \ddots & \Lambda_{lD} \\
        \vdots & \ddots & \ddots & \ddots & \ddots & \Lambda_{l+1, D} \\
        \vdots & \ddots & \ddots & \ddots & \Lambda_{D-1,D-1} & \vdots \\
        0 & \dots & \dots & \dots & 0 &  \Lambda_{DD}
        \end{pmatrix},
\end{eqnarray}
where 
\begin{equation}
    \Lambda_{l D} = \Lambda^{(l, D)}_{l D} + \Lambda_{ll} O_{l D} - O_{l D} \Lambda_{DD},
\label{offdiageqngen2}
\end{equation}
and $\Lambda^{(l-1, D)}_{kn}$'s are irrelevant matrices.
Once again the similarity transformation using $\bO_{lD}$ only changes the first $l$ blocks on the $D$-th column, leaving the rest of the blocks unchanged. 
Eq.~(\ref{offdiageqngen2}) has the form of Eq.~(\ref{finaleqnform}) and thus, using Eq.~(\ref{usefulresult}), $\Lambda_{lD}$ satisfies the property of Eq.~(\ref{reqdproperty}). 

\subsubsection{Summary}
In summary, to obtain $\bLam$ of Eq.~(\ref{bLamkstruct}) from Eq.~(\ref{Lam0struct}), a sequence of $D (D -1)/2 - 1$ similarity transformations is applied to $\bLam^{(1,2)}$, where each one transforms a single off-diagonal block into an off-diagonal block of $\bLam$ that satisfies the property of Eq.~(\ref{reqdproperty}). This operation is applied column-wise starting from second column, and row-wise in each column starting from the off-diagonal block closest to the diagonal.
Thus, the sequence of similarity transformations that leads to $\bLam$ reads:
\begin{eqnarray}
    &&\bLam^{(1,2)} \xrightarrow{\bO_{12}} \bLam^{(2,3)} \xrightarrow{\bO_{23}} \bLam^{(1,3)} \xrightarrow{\bO_{13}} \bLam^{(3,4)} \xrightarrow{\bO_{34}} \cdots \nn \\
    &&\cdots \xrightarrow{\bO_{ij}} \cdots  \xrightarrow{\bO_{2 D}}\bLam^{\left(1, D\right)} \xrightarrow{\bO_{1 D}} \bLam, 
\label{similaritysequence}
\end{eqnarray}
where we have used the notation
\begin{equation}
    A \xrightarrow{B} C  \implies C = B^{\mm} A B,
\label{similaritynotation}
\end{equation}
Thus the similarity transformation from $\bLam^{(1,2)}$ to $\bLam$ has the form of Eq.~(\ref{stepbsimil}), where Eq.~(\ref{bOstruct}) holds. 
At each step the matrix $O_{mn}$ and $\Lambda_{mn}$ are determined as solutions to Eq.~(\ref{finaleqnform}) of the form Eq.~(\ref{usefulresult}), where
\begin{eqnarray}
    &C = \Lambda^{(m, n)}_{mn},\;\;\;\Theta_1 = \Lambda_{mm}, \;\;\; \Theta_2 = \Lambda_{nn},  \nn \\
    &X = O_{mn}, \;\;\; Y = \Lambda_{mn}.
\label{CthetaXY}
\end{eqnarray}
Thus, 
\begin{eqnarray}
    &&\left(O_{mn}\right)_{\alpha\beta} = \twopartdef{0}{\left(\Lambda_{mm}\right)_{\alpha\alpha} = \left(\Lambda_{nn}\right)_{\beta\beta}}{\frac{\left(\Lambda^{(m,n)}_{mn}\right)_{\alpha\beta}}{\left(\Lambda_{nn}\right)_{\beta\beta} - \left(\Lambda_{mm}\right)_{\alpha\alpha}}}{\left(\Lambda_{mm}\right)_{\alpha\alpha} \neq \left(\Lambda_{nn}\right)_{\beta\beta}} \nn  \\
    &&\left(\Lambda_{mn}\right)_{\alpha\beta} \equiv \twopartdef{\left(\Lambda^{(m,n)}_{mn}\right)_{\alpha\beta}}{\left(\Lambda_{mm}\right)_{\alpha\alpha} = \left(\Lambda_{nn}\right)_{\beta\beta}}{0}{\left(\Lambda_{mm}\right)_{\alpha\alpha} \neq \left(\Lambda_{nn}\right)_{\beta\beta}}.\nn \\
\label{finaldep}
\end{eqnarray}
For future convenience, the second line in Eq.~(\ref{finaldep}) can be written as 
\begin{equation}
    \Lambda_{mn} \equiv T\left[\Lambda^{(m,n)}_{mn}, \Lambda_{mm}, \Lambda_{nn}\right], 
\label{Teqn}
\end{equation}
where we have defined a function $T[A, B, C]$ that acts on matrices $A$, $B$, $C$:
\begin{equation}
    \left(T[A,B,C]\right)_{\alpha\beta}= \twopartdef{A_{\alpha\beta}}{B_{\alpha\alpha} = C_{\beta\beta}}{0}{A_{\alpha\alpha} \neq B_{\beta\beta}}.
\label{Tdef2}
\end{equation}
We now discuss a few properties of $\Lambda_{mn}$ that will be useful later in the paper.
To determine the structure of $\Lambda_{mn}$ in Eq.~(\ref{finaldep}), it is thus useful to study the dependence of $\Lambda^{(m,n)}_{mn}$ on the blocks of $\bLam^{(1,2)}$.
In Eq.~(\ref{similaritysequence}), if $\bLam^{(i,j)} \xrightarrow{\bO_{ij}} \bLam^{(i', j')}$ ($\bLam^{(i', j')} = \bO_{ij}^{\mm}\bLam^{(i,j)}\bO_{ij}$) , then using Eqs.~(\ref{bLammijstruct}) and (\ref{similarity2}), we obtain 
\begin{equation}
    \Lambda^{(i',j')}_{st} =  \threepartdeftwo{\Lambda^{(i,j)}_{st} + \Lambda^{(i,j)}_{si} O_{it}}{s < i,\ t = j}{\Lambda^{(i,j)}_{st} - O_{sj} \Lambda^{(i,j)}_{jt}}{s = i,\ t > j}{\Lambda^{(i,j)}_{st}},     
\label{depeqn}
\end{equation}
where, by abuse of notation, $\Lambda^{(i,j)}_{st}$ is the block on the $s$-th row and $t$-th column of $\bLam^{(i,j)}$. 
When the blocks of $\bLam^{(i',j')}$ are written in terms of the blocks of $\bLam^{(i,j)}$, we observe the following properties from Eq.~(\ref{depeqn}):
\begin{enumerate}
    \item[(P1)] $O_{ij}$ appears only in the expressions for the blocks $\Lambda^{(i',j')}_{it}$'s for $t > j$ and $\Lambda^{(i',j')}_{sj}$'s for $i > s$.
    \item[(P2)] $\Lambda^{(i',j')}_{st}$ depends only on the blocks $\Lambda^{(i,j)}_{si}$ and $\Lambda^{(i,j)}_{jt}$ of $\bLam^{(i,j)}$.
\end{enumerate}
As a consequence of property (P1), the similarity transformations $\bO_{ij}$ modify $\Lambda^{(1,2)}_{mn}$ only when $i = m, j < n$ or $i > m, j = n$, i.e. when $(i,j)$ is directly below or directly to the left of $(m, n)$.
Thus, using the sequence of similarity transformations of Eq.~(\ref{similaritysequence}) and the structure of $\bLam^{(i,j)}$ in Eq.~(\ref{bLammijstruct}), the expression for $\Lambda^{(m,n)}_{mn}$ can be written as follows:
\begin{equation}
    \Lambda^{(m,n)}_{mn} = \Lambda^{(1,2)}_{mn} + \sumal{t = m}{n-1}\Lambda_{mt}O_{tn} - \sumal{t = m}{n-1}O_{mt} \Lambda^{(m,t)}_{tn}.
\label{Lambdafullexp}
\end{equation}
As a consequence of Eq.~(\ref{Lambdafullexp}) and property (P2), when the blocks of $\bLam^{(m,n)}$ are written in terms of the blocks of $\bLam^{(1,2)}$ and $\{O_{ij}\}$ using the sequence of similarity transformations of Eq.~(\ref{similaritysequence}), $\Lambda^{(m,n)}_{mn}$ is of the form 
\begin{eqnarray}
    &\Lambda^{(m,n)}_{mn} = \Lambda^{(1,2)}_{mn} +  f(\{\Lambda^{(1,2)}_{ij}\}; \{\Lambda_{kk}\}),  \nn \\
    &m \leq i < j \leq n,\; m \leq k \leq n, \;\;(i,j) \neq (m,n),  
\label{lamdependencies2}
\end{eqnarray}
where $f$ is a function of matrices that depends on the blocks within the following boxed region of $\bLam^{(1,2)}$:
\begin{equation}
\hspace{-10mm}
\left(
\begin{array}{cccccccccc}
 \Lambda_{11} & \Lambda^{(1,2)}_{12} & \cdots & \cdots & \cdots  & \cdots &  \cdots & \cdots & \cdots & \Lambda^{(1,2)}_{1D} \\
0 & \ddots & \ddots & \ddots & \ddots & \ddots & \ddots & \ddots  & \ddots &  \vdots \\
\vdots & \ddots & \ddots & \ddots & \ddots & \ddots & \ddots & \ddots & \ddots & \vdots \\\cline{4-3}\cline{5-3}\cline{6-3}\cline{7-3}
\vdots & \ddots & \bordr & \Lambda_{mm} & \Lambda^{(1,2)}_{m,m+1} & \cdots & \bordr & \Lambda^{(1,2)}_{mn} & \ddots & \vdots \\\cline{4-4}\cline{8-7}
\vdots & \ddots & \ddots & \bordr & \ddots  & \ddots & \ddots & \vdots&  \bordl  &\vdots \\\cline{5-5}
\vdots & \ddots & \ddots &\ddots &  \bordr & \ddots & \ddots & \vdots &  \bordl &   \vdots \\\cline{6-6}
\vdots & \ddots & \ddots & \ddots & \ddots & \bordr & \ddots &   \Lambda^{(1,2)}_{n-1,n} & \bordl &\vdots \\ \cline{7-7}
\vdots & \ddots & \ddots & \ddots & \ddots & \ddots & \bordr & \Lambda_{nn} &   \bordl &\vdots \\\cline{8-8}
\vdots & \ddots & \ddots & \ddots & \ddots & \ddots& \ddots & \ddots & \ddots & \Lambda^{(1,2)}_{D-1,D}\\
0 & \dots & \dots & \dots & \dots & \dots & \dots & \dots & 0 & \Lambda_{DD}
\end{array}\right).
\label{boxedregion}
\end{equation}
Note that we could expect the function $f$ in Eq.~(\ref{lamdependencies2}) to depend on $O_{ij}$'s involved in the sequence of similarity transformations in Eq.~(\ref{similaritysequence}).
However, every $O_{ij}$ is determined using Eq.~(\ref{finaldep}), and thus it depends on $\Lambda^{(i,j)}_{ij}$, $\Lambda_{ii}$ and $\Lambda_{jj}$, that are already included in $\{\Lambda^{(1,2)}_{ij}\}$ and $\{\Lambda_{kk}\}$ in Eq.~(\ref{lamdependencies2}).
We now derive a useful property of the function $f$ in Eq.~(\ref{lamdependencies2}). For simplicity, we refer to the resulting matrix as $f$, i.e. $f(\{\Lambda^{(1,2)}_{ij}\}; \{\Lambda_{kk}\}) \equiv f$. As evident from Eq.~(\ref{finaldep}), the block structure of $\Lambda^{(i,j)}_{ij}$ is preserved in $O_{ij}$ and $\Lambda_{ij}$. Thus, by repeated applications of Eqs.~(\ref{Lambdafullexp}) and (\ref{lamdependencies2}) we deduce the following:
\begin{enumerate}
    \item[(f1)]  If every off-diagonal block $\Lambda^{(1,2)}_{ij}$ that appears in the argument of $f$ can be written as $\Lambda^{(1,2)}_{ij} = 0 \oplus L_{ij}$, where $0$ is the zero matrix and $L_{ij}$'s are some non-zero matrices with identical dimensions, then $f$ can be written as $f = 0 \oplus g$ where $g$ has the same dimension as the $L_{ij}$'s.  
\end{enumerate}
For example, if all the off-diagonal blocks within the boxed region in Eq.~(\ref{boxedregion}) are of the form
    \begin{equation}
        \Lambda^{(1,2)}_{ij} = 
        \begin{pmatrix}
            0 & 0 & \cdots & 0 \\
            0 & \ast & \cdots & \ast \\
            \vdots & \vdots & \ddots & \ast \\
            0 & \ast & \cdots  & \ast 
        \end{pmatrix}, \;\;\; m \leq i < j \leq n, \;\; (i,j) \neq (m,n),
    \label{lamexample}
    \end{equation}
then $f\left(\{\Lambda^{(1,2)}_{ij}\}; \{\Lambda_{kk}\}\right)$ of Eq.~(\ref{lamdependencies2}) has the same structure as the $\Lambda^{(1,2)}_{ij}$'s in Eq.~(\ref{lamexample}) irrespective of the $\Lambda_{kk}$'s.  
Thus, using Eq.~(\ref{Teqn}), the block $\Lambda_{mn}$ of $\bLam$ is related to the blocks of $\Lambda^{(1,2)}_{mn}$ as
\begin{eqnarray}
    &\Lambda_{mn} = T\left[\Lambda^{(1,2)}_{mn} +  f\left(\{\Lambda^{(1,2)}_{ij}\}; \{\Lambda_{kk}\} \right), \Lambda_{mm}, \Lambda_{nn}\right]\nn \\
    &m \leq i < j \leq n, \; m \leq k \leq n, \;\;(i,j) \neq (m,n), \nn \\
\label{lamdependencies}
\end{eqnarray}
where the function $T$ is defined in Eq.~(\ref{Tdef2}) and the function $f$ satisfies property (f1).
\subsection{Step (III)}
We now proceed to the final step of similarity transformations to obtain the Jordan normal form $\bJ$. 
Note that Eq.~(\ref{reqdproperty}) imposes a direct sum structure on $\bLam$, which we write as:
\begin{equation}
    \bLam = \bigoplus_k \Lambda_k,
\label{lamdirsumorig}
\end{equation}
where $\bLam_k$ is an upper triangular matrix with all its diagonal entries $\lambda_k$, an eigenvalue of $\bM$. 
Consequently, similarity transformations can be applied separately to each of the $\bLam_k$'s to obtain the Jordan normal form.
That is, one can apply a similarity transformation on $\bLam$ using
\begin{equation}
    \bS \equiv \bigoplus_k S_k
\label{bSdirsumorig}
\end{equation}
such that the Jordan normal form $\bJ$ of $\bM$ reads
\begin{equation}
    \bJ = \bS^{\mm} \bLam \bS = \bigoplus_k J_k,
\label{JDfinal}
\end{equation}
where
\begin{eqnarray}
    J_k &=& S^{\mm}_k \Lambda_k S_k.
\label{J1Jrest}
\end{eqnarray}
$J_k$ is a Jordan block of $\bM$ corresponding to eigenvalue $\lambda_k$.
Combining Eqs.~(\ref{stepasimil}), (\ref{stepbsimil}) and (\ref{stepcsimil}), $\bM$ can be written as 
\begin{equation}
    \bM = \bP \bJ \bP^{\mm},
\label{Mjordan}
\end{equation}
where
\begin{equation}
    \bP \equiv \bQ \bS,\;\;\; \bQ \equiv \bD \prodal{j = 2}{D}{\left(\prodal{i = j - 1}{1}{\bO_{ij}}\right)}.
\label{bPformbQdefn}
\end{equation}
Using Eqs.~(\ref{stepasimil}), (\ref{Deltadef}), (\ref{similarity2}) and (\ref{bPformbQdefn}), $\bQ$ and $\bQ^{\mm}$ read
\begin{eqnarray}
    &\bQ = \bD \prodal{j = 2}{D}{\prodal{i = j - 1}{1}{\bO_{ij}}} = 
    \begin{pmatrix}
        \Delta_{11} & \ast & \cdots & \cdots & \ast \\
        0 & \Delta_{22} & \ddots & \ddots & \ast \\
        \vdots & \ddots & \ddots & \ddots & \vdots \\
        \vdots & \ddots & 0 & \Delta_{D-1, D-1} & \ast \\
        0 & \cdots & \cdots & 0 & \Delta_{DD}
    \end{pmatrix}, \;\;\; \nn \\
    &\bQ^{\mm} = 
    \prodal{j = D}{2}{\prodal{i = 1}{j-1}{\bO^{\mm}_{ij}}} \bD^{\mm} = 
    \begin{pmatrix}
        \Delta^{\mm}_{11} & \ast & \cdots & \cdots & \ast \\
        0 & \Delta^{\mm}_{22} & \ddots & \ddots & \ast \\
        \vdots & \ddots & \ddots & \ddots & \vdots \\
        \vdots & \ddots & 0 & \Delta^{\mm}_{D-1, D-1} & \ast \\
        0 & \cdots & \cdots & 0 & \Delta^{\mm}_{DD}
    \end{pmatrix}, \nn \\
\label{Qstruct}
\end{eqnarray} 
where $\ast$'s are matrices whose structure we will not need for the discussion in the main text.
\subsection{Structure of generalized eigenvectors}
We now study the columns (resp. rows) of $\bP$ (resp. $\bP^{\mm}$) that are the generalized eigenvectors corresponding to the generalized eigenvalues of unit magnitude. 
For $\bLam, \bS, \bJ$, it is convenient to write Eqs.~(\ref{lamdirsumorig}), (\ref{bSdirsumorig}) and (\ref{JDfinal}) as
\begin{eqnarray}
    \bLam = \wLam \oplus \rLam,\nn \\
    \bS = \wS \oplus \rS,\nn \\
    \bJ = \wJ \oplus \rJ,\nn \\
\label{lamdirsum}
\end{eqnarray}
where 
\begin{eqnarray}
    \wLam = \bigoplus_{\{k : |\lambda_k| = 1\}} \Lambda_k,\;\;\; \rLam = \bigoplus_{\{k: |\lambda_k| \neq 1\}} \Lambda_k\nn \\
    \wS = \bigoplus_{\{k : |\lambda_k| = 1\}} S_k,\;\;\; \rS = \bigoplus_{\{k: |\lambda_k| \neq 1\}} S_k\nn \\
    \wJ = \bigoplus_{\{k : |\lambda_k| = 1\}} J_k,\;\;\; \rJ = \bigoplus_{\{k: |\lambda_k| \neq 1\}} J_k. \nn \\
\label{lamSJdirsums}
\end{eqnarray}
Indeed, only the generalized eigenvalues of magnitude one are relevant in our case and we have assumed that each $M_{ii}$ has at most one such eigenvalue (note that this does not mean that these eigenvalues are identical). 
Since $\bLam$ has a direct sum structure shown in Eq~(\ref{lamdirsum}), $\bM$ can be written as 
\begin{equation}
    \bM = \wQ \wLam \wwQ + \rQ \rLam \wrQ
\label{Mdirexpand}
\end{equation}
where 
$\wQ$, $\wwQ$, $\rQ$, and $\wrQ$ are rectangular matrices such that the columns (resp. rows) of $\wQ$ (resp. $\wwQ$) and $\rQ$ (resp. $\wrQ$) act on the subspaces of $\wLam$ and $\rLam$ respectively.
Since $\bLam$ has the structure shown in Eq.~(\ref{blamstruct}) and only the blocks $M_{\sigma(i), \sigma(i)}$ (and consequently $\Lambda_{\sigma(i),\sigma(i)}$), $1 \leq i \leq d$ contain eigenvalues of magnitude 1, using Eq.~(\ref{Qstruct}) $\wQ$ and $\wwQ$ have $d$ rows and columns respectively and are of the forms 
\begin{eqnarray}
    &\wQ = 
    \begin{pmatrix}
        q_{\sigma(1)} & q_{\sigma(2)} & \cdots & q_{\sigma(d)}
    \end{pmatrix} \nn \\
    &\wwQ = 
    \begin{pmatrix}
        \tilde{q}_{\sigma(1)} \\
        \tilde{q}_{\sigma(2)} \\
        \vdots \\
        \tilde{q}_{\sigma(d)}
    \end{pmatrix},
\label{bQforms}
\end{eqnarray}
where
\begin{eqnarray}
    &q_{\sigma(i)} = 
    \begin{pmatrix}
        \ast \\
        \vdots \\
        \ast \\
        r_{\sigma(i)} \\
        0 \\
        \vdots \\
        0
    \end{pmatrix}\nn \\
    &
    \begin{array}{l}
  \\ [-65mm]\hspace{25mm}\rdelim\}{3}{0mm}[\textrm{$\sigma(i) - 1$}] \\
\end{array} \nn \\
    & 
    \begin{array}{ccccccc}
    \tilde{q}_{\sigma(i)} =
    \begin{pmatrix}
        0 & \cdots & 0 & l^T_{\sigma(i)} & \ast & \cdots & \ast \\
    \end{pmatrix},\\ [-1ex]\hspace{35mm}
    \hexbrace{12mm}{\textrm{$d - \sigma(i) + 1$}}
    \end{array}
\label{bQforms1}
\end{eqnarray}
where $\{r_j\}$ (resp. $\{l_j\}$) are right (resp. left) generalized eigenvectors of $\{M_{jj}\}$ corresponding to the generalized eigenvalues of unit magnitude. 
Further, using Eq.~(\ref{J1Jrest}), Eq.~(\ref{Mdirexpand}) can be written as
\begin{equation}
    \bM = \wP \wJ \wwP + \rP \rJ  \wrP,
\label{MJBexpand}
\end{equation}
where
\begin{eqnarray}
    &\wP \equiv \wQ \wS, \;\;\; \wwP \equiv  \wS^{\mm} \wwQ, \nn \\ 
    &\rP \equiv \rQ \rS, \;\;\; \wrP \equiv  \rS^{\mm} \wrQ. 
\label{Ptilddefn}
\end{eqnarray}
Since $\wQ$ and $\wwQ$ have the forms shown in Eq.~(\ref{bQforms}) and (\ref{bQforms1}), using Eq.~(\ref{Ptilddefn}), we obtain that
\begin{eqnarray}
    &\wP = 
    \begin{pmatrix}
        s_{\sigma(1)} q_{\sigma(1)} & s_{\sigma(2)} q_{\sigma(2)} & \cdots & s_{\sigma(d)} q_{\sigma(d)} \\
    \end{pmatrix}\nn \\
    &\wwP = 
    \begin{pmatrix}
        \frac{\tilde{q}_{\sigma(1)}}{s_{\sigma(1)}} \\
        \frac{\tilde{q}_{\sigma(2)}}{s_{\sigma(2)}} \\
        \vdots \\
        \frac{\tilde{q}_{\sigma(d)}}{s_{\sigma(d)}}
    \end{pmatrix}, \nn \\
    &\textrm{if $\wS$ is upper triangular and has the form} \nn \\
    &\wS = 
    \begin{pmatrix}
        s_{\sigma(1)} & \ast & \cdots & \cdots & \ast \\
        0 & s_{\sigma(2)} & \ddots & \ddots & \ast \\
        \vdots & \ddots & \ddots & \ddots & \vdots \\
        \vdots & \ddots & 0 & s_{\sigma(d-1)} & \ast \\
        0 & \cdots & \cdots & 0 & s_{\sigma(d)}
    \end{pmatrix},
\label{bPforms}
\end{eqnarray}
where the $s_i$'s are non-zero constants.
Thus, {\bf when }$\wS$ {\bf is upper triangular with all non-zero diagonal entries on its diagonal} the left and right generalized eigenvectors corresponding to generalized eigenvalues of unit magnitude have the following forms: 
\begin{equation}
    \bv_{j} =
    \begin{pmatrix}
    \ast \\
    \vdots \\
    \ast \\
    c_j r_{j} \\
    0 \\
    \vdots \\
    0
    \end{pmatrix},\;\;\;
    \bh_{j} =
    \begin{pmatrix}
    0 \\
    \vdots \\
    0 \\
    \frac{l_{j}}{c_j} \\
    \ast \\
    \vdots \\
    \ast
    \end{pmatrix},
\label{eigenvectorforms} 
\end{equation}
where $r_{j}$ and $l_{j}$ are the left and right eigenvectors of $M_{jj}$ corresponding to eigenvalue of unit magnitude, and $c_j$ is a non-zero constant that need not be the same as $s_j$ since $l_j$ and $r_j$ can be rescaled freely in a way that $l_j^T r_j = \mathds{1}$. 
In App.~\ref{sec:Jordanexample}, we show that for the examples we work with in the text, $\wS$ is indeed a diagonal matrix, thus imposing the forms of Eq.~(\ref{eigenvectorforms}) on the left and right generalized eigenvectors of those transfer matrices. 
For the construction of the matrix $\bS$ in general, we refer to discussions in Ref.~[\onlinecite{bartlett2013jordan}].

\subsection{Exact example with $D = 2$}
We now illustrate the above results with the help of an example. We consider the following block-upper triangular matrix:
\begin{equation}
    \bM = 
    \begin{pmatrix}
        4 & \mm & 4 & 3 \\
        6 & \mm & 2 & 6 \\
        0 & 0 & 46 & \text{-}30 \\
        0 & 0 & 63 & \text{-}41
    \end{pmatrix},
\label{degexample1}
\end{equation}
where the diagonal blocks are
\begin{equation}
    M_{11} =
    \begin{pmatrix}
        4  & \mm \\
        6 & \mm
    \end{pmatrix}\;\;
    M_{22} = 
    \begin{pmatrix}
        46 & \text{-}30 \\
        63 & \text{-}41
    \end{pmatrix},
\label{degexamplediagblocks}
\end{equation}
and the off-diagonal block $M_{12}$ is
\begin{equation}
    M_{12} = 
    \begin{pmatrix}
        4 & 3 \\
        2 & 6
    \end{pmatrix}.
\end{equation}
The eigenvalue decompositions of $M_{11}$ and $M_{22}$ read
\begin{equation}
    M_{11} = \Delta_{11} \Lambda_{11} \Delta^{\mm}_{11},\;\;\; M_{22} = \Delta_{22} \Lambda_{22} \Delta^{\mm}_{22},
\end{equation}
where
\begin{eqnarray}
    \Lambda_{11} = 
    \begin{pmatrix}
        1 & 0 \\
        0 & 2 
    \end{pmatrix},\;\;\;
    \Lambda_{22} = 
    \begin{pmatrix}
        1 & 0 \\
        0 & 4 
    \end{pmatrix}, \nn \\
    \Delta_{11} = 
    \begin{pmatrix}
        1 & 1 \\
        3 & 2 
    \end{pmatrix},\;\;\;
    \Delta_{22} = 
    \begin{pmatrix}
        4 & 5 \\
        6 & 7 
    \end{pmatrix}. \nn \\
\end{eqnarray}
Since both $M_{11}$ and $M_{22}$ have eigenvalues $1$, we have $\sigma(i) = i$ and $d = D$ in Eq.~(\ref{sigmadef}).
Consequently, applying the similarity transformation of Eq.~(\ref{stepasimil}) using
\begin{equation}
    \bD = 
    \begin{pmatrix}
        1 & 1 & 0 & 0 \\
        3 & 2 & 0 & 0 \\
        0 & 0 & 4 & 5 \\
        0 & 0 & 6 & 7
    \end{pmatrix}, 
\end{equation}
we obtain
\begin{equation}
    \bLam^{(1,2)} = \bD^{\mm} \bM \bD = 
    \begin{pmatrix}
        1 & 0 & -24 & -30 \\
        0 & 2 & 58 & 71 \\
        0 & 0 & 1 & 0 \\
        0 & 0 & 0 & 4
    \end{pmatrix}.
\end{equation}
Using Eqs.~(\ref{O12struct}), (\ref{D2subs}) and (\ref{usefulresult}), $\bO$ of Eq.~(\ref{bOstruct}) reads 
\begin{equation}
    \bO =
    \begin{pmatrix}
        1 & 0 & 0 & -10 \\
        0 & 1 & -58 & \frac{71}{2} \\
        0 & 0 & 1 & 0 \\
        0 & 0 & 0 & 1
    \end{pmatrix}.
\label{bO12example}
\end{equation}
Then, using the similarity transformation of Eq.~(\ref{lamstructD2}), we obtain
\begin{equation}
    \bLam = \bO^{\mm} \bLam^{(1,2)} \bO =  
    \begin{pmatrix}
        1 & 0 & -24 & 0 \\
        0 & 2 & 0 & 0 \\
        0 & 0 & 1 & 0 \\
        0 & 0 & 0 & 4
    \end{pmatrix}.
\label{blamexample}
\end{equation}
$\bLam$ in Eq.~(\ref{blamexample}) has the direct sum structure of Eq.~(\ref{lamdirsum}), where
\begin{equation}
    \wLam =
    \begin{pmatrix}
        1 & -24 \\
        0 & 1
    \end{pmatrix},\;\;\;
    \rLam = 
    \begin{pmatrix}
        2 & 0 \\
        0 & 4
    \end{pmatrix}.
\end{equation}
To obtain the Jordan normal form, similarity transformation of the form of Eq.~(\ref{JDfinal}) is applied to $\bLam$, where
\begin{equation}
    \bS = 
    \begin{pmatrix}
        -24 & 0 & 0 & 0 \\
        0 & 1 & 0 & 0 \\
        0 & 0 & 1 & 0 \\
        0 & 0 & 0 & 1
    \end{pmatrix}.
\end{equation}
The matrix $\bS$ acts on $\wLam$ and $\rLam$ separately, $\bS = \wS \oplus \rS$, where
\begin{equation}
    \wS = 
    \begin{pmatrix}
        -24 & 0 \\
        0 & 1 
    \end{pmatrix},\;\;\;
    \rS = 
    \begin{pmatrix}
        1 & 0 \\
        0 & 1
    \end{pmatrix}.
\label{bSstruct}
\end{equation}
Thus, the Jordan normal form $\bJ = \wJ \oplus \rJ$, where
\begin{eqnarray}
    &\wJ = 
    \begin{pmatrix}
        1 & 1 \\
        0 & 1
    \end{pmatrix},\;\;\;
    \rJ = 
    \begin{pmatrix}
        2 & 0 \\
        0 & 4
    \end{pmatrix},\nn \\
    &\bJ = 
    \begin{pmatrix}
        1 & 0 & 1 & 0 \\
        0 & 2 & 0 & 0 \\
        0 & 0 & 1 & 0 \\
        0 & 0 & 0 & 4
    \end{pmatrix}.
\label{jform}
\end{eqnarray}
To write $\bJ$ in the conventional form with the Jordan blocks consisting of $1$'s on the superdiagonal, the generalized eigenvalues can always be rearranged by a unitary transformation.
However, for our purposes, it is easier to work with $\bJ$ of the form of Eq.~(\ref{jform}).
Using Eqs.~(\ref{bPforms}), $\bP$ reads
\begin{eqnarray}
    &&\bP = 
    \begin{pmatrix}
        -24 & 1 & -58 & \frac{51}{2} \\
        -72 & 2 & -116 & 41\\
        0 & 0 & 4 & 5\\
        0 & 0 & 6 & 7
    \end{pmatrix} \nn \\
    &&\bP^{\mm} = 
    \begin{pmatrix}
        \frac{1}{12} & \text{-}\frac{1}{24} & \text{-}\frac{5}{4} & \frac{5}{6} \\
        3 & \mm & \text{-}\frac{619}{2} & 216\\
        0 & 0 & \text{-}\frac{7}{2} & \frac{5}{2}\\
        0 & 0 & 3 & \text{-}2
    \end{pmatrix} \nn \\
\end{eqnarray}
Note that the generalized eigenvectors corresponding to the eigenvalue $1$ have the structure of Eq.~(\ref{eigenvectorforms}). This is a direct consequence of the fact that $\wS$ in Eq.~(\ref{bSstruct}) is a diagonal (and hence upper triangular) matrix. 

\section{Structure of generalized transfer Matrices of operators in the AKLT MPS}\label{sec:genTransferstructure}
To compute the entanglement spectra of the spin-$2S$ magnon and the tower of states, we need the structure of the generalized transfer matrices Eq.~(\ref{genTransfer}) for the operators $({S^-})^{2S}$, $({S^+})^{2S}$ and $(S^-)^{2S}(S^+)^{2S}$. In the spin-$S$ basis these operators have the following representations (up to overall constants):
\begin{eqnarray}
    &({S^+}^{2S})_{mn} \sim \delta_{m,S}\delta_{n, -S} \nn \\
    &({S^-}^{2S})_{mn} \sim \delta_{m,-S}\delta_{n,S} \nn \\
    &((S^-)^{2S} (S^+)^{2S})_{mn} \sim \delta_{m,-S}\delta_{n,-S},    
\label{opform}
\end{eqnarray} 
where $-S \leq m, n \leq S$. 
Using the expression of the spin-$S$ AKLT ground state MPS ($\chi = S + 1$) of Eq.~(\ref{akltform}), the $\chi^2 \times \chi^2$ generalized transfer matrices $E_+, E_-, E_{-+}$ corresponding to the operators $(S^+)^{2S}$, $(S^-)^{2S}$ and $(S^-)^{2S} (S^+)^{2S}$ read 
\begin{eqnarray}
    &&({E_+})_{ij} \sim \delta_{i,\chi}\delta_{j,\chi^2+1-\chi} \nn \\
    &&({E_-})_{ij} \sim \delta_{i,\chi^2+1-\chi}\delta_{j,\chi} \nn \\
    &&({E_{-+}})_{ij} \sim \delta_{i,\chi^2}\delta_{j,1},  
\label{genTransferform}
\end{eqnarray}
where $1 \leq i, j \leq \chi^2$.
For example, for the spin-1 AKLT MPS of Eq.~(\ref{AKLTMPS}), the form of these generalized transfer matrices read
\begin{equation}
\hspace{-5mm}
    E_+
    =
    \begin{pmatrix}
        0 & 0 & 0 & 0 \\
        0 & 0 & -\frac{2}{3} & 0 \\
        0 & 0 & 0 & 0 \\
        0 & 0 & 0 & 0
    \end{pmatrix}
    E_-
    =
    \begin{pmatrix}
        0 & 0 & 0 & 0 \\
        0 & 0 & 0 & 0 \\
        0 & \frac{2}{3} & 0 & 0 \\
        0 & 0 & 0 & 0
    \end{pmatrix}
    E_{-+}
    =
    \begin{pmatrix}
        0 & 0 & 0 & 0 \\
        0 & 0 & 0 & 0 \\
        0 & 0 & 0 & 0 \\
        \frac{2}{3} & 0 & 0 & 0
    \end{pmatrix}.
\label{spin1transfer}
\end{equation}
As mentioned in Eq.~(\ref{Transferblockdiagonal}), the AKLT ground state transfer matrix can be written as a direct sum of $(S+2)$ blocks $\{E_p\}$, where the block $E_p$ is the submatrix of $E$ consists sets of rows and columns in Eq.~(\ref{indices}).
Using Eq.~(\ref{genTransferform}) and the fact that the left eigenvector $e_L$ corresponding to the largest eigenvalue is located in block $E_0$, we directly obtain
\begin{eqnarray}
E_{+} e_R = E_{-} e_R =0,\;\; e_L^T E_{+} = e_L^T E_{-} = 0, \;\;  e_L^T E_{-+} e_R \neq 0, \nn \\
\label{AKLTspec}
\end{eqnarray}
where $0$ denotes the zero vector of appropriate dimensions.
For example, the eigenvectors for the spin-1 AKLT transfer matrix of Eq.~(\ref{AKLTTransfer}) have the forms
\begin{equation}
    e_L^T = e_R^T = 
    \begin{pmatrix}
    \frac{1}{\sqrt{2}} & 0 & 0 & \frac{1}{\sqrt{2}}
    \end{pmatrix}
\end{equation}
and Eq.~(\ref{AKLTspec}) is directly verified using Eq.~(\ref{spin1transfer}).
As we will show in App.~\ref{sec:JBexampletower}, the properties of Eq.~(\ref{AKLTspec}) determine the Jordan normal form of the transfer matrix for the tower of states in Sec.~\ref{sec:Towerofstates}.
\section{Examples of Jordan norm form of block upper triangular matrices}\label{sec:Jordanexample}
In this section, we show examples of determining the Jordan normal forms of block upper triangular matrices.
\subsection{Single-mode excitation Transfer matrix}\label{sec:JBexamplesma}
Our first example is the transfer matrix of Eq.~(\ref{nestedtriangulartransfer}) for a single-mode excitation with a generic operator:
\begin{equation}
    \bM \equiv F = 
    \begin{pmatrix}
        E & E_{\hat{C}} & E_{\hat{C}^\dagger} & E_{\hat{C}^\dagger \hat{C}} \\
        0 & -E & 0 & -E_{\hat{C}^\dagger} \\
        0 & 0 & -E & -E_{\hat{C}} \\
        0 & 0 & 0 & E
    \end{pmatrix}.
\label{nestedtriangulartransferapp}
\end{equation}
If $E$ has the eigenvalue decomposition
\begin{equation}
    E = P_E \Lambda_E P^{\mm}_E, 
\label{Edecomp}
\end{equation}
where $\Lambda_E$, $P_E$ and $P^{\mm}_E$ read
\begin{eqnarray}
    &\Lambda_E = 
    \begin{pmatrix}
        1 & 0 & \cdots & 0 \\
        0 & \lambda_1 & \ddots & \vdots \\
        \vdots & \ddots & \ddots & \vdots \\
        0 & \cdots & \cdots & \lambda_{\chi^2 - 1}
    \end{pmatrix},\;\;\; |\lambda_i| < 1 \nn \\
    &P_E = 
    \begin{pmatrix}
        e_R & \ast & \cdots & \ast 
    \end{pmatrix}, \;\;\; P^{\mm}_E = 
    \begin{pmatrix}
        e_L^T \\
        \ast \\
        \vdots \\
        \ast
    \end{pmatrix}, 
\label{DEstruct}
\end{eqnarray}
where $\ast$'s are left and right eigenvectors of $E$ corresponding to the eigenvalues of magnitude less than 1. 
Using Eqs.~(\ref{stepasimil}) and (\ref{Deltadef}), we obtain
\begin{equation}
    \bLam^{(1,2)} =
    \begin{pmatrix}
        \Lambda_E & \Lambda^{(1,2)}_{\hat{C}} & \Lambda^{(1,2)}_{\hat{C}^\dagger} & \Lambda^{(1,2)}_{\hat{C}^\dagger \hat{C}} \\
        0 & -\Lambda_E & 0 & -\Lambda^{(1,2)}_{\hat{C}^\dagger} \\
        0 & 0 & -\Lambda_E & -\Lambda^{(1,2)}_{\hat{C}} \\
        0 & 0 & 0 & \Lambda_E
    \end{pmatrix},
\end{equation}
where
\begin{equation}
    \Lambda^{(1,2)}_{\hat{C}} \equiv P^{\mm}_E E_{\hat{C}} P_E,\;\; \Lambda^{(1,2)}_{\hat{C}^\dagger} \equiv P^{\mm}_E E_{\hat{C}^\dagger} P_E,\;\; \Lambda^{(1,2)}_{\hat{C}^\dagger} \equiv P^{\mm}_E E_{\hat{C}^\dagger \hat{C}} P_E. 
\end{equation}
Using the procedure described in App.~\ref{subsec:stepII}, since $\lambda_i \neq 1$ we obtain the matrix $\bLam$ that reads
\begin{equation}
    \bLam =
    \begin{pmatrix}
        \Lambda_E & \Lambda_{\hat{C}} & \Lambda_{\hat{C}^\dagger} & \Lambda_{\hat{C}^\dagger \hat{C}} \\
        0 & -\Lambda_E & 0 & -\Lambda_{\hat{C}^\dagger} \\
        0 & 0 & -\Lambda_E & -\Lambda_{\hat{C}} \\
        0 & 0 & 0 & \Lambda_E
    \end{pmatrix},
\label{bLamsmastruct}
\end{equation}
where
\begin{eqnarray}
    \Lambda_{\hat{C}} &=& 
    \begin{pmatrix}
        0 & \cdots & \cdots & 0 \\
        \vdots & \ast & \cdots & \ast \\
        \vdots & \vdots & \ddots & \vdots \\
        0 & \ast & \cdots & \ast
    \end{pmatrix}, \nn \\
    \Lambda_{\hat{C}^\dagger} &=& 
    \begin{pmatrix}
        0 & \cdots & \cdots & 0 \\
        \vdots & \ast & \cdots & \ast \\
        \vdots & \vdots & \ddots & \vdots \\
        0 & \ast & \cdots & \ast
    \end{pmatrix}, \nn \\
    \Lambda_{\hat{C}^\dagger \hat{C}} &=& 
    \begin{pmatrix}
        s & 0 & \cdots & 0 \\
        0 & \ast & \cdots & \ast \\
        \vdots & \vdots & \ddots & \vdots \\
        0 & \ast & \cdots & \ast
    \end{pmatrix},
\label{lambdasmastruct}
\end{eqnarray}
where the $\ast$'s are irrelevant. The forms of $\Lambda_{\hat{C}}$, $\Lambda_{\hat{C}^\dagger}$, $\Lambda_{\hat{C}^\dagger \hat{C}}$ in Eq.~(\ref{lambdasmastruct}) are a consequence of condition of Eq.~(\ref{reqdproperty}) applied to $\bLam$ in Eq.~(\ref{bLamsmastruct}). 
In Eq.~(\ref{lambdasmastruct}), the matrix element $s$ in $\Lambda_{\hat{C}^\dagger \hat{C}}$ involves components of $E_{\hat{C}^\dagger}$, $E_{\hat{C}}$, $E_{\hat{C}^\dagger \hat{C}}$ and $P_E$, and it does not have a simple expression in general.
Using Eqs.~(\ref{bLamsmastruct}) and (\ref{lambdasmastruct}), $\wLam$ (defined in Eq.~(\ref{lamdirsum})) reads
\begin{equation}
    \wLam =
    \begin{pmatrix}
        1 & 0 & 0 & s \\
        0 & \mm & 0 & 0 \\
        0 & 0 & \mm & 0 \\
        0 & 0 & 0 & 1
    \end{pmatrix},
\end{equation}
where $s$ in general does not have a simple expression in terms of the generalized transfer matrices.
Thus, a Jordan block is formed between the generalized eigenvalues $+1$ iff 
\begin{equation}
    s \neq 0.
\label{JBconditionsma}
\end{equation}
Eq.~(\ref{JBconditionsma}) holds for general operators $\hat{C}$, in which case one can rescale $s$ to $1$ by means of a similarity transformation $\wS$ that reads
\begin{equation}
    \wS = 
    \begin{pmatrix}
        s & 0 & 0 & 0 \\
        0 & 1 & 0 & 0 \\
        0 & 0 & 1 & 0 \\
        0 & 0 & 0 & 1
    \end{pmatrix}, 
\label{bS1smastruct}
\end{equation}
such that the Jordan block of the generalized eigenvalues of unit magnitude reads
\begin{eqnarray}
    \wJ &=& \wS^{\mm} \wLam \wS \nn \\  
    &=& 
    \begin{pmatrix}
        1 & 0 & 0 & 1 \\
        0 & \mm & 0 & 0 \\
        0 & 0 & \mm & 0 \\
        0 & 0 & 0 & 1
    \end{pmatrix}.
\label{JBsma}
\end{eqnarray}
Furthermore, since $\wS$ is an upper triangular matrix, the right and left generalized eigenvectors of $\bM$ corresponding to generalized eigenvalues of unit magnitude have the forms given by Eq.~(\ref{eigenvectorforms}) 
\begin{equation}
    r_1 = 
    \begin{pmatrix}
    c_1 e_R \\
    0 \\
    0 \\
    0
    \end{pmatrix}\;
    r_2 = 
    \begin{pmatrix}
    \ast \\
    c_2 e_R \\
    0 \\
    0
    \end{pmatrix}\;
    r_3 = 
    \begin{pmatrix}
    \ast \\
    \ast \\
    c_3 e_R \\
    0
    \end{pmatrix}\;
    r_4 = 
    \begin{pmatrix}
    \ast \\
    \ast \\
    \ast \\
    c_4 e_R
    \end{pmatrix}\;
\label{smarightcolumnapp}
\end{equation}
and 
\begin{equation}
    l_1 = 
    \begin{pmatrix}
    \frac{e_L}{c_1} \\
    \ast \\
    \ast \\
    \ast
    \end{pmatrix}\;
    l_2 = 
    \begin{pmatrix}
    0 \\
    \frac{e_L}{c_2} \\
    \ast \\
    \ast
    \end{pmatrix}\;
    l_3 = 
    \begin{pmatrix}
    0 \\
    0 \\
    \frac{e_L}{c_3} \\
    \ast
    \end{pmatrix}\;
    l_4 = 
    \begin{pmatrix}
    0 \\
    0 \\
    0 \\
    \frac{e_L}{c_4}
    \end{pmatrix}\;
\label{smaleftcolumnapp}
\end{equation}
respectively.

\subsection{Spin-$S$ Tower of states transfer matrix with $N = 2$}\label{sec:JBexampletower}
Our second example is the tower of states transfer matrix shown in Eq.~(\ref{N=2Transfer}),
\begin{equation}
    \hspace{-5mm}
    \bM \equiv F = 
    \begin{pmatrix}
        E & E_+ & 0 & E_- & E_{-+} & 0 & 0 & 0 & 0 \\
        0 & -E & -E_+ & 0 & -E_- & E_{-+} & 0 & 0 & 0 \\ 
        0 & 0 & E & 0 & 0 & E_- & 0 & 0 & 0 \\
        0 & 0 & 0 & -E & E_+ & 0 & E_- & E_{-+} & 0 \\
        0 & 0 & 0 & 0 & E & E_{+} & 0 & -E_- & E_{-+} \\
        0 & 0 & 0 & 0 & 0 & -E & 0 & 0 & E_- \\
        0 & 0 & 0 & 0 & 0 & 0 & E & E_+ & 0 \\
        0 & 0 & 0 & 0 & 0 & 0 & 0 & -E & -E_+ \\
        0 & 0 & 0 & 0 & 0 & 0 & 0 & 0 & E \\
    \end{pmatrix}.
\label{bMN=2Transfer}
\end{equation}
The eigenvalue decomposition of $E$ for the spin-$S$ AKLT ground state transfer matrix given by Eq.~(\ref{Edecomp}) and the diagonal matrix $\Lambda_E$ has the structure shown in Eq.~(\ref{DEstruct}).
Consequently, we obtain (using Eqs.~(\ref{stepasimil}) and (\ref{Deltadef}))
\begin{equation}
    \hspace{-10mm}
    \bLam^{(1,2)} =
    \begin{pmatrix}
        \Lambda_E & \Lambda_+ & 0 & \Lambda_- & \fbox{$\Lambda_{-+}$} & 0 & 0 & 0 & 0 \\
        0 & -\Lambda_E & -\Lambda_+ & 0 & -\Lambda_- & \fbox{$\Lambda_{-+}$} & 0 & 0 & 0 \\ 
        0 & 0 & \Lambda_E & 0 & 0 & \Lambda_- & 0 & 0 & 0 \\
        0 & 0 & 0 & -\Lambda_E & \Lambda_+ & 0 & \Lambda_- & \fbox{$\Lambda_{-+}$} & 0 \\
        0 & 0 & 0 & 0 & \Lambda_E & \Lambda_{+} & 0 & -\Lambda_- & \fbox{$\Lambda_{-+}$} \\
        0 & 0 & 0 & 0 & 0 & -\Lambda_E & 0 & 0 & \Lambda_- \\
        0 & 0 & 0 & 0 & 0 & 0 & \Lambda_E & \Lambda_+ & 0 \\
        0 & 0 & 0 & 0 & 0 & 0 & 0 & -\Lambda_E & -\Lambda_+ \\
        0 & 0 & 0 & 0 & 0 & 0 & 0 & 0 & \Lambda_E \\
    \end{pmatrix},
\label{bLam0N=2}
\end{equation}
where
\begin{equation}
    \Lambda_+ \equiv P^{\mm}_E E_+ P_E,\;\; \Lambda_- \equiv P^{\mm}_E E_- P_E, \;\; \Lambda_{-+} \equiv P^{\mm}_E E_{-+} P_E. 
\label{Npm}
\end{equation}
Using the properties of Eq.~(\ref{AKLTspec}) and the structures of $P_E$ and $P^{\mm}_E$ in Eq.~(\ref{DEstruct}), the matrices of Eq.~(\ref{Npm}) read
\begin{eqnarray}
    \Lambda_+ &=& 
    \begin{pmatrix}
        0 & 0 & \cdots & 0 \\
        0 & \ast & \ddots & \ast\\
        \vdots & \ddots & \ddots & \vdots \\
        0 & \ast & \cdots  & \ast
    \end{pmatrix} \nn \\
    \Lambda_- &=& 
    \begin{pmatrix}
        0 & 0 & \cdots & 0 \\
        0 & \ast & \ddots & \ast\\
        \vdots & \ddots & \ddots & \vdots \\
        0 & \ast & \cdots  & \ast
    \end{pmatrix} \nn \\
    \Lambda_{-+} &=& 
    \begin{pmatrix}
        s & \ast & \cdots & \ast \\
        \ast & \ast & \ddots & \ast\\
        \vdots & \ddots & \ddots & \vdots \\
        \ast & \cdots & \cdots  & \ast
    \end{pmatrix},
\label{Lam0forms}
\end{eqnarray}
where $\ast$'s are irrelevant values and the matrix element $s$ is given by
\begin{equation}
    s = e_L^T E_{-+} e_R \neq 0,
\label{sdefn}
\end{equation}
where $e_L$ and $e_R$ are the left and right eigenvectors of $E$ corresponding to the eigenvalue $1$.
A matrix $\bLam$ that satisfies Eq.~(\ref{reqdproperty}) can be obtained from $\bLam^{(1,2)}$ of Eq.~(\ref{bLam0N=2}) using the procedure described in App.~\ref{subsec:stepII}.  
We obtain the form of the blocks of $\bLam$ using its dependence on $\{\Lambda^{(1,2)}_{ij}\}$ as shown in Eqs.~(\ref{lamdependencies}) and (\ref{boxedregion}), and the forms of $\Lambda_{+}$, $\Lambda_{-}$ and $\Lambda_{-+}$ of Eq.~(\ref{Lam0forms}). 
An important result is that the function $f$ in Eq.~(\ref{lamdependencies}) preserves the direct sum structure, any $\Lambda_{mn}$ that only depends on $\Lambda_+$ and $\Lambda_-$ keeps the same form as the ones of $\Lambda_+$ and $\Lambda_-$ in Eq.~(\ref{Lam0forms}). 
We first define three matrix types, and then show that the off-diagonal blocks of $\bLam$ that we obtain from $\bLam^{(1,2)}$ of Eq.~(\ref{bLam0N=2}) fall into one of these types:
\begin{eqnarray}
    A &=& 
    \begin{pmatrix}
        0 & 0 & \cdots & 0 \\
        0 & \ast & \cdots & \ast\\
        \vdots & \vdots & \ddots & \vdots \\
        0 & \ast & \cdots  & \ast
    \end{pmatrix} \nn \\
    B &=& 
    \begin{pmatrix}
        s & 0 & \cdots & 0 \\
        0 & \ast & \cdots & \ast\\
        \vdots & \vdots & \ddots & \vdots \\
        0 & \ast & \cdots  & \ast
    \end{pmatrix} \nn \\
    C &=& 
    \begin{pmatrix}
        \ast & 0 & \cdots & 0 \\
        0 & \ast & \cdots & \ast\\
        \vdots & \vdots & \ddots & \vdots \\
        0 & \ast & \cdots  & \ast
    \end{pmatrix},
\label{ABCstructures}
\end{eqnarray}
where, as we will show, the $\ast$'s are not relevant to the Jordan normal form of the eigenvalues of unit magnitude. 
In Eq.~(\ref{bLam0N=2}), note that the blocks $\Lambda_{-+}$ all lie on a single diagonal of $\bLam^{(1,2)}$, which we call $\mathcal{D}$.
As we will show, these blocks determine the Jordan normal form of $\bM$.
We now consider the structure of various blocks of $\bLam^{(1,2)}$ of Eq.~(\ref{bLam0N=2}) and obtain the structure of the corresponding block in $\bLam$ using the properties of $f$ in property (f1) in Sec.~\ref{subsec:stepII}, the definition of $T$ in Eq.~(\ref{Tdef2}), the forms of the blocks $\Lambda_+$ and $\Lambda_-$ in Eq.~(\ref{Lam0forms}), and the form of $\Lambda_E$ in Eq.~(\ref{DEstruct}):
\begin{enumerate}
    \item[(c1)] Blocks to the left of the diagonal $\mathcal{D}$, and the blocks on $\mathcal{D}$ that are not $\Lambda_{-+}$ in $\bLam^{(1,2)}$: According to Eqs.~(\ref{lamdependencies}), (\ref{boxedregion}) and (\ref{bLam0N=2}), the expressions for these blocks can be written in one of the following forms:
    \begin{eqnarray}
        &&T\left[f\left(\{\Lambda_+, \Lambda_-\}; \{\Lambda_E\}\right), \pm \Lambda_E, \mp \Lambda_E\right] \sim A\nn \\
        &&T\left[f\left(\{\Lambda_+, \Lambda_-\}; \{\Lambda_E\}\right), \pm \Lambda_E, \pm \Lambda_E\right] \sim A,
        \label{case1}
    \end{eqnarray}
    where we have used the fact that
    \begin{equation}
        f\left(\{\Lambda_+, \Lambda_-\}; \{\Lambda_E\}\right) \sim A
    \label{f1form}
    \end{equation}
    as a consequence of the structures of $\Lambda_+$ and $\Lambda_-$ in Eq.~(\ref{Lam0forms}) and property (f1) in Sec.~\ref{subsec:stepII}. 
    \item[(c2)] Blocks on the diagonal $\mathcal{D}$ that are $\Lambda_{-+}$ in $\bLam^{(1,2)}$: These blocks of $\bLam$ are of the form
    \begin{equation}
        T\left[\Lambda_{-+} + f\left(\{\Lambda_+, \Lambda_-\}; \{\Lambda_E\}\right), \pm \Lambda_{E}, \pm \Lambda_E \right] \sim B,
    \end{equation}
    where we have used Eq.~(\ref{f1form}) and the structure of $\Lambda_{-+}$ in Eq.~(\ref{Lam0forms}) to deduce that $\Lambda_{-+} + f(\{\Lambda_+, \Lambda_-\}; \{\Lambda_E\})$ has the same structure as $\Lambda_{-+}$ in Eq.~(\ref{Lam0forms}), and subsequently used the definition of $T$ in Eq.~(\ref{Tdef}).  
    \item[(c3)] Blocks to the right of $\mathcal{D}$ on rows that have a $\Lambda_{-+}$ in $\bLam^{(1,2)}$: Here, the blocks are one of two forms:
    \begin{eqnarray}
        &&T\left[f\left(\{\Lambda_+, \Lambda_-, \Lambda_{-+}\}; \{\Lambda_E\}\right), \pm \Lambda_E, \mp \Lambda_E \right] \sim A \nn \\
        &&T\left[f\left(\{\Lambda_+, \Lambda_-, \Lambda_{-+}\}; \{\Lambda_E\}\right), \pm \Lambda_E, \pm \Lambda_E \right] \sim C,
    \end{eqnarray}
    which is true irrespective of the structure of $f\left(\{\Lambda_-, \Lambda_+, \Lambda_{-+}\}; \{\Lambda_E\}\right)$ due to the definition of $T$ in Eq.~(\ref{Tdef}) and the structure of $\Lambda_E$ in Eq.~(\ref{DEstruct}).
    \item[(c4)] Blocks to the right of $\mathcal{D}$ on rows that do not have a $\Lambda_{-+}$ in $\bLam^{(1,2)}$: We first show via induction on $n$ that any $\Lambda_{mn}$ in such a row (that does not have a $\Lambda_{-+}$ in $\bLam^{(1,2)}$) is always of the form of $A$ in Eq.~(\ref{ABCstructures}). We start with the induction hypothesis that
    \begin{eqnarray}
        &\Lambda^{(m,t)}_{mt}, O_{mt} \sim 
        \begin{pmatrix}
            0 & 0 & \cdots & 0 \\
            \ast & \cdots & \cdots & \ast \\
            \vdots & \ddots & \ddots & \vdots \\
            \ast & \cdots & \cdots & \ast
        \end{pmatrix},\; \Lambda_{mt} \sim 
        \begin{pmatrix}
            0 & 0 & \cdots & 0 \\
            0 & \ast & \cdots & \ast \\
            \vdots & \vdots & \ddots & \vdots \\
            0 & \ast & \cdots & \ast
        \end{pmatrix}, \nn \\
        &\; \forall t,\  m + 1 \leq t \leq n -1,
    \label{form}
    \end{eqnarray}
    which is true for $n = m + 2$ due to the case (c1).
    Using Eq.~(\ref{Lambdafullexp}), (\ref{form}) and the fact that $\Lambda^{(1,2)}_{mn}$ is either $0$, $\Lambda_+$ or $\Lambda_-$, we directly obtain that $\Lambda^{(m,n)}_{mn}$ is of the form of $\Lambda^{(m,t)}_{mt}$ shown in Eq.~(\ref{form}), irrespective of the structures of $O_{tn}$ and $\Lambda^{(m,t)}_{tn}$.
    As a consequence of Eq.~(\ref{finaldep}), $O_{mn}$ and $\Lambda_{mn}$ have the forms of $O_{mt}$ and $\Lambda_{mt}$ shown in Eq.~(\ref{form}).
    Thus, all the blocks on such a row $\Lambda_{mn}$ are of the form of $A$ in Eq.~(\ref{ABCstructures}).
\end{enumerate}
Thus, as a consequence of the cases (c1) through (c4), $\bLam$ obtained from $\bLam^{(1,2)}$ of Eq.~(\ref{bLam0N=2}) has the following structure:
\begin{equation}
    \bLam \sim
    \begin{pmatrix}
        \Lambda_E & A & A & A & B & A & C & A & C \\
        0 & -\Lambda_E & A & A & A & B & A & C & A \\ 
        0 & 0 & \Lambda_E & A & A & A & A & A & A \\
        0 & 0 & 0 & -\Lambda_E & A & A & A & B & C \\
        0 & 0 & 0 & 0 & \Lambda_E & A & A & A & B \\
        0 & 0 & 0 & 0 & 0 & -\Lambda_E & A & A & A \\
        0 & 0 & 0 & 0 & 0 & 0 & \Lambda_E & A & A \\
        0 & 0 & 0 & 0 & 0 & 0 & 0 & -\Lambda_E & A \\
        0 & 0 & 0 & 0 & 0 & 0 & 0 & 0 & \Lambda_E \\
    \end{pmatrix},
\label{bLamN=2}
\end{equation}
where the structures of $A$, $B$ and $C$ matrix types are shown in Eqs.~(\ref{ABCstructures}).
In Eq.~(\ref{bLamN=2}), $A$, $B$, and $C$ denote only the structures of the matrices (shown in Eq.~(\ref{ABCstructures})) and not the matrices themselves.
That is, the $\ast$'s in different copies of $A$'s are not guaranteed to be identical, and similarly for the $B$'s and the $C$'s.
As we will show, only the element $s$ in matrix $B$ is relevant to the Jordan normal form. 
This element originates from the $\Lambda_{-+}$ block in Eq.~(\ref{bLam0N=2}) due to the dependencies of blocks of $\bLam$ on the blocks of $\bLam^{(1,2)}$ shown in Eqs.~(\ref{lamdependencies}) and (\ref{boxedregion}).
Consequently, $\wLam$ reads
\begin{equation}
    \wLam = 
    \begin{pmatrix}
        1 & 0 & 0 & 0 & s & 0 & \ast & 0 & \ast \\
        0 & \mm & 0 & 0 & 0 & s & 0 & \ast & 0 \\
        0 & 0 & 1 & 0 & 0 & 0 & 0 & 0 & 0 \\
        0 & 0 & 0 & \mm & 0 & 0 & 0 & s & 0 \\
        0 & 0 & 0 & 0 & 1 & 0 & 0 & 0 & s \\
        0 & 0 & 0 & 0 & 0 & \mm & 0 & 0 & 0 \\
        0 & 0 & 0 & 0 & 0 & 0 & 1 & 0 & 0 \\
        0 & 0 & 0 & 0 & 0 & 0 & 0 & \mm & 0 \\
        0 & 0 & 0 & 0 & 0 & 0 & 0 & 0 & 1
    \end{pmatrix}.
\label{lambdaunitform}
\end{equation}
As we now show, the $\ast$ values are not relevant to the Jordan normal form (and are in general not identical). 
To show that and  transform $\Lambda_{\textrm{unit}}$ to the Jordan normal form, we first prove a useful Lemma.
\begin{lemma}\label{usefullemma2}
    Consider an upper triangular matrix $R$ that satisfies the following conditions:
    \begin{enumerate}
        \item [(C1)] The diagonal entries $R_{ii}$'s are all equal.
        \item [(C2)] For any $i < j$ such that $R_{ij}\, \neq 0$ and $R_{ik} = 0\;\; \forall\ k, \; i < k < j$, the entries of $R$ satisfy $R_{mj} = 0\;\; \forall\ m,\; i < m < j$.  
    \end{enumerate}
    Condition~(C2) translates to the following: the leftmost non-zero off-diagonal element on any row of $R$ should also be the bottommost non-zero off-diagonal element of its column. For example, this condition is satisfied by $\Lambda_{\textrm{unit}}$ of Eq.~(\ref{lambdaunitform}).
    The Jordan decomposition of $R$ satisfying these conditions reads $R = S J S^{\mm}$ where $S$ is an upper triangular matrix with all its diagonal entries non-zero, and $J$ is the Jordan normal form of $R$ that has the property:
    \begin{enumerate}
        \item [(P1)] $J_{ij} = 1$ for some $i < j$ only if $R_{ij} \neq 0$ and $R_{ik} = 0\;\; \forall\; i < k < j$. 
    \end{enumerate}
    The property of Jordan normal form $J$ translates to the following: the non-zero off-diagonal elements of $J$ are in the same positions as the leftmost non-zero off-diagonal elements in any row of $R$. Thus, for $R$ satisfying conditions (C1) and (C2), the Jordan normal form is obtained by replacing the first non-zero off-diagonal element in each row by $1$. 
\end{lemma}
\begin{proof}
    We proceed via induction on the matrix dimension $d$. We assume that Lemma~\ref{usefullemma2} holds for $(d - 1)$-dimensional matrices and show that it holds for $d$-dimensional matrices.
    That is, for a $d$-dimensional matrix $R$, we assume that the $(d-1)$-dimensional submatrix formed by the first $(d-1)$ rows and $(d-1)$ columns is a Jordan normal form (i.e. the only off-diagonal elements are $1$).  
    We then focus on the last column of the $d$-dimensional matrix $R$ and focus on one element at a time starting from $R_{d-1, d}$ and working up the column to $R_{1d}$. 
    At any step, if $R_{md} = t \neq 0$ for some $1 \leq m < d$, there are two possible cases:
    \begin{enumerate}
        \item $R_{mj} = 0\;\;\forall\ j,\; m < j \leq d - 1$.\\
        In this case, we know that $R_{nd} = 0\;\; \forall\; m < n < d - 1$ because of condition 2 in Lemma~\ref{usefullemma2}. 
        We apply a similarity transformation to $R$ using a diagonal matrix $\Delta$ whose components read
        \begin{equation}
            \Delta_{ii} = \twopartdef{\frac{1}{t}}{i = d}{1}{i \neq d}.
        \label{Deldef}
        \end{equation}
        The resulting matrix 
        \begin{equation}
            R' = \Delta^{\mm} R \Delta
        \label{Rpdef}
        \end{equation} 
        has the property that $\left(R'\right)_{md} = 1$. 
        For example, we consider $R$ reads ($d = 4$) 
        \begin{equation}
            R = 
            \begin{pmatrix}
                \lambda & 1 & 0 & 5 \\
                0 & \lambda & 0 & 2 \\
                0 & 0 & \lambda & 0 \\
                0 & 0 & 0 & \lambda
            \end{pmatrix}, 
        \end{equation}
        and focus on $m = 2$. Thus,  $R_{24} \neq 0$ and $R_{2j} = 0\;\;\forall\ j,\; 2 < j \leq 3$. Using $\Delta$ that reads
        \begin{equation}
            \Delta =
            \begin{pmatrix}
                1 & 0 & 0 & 0 \\
                0 & 1 & 0 & 0 \\
                0 & 0 & 1 & 0 \\
                0 & 0 & 0 & \frac{1}{2}
            \end{pmatrix}, 
        \label{Deltadefexamp}
        \end{equation}
        we obtain $R'$ of Eq.~(\ref{Rpdef}) reads
        \begin{equation}
        R' = 
            \begin{pmatrix}
                \lambda & 1 & 0 & \frac{5}{2} \\
                0 & \lambda & 0 & 1 \\
                0 & 0 & \lambda & 0 \\
                0 & 0 & 0 & \lambda
            \end{pmatrix}.
        \end{equation}

        \item $R_{mn} = 1$ for one $n$, $m < n \leq d - 1$. \\
        We never obtain the case $R_{mn} = 1$ for more than one $n$, $m < n \leq d -1$ because the submatrix consisting of the first $(d-1)$ rows and first $(d-1)$ columns is a Jordan normal form due to the induction hypothesis. 
        Here we apply a similarity transformation using an upper triangular matrix $T$ whose components read
        \begin{equation}
            T_{ij} = \delta_{ij} - t \delta_{in}\delta_{jd}.
        \label{Tdef}
        \end{equation}
        The resulting matrix
        \begin{equation}
            R' = T^{\mm} R T
        \label{Rpdef2}
        \end{equation}
        has the property $R'_{md} = 0$. 
        For example, we consider $R$ that reads ($d = 4$)
        \begin{equation}
            R = 
            \begin{pmatrix}
                \lambda & 1 & 0 & 5 \\
                0 & \lambda & 1 & 0 \\
                0 & 0 & \lambda & 1 \\
                0 & 0 & 0 & \lambda
            \end{pmatrix},
        \end{equation}
        and we focus on $m = 1$. Thus, $R_{14} \neq 0$ and $R_{12} = 1$. The corresponding $T$ is
        \begin{equation}
            T = 
            \begin{pmatrix}
                1 & 0 & 0 & 0 \\
                0 & 1 & 0 & -5 \\
                0 & 0 & 1 & 0 \\
                0 & 0 & 0 & 1
            \end{pmatrix}.
        \end{equation}
        We then obtain the following expression for $R'$ of Eq.~(\ref{Rpdef2}):
        \begin{equation}
            R' = 
            \begin{pmatrix}
                \lambda & 1 & 0 & 0 \\
                0 & \lambda & 1 & 0 \\
                0 & 0 & \lambda & 1 \\
                0 & 0 & 0 & \lambda
            \end{pmatrix}.
        \end{equation}
    \end{enumerate}
Thus, by sequentially applying similarity transformations Eqs.~(\ref{Rpdef}) and (\ref{Rpdef2}), we transform the entries of the last column of $R$ to either $1$ or $0$, resulting in a matrix $J$ that satisfies property (P1) of Lemma~\ref{usefullemma2}.
Since the full similarity transformation $S$ is a product of diagonal matrices with only non-zero elements on its diagonal ($\Delta$'s of Eq.~(\ref{Deldef})) and upper triangular matrices ($T$'s of Eq.~(\ref{Tdef})) with only non-zero elements on its diagonal, we obtain %
\begin{equation}
    R = S J S^{\mm}, 
\end{equation}
where $S$ is an upper triangular matrix with only non-zero elements along its diagonal and $J$ is the Jordan normal form of $M$.
This shows that 
\end{proof}
$\wLam$ of Eq.~(\ref{lambdaunitform}) is a direct sum of two matrices (one for the generalized eigenvalues +1, one for the generalized eigenvalues -1), both of which satisfy the conditions of the Lemma~\ref{usefullemma2}.
This validates that only the off-diagonal matrix elements $s$ (first non-zero off-diagonal elements in each row) in $\wLam$ are relevant when finding the non-zero upper-diagonal element in the Jordan normal form.
Moreover, $\bS_{unit}$ of Eq.~(\ref{lamSJdirsums}) is an upper triangular matrix with non-zero elements on its diagonal as a consequence of Lemma~\ref{usefullemma2}.
Thus, $\wP$ and $\wwP$ have the forms of Eq.~(\ref{bPforms}), and the left and right generalized eigenvectors of $\bM$ of Eq.~(\ref{bMN=2Transfer}) have the forms of Eq.~(\ref{eigenvectorforms}).
Furthermore, applying Lemma~\ref{usefullemma2}, $\wJ$ reads
\begin{eqnarray}
\wJ =
    \begin{pmatrix}
        1 & 0 & 0 & 0 & 1 & 0 & 0 & 0 & 0 \\
        0 & -1 & 0 & 0 & 0 & 1 & 0 & 0 & 0 \\ 
        0 & 0 & 1 & 0 & 0 & 0 & 0 & 0 & 0 \\
        0 & 0 & 0 & -1 & 0 & 0 & 0 & 1 & 0 \\
        0 & 0 & 0 & 0 & 1 & 0 & 0 & 0 & 1 \\
        0 & 0 & 0 & 0 & 0 & -1 & 0 & 0 & 0 \\
        0 & 0 & 0 & 0 & 0 & 0 & 1 & 0 & 0 \\
        0 & 0 & 0 & 0 & 0 & 0 & 0 & -1 & 0 \\
        0 & 0 & 0 & 0 & 0 & 0 & 0 & 0 & 1 \\
    \end{pmatrix}.
\label{JBtower}
\end{eqnarray}
\section{Asymptotic behavior of the tower of states entanglement entropy}\label{sec:towerentropy}
To obtain the large-$N$ behavior of Eq.~(\ref{entropylowk}), we first use Stirling approximation to obtain
\begin{equation}
    \binom{N}{\alpha} \sim e^{N H\left(\frac{\alpha}{N}\right)} \sqrt{\frac{1}{2 \pi \alpha \left(1 - \frac{\alpha}{N}\right)}}
\label{stirlingcomb}
\end{equation}
where $H(x) \equiv -x\log x - (1-x)\log(1-x)$ is the Shannon entropy function.
The sum in Eq.~(\ref{entropylowk}) can then be written as
\begin{eqnarray}
    &I \equiv \sumal{\alpha = 0}{N}{\binom{N}{\alpha} \log \binom{N}{\alpha}}\nn \\
    &\sim \sumal{\alpha = 0}{N}{\frac{e^{N H\left(\frac{\alpha}{N}\right)}}{\sqrt{2 \pi \alpha \left(1 - \frac{\alpha}{N}\right)}} \left(N H\left(\frac{\alpha}{N}\right) - \frac{1}{2}\log\left(2 \pi \alpha \left(1 - \frac{\alpha}{N}\right)\right)\right)} \nn \\
    &\approx N \int_{0}^1{dp\ \frac{e^{N H(p)}}{\sqrt{2\pi N p (1-p)}} \left(N H(p) - \frac{1}{2}\log(2\pi N p (1-p))\right)},\nn \\
\label{integral}
\end{eqnarray}
where $p = \frac{\alpha}{N}$.
To evaluate Eq.~(\ref{integral}) for large $N$, we can use a saddle point approximation: 
\begin{eqnarray}
    \int_{a}^{b}{dx\ g(x) e^{N f(x)}} &\approx& \int_{a}^{b}{dx\ g(x) e^{N f(x_0) + \frac{N}{2}f''(x_0) (x - x_0)^2}} \nn \\
    &=& g(x_0) e^{N f(x_0)} \sqrt{\frac{2 \pi}{N |f''(x_0)|}},
\end{eqnarray}
where $f'(x_0) = 0$ such that $a < x_0 < b$ and $f''(x_0) < 0$.
Thus, we obtain
\begin{eqnarray}
    &I = N \frac{e^{N H(p_0)}}{\sqrt{2 \pi N p_0 (1-p_0)}} \sqrt{\frac{2 \pi}{N |H''(p_0)|}}\nn \\
    &\times \left(N H(p_0) - \frac{1}{2}\log(2 \pi N p_0 (1 - p_0))\right)
\label{saddlepointI}
\end{eqnarray}
where $p_0$ is defined by $H'(p_0) = 0$.
Substituting $p_0 = 1/2$, $H(p_0) = \log 2$ and $H''(p_0) = -4$, Eq.~(\ref{saddlepointI}) simplifies to
\begin{equation}
    I = 2^N\left( N \log 2 - \frac{1}{2}\log\left(\frac{\pi N}{2}\right)\right) 
\end{equation}
Substituting $I$ into Eq.~(\ref{entropylowk}), we obtain Eq.~(\ref{entropylargeN}).

\section{Breakdown of Eq.~(\ref{rhocopies}) in the finite density limit}\label{sec:breakdown}
In Sec.~\ref{sec:finitedensity}, we mentioned that terms weighted by $\binom{n}{a}\binom{n}{k-a}$ do not suppress the terms that appear with the factor $\binom{n}{a}\binom{n}{k-a-b}$.
To see that this is indeed the case, we write $N = p n$, where $p > 0$, and use the asymptotic form
\begin{equation}
    \binom{n}{pn} \sim e^{n H(p)}
\label{finitedensityapprox}
\end{equation}
where $H(x) \equiv -x \log x - (1-x)\log(1-x)$, the Shannon entropy function. 
Expansions in orders of $n$ breaks down if, for some finite $b$,
\begin{equation}
    \exists\; k_1, k_2\; \ni\; \binom{n}{k_1}\binom{n}{k-k_1} < \binom{n}{k_2}\binom{n}{k-b-k_2}.
\label{nomixing}
\end{equation}
In terms of the $H$-function, condition Eq.~(\ref{nomixing}) translates to
\begin{equation}
    \exists\; p_1, p_2\; \ni\; H(p_1) + H(p - p_1) < H(p_2) + H(p-p_2 - \frac{b}{n}),
\label{nomixingp}
\end{equation}
where $p_1 \equiv k_1/n$, $p_2 \equiv k_2/n$ and $p \equiv k/n$.
However, by using $p_1 \rightarrow 0$ and $p_2 \rightarrow \frac{p}{2} - \frac{b}{2n}$, the condition Eq.~(\ref{nomixingp}) reduces to
\begin{equation}
    H(p) < 2 H(\frac{p}{2} - \frac{b}{2n}).
\label{mixingconditionreduced}
\end{equation}
Since $H(x)$ is a strictly concave function for $x \in [0, 1]$, we know that
\begin{equation}
    \frac{H(x) + H(y)}{2} < H(\frac{x+y}{2}).
\end{equation}
Using $x = p$ and $y = 0$, for any non-zero $p$ we obtain
\begin{equation}
    H(p) < 2 H(\frac{p}{2}),
\end{equation}
which is the same as Eq.~(\ref{mixingconditionreduced}) in the limit $n \rightarrow \infty$ and $b$ finite.
Thus, the replica structure of the ground state and excited state entanglement spectra breaks down at any non-zero energy density. 
For small densities $p$, we expect the prefactor in Eq.~(\ref{prefactor}) to be $P = 1/2$, the same as the one in the zero density limit.
The entropy contribution in Eq.~(\ref{entropylargeN}) is due to the copies of the ground state with $\alpha \sim N/2$ in Eq.~(\ref{entropylowk}).
Using Eqs.~(\ref{generalNeig}), (\ref{N=2CR}) and (\ref{N=2CL}), the dominant corrections to those eigenvalues of $\rhored$ are due to products of the form
\begin{equation}
    l_{\alpha,\beta} r_{\gamma, \delta}^T,\;\;\;\textrm{where}\;\; \alpha, \beta, \gamma, \delta \sim N/2 - \epsilon n.
\end{equation}
However, using Eq.~(\ref{finitedensityapprox}), such terms are suppressed by a factor of 
\begin{eqnarray}
    \frac{\binom{n}{N/2 - \epsilon n}^2}{\binom{n}{N/2}^2} &\sim& e^{4 (H\left(\frac{p}{2} - \epsilon \right) - H\left(\frac{p}{2} \right))} \nn \\
    &\sim& e^{-4 H'(p/2) \epsilon n}.
\end{eqnarray}
Thus we do not expect the saddle point form of the entropy in Eq.~(\ref{entropylargeN}) to change for small $p$. 

\section{Transformation of MPOs corresponding to the AKLT excited states under various symmetries}\label{sec:symtransform}
In this section, we describe the transformation of the MPOs corresponding to the AKLT excited states under inversion, time-reversal and $\mathbb{Z}_2 \times \mathbb{Z}_2$ rotation symmetries.
\subsection{Inversion symmetry}
Under Inversion symmetry, all the physical operators are mapped to themselves.
\begin{equation}
    \mathcal{I}: (S^+, S^-, S^z) \rightarrow (S^+, S^-, S^z)
\end{equation}
The MPOs transform under inversion as described in Eq.~(\ref{mpoinversiontransform}).
The Arovas A MPO $M_A$ of Eq.~(\ref{ArovasA}) satisfies
\begin{equation}
    \Sigma_I^A M_A {\Sigma_I^A}^\dagger = M_A^T,
\end{equation}
where by brute force we obtain
\begin{equation}
    \Sigma_I^A =
    \begin{pmatrix}
    0 & 0 & 0 & 0 & 1 \\
    0 & 0 & -1 & 0 & 0 \\
    0 & -1 & 0 & 0 & 0 \\
    0 & 0 & 0 & 1 & 0 \\
    1 & 0 & 0 & 0 & 0 \\
    \end{pmatrix}.
\label{firstarovasinv}  
\end{equation}
Since $\Sigma_I^A {\Sigma_I^A}^\ast = +\mathds{1}$, the Arovas A MPO transforms linearly under inversion.
Similarly, the Arovas B MPO $M_B$ of Eq.~(\ref{secondarovasmpo}) satisfies
\begin{equation}
    \Sigma_I^B M_B {\Sigma_I^B}^\dagger = -M_B^T,
\end{equation}
where
\begin{equation}
    \Sigma_I^B = 
    \begin{pmatrix}
        0 & 0 & 0 & 0 & 0 & 0 & 0 & \mm \\
        0 & 0 & 0 & 0 & 0 & \mm & 0 & 0 \\
        0 & 0 & 0 & 0 & 1 & 0 & 0 & 0 \\
        0 & 0 & 0 & 0 & 0 & 0 & 1 & 0 \\
        0 & 0 & \mm & 0 & 0 & 0 & 0 & 0 \\
        0 & \mm & 0 & 0 & 0 & 0 & 0 & 0 \\
        0 & 0 & 0 & \mm & 0 & 0 & 0 & 0 \\
        1 & 0 & 0 & 0 & 0 & 0 & 0 & 0 
    \end{pmatrix}.
    \label{secondarovasinv}
\end{equation}
Since $\Sigma_I^B {\Sigma_I^B}^\ast = -\mathds{1}$, the Arovas B MPO tranforms projectively under inversion.
The tower of states MPO $M_{SS_{2N}}$ of Eq.~(\ref{TowerofstatesMPO}) transforms as
\begin{equation}
    \Sigma_I^t M_{SS_{2N}} {\Sigma_I^t}^\dagger = -M_{SS_{2N}}^T
\end{equation}
where 
\begin{equation}
    \Sigma_I^t = e^{i \pi S^y_a}
\label{towerinv}
\end{equation}
where $S^y_a$ is the spin-$N/2$ operator that acts on the $(N+1)$-dimensional ancilla.
For $N = 1$, $\Sigma_I^t = i \sigma_y$. 
Since
\begin{equation}
    \Sigma_I^t {\Sigma_I^t}^\ast = (-1)^{N}\mathds{1},
\label{towerlinearproj}
\end{equation}
the tower of states MPO transforms linearly for even $N$ and projectively for odd $N$.

\subsection{Time-reversal symmetry}
Under time-reversal, the physical integer spin operators transform as 
\begin{equation}
    \mathcal{T}:\;\;(S^+, S^-, S^z)\;\;\rightarrow\;\; (-S^-, -S^+, -S^z).
\label{TRaction}
\end{equation}
Thus the Arovas A MPO $M_A$ transforms as
\begin{equation}
    \Sigma_T^A M_A {\Sigma_T^A}^\dagger = \mathcal{T}(M_A),   
\end{equation}
where $\mathcal{T}$ transforms the physical operator in the MPO under Eq.~(\ref{TRaction}), and
\begin{equation}
    \Sigma_T^A = 
    \begin{pmatrix}
    1 & 0 & 0 & 0 & 0 \\
    0 & 0 & \mm & 0 & 0 \\
    0 & \mm & 0 & 0 & 0 \\
    0 & 0 & 0 & \mm & 0 \\
    0 & 0 & 0 & 0 & 0 \\
    \end{pmatrix}.
\label{firstarovastr}
\end{equation}
Since $\Sigma_T^A {\Sigma_T^A}^\ast = +\mathds{1}$, this is a linear transformation.
The Arovas B MPO $M_B$ transforms linearly as well with 
\begin{equation}
    \Sigma_T^B M_B {\Sigma_T^B}^\dagger = \mathcal{T}(M_B),   
\end{equation}
where
\begin{equation}
    \Sigma_T^B = 
    \begin{pmatrix}
        1 & 0 & 0 & 0 & 0 & 0 & 0 & 0 \\
        0 & 0 & \mm & 0 & 0 & 0 & 0 & 0 \\
        0 & \mm & 0 & 0 & 0 & 0 & 0 & 0 \\
        0 & 0 & 0 & \mm & 0 & 0 & 0 & 0 \\
        0 & 0 & 0 & 0 & 0 & \mm & 0 & 0 \\
        0 & 0 & 0 & 0 & \mm & 0 & 0 & 0 \\
        0 & 0 & 0 & 0 & 0 & 0 & \mm & 0 \\
        0 & 0 & 0 & 0 & 0 & 0 & 0 & 1 \\
    \end{pmatrix}.
\label{secondarovastr}
\end{equation}
The spin-$S$ AKLT tower of states we have considered have $S_z \neq 0$ and thus explicitly break time-reversal symmetry.

\subsection{Rotation symmetry}
Under $\pi$-rotations about the $x$ and $z$ axes, the physical integer spin operators transform as
\begin{eqnarray}
    \mathcal{R}_x:\;\;(S^+, S^-, S^z)\;\;\rightarrow\;\; (S^-, S^+, -S^z) \nn \\
    \mathcal{R}_z:\;\;(S^+, S^-, S^z)\;\;\rightarrow\;\; (-S^+, -S^-, S^z) \nn \\
\end{eqnarray}
Consequently, the Arovas A MPO transforms as
\begin{equation}
    (\Sigma_\sigma^A) M_A {\Sigma_\sigma^A}^\dagger = \mathcal{R}_\sigma M_A \;\; \sigma = x, z.
\label{arovasArot}
\end{equation}
where
\begin{eqnarray}
    &\Sigma_x^A = 
    \begin{pmatrix}
     1 & 0 & 0 & 0 & 0 \\
     0 & 0 & 1 & 0 & 0 \\
     0 & 1 & 0 & 0 & 0 \\
     0 & 0 & 0 & \mm & 0 \\
     0 & 0 & 0 & 0 & 1 \\
    \end{pmatrix} \nn \\
    &\Sigma_z^A = 
    \begin{pmatrix}
    1 & 0 & 0 & 0 & 0 \\
 0 & \mm & 0 & 0 & 0 \\
 0 & 0 & \mm & 0 & 0 \\
 0 & 0 & 0 & 1 & 0 \\
 0 & 0 & 0 & 0 & 1 \\
    \end{pmatrix}
\label{firstarovasrot}
\end{eqnarray}
This is a linear transformation since $\Sigma^A_x \Sigma^A_z ({\Sigma^A_x} {\Sigma^A_z})^\ast = + \mathds{1}$.
The Arovas B MPO also transforms similar to Eq.~(\ref{arovasArot}) where
\begin{eqnarray}
    &\Sigma_x^B =
    \begin{pmatrix}
         1 & 0 & 0 & 0 & 0 & 0 & 0 & 0 \\
         0 & 0 & 1 & 0 & 0 & 0 & 0 & 0 \\
         0 & 1 & 0 & 0 & 0 & 0 & 0 & 0 \\
         0 & 0 & 0 & \mm & 0 & 0 & 0 & 0 \\
         0 & 0 & 0 & 0 & 0 & 1 & 0 & 0 \\
         0 & 0 & 0 & 0 & 1 & 0 & 0 & 0 \\
         0 & 0 & 0 & 0 & 0 & 0 & \mm & 0 \\
         0 & 0 & 0 & 0 & 0 & 0 & 0 & 1 \\
    \end{pmatrix} \nn \\
    &\Sigma_z^B =
        \begin{pmatrix}
         1 & 0 & 0 & 0 & 0 & 0 & 0 & 0 \\
         0 & \mm & 0 & 0 & 0 & 0 & 0 & 0 \\
         0 & 0 & \mm & 0 & 0 & 0 & 0 & 0 \\
         0 & 0 & 0 & 1 & 0 & 0 & 0 & 0 \\
         0 & 0 & 0 & 0 & \mm & 0 & 0 & 0 \\
         0 & 0 & 0 & 0 & 0 & \mm & 0 & 0 \\
         0 & 0 & 0 & 0 & 0 & 0 & 1 & 0 \\
         0 & 0 & 0 & 0 & 0 & 0 & 0 & 1 \\
        \end{pmatrix}
\label{secondarovasrot}
\end{eqnarray}
Since $\Sigma^B_x \Sigma^B_z ({\Sigma^B_x} {\Sigma^B_z})^\ast = +\mathds{1}$, this is a linear transformation.
Since the tower of states does not have $S_z = 0$, the states are not invariant under the $\mathbb{Z}_2 \times \mathbb{Z}_2$ rotation symmetry.

\section{Symmetry-protected degeneracies in the entanglement spectrum for a finite system}\label{sec:InvSym}
In this section, we show that under certain conditions, degeneracies in the entanglement spectrum are protected (i.e. without any finite-size splitting, not even exponential) due to symmetries.
For example, if the system is inversion symmetric, we consider the case where the left and right boundary vectors of the MPS are related by
\begin{equation}
    U_I^\dagger b^r_A = b^l_A.
\label{boundaryspt}
\end{equation}
Here $U_I$ is the action of inversion symmetry on the ancilla, defined in Eq.~(\ref{inversiontransform}).
Using Eqs.~(\ref{MPSLRdefn}) and (\ref{inversiontransform}), we obtain the following property for the transfer matrix,
\begin{equation}
    E^T = (U_I^\dagger \otimes U_I^T) E (U_I \otimes U_I^\ast).
\end{equation}
Consequently, using Eqs.~(\ref{boundaryspt}) and (\ref{MPSLRdefn}), if the left and right subsystems have an equal size, we obtain
\begin{equation}
     \ml =  U_I^\dagger \mr U_I.
\end{equation}
Using the definition of $\rhored$ in Eq.~(\ref{rhored}),  we obtain
\begin{equation}
    \rhored = \ml \mr^T = U_I^\dagger \mr U_I U_I^\ast \ml^T U_I^T.
\label{rhocommute1}
\end{equation}
Since $U_I U_I^\ast = \pm 1$, and consequently $U_I^T = \pm U_I$, we obtain
\begin{equation}
    U_I \rhored = \rhored^T U_I.
\label{rhocommute}
\end{equation}
Suppose $\bR_\lambda$ is the right eigenvector of $\rhored$ corresponding to an eigenvalue $\lambda$, using Eq.~(\ref{rhocommute}), we obtain
\begin{equation}
    \rhored^T U_I \bR_\lambda = U_I \rhored \bR_\lambda = \lambda U_I \bR_\lambda. 
\end{equation}
Thus, $\bL_\lambda \equiv U_I \bR_\lambda$ is a right eigenvector of $\rhored^T$, and hence a left eigenvector of $\rhored$ corresponding to the eigenvalue $\lambda$. 
If $U_I^T = -U_I$, we show that 
\begin{eqnarray}
    \bL_\lambda^T \bR_\lambda &=& \bR_\lambda^T U_I^T \bR_\lambda \nn\\
    &=& -\bR_\lambda^T U_I \bR_\lambda \nn \\
    &=&  -\bR_\lambda^T \bL_\lambda = -\bL_\lambda^T \bR_\lambda.
\label{Uinvcont}
\end{eqnarray}
Thus, we obtain $\bL_\lambda^T \bR_\lambda = 0$, which is impossible if $\lambda$ is non-degenerate.
Consequently, all the eigenvalues of $\rhored$ are doubly degenerate. 
For other unitary symmetries we have considered, such as time-reversal and $\mathbb{Z}_2 \times \mathbb{Z}_2$ rotation, the boundary conditions satisfy
\begin{equation}
    U b^r_A = b^r_A\;\; \textrm{and} \;\; U b^l_A = b^l_A.
\end{equation}
Consequently, using Eq.~(\ref{sptsym}) and (\ref{MPSLRdefn}), we obtain
\begin{equation}
    \mr = U^\dagger \mr U\;\;\textrm{and}\;\;\ml = U^\dagger \ml U.
\end{equation}
Since $[\ml, U] = 0$ and $[\mr, U] = 0$, we obtain $[\rhored, U] = 0$.
$U$ can thus be block-diagonalized into blocks $U_\lambda$ (of dimension $D_\lambda$) labelled by the eigenvalues $\lambda$ of $\rhored$.
Since $U_\lambda$ is antisymmetic, it satisfies
\begin{equation}
    \det(U_\lambda) = \det(U_\lambda^T) = (-1)^{D_\lambda} \det(U_\lambda)
\end{equation}
However, $U_\lambda$ is also unitary and thus $\det(U_\lambda) \neq 0$, requiring $D_\lambda$ to be even.

\bibliography{aklt_entanglement}

\begin{thebibliography}{98}%
\makeatletter
\providecommand \@ifxundefined [1]{%
 \@ifx{#1\undefined}
}%
\providecommand \@ifnum [1]{%
 \ifnum #1\expandafter \@firstoftwo
 \else \expandafter \@secondoftwo
 \fi
}%
\providecommand \@ifx [1]{%
 \ifx #1\expandafter \@firstoftwo
 \else \expandafter \@secondoftwo
 \fi
}%
\providecommand \natexlab [1]{#1}%
\providecommand \enquote  [1]{``#1''}%
\providecommand \bibnamefont  [1]{#1}%
\providecommand \bibfnamefont [1]{#1}%
\providecommand \citenamefont [1]{#1}%
\providecommand \href@noop [0]{\@secondoftwo}%
\providecommand \href [0]{\begingroup \@sanitize@url \@href}%
\providecommand \@href[1]{\@@startlink{#1}\@@href}%
\providecommand \@@href[1]{\endgroup#1\@@endlink}%
\providecommand \@sanitize@url [0]{\catcode `\\12\catcode `\$12\catcode
  `\&12\catcode `\#12\catcode `\^12\catcode `\_12\catcode `\%12\relax}%
\providecommand \@@startlink[1]{}%
\providecommand \@@endlink[0]{}%
\providecommand \url  [0]{\begingroup\@sanitize@url \@url }%
\providecommand \@url [1]{\endgroup\@href {#1}{\urlprefix }}%
\providecommand \urlprefix  [0]{URL }%
\providecommand \Eprint [0]{\href }%
\providecommand \doibase [0]{http://dx.doi.org/}%
\providecommand \selectlanguage [0]{\@gobble}%
\providecommand \bibinfo  [0]{\@secondoftwo}%
\providecommand \bibfield  [0]{\@secondoftwo}%
\providecommand \translation [1]{[#1]}%
\providecommand \BibitemOpen [0]{}%
\providecommand \bibitemStop [0]{}%
\providecommand \bibitemNoStop [0]{.\EOS\space}%
\providecommand \EOS [0]{\spacefactor3000\relax}%
\providecommand \BibitemShut  [1]{\csname bibitem#1\endcsname}%
\let\auto@bib@innerbib\@empty
\bibitem [{\citenamefont {Rigol}\ \emph {et~al.}(2008)\citenamefont {Rigol},
  \citenamefont {Dunjko},\ and\ \citenamefont
  {Olshanii}}]{rigol2008thermalization}%
  \BibitemOpen
  \bibfield  {author} {\bibinfo {author} {\bibfnamefont {M.}~\bibnamefont
  {Rigol}}, \bibinfo {author} {\bibfnamefont {V.}~\bibnamefont {Dunjko}}, \
  and\ \bibinfo {author} {\bibfnamefont {M.}~\bibnamefont {Olshanii}},\
  }\href@noop {} {\bibfield  {journal} {\bibinfo  {journal} {Nature}\ }\textbf
  {\bibinfo {volume} {452}},\ \bibinfo {pages} {854} (\bibinfo {year}
  {2008})}\BibitemShut {NoStop}%
\bibitem [{\citenamefont {Basko}\ \emph {et~al.}(2006)\citenamefont {Basko},
  \citenamefont {Aleiner},\ and\ \citenamefont {Altshuler}}]{basko2006metal}%
  \BibitemOpen
  \bibfield  {author} {\bibinfo {author} {\bibfnamefont {D.}~\bibnamefont
  {Basko}}, \bibinfo {author} {\bibfnamefont {I.}~\bibnamefont {Aleiner}}, \
  and\ \bibinfo {author} {\bibfnamefont {B.}~\bibnamefont {Altshuler}},\
  }\href@noop {} {\bibfield  {journal} {\bibinfo  {journal} {Annals of
  physics}\ }\textbf {\bibinfo {volume} {321}},\ \bibinfo {pages} {1126}
  (\bibinfo {year} {2006})}\BibitemShut {NoStop}%
\bibitem [{\citenamefont {Nandkishore}\ and\ \citenamefont
  {Huse}(2015)}]{nandkishore2014many}%
  \BibitemOpen
  \bibfield  {author} {\bibinfo {author} {\bibfnamefont {R.}~\bibnamefont
  {Nandkishore}}\ and\ \bibinfo {author} {\bibfnamefont {D.~A.}\ \bibnamefont
  {Huse}},\ }\href@noop {} {\bibfield  {journal} {\bibinfo  {journal} {Annu.
  Rev. Condens. Matter Phys.}\ }\textbf {\bibinfo {volume} {6}},\ \bibinfo
  {pages} {15} (\bibinfo {year} {2015})}\BibitemShut {NoStop}%
\bibitem [{\citenamefont {Khemani}\ \emph {et~al.}(2017)\citenamefont
  {Khemani}, \citenamefont {Sheng},\ and\ \citenamefont
  {Huse}}]{khemani2017two}%
  \BibitemOpen
  \bibfield  {author} {\bibinfo {author} {\bibfnamefont {V.}~\bibnamefont
  {Khemani}}, \bibinfo {author} {\bibfnamefont {D.}~\bibnamefont {Sheng}}, \
  and\ \bibinfo {author} {\bibfnamefont {D.~A.}\ \bibnamefont {Huse}},\
  }\href@noop {} {\bibfield  {journal} {\bibinfo  {journal} {Physical Review
  Letters}\ }\textbf {\bibinfo {volume} {119}},\ \bibinfo {pages} {075702}
  (\bibinfo {year} {2017})}\BibitemShut {NoStop}%
\bibitem [{\citenamefont {Pal}\ and\ \citenamefont {Huse}(2010)}]{pal2010many}%
  \BibitemOpen
  \bibfield  {author} {\bibinfo {author} {\bibfnamefont {A.}~\bibnamefont
  {Pal}}\ and\ \bibinfo {author} {\bibfnamefont {D.~A.}\ \bibnamefont {Huse}},\
  }\href@noop {} {\bibfield  {journal} {\bibinfo  {journal} {Physical Review
  b}\ }\textbf {\bibinfo {volume} {82}},\ \bibinfo {pages} {174411} (\bibinfo
  {year} {2010})}\BibitemShut {NoStop}%
\bibitem [{\citenamefont {Kj{\"a}ll}\ \emph {et~al.}(2014)\citenamefont
  {Kj{\"a}ll}, \citenamefont {Bardarson},\ and\ \citenamefont
  {Pollmann}}]{kjall2014many}%
  \BibitemOpen
  \bibfield  {author} {\bibinfo {author} {\bibfnamefont {J.~A.}\ \bibnamefont
  {Kj{\"a}ll}}, \bibinfo {author} {\bibfnamefont {J.~H.}\ \bibnamefont
  {Bardarson}}, \ and\ \bibinfo {author} {\bibfnamefont {F.}~\bibnamefont
  {Pollmann}},\ }\href@noop {} {\bibfield  {journal} {\bibinfo  {journal}
  {Physical Review Letters}\ }\textbf {\bibinfo {volume} {113}},\ \bibinfo
  {pages} {107204} (\bibinfo {year} {2014})}\BibitemShut {NoStop}%
\bibitem [{\citenamefont {{\v{Z}}nidari{\v{c}}}\ \emph
  {et~al.}(2008)\citenamefont {{\v{Z}}nidari{\v{c}}}, \citenamefont {Prosen},\
  and\ \citenamefont {Prelov{\v{s}}ek}}]{vznidarivc2008many}%
  \BibitemOpen
  \bibfield  {author} {\bibinfo {author} {\bibfnamefont {M.}~\bibnamefont
  {{\v{Z}}nidari{\v{c}}}}, \bibinfo {author} {\bibfnamefont {T.}~\bibnamefont
  {Prosen}}, \ and\ \bibinfo {author} {\bibfnamefont {P.}~\bibnamefont
  {Prelov{\v{s}}ek}},\ }\href@noop {} {\bibfield  {journal} {\bibinfo
  {journal} {Physical Review B}\ }\textbf {\bibinfo {volume} {77}},\ \bibinfo
  {pages} {064426} (\bibinfo {year} {2008})}\BibitemShut {NoStop}%
\bibitem [{\citenamefont {Luitz}\ \emph {et~al.}(2015)\citenamefont {Luitz},
  \citenamefont {Laflorencie},\ and\ \citenamefont {Alet}}]{luitz2015many}%
  \BibitemOpen
  \bibfield  {author} {\bibinfo {author} {\bibfnamefont {D.~J.}\ \bibnamefont
  {Luitz}}, \bibinfo {author} {\bibfnamefont {N.}~\bibnamefont {Laflorencie}},
  \ and\ \bibinfo {author} {\bibfnamefont {F.}~\bibnamefont {Alet}},\
  }\href@noop {} {\bibfield  {journal} {\bibinfo  {journal} {Physical Review
  B}\ }\textbf {\bibinfo {volume} {91}},\ \bibinfo {pages} {081103} (\bibinfo
  {year} {2015})}\BibitemShut {NoStop}%
\bibitem [{\citenamefont {Bardarson}\ \emph {et~al.}(2012)\citenamefont
  {Bardarson}, \citenamefont {Pollmann},\ and\ \citenamefont
  {Moore}}]{bardarson2012unbounded}%
  \BibitemOpen
  \bibfield  {author} {\bibinfo {author} {\bibfnamefont {J.~H.}\ \bibnamefont
  {Bardarson}}, \bibinfo {author} {\bibfnamefont {F.}~\bibnamefont {Pollmann}},
  \ and\ \bibinfo {author} {\bibfnamefont {J.~E.}\ \bibnamefont {Moore}},\
  }\href@noop {} {\bibfield  {journal} {\bibinfo  {journal} {Physical Review
  Letters}\ }\textbf {\bibinfo {volume} {109}},\ \bibinfo {pages} {017202}
  (\bibinfo {year} {2012})}\BibitemShut {NoStop}%
\bibitem [{\citenamefont {Iyer}\ \emph {et~al.}(2013)\citenamefont {Iyer},
  \citenamefont {Oganesyan}, \citenamefont {Refael},\ and\ \citenamefont
  {Huse}}]{iyer2013many}%
  \BibitemOpen
  \bibfield  {author} {\bibinfo {author} {\bibfnamefont {S.}~\bibnamefont
  {Iyer}}, \bibinfo {author} {\bibfnamefont {V.}~\bibnamefont {Oganesyan}},
  \bibinfo {author} {\bibfnamefont {G.}~\bibnamefont {Refael}}, \ and\ \bibinfo
  {author} {\bibfnamefont {D.~A.}\ \bibnamefont {Huse}},\ }\href@noop {}
  {\bibfield  {journal} {\bibinfo  {journal} {Physical Review B}\ }\textbf
  {\bibinfo {volume} {87}},\ \bibinfo {pages} {134202} (\bibinfo {year}
  {2013})}\BibitemShut {NoStop}%
\bibitem [{\citenamefont {Rigol}(2009)}]{rigol2009breakdown}%
  \BibitemOpen
  \bibfield  {author} {\bibinfo {author} {\bibfnamefont {M.}~\bibnamefont
  {Rigol}},\ }\href@noop {} {\bibfield  {journal} {\bibinfo  {journal}
  {Physical Review Letters}\ }\textbf {\bibinfo {volume} {103}},\ \bibinfo
  {pages} {100403} (\bibinfo {year} {2009})}\BibitemShut {NoStop}%
\bibitem [{\citenamefont {Geraedts}\ \emph {et~al.}(2017)\citenamefont
  {Geraedts}, \citenamefont {Regnault},\ and\ \citenamefont
  {Nandkishore}}]{geraedts2017characterizing}%
  \BibitemOpen
  \bibfield  {author} {\bibinfo {author} {\bibfnamefont {S.~D.}\ \bibnamefont
  {Geraedts}}, \bibinfo {author} {\bibfnamefont {N.}~\bibnamefont {Regnault}},
  \ and\ \bibinfo {author} {\bibfnamefont {R.~M.}\ \bibnamefont
  {Nandkishore}},\ }\href@noop {} {\bibfield  {journal} {\bibinfo  {journal}
  {New Journal of Physics}\ }\textbf {\bibinfo {volume} {19}},\ \bibinfo
  {pages} {113021} (\bibinfo {year} {2017})}\BibitemShut {NoStop}%
\bibitem [{\citenamefont {Chandran}\ \emph {et~al.}(2014)\citenamefont
  {Chandran}, \citenamefont {Khemani}, \citenamefont {Laumann},\ and\
  \citenamefont {Sondhi}}]{chandran2014many}%
  \BibitemOpen
  \bibfield  {author} {\bibinfo {author} {\bibfnamefont {A.}~\bibnamefont
  {Chandran}}, \bibinfo {author} {\bibfnamefont {V.}~\bibnamefont {Khemani}},
  \bibinfo {author} {\bibfnamefont {C.}~\bibnamefont {Laumann}}, \ and\
  \bibinfo {author} {\bibfnamefont {S.}~\bibnamefont {Sondhi}},\ }\href@noop {}
  {\bibfield  {journal} {\bibinfo  {journal} {Physical Review B}\ }\textbf
  {\bibinfo {volume} {89}},\ \bibinfo {pages} {144201} (\bibinfo {year}
  {2014})}\BibitemShut {NoStop}%
\bibitem [{\citenamefont {Fendley}(2016)}]{fendley2016strong}%
  \BibitemOpen
  \bibfield  {author} {\bibinfo {author} {\bibfnamefont {P.}~\bibnamefont
  {Fendley}},\ }\href@noop {} {\bibfield  {journal} {\bibinfo  {journal}
  {Journal of Physics A: Mathematical and Theoretical}\ }\textbf {\bibinfo
  {volume} {49}},\ \bibinfo {pages} {30LT01} (\bibinfo {year}
  {2016})}\BibitemShut {NoStop}%
\bibitem [{\citenamefont {Deutsch}(1991)}]{deutsch1991quantum}%
  \BibitemOpen
  \bibfield  {author} {\bibinfo {author} {\bibfnamefont {J.}~\bibnamefont
  {Deutsch}},\ }\href@noop {} {\bibfield  {journal} {\bibinfo  {journal}
  {Physical Review A}\ }\textbf {\bibinfo {volume} {43}},\ \bibinfo {pages}
  {2046} (\bibinfo {year} {1991})}\BibitemShut {NoStop}%
\bibitem [{\citenamefont {Srednicki}(1994)}]{srednicki1994chaos}%
  \BibitemOpen
  \bibfield  {author} {\bibinfo {author} {\bibfnamefont {M.}~\bibnamefont
  {Srednicki}},\ }\href@noop {} {\bibfield  {journal} {\bibinfo  {journal}
  {Physical Review E}\ }\textbf {\bibinfo {volume} {50}},\ \bibinfo {pages}
  {888} (\bibinfo {year} {1994})}\BibitemShut {NoStop}%
\bibitem [{\citenamefont {Huse}\ \emph {et~al.}(2013)\citenamefont {Huse},
  \citenamefont {Nandkishore}, \citenamefont {Oganesyan}, \citenamefont {Pal},\
  and\ \citenamefont {Sondhi}}]{huse2013localization}%
  \BibitemOpen
  \bibfield  {author} {\bibinfo {author} {\bibfnamefont {D.~A.}\ \bibnamefont
  {Huse}}, \bibinfo {author} {\bibfnamefont {R.}~\bibnamefont {Nandkishore}},
  \bibinfo {author} {\bibfnamefont {V.}~\bibnamefont {Oganesyan}}, \bibinfo
  {author} {\bibfnamefont {A.}~\bibnamefont {Pal}}, \ and\ \bibinfo {author}
  {\bibfnamefont {S.}~\bibnamefont {Sondhi}},\ }\href@noop {} {\bibfield
  {journal} {\bibinfo  {journal} {Physical Review B}\ }\textbf {\bibinfo
  {volume} {88}},\ \bibinfo {pages} {014206} (\bibinfo {year}
  {2013})}\BibitemShut {NoStop}%
\bibitem [{\citenamefont {Huang}(2017)}]{huang2017universal}%
  \BibitemOpen
  \bibfield  {author} {\bibinfo {author} {\bibfnamefont {Y.}~\bibnamefont
  {Huang}},\ }\href@noop {} {\bibfield  {journal} {\bibinfo  {journal} {arXiv
  preprint arXiv:1708.08607}\ } (\bibinfo {year} {2017})}\BibitemShut {NoStop}%
\bibitem [{\citenamefont {Lu}\ and\ \citenamefont
  {Grover}(2017)}]{lu2017renyi}%
  \BibitemOpen
  \bibfield  {author} {\bibinfo {author} {\bibfnamefont {T.-C.}\ \bibnamefont
  {Lu}}\ and\ \bibinfo {author} {\bibfnamefont {T.}~\bibnamefont {Grover}},\
  }\href@noop {} {\bibfield  {journal} {\bibinfo  {journal} {arXiv preprint
  arXiv:1709.08784}\ } (\bibinfo {year} {2017})}\BibitemShut {NoStop}%
\bibitem [{\citenamefont {Huang}\ and\ \citenamefont
  {Gu}(2017)}]{huang2017eigenstate}%
  \BibitemOpen
  \bibfield  {author} {\bibinfo {author} {\bibfnamefont {Y.}~\bibnamefont
  {Huang}}\ and\ \bibinfo {author} {\bibfnamefont {Y.}~\bibnamefont {Gu}},\
  }\href@noop {} {\bibfield  {journal} {\bibinfo  {journal} {arXiv preprint
  arXiv:1709.09160}\ } (\bibinfo {year} {2017})}\BibitemShut {NoStop}%
\bibitem [{\citenamefont {Liu}\ \emph {et~al.}(2018)\citenamefont {Liu},
  \citenamefont {Chen},\ and\ \citenamefont {Balents}}]{liu2018quantum}%
  \BibitemOpen
  \bibfield  {author} {\bibinfo {author} {\bibfnamefont {C.}~\bibnamefont
  {Liu}}, \bibinfo {author} {\bibfnamefont {X.}~\bibnamefont {Chen}}, \ and\
  \bibinfo {author} {\bibfnamefont {L.}~\bibnamefont {Balents}},\ }\href@noop
  {} {\bibfield  {journal} {\bibinfo  {journal} {Physical Review B}\ }\textbf
  {\bibinfo {volume} {97}},\ \bibinfo {pages} {245126} (\bibinfo {year}
  {2018})}\BibitemShut {NoStop}%
\bibitem [{\citenamefont {Haegeman}\ \emph
  {et~al.}(2013{\natexlab{a}})\citenamefont {Haegeman}, \citenamefont
  {Michalakis}, \citenamefont {Nachtergaele}, \citenamefont {Osborne},
  \citenamefont {Schuch},\ and\ \citenamefont
  {Verstraete}}]{haegeman2013elementary}%
  \BibitemOpen
  \bibfield  {author} {\bibinfo {author} {\bibfnamefont {J.}~\bibnamefont
  {Haegeman}}, \bibinfo {author} {\bibfnamefont {S.}~\bibnamefont
  {Michalakis}}, \bibinfo {author} {\bibfnamefont {B.}~\bibnamefont
  {Nachtergaele}}, \bibinfo {author} {\bibfnamefont {T.~J.}\ \bibnamefont
  {Osborne}}, \bibinfo {author} {\bibfnamefont {N.}~\bibnamefont {Schuch}}, \
  and\ \bibinfo {author} {\bibfnamefont {F.}~\bibnamefont {Verstraete}},\
  }\href@noop {} {\bibfield  {journal} {\bibinfo  {journal} {Physical Review
  Letters}\ }\textbf {\bibinfo {volume} {111}},\ \bibinfo {pages} {080401}
  (\bibinfo {year} {2013}{\natexlab{a}})}\BibitemShut {NoStop}%
\bibitem [{\citenamefont {Pizorn}(2012)}]{pizorn2012universality}%
  \BibitemOpen
  \bibfield  {author} {\bibinfo {author} {\bibfnamefont {I.}~\bibnamefont
  {Pizorn}},\ }\href@noop {} {\bibfield  {journal} {\bibinfo  {journal} {arXiv
  preprint arXiv:1202.3336}\ } (\bibinfo {year} {2012})}\BibitemShut {NoStop}%
\bibitem [{\citenamefont {Haegeman}\ \emph {et~al.}(2012)\citenamefont
  {Haegeman}, \citenamefont {Pirvu}, \citenamefont {Weir}, \citenamefont
  {Cirac}, \citenamefont {Osborne}, \citenamefont {Verschelde},\ and\
  \citenamefont {Verstraete}}]{haegeman2012variational}%
  \BibitemOpen
  \bibfield  {author} {\bibinfo {author} {\bibfnamefont {J.}~\bibnamefont
  {Haegeman}}, \bibinfo {author} {\bibfnamefont {B.}~\bibnamefont {Pirvu}},
  \bibinfo {author} {\bibfnamefont {D.~J.}\ \bibnamefont {Weir}}, \bibinfo
  {author} {\bibfnamefont {J.~I.}\ \bibnamefont {Cirac}}, \bibinfo {author}
  {\bibfnamefont {T.~J.}\ \bibnamefont {Osborne}}, \bibinfo {author}
  {\bibfnamefont {H.}~\bibnamefont {Verschelde}}, \ and\ \bibinfo {author}
  {\bibfnamefont {F.}~\bibnamefont {Verstraete}},\ }\href@noop {} {\bibfield
  {journal} {\bibinfo  {journal} {Physical Review B}\ }\textbf {\bibinfo
  {volume} {85}},\ \bibinfo {pages} {100408} (\bibinfo {year}
  {2012})}\BibitemShut {NoStop}%
\bibitem [{\citenamefont {Haegeman}\ \emph
  {et~al.}(2013{\natexlab{b}})\citenamefont {Haegeman}, \citenamefont
  {Osborne},\ and\ \citenamefont {Verstraete}}]{haegeman2013post}%
  \BibitemOpen
  \bibfield  {author} {\bibinfo {author} {\bibfnamefont {J.}~\bibnamefont
  {Haegeman}}, \bibinfo {author} {\bibfnamefont {T.~J.}\ \bibnamefont
  {Osborne}}, \ and\ \bibinfo {author} {\bibfnamefont {F.}~\bibnamefont
  {Verstraete}},\ }\href@noop {} {\bibfield  {journal} {\bibinfo  {journal}
  {Physical Review B}\ }\textbf {\bibinfo {volume} {88}},\ \bibinfo {pages}
  {075133} (\bibinfo {year} {2013}{\natexlab{b}})}\BibitemShut {NoStop}%
\bibitem [{\citenamefont {Haegeman}\ and\ \citenamefont
  {Verstraete}(2017)}]{haegeman2017diagonalizing}%
  \BibitemOpen
  \bibfield  {author} {\bibinfo {author} {\bibfnamefont {J.}~\bibnamefont
  {Haegeman}}\ and\ \bibinfo {author} {\bibfnamefont {F.}~\bibnamefont
  {Verstraete}},\ }\href@noop {} {\bibfield  {journal} {\bibinfo  {journal}
  {Annual Review of Condensed Matter Physics}\ }\textbf {\bibinfo {volume}
  {8}},\ \bibinfo {pages} {355} (\bibinfo {year} {2017})}\BibitemShut {NoStop}%
\bibitem [{\citenamefont {Zauner}\ \emph {et~al.}(2015)\citenamefont {Zauner},
  \citenamefont {Draxler}, \citenamefont {Vanderstraeten}, \citenamefont
  {Degroote}, \citenamefont {Haegeman}, \citenamefont {Rams}, \citenamefont
  {Stojevic}, \citenamefont {Schuch},\ and\ \citenamefont
  {Verstraete}}]{zauner2015transfer}%
  \BibitemOpen
  \bibfield  {author} {\bibinfo {author} {\bibfnamefont {V.}~\bibnamefont
  {Zauner}}, \bibinfo {author} {\bibfnamefont {D.}~\bibnamefont {Draxler}},
  \bibinfo {author} {\bibfnamefont {L.}~\bibnamefont {Vanderstraeten}},
  \bibinfo {author} {\bibfnamefont {M.}~\bibnamefont {Degroote}}, \bibinfo
  {author} {\bibfnamefont {J.}~\bibnamefont {Haegeman}}, \bibinfo {author}
  {\bibfnamefont {M.~M.}\ \bibnamefont {Rams}}, \bibinfo {author}
  {\bibfnamefont {V.}~\bibnamefont {Stojevic}}, \bibinfo {author}
  {\bibfnamefont {N.}~\bibnamefont {Schuch}}, \ and\ \bibinfo {author}
  {\bibfnamefont {F.}~\bibnamefont {Verstraete}},\ }\href@noop {} {\bibfield
  {journal} {\bibinfo  {journal} {New Journal of Physics}\ }\textbf {\bibinfo
  {volume} {17}},\ \bibinfo {pages} {053002} (\bibinfo {year}
  {2015})}\BibitemShut {NoStop}%
\bibitem [{\citenamefont {Zauner-Stauber}\ \emph {et~al.}(2018)\citenamefont
  {Zauner-Stauber}, \citenamefont {Vanderstraeten}, \citenamefont {Haegeman},
  \citenamefont {McCulloch},\ and\ \citenamefont
  {Verstraete}}]{zauner2018topological}%
  \BibitemOpen
  \bibfield  {author} {\bibinfo {author} {\bibfnamefont {V.}~\bibnamefont
  {Zauner-Stauber}}, \bibinfo {author} {\bibfnamefont {L.}~\bibnamefont
  {Vanderstraeten}}, \bibinfo {author} {\bibfnamefont {J.}~\bibnamefont
  {Haegeman}}, \bibinfo {author} {\bibfnamefont {I.}~\bibnamefont {McCulloch}},
  \ and\ \bibinfo {author} {\bibfnamefont {F.}~\bibnamefont {Verstraete}},\
  }\href@noop {} {\bibfield  {journal} {\bibinfo  {journal} {Physical Review
  B}\ }\textbf {\bibinfo {volume} {97}},\ \bibinfo {pages} {235155} (\bibinfo
  {year} {2018})}\BibitemShut {NoStop}%
\bibitem [{\citenamefont {Thomale}\ \emph {et~al.}(2015)\citenamefont
  {Thomale}, \citenamefont {Rachel}, \citenamefont {Bernevig},\ and\
  \citenamefont {Arovas}}]{thomale2015entanglement}%
  \BibitemOpen
  \bibfield  {author} {\bibinfo {author} {\bibfnamefont {R.}~\bibnamefont
  {Thomale}}, \bibinfo {author} {\bibfnamefont {S.}~\bibnamefont {Rachel}},
  \bibinfo {author} {\bibfnamefont {B.~A.}\ \bibnamefont {Bernevig}}, \ and\
  \bibinfo {author} {\bibfnamefont {D.~P.}\ \bibnamefont {Arovas}},\
  }\href@noop {} {\bibfield  {journal} {\bibinfo  {journal} {Journal of
  Statistical Mechanics: Theory and Experiment}\ }\textbf {\bibinfo {volume}
  {2015}},\ \bibinfo {pages} {P07017} (\bibinfo {year} {2015})}\BibitemShut
  {NoStop}%
\bibitem [{\citenamefont {Schollw{\"o}ck}(2011)}]{schollwock2011density}%
  \BibitemOpen
  \bibfield  {author} {\bibinfo {author} {\bibfnamefont {U.}~\bibnamefont
  {Schollw{\"o}ck}},\ }\href@noop {} {\bibfield  {journal} {\bibinfo  {journal}
  {Annals of Physics}\ }\textbf {\bibinfo {volume} {326}},\ \bibinfo {pages}
  {96} (\bibinfo {year} {2011})}\BibitemShut {NoStop}%
\bibitem [{\citenamefont {Or{\'u}s}(2014)}]{orus2014practical}%
  \BibitemOpen
  \bibfield  {author} {\bibinfo {author} {\bibfnamefont {R.}~\bibnamefont
  {Or{\'u}s}},\ }\href@noop {} {\bibfield  {journal} {\bibinfo  {journal}
  {Annals of Physics}\ }\textbf {\bibinfo {volume} {349}},\ \bibinfo {pages}
  {117} (\bibinfo {year} {2014})}\BibitemShut {NoStop}%
\bibitem [{\citenamefont {Verstraete}\ and\ \citenamefont
  {Cirac}(2006)}]{verstraete2006matrix}%
  \BibitemOpen
  \bibfield  {author} {\bibinfo {author} {\bibfnamefont {F.}~\bibnamefont
  {Verstraete}}\ and\ \bibinfo {author} {\bibfnamefont {J.~I.}\ \bibnamefont
  {Cirac}},\ }\href@noop {} {\bibfield  {journal} {\bibinfo  {journal}
  {Physical Review B}\ }\textbf {\bibinfo {volume} {73}},\ \bibinfo {pages}
  {094423} (\bibinfo {year} {2006})}\BibitemShut {NoStop}%
\bibitem [{\citenamefont {Friesdorf}\ \emph {et~al.}(2015)\citenamefont
  {Friesdorf}, \citenamefont {Werner}, \citenamefont {Brown}, \citenamefont
  {Scholz},\ and\ \citenamefont {Eisert}}]{friesdorf2015many}%
  \BibitemOpen
  \bibfield  {author} {\bibinfo {author} {\bibfnamefont {M.}~\bibnamefont
  {Friesdorf}}, \bibinfo {author} {\bibfnamefont {A.~H.}\ \bibnamefont
  {Werner}}, \bibinfo {author} {\bibfnamefont {W.}~\bibnamefont {Brown}},
  \bibinfo {author} {\bibfnamefont {V.~B.}\ \bibnamefont {Scholz}}, \ and\
  \bibinfo {author} {\bibfnamefont {J.}~\bibnamefont {Eisert}},\ }\href@noop {}
  {\bibfield  {journal} {\bibinfo  {journal} {Physical Review Letters}\
  }\textbf {\bibinfo {volume} {114}},\ \bibinfo {pages} {170505} (\bibinfo
  {year} {2015})}\BibitemShut {NoStop}%
\bibitem [{\citenamefont {Khemani}\ \emph {et~al.}(2016)\citenamefont
  {Khemani}, \citenamefont {Pollmann},\ and\ \citenamefont
  {Sondhi}}]{khemani2016obtaining}%
  \BibitemOpen
  \bibfield  {author} {\bibinfo {author} {\bibfnamefont {V.}~\bibnamefont
  {Khemani}}, \bibinfo {author} {\bibfnamefont {F.}~\bibnamefont {Pollmann}}, \
  and\ \bibinfo {author} {\bibfnamefont {S.}~\bibnamefont {Sondhi}},\
  }\href@noop {} {\bibfield  {journal} {\bibinfo  {journal} {Physical Review
  Letters}\ }\textbf {\bibinfo {volume} {116}},\ \bibinfo {pages} {247204}
  (\bibinfo {year} {2016})}\BibitemShut {NoStop}%
\bibitem [{\citenamefont {Yu}\ \emph {et~al.}(2017)\citenamefont {Yu},
  \citenamefont {Pekker},\ and\ \citenamefont {Clark}}]{yu2017finding}%
  \BibitemOpen
  \bibfield  {author} {\bibinfo {author} {\bibfnamefont {X.}~\bibnamefont
  {Yu}}, \bibinfo {author} {\bibfnamefont {D.}~\bibnamefont {Pekker}}, \ and\
  \bibinfo {author} {\bibfnamefont {B.~K.}\ \bibnamefont {Clark}},\ }\href@noop
  {} {\bibfield  {journal} {\bibinfo  {journal} {Physical Review Letters}\
  }\textbf {\bibinfo {volume} {118}},\ \bibinfo {pages} {017201} (\bibinfo
  {year} {2017})}\BibitemShut {NoStop}%
\bibitem [{\citenamefont {Shiraishi}\ and\ \citenamefont
  {Mori}(2017{\natexlab{a}})}]{shiraishi2017systematic}%
  \BibitemOpen
  \bibfield  {author} {\bibinfo {author} {\bibfnamefont {N.}~\bibnamefont
  {Shiraishi}}\ and\ \bibinfo {author} {\bibfnamefont {T.}~\bibnamefont
  {Mori}},\ }\href@noop {} {\bibfield  {journal} {\bibinfo  {journal} {Physical
  Review Letters}\ }\textbf {\bibinfo {volume} {119}},\ \bibinfo {pages}
  {030601} (\bibinfo {year} {2017}{\natexlab{a}})}\BibitemShut {NoStop}%
\bibitem [{\citenamefont {Mondaini}\ \emph {et~al.}(2017)\citenamefont
  {Mondaini}, \citenamefont {Mallayya}, \citenamefont {Santos},\ and\
  \citenamefont {Rigol}}]{mondaini2017comment}%
  \BibitemOpen
  \bibfield  {author} {\bibinfo {author} {\bibfnamefont {R.}~\bibnamefont
  {Mondaini}}, \bibinfo {author} {\bibfnamefont {K.}~\bibnamefont {Mallayya}},
  \bibinfo {author} {\bibfnamefont {L.~F.}\ \bibnamefont {Santos}}, \ and\
  \bibinfo {author} {\bibfnamefont {M.}~\bibnamefont {Rigol}},\ }\href@noop {}
  {\bibfield  {journal} {\bibinfo  {journal} {arXiv preprint arXiv:1711.06279}\
  } (\bibinfo {year} {2017})}\BibitemShut {NoStop}%
\bibitem [{\citenamefont {Shiraishi}\ and\ \citenamefont
  {Mori}(2017{\natexlab{b}})}]{shiraishi2017reply}%
  \BibitemOpen
  \bibfield  {author} {\bibinfo {author} {\bibfnamefont {N.}~\bibnamefont
  {Shiraishi}}\ and\ \bibinfo {author} {\bibfnamefont {T.}~\bibnamefont
  {Mori}},\ }\href@noop {} {\bibfield  {journal} {\bibinfo  {journal} {arXiv
  preprint arXiv:1712.01999}\ } (\bibinfo {year}
  {2017}{\natexlab{b}})}\BibitemShut {NoStop}%
\bibitem [{\citenamefont {Turner}\ \emph
  {et~al.}(2018{\natexlab{a}})\citenamefont {Turner}, \citenamefont
  {Michailidis}, \citenamefont {Abanin}, \citenamefont {Serbyn},\ and\
  \citenamefont {Papi{\'c}}}]{turner2018weak}%
  \BibitemOpen
  \bibfield  {author} {\bibinfo {author} {\bibfnamefont {C.}~\bibnamefont
  {Turner}}, \bibinfo {author} {\bibfnamefont {A.}~\bibnamefont {Michailidis}},
  \bibinfo {author} {\bibfnamefont {D.}~\bibnamefont {Abanin}}, \bibinfo
  {author} {\bibfnamefont {M.}~\bibnamefont {Serbyn}}, \ and\ \bibinfo {author}
  {\bibfnamefont {Z.}~\bibnamefont {Papi{\'c}}},\ }\href@noop {} {\bibfield
  {journal} {\bibinfo  {journal} {Nature Physics}\ } (\bibinfo {year}
  {2018}{\natexlab{a}})}\BibitemShut {NoStop}%
\bibitem [{\citenamefont {Schecter}\ and\ \citenamefont
  {Iadecola}(2018)}]{schecter2018many}%
  \BibitemOpen
  \bibfield  {author} {\bibinfo {author} {\bibfnamefont {M.}~\bibnamefont
  {Schecter}}\ and\ \bibinfo {author} {\bibfnamefont {T.}~\bibnamefont
  {Iadecola}},\ }\href@noop {} {\bibfield  {journal} {\bibinfo  {journal}
  {arXiv preprint arXiv:1801.03101}\ } (\bibinfo {year} {2018})}\BibitemShut
  {NoStop}%
\bibitem [{\citenamefont {Vafek}\ \emph {et~al.}(2017)\citenamefont {Vafek},
  \citenamefont {Regnault},\ and\ \citenamefont
  {Bernevig}}]{vafek2017entanglement}%
  \BibitemOpen
  \bibfield  {author} {\bibinfo {author} {\bibfnamefont {O.}~\bibnamefont
  {Vafek}}, \bibinfo {author} {\bibfnamefont {N.}~\bibnamefont {Regnault}}, \
  and\ \bibinfo {author} {\bibfnamefont {B.~A.}\ \bibnamefont {Bernevig}},\
  }\href@noop {} {\bibfield  {journal} {\bibinfo  {journal} {SciPost Physics}\
  }\textbf {\bibinfo {volume} {3}},\ \bibinfo {pages} {043} (\bibinfo {year}
  {2017})}\BibitemShut {NoStop}%
\bibitem [{\citenamefont {Moudgalya}\ \emph {et~al.}(2017)\citenamefont
  {Moudgalya}, \citenamefont {Rachel}, \citenamefont {Bernevig},\ and\
  \citenamefont {Regnault}}]{self}%
  \BibitemOpen
  \bibfield  {author} {\bibinfo {author} {\bibfnamefont {S.}~\bibnamefont
  {Moudgalya}}, \bibinfo {author} {\bibfnamefont {S.}~\bibnamefont {Rachel}},
  \bibinfo {author} {\bibfnamefont {B.~A.}\ \bibnamefont {Bernevig}}, \ and\
  \bibinfo {author} {\bibfnamefont {N.}~\bibnamefont {Regnault}},\ }\href@noop
  {} {\bibfield  {journal} {\bibinfo  {journal} {arXiv preprint
  arXiv:1708.05021}\ } (\bibinfo {year} {2017})}\BibitemShut {NoStop}%
\bibitem [{\citenamefont {Fan}\ \emph {et~al.}(2004)\citenamefont {Fan},
  \citenamefont {Korepin},\ and\ \citenamefont
  {Roychowdhury}}]{fan2004entanglement}%
  \BibitemOpen
  \bibfield  {author} {\bibinfo {author} {\bibfnamefont {H.}~\bibnamefont
  {Fan}}, \bibinfo {author} {\bibfnamefont {V.}~\bibnamefont {Korepin}}, \ and\
  \bibinfo {author} {\bibfnamefont {V.}~\bibnamefont {Roychowdhury}},\
  }\href@noop {} {\bibfield  {journal} {\bibinfo  {journal} {Physical Review
  Letters}\ }\textbf {\bibinfo {volume} {93}},\ \bibinfo {pages} {227203}
  (\bibinfo {year} {2004})}\BibitemShut {NoStop}%
\bibitem [{\citenamefont {Korepin}\ and\ \citenamefont
  {Xu}(2010)}]{korepin2010entanglement}%
  \BibitemOpen
  \bibfield  {author} {\bibinfo {author} {\bibfnamefont {V.~E.}\ \bibnamefont
  {Korepin}}\ and\ \bibinfo {author} {\bibfnamefont {Y.}~\bibnamefont {Xu}},\
  }\href@noop {} {\bibfield  {journal} {\bibinfo  {journal} {International
  Journal of Modern Physics B}\ }\textbf {\bibinfo {volume} {24}},\ \bibinfo
  {pages} {1361} (\bibinfo {year} {2010})}\BibitemShut {NoStop}%
\bibitem [{\citenamefont {Katsura}\ \emph {et~al.}(2007)\citenamefont
  {Katsura}, \citenamefont {Hirano},\ and\ \citenamefont
  {Hatsugai}}]{katsura2007exact}%
  \BibitemOpen
  \bibfield  {author} {\bibinfo {author} {\bibfnamefont {H.}~\bibnamefont
  {Katsura}}, \bibinfo {author} {\bibfnamefont {T.}~\bibnamefont {Hirano}}, \
  and\ \bibinfo {author} {\bibfnamefont {Y.}~\bibnamefont {Hatsugai}},\
  }\href@noop {} {\bibfield  {journal} {\bibinfo  {journal} {Physical Review
  B}\ }\textbf {\bibinfo {volume} {76}},\ \bibinfo {pages} {012401} (\bibinfo
  {year} {2007})}\BibitemShut {NoStop}%
\bibitem [{\citenamefont {Santos}\ \emph {et~al.}(2011)\citenamefont {Santos},
  \citenamefont {Korepin},\ and\ \citenamefont {Bose}}]{santos2011negativity}%
  \BibitemOpen
  \bibfield  {author} {\bibinfo {author} {\bibfnamefont {R.~A.}\ \bibnamefont
  {Santos}}, \bibinfo {author} {\bibfnamefont {V.}~\bibnamefont {Korepin}}, \
  and\ \bibinfo {author} {\bibfnamefont {S.}~\bibnamefont {Bose}},\ }\href@noop
  {} {\bibfield  {journal} {\bibinfo  {journal} {Physical Review A}\ }\textbf
  {\bibinfo {volume} {84}},\ \bibinfo {pages} {062307} (\bibinfo {year}
  {2011})}\BibitemShut {NoStop}%
\bibitem [{\citenamefont {Santos}\ and\ \citenamefont
  {Korepin}(2012)}]{santos2012entanglement}%
  \BibitemOpen
  \bibfield  {author} {\bibinfo {author} {\bibfnamefont {R.~A.}\ \bibnamefont
  {Santos}}\ and\ \bibinfo {author} {\bibfnamefont {V.~E.}\ \bibnamefont
  {Korepin}},\ }\href@noop {} {\bibfield  {journal} {\bibinfo  {journal}
  {Journal of Physics A: Mathematical and Theoretical}\ }\textbf {\bibinfo
  {volume} {45}},\ \bibinfo {pages} {125307} (\bibinfo {year}
  {2012})}\BibitemShut {NoStop}%
\bibitem [{\citenamefont {Santos}\ \emph
  {et~al.}(2012{\natexlab{a}})\citenamefont {Santos}, \citenamefont {Paraan},
  \citenamefont {Korepin},\ and\ \citenamefont
  {Kl{\"u}mper}}]{santos2012entanglement2}%
  \BibitemOpen
  \bibfield  {author} {\bibinfo {author} {\bibfnamefont {R.~A.}\ \bibnamefont
  {Santos}}, \bibinfo {author} {\bibfnamefont {F.~N.}\ \bibnamefont {Paraan}},
  \bibinfo {author} {\bibfnamefont {V.~E.}\ \bibnamefont {Korepin}}, \ and\
  \bibinfo {author} {\bibfnamefont {A.}~\bibnamefont {Kl{\"u}mper}},\
  }\href@noop {} {\bibfield  {journal} {\bibinfo  {journal} {EPL (Europhysics
  Letters)}\ }\textbf {\bibinfo {volume} {98}},\ \bibinfo {pages} {37005}
  (\bibinfo {year} {2012}{\natexlab{a}})}\BibitemShut {NoStop}%
\bibitem [{\citenamefont {Santos}\ \emph
  {et~al.}(2012{\natexlab{b}})\citenamefont {Santos}, \citenamefont {Paraan},
  \citenamefont {Korepin},\ and\ \citenamefont
  {Kl{\"u}mper}}]{santos2012entanglement3}%
  \BibitemOpen
  \bibfield  {author} {\bibinfo {author} {\bibfnamefont {R.~A.}\ \bibnamefont
  {Santos}}, \bibinfo {author} {\bibfnamefont {F.~N.}\ \bibnamefont {Paraan}},
  \bibinfo {author} {\bibfnamefont {V.~E.}\ \bibnamefont {Korepin}}, \ and\
  \bibinfo {author} {\bibfnamefont {A.}~\bibnamefont {Kl{\"u}mper}},\
  }\href@noop {} {\bibfield  {journal} {\bibinfo  {journal} {Journal of Physics
  A: Mathematical and Theoretical}\ }\textbf {\bibinfo {volume} {45}},\
  \bibinfo {pages} {175303} (\bibinfo {year} {2012}{\natexlab{b}})}\BibitemShut
  {NoStop}%
\bibitem [{\citenamefont {Santos}\ and\ \citenamefont
  {Korepin}(2016)}]{santos2016negativity}%
  \BibitemOpen
  \bibfield  {author} {\bibinfo {author} {\bibfnamefont {R.~A.}\ \bibnamefont
  {Santos}}\ and\ \bibinfo {author} {\bibfnamefont {V.}~\bibnamefont
  {Korepin}},\ }\href@noop {} {\bibfield  {journal} {\bibinfo  {journal}
  {Quantum Information Processing}\ }\textbf {\bibinfo {volume} {15}},\
  \bibinfo {pages} {4581} (\bibinfo {year} {2016})}\BibitemShut {NoStop}%
\bibitem [{\citenamefont {Katsura}\ \emph {et~al.}(2008)\citenamefont
  {Katsura}, \citenamefont {Hirano},\ and\ \citenamefont
  {Korepin}}]{katsura2008entanglement}%
  \BibitemOpen
  \bibfield  {author} {\bibinfo {author} {\bibfnamefont {H.}~\bibnamefont
  {Katsura}}, \bibinfo {author} {\bibfnamefont {T.}~\bibnamefont {Hirano}}, \
  and\ \bibinfo {author} {\bibfnamefont {V.~E.}\ \bibnamefont {Korepin}},\
  }\href@noop {} {\bibfield  {journal} {\bibinfo  {journal} {Journal of Physics
  A: Mathematical and Theoretical}\ }\textbf {\bibinfo {volume} {41}},\
  \bibinfo {pages} {135304} (\bibinfo {year} {2008})}\BibitemShut {NoStop}%
\bibitem [{\citenamefont {Xu}\ \emph {et~al.}(2008{\natexlab{a}})\citenamefont
  {Xu}, \citenamefont {Katsura}, \citenamefont {Hirano},\ and\ \citenamefont
  {Korepin}}]{xu2008block}%
  \BibitemOpen
  \bibfield  {author} {\bibinfo {author} {\bibfnamefont {Y.}~\bibnamefont
  {Xu}}, \bibinfo {author} {\bibfnamefont {H.}~\bibnamefont {Katsura}},
  \bibinfo {author} {\bibfnamefont {T.}~\bibnamefont {Hirano}}, \ and\ \bibinfo
  {author} {\bibfnamefont {V.~E.}\ \bibnamefont {Korepin}},\ }\href@noop {}
  {\bibfield  {journal} {\bibinfo  {journal} {Quantum Information Processing}\
  }\textbf {\bibinfo {volume} {7}},\ \bibinfo {pages} {153} (\bibinfo {year}
  {2008}{\natexlab{a}})}\BibitemShut {NoStop}%
\bibitem [{\citenamefont {Xu}\ \emph {et~al.}(2008{\natexlab{b}})\citenamefont
  {Xu}, \citenamefont {Katsura}, \citenamefont {Hirano},\ and\ \citenamefont
  {Korepin}}]{xu2008entanglement}%
  \BibitemOpen
  \bibfield  {author} {\bibinfo {author} {\bibfnamefont {Y.}~\bibnamefont
  {Xu}}, \bibinfo {author} {\bibfnamefont {H.}~\bibnamefont {Katsura}},
  \bibinfo {author} {\bibfnamefont {T.}~\bibnamefont {Hirano}}, \ and\ \bibinfo
  {author} {\bibfnamefont {V.~E.}\ \bibnamefont {Korepin}},\ }\href@noop {}
  {\bibfield  {journal} {\bibinfo  {journal} {Journal of Statistical Physics}\
  }\textbf {\bibinfo {volume} {133}},\ \bibinfo {pages} {347} (\bibinfo {year}
  {2008}{\natexlab{b}})}\BibitemShut {NoStop}%
\bibitem [{\citenamefont {Kim}\ \emph {et~al.}(2014)\citenamefont {Kim},
  \citenamefont {Ikeda},\ and\ \citenamefont {Huse}}]{kim2014testing}%
  \BibitemOpen
  \bibfield  {author} {\bibinfo {author} {\bibfnamefont {H.}~\bibnamefont
  {Kim}}, \bibinfo {author} {\bibfnamefont {T.~N.}\ \bibnamefont {Ikeda}}, \
  and\ \bibinfo {author} {\bibfnamefont {D.~A.}\ \bibnamefont {Huse}},\
  }\href@noop {} {\bibfield  {journal} {\bibinfo  {journal} {Physical Review
  E}\ }\textbf {\bibinfo {volume} {90}},\ \bibinfo {pages} {052105} (\bibinfo
  {year} {2014})}\BibitemShut {NoStop}%
\bibitem [{\citenamefont {Garrison}\ and\ \citenamefont
  {Grover}(2018)}]{garrison2018does}%
  \BibitemOpen
  \bibfield  {author} {\bibinfo {author} {\bibfnamefont {J.~R.}\ \bibnamefont
  {Garrison}}\ and\ \bibinfo {author} {\bibfnamefont {T.}~\bibnamefont
  {Grover}},\ }\href@noop {} {\bibfield  {journal} {\bibinfo  {journal}
  {Physical Review X}\ }\textbf {\bibinfo {volume} {8}},\ \bibinfo {pages}
  {021026} (\bibinfo {year} {2018})}\BibitemShut {NoStop}%
\bibitem [{\citenamefont {Verstraete}\ \emph {et~al.}(2008)\citenamefont
  {Verstraete}, \citenamefont {Murg},\ and\ \citenamefont
  {Cirac}}]{verstraete2008matrix}%
  \BibitemOpen
  \bibfield  {author} {\bibinfo {author} {\bibfnamefont {F.}~\bibnamefont
  {Verstraete}}, \bibinfo {author} {\bibfnamefont {V.}~\bibnamefont {Murg}}, \
  and\ \bibinfo {author} {\bibfnamefont {J.~I.}\ \bibnamefont {Cirac}},\
  }\href@noop {} {\bibfield  {journal} {\bibinfo  {journal} {Advances in
  Physics}\ }\textbf {\bibinfo {volume} {57}},\ \bibinfo {pages} {143}
  (\bibinfo {year} {2008})}\BibitemShut {NoStop}%
\bibitem [{\citenamefont {Perez-Garcia}\ \emph {et~al.}(2007)\citenamefont
  {Perez-Garcia}, \citenamefont {Verstraete}, \citenamefont {Wolf},\ and\
  \citenamefont {Cirac}}]{perez2006matrix}%
  \BibitemOpen
  \bibfield  {author} {\bibinfo {author} {\bibfnamefont {D.}~\bibnamefont
  {Perez-Garcia}}, \bibinfo {author} {\bibfnamefont {F.}~\bibnamefont
  {Verstraete}}, \bibinfo {author} {\bibfnamefont {M.~M.}\ \bibnamefont
  {Wolf}}, \ and\ \bibinfo {author} {\bibfnamefont {J.~I.}\ \bibnamefont
  {Cirac}},\ }\href@noop {} {\bibfield  {journal} {\bibinfo  {journal} {Quantum
  Inf. Comput. 7, 401}\ } (\bibinfo {year} {2007})}\BibitemShut {NoStop}%
\bibitem [{\citenamefont {Cirac}\ \emph {et~al.}(2011)\citenamefont {Cirac},
  \citenamefont {Poilblanc}, \citenamefont {Schuch},\ and\ \citenamefont
  {Verstraete}}]{cirac2011entanglement}%
  \BibitemOpen
  \bibfield  {author} {\bibinfo {author} {\bibfnamefont {J.~I.}\ \bibnamefont
  {Cirac}}, \bibinfo {author} {\bibfnamefont {D.}~\bibnamefont {Poilblanc}},
  \bibinfo {author} {\bibfnamefont {N.}~\bibnamefont {Schuch}}, \ and\ \bibinfo
  {author} {\bibfnamefont {F.}~\bibnamefont {Verstraete}},\ }\href@noop {}
  {\bibfield  {journal} {\bibinfo  {journal} {Physical Review B}\ }\textbf
  {\bibinfo {volume} {83}},\ \bibinfo {pages} {245134} (\bibinfo {year}
  {2011})}\BibitemShut {NoStop}%
\bibitem [{\citenamefont {Affleck}\ \emph {et~al.}(1987)\citenamefont
  {Affleck}, \citenamefont {Kennedy}, \citenamefont {Lieb},\ and\ \citenamefont
  {Tasaki}}]{aklt1987rigorous}%
  \BibitemOpen
  \bibfield  {author} {\bibinfo {author} {\bibfnamefont {I.}~\bibnamefont
  {Affleck}}, \bibinfo {author} {\bibfnamefont {T.}~\bibnamefont {Kennedy}},
  \bibinfo {author} {\bibfnamefont {E.~H.}\ \bibnamefont {Lieb}}, \ and\
  \bibinfo {author} {\bibfnamefont {H.}~\bibnamefont {Tasaki}},\ }\href@noop {}
  {\bibfield  {journal} {\bibinfo  {journal} {Physical Review Letters}\
  }\textbf {\bibinfo {volume} {59}},\ \bibinfo {pages} {799} (\bibinfo {year}
  {1987})}\BibitemShut {NoStop}%
\bibitem [{\citenamefont {Kl{\"u}mper}\ \emph {et~al.}(1993)\citenamefont
  {Kl{\"u}mper}, \citenamefont {Schadschneider},\ and\ \citenamefont
  {Zittartz}}]{klumper1993matrix}%
  \BibitemOpen
  \bibfield  {author} {\bibinfo {author} {\bibfnamefont {A.}~\bibnamefont
  {Kl{\"u}mper}}, \bibinfo {author} {\bibfnamefont {A.}~\bibnamefont
  {Schadschneider}}, \ and\ \bibinfo {author} {\bibfnamefont {J.}~\bibnamefont
  {Zittartz}},\ }\href@noop {} {\bibfield  {journal} {\bibinfo  {journal} {EPL
  (Europhysics Letters)}\ }\textbf {\bibinfo {volume} {24}},\ \bibinfo {pages}
  {293} (\bibinfo {year} {1993})}\BibitemShut {NoStop}%
\bibitem [{\citenamefont {Pirvu}\ \emph {et~al.}(2010)\citenamefont {Pirvu},
  \citenamefont {Murg}, \citenamefont {Cirac},\ and\ \citenamefont
  {Verstraete}}]{pirvu2010matrix}%
  \BibitemOpen
  \bibfield  {author} {\bibinfo {author} {\bibfnamefont {B.}~\bibnamefont
  {Pirvu}}, \bibinfo {author} {\bibfnamefont {V.}~\bibnamefont {Murg}},
  \bibinfo {author} {\bibfnamefont {J.~I.}\ \bibnamefont {Cirac}}, \ and\
  \bibinfo {author} {\bibfnamefont {F.}~\bibnamefont {Verstraete}},\
  }\href@noop {} {\bibfield  {journal} {\bibinfo  {journal} {New Journal of
  Physics}\ }\textbf {\bibinfo {volume} {12}},\ \bibinfo {pages} {025012}
  (\bibinfo {year} {2010})}\BibitemShut {NoStop}%
\bibitem [{\citenamefont {McCulloch}(2007)}]{mcculloch2007density}%
  \BibitemOpen
  \bibfield  {author} {\bibinfo {author} {\bibfnamefont {I.~P.}\ \bibnamefont
  {McCulloch}},\ }\href@noop {} {\bibfield  {journal} {\bibinfo  {journal}
  {Journal of Statistical Mechanics: Theory and Experiment}\ }\textbf {\bibinfo
  {volume} {2007}},\ \bibinfo {pages} {P10014} (\bibinfo {year}
  {2007})}\BibitemShut {NoStop}%
\bibitem [{\citenamefont {Verstraete}\ \emph {et~al.}(2004)\citenamefont
  {Verstraete}, \citenamefont {Garcia-Ripoll},\ and\ \citenamefont
  {Cirac}}]{verstraete2004matrix}%
  \BibitemOpen
  \bibfield  {author} {\bibinfo {author} {\bibfnamefont {F.}~\bibnamefont
  {Verstraete}}, \bibinfo {author} {\bibfnamefont {J.~J.}\ \bibnamefont
  {Garcia-Ripoll}}, \ and\ \bibinfo {author} {\bibfnamefont {J.~I.}\
  \bibnamefont {Cirac}},\ }\href@noop {} {\bibfield  {journal} {\bibinfo
  {journal} {Physical Review Letters}\ }\textbf {\bibinfo {volume} {93}},\
  \bibinfo {pages} {207204} (\bibinfo {year} {2004})}\BibitemShut {NoStop}%
\bibitem [{\citenamefont {Crosswhite}\ and\ \citenamefont
  {Bacon}(2008)}]{crosswhite2008finite}%
  \BibitemOpen
  \bibfield  {author} {\bibinfo {author} {\bibfnamefont {G.~M.}\ \bibnamefont
  {Crosswhite}}\ and\ \bibinfo {author} {\bibfnamefont {D.}~\bibnamefont
  {Bacon}},\ }\href@noop {} {\bibfield  {journal} {\bibinfo  {journal}
  {Physical Review A}\ }\textbf {\bibinfo {volume} {78}},\ \bibinfo {pages}
  {012356} (\bibinfo {year} {2008})}\BibitemShut {NoStop}%
\bibitem [{\citenamefont {Motruk}\ \emph {et~al.}(2016)\citenamefont {Motruk},
  \citenamefont {Zaletel}, \citenamefont {Mong},\ and\ \citenamefont
  {Pollmann}}]{motruk2016density}%
  \BibitemOpen
  \bibfield  {author} {\bibinfo {author} {\bibfnamefont {J.}~\bibnamefont
  {Motruk}}, \bibinfo {author} {\bibfnamefont {M.~P.}\ \bibnamefont {Zaletel}},
  \bibinfo {author} {\bibfnamefont {R.~S.}\ \bibnamefont {Mong}}, \ and\
  \bibinfo {author} {\bibfnamefont {F.}~\bibnamefont {Pollmann}},\ }\href@noop
  {} {\bibfield  {journal} {\bibinfo  {journal} {Physical Review B}\ }\textbf
  {\bibinfo {volume} {93}},\ \bibinfo {pages} {155139} (\bibinfo {year}
  {2016})}\BibitemShut {NoStop}%
\bibitem [{\citenamefont {Arovas}(1989)}]{arovas1989two}%
  \BibitemOpen
  \bibfield  {author} {\bibinfo {author} {\bibfnamefont {D.~P.}\ \bibnamefont
  {Arovas}},\ }\href@noop {} {\bibfield  {journal} {\bibinfo  {journal}
  {Physics Letters A}\ }\textbf {\bibinfo {volume} {137}},\ \bibinfo {pages}
  {431} (\bibinfo {year} {1989})}\BibitemShut {NoStop}%
\bibitem [{\citenamefont {Vanderstraeten}\ \emph
  {et~al.}(2015{\natexlab{a}})\citenamefont {Vanderstraeten}, \citenamefont
  {Mari{\"e}n}, \citenamefont {Verstraete},\ and\ \citenamefont
  {Haegeman}}]{vanderstraeten2015excitations}%
  \BibitemOpen
  \bibfield  {author} {\bibinfo {author} {\bibfnamefont {L.}~\bibnamefont
  {Vanderstraeten}}, \bibinfo {author} {\bibfnamefont {M.}~\bibnamefont
  {Mari{\"e}n}}, \bibinfo {author} {\bibfnamefont {F.}~\bibnamefont
  {Verstraete}}, \ and\ \bibinfo {author} {\bibfnamefont {J.}~\bibnamefont
  {Haegeman}},\ }\href@noop {} {\bibfield  {journal} {\bibinfo  {journal}
  {Physical Review B}\ }\textbf {\bibinfo {volume} {92}},\ \bibinfo {pages}
  {201111} (\bibinfo {year} {2015}{\natexlab{a}})}\BibitemShut {NoStop}%
\bibitem [{\citenamefont {Pirvu}\ \emph {et~al.}(2012)\citenamefont {Pirvu},
  \citenamefont {Haegeman},\ and\ \citenamefont
  {Verstraete}}]{pirvu2012matrix}%
  \BibitemOpen
  \bibfield  {author} {\bibinfo {author} {\bibfnamefont {B.}~\bibnamefont
  {Pirvu}}, \bibinfo {author} {\bibfnamefont {J.}~\bibnamefont {Haegeman}}, \
  and\ \bibinfo {author} {\bibfnamefont {F.}~\bibnamefont {Verstraete}},\
  }\href@noop {} {\bibfield  {journal} {\bibinfo  {journal} {Physical Review
  B}\ }\textbf {\bibinfo {volume} {85}},\ \bibinfo {pages} {035130} (\bibinfo
  {year} {2012})}\BibitemShut {NoStop}%
\bibitem [{\citenamefont {Arovas}\ \emph {et~al.}(1988)\citenamefont {Arovas},
  \citenamefont {Auerbach},\ and\ \citenamefont
  {Haldane}}]{arovas1988extended}%
  \BibitemOpen
  \bibfield  {author} {\bibinfo {author} {\bibfnamefont {D.~P.}\ \bibnamefont
  {Arovas}}, \bibinfo {author} {\bibfnamefont {A.}~\bibnamefont {Auerbach}}, \
  and\ \bibinfo {author} {\bibfnamefont {F.}~\bibnamefont {Haldane}},\
  }\href@noop {} {\bibfield  {journal} {\bibinfo  {journal} {Physical Review
  Letters}\ }\textbf {\bibinfo {volume} {60}},\ \bibinfo {pages} {531}
  (\bibinfo {year} {1988})}\BibitemShut {NoStop}%
\bibitem [{\citenamefont {Yi}\ \emph {et~al.}(2002)\citenamefont {Yi},
  \citenamefont {M{\"u}stecapl{\i}o{\u{g}}lu}, \citenamefont {Sun},\ and\
  \citenamefont {You}}]{yi2002single}%
  \BibitemOpen
  \bibfield  {author} {\bibinfo {author} {\bibfnamefont {S.}~\bibnamefont
  {Yi}}, \bibinfo {author} {\bibfnamefont {{\"O}.}~\bibnamefont
  {M{\"u}stecapl{\i}o{\u{g}}lu}}, \bibinfo {author} {\bibfnamefont {C.-P.}\
  \bibnamefont {Sun}}, \ and\ \bibinfo {author} {\bibfnamefont
  {L.}~\bibnamefont {You}},\ }\href@noop {} {\bibfield  {journal} {\bibinfo
  {journal} {Physical Review A}\ }\textbf {\bibinfo {volume} {66}},\ \bibinfo
  {pages} {011601} (\bibinfo {year} {2002})}\BibitemShut {NoStop}%
\bibitem [{\citenamefont {Arovas}(2008)}]{arovas2008simplex}%
  \BibitemOpen
  \bibfield  {author} {\bibinfo {author} {\bibfnamefont {D.~P.}\ \bibnamefont
  {Arovas}},\ }\href@noop {} {\bibfield  {journal} {\bibinfo  {journal}
  {Physical Review B}\ }\textbf {\bibinfo {volume} {77}},\ \bibinfo {pages}
  {104404} (\bibinfo {year} {2008})}\BibitemShut {NoStop}%
\bibitem [{\citenamefont {Girvin}\ \emph {et~al.}(1986)\citenamefont {Girvin},
  \citenamefont {MacDonald},\ and\ \citenamefont
  {Platzman}}]{girvin1986magneto}%
  \BibitemOpen
  \bibfield  {author} {\bibinfo {author} {\bibfnamefont {S.}~\bibnamefont
  {Girvin}}, \bibinfo {author} {\bibfnamefont {A.}~\bibnamefont {MacDonald}}, \
  and\ \bibinfo {author} {\bibfnamefont {P.}~\bibnamefont {Platzman}},\
  }\href@noop {} {\bibfield  {journal} {\bibinfo  {journal} {Physical Review
  B}\ }\textbf {\bibinfo {volume} {33}},\ \bibinfo {pages} {2481} (\bibinfo
  {year} {1986})}\BibitemShut {NoStop}%
\bibitem [{\citenamefont {D'Alessio}\ \emph {et~al.}(2016)\citenamefont
  {D'Alessio}, \citenamefont {Kafri}, \citenamefont {Polkovnikov},\ and\
  \citenamefont {Rigol}}]{d2016quantum}%
  \BibitemOpen
  \bibfield  {author} {\bibinfo {author} {\bibfnamefont {L.}~\bibnamefont
  {D'Alessio}}, \bibinfo {author} {\bibfnamefont {Y.}~\bibnamefont {Kafri}},
  \bibinfo {author} {\bibfnamefont {A.}~\bibnamefont {Polkovnikov}}, \ and\
  \bibinfo {author} {\bibfnamefont {M.}~\bibnamefont {Rigol}},\ }\href@noop {}
  {\bibfield  {journal} {\bibinfo  {journal} {Advances in Physics}\ }\textbf
  {\bibinfo {volume} {65}},\ \bibinfo {pages} {239} (\bibinfo {year}
  {2016})}\BibitemShut {NoStop}%
\bibitem [{\citenamefont {Turner}\ \emph
  {et~al.}(2018{\natexlab{b}})\citenamefont {Turner}, \citenamefont
  {Michailidis}, \citenamefont {Abanin}, \citenamefont {Serbyn},\ and\
  \citenamefont {Papi{\'c}}}]{turner2018quantum}%
  \BibitemOpen
  \bibfield  {author} {\bibinfo {author} {\bibfnamefont {C.}~\bibnamefont
  {Turner}}, \bibinfo {author} {\bibfnamefont {A.}~\bibnamefont {Michailidis}},
  \bibinfo {author} {\bibfnamefont {D.}~\bibnamefont {Abanin}}, \bibinfo
  {author} {\bibfnamefont {M.}~\bibnamefont {Serbyn}}, \ and\ \bibinfo {author}
  {\bibfnamefont {Z.}~\bibnamefont {Papi{\'c}}},\ }\href@noop {} {\bibfield
  {journal} {\bibinfo  {journal} {arXiv preprint arXiv:1806.10933}\ } (\bibinfo
  {year} {2018}{\natexlab{b}})}\BibitemShut {NoStop}%
\bibitem [{\citenamefont {Bernien}\ \emph {et~al.}(2017)\citenamefont
  {Bernien}, \citenamefont {Schwartz}, \citenamefont {Keesling}, \citenamefont
  {Levine}, \citenamefont {Omran}, \citenamefont {Pichler}, \citenamefont
  {Choi}, \citenamefont {Zibrov}, \citenamefont {Endres}, \citenamefont
  {Greiner} \emph {et~al.}}]{bernien2017probing}%
  \BibitemOpen
  \bibfield  {author} {\bibinfo {author} {\bibfnamefont {H.}~\bibnamefont
  {Bernien}}, \bibinfo {author} {\bibfnamefont {S.}~\bibnamefont {Schwartz}},
  \bibinfo {author} {\bibfnamefont {A.}~\bibnamefont {Keesling}}, \bibinfo
  {author} {\bibfnamefont {H.}~\bibnamefont {Levine}}, \bibinfo {author}
  {\bibfnamefont {A.}~\bibnamefont {Omran}}, \bibinfo {author} {\bibfnamefont
  {H.}~\bibnamefont {Pichler}}, \bibinfo {author} {\bibfnamefont
  {S.}~\bibnamefont {Choi}}, \bibinfo {author} {\bibfnamefont {A.~S.}\
  \bibnamefont {Zibrov}}, \bibinfo {author} {\bibfnamefont {M.}~\bibnamefont
  {Endres}}, \bibinfo {author} {\bibfnamefont {M.}~\bibnamefont {Greiner}},
  \emph {et~al.},\ }\href@noop {} {\bibfield  {journal} {\bibinfo  {journal}
  {Nature}\ }\textbf {\bibinfo {volume} {551}},\ \bibinfo {pages} {579}
  (\bibinfo {year} {2017})}\BibitemShut {NoStop}%
\bibitem [{\citenamefont {Pollmann}\ \emph {et~al.}(2010)\citenamefont
  {Pollmann}, \citenamefont {Turner}, \citenamefont {Berg},\ and\ \citenamefont
  {Oshikawa}}]{pollmann2010entanglement}%
  \BibitemOpen
  \bibfield  {author} {\bibinfo {author} {\bibfnamefont {F.}~\bibnamefont
  {Pollmann}}, \bibinfo {author} {\bibfnamefont {A.~M.}\ \bibnamefont
  {Turner}}, \bibinfo {author} {\bibfnamefont {E.}~\bibnamefont {Berg}}, \ and\
  \bibinfo {author} {\bibfnamefont {M.}~\bibnamefont {Oshikawa}},\ }\href@noop
  {} {\bibfield  {journal} {\bibinfo  {journal} {Physical Review B}\ }\textbf
  {\bibinfo {volume} {81}},\ \bibinfo {pages} {064439} (\bibinfo {year}
  {2010})}\BibitemShut {NoStop}%
\bibitem [{\citenamefont {Pollmann}\ \emph {et~al.}(2012)\citenamefont
  {Pollmann}, \citenamefont {Berg}, \citenamefont {Turner},\ and\ \citenamefont
  {Oshikawa}}]{pollmann2012symmetry}%
  \BibitemOpen
  \bibfield  {author} {\bibinfo {author} {\bibfnamefont {F.}~\bibnamefont
  {Pollmann}}, \bibinfo {author} {\bibfnamefont {E.}~\bibnamefont {Berg}},
  \bibinfo {author} {\bibfnamefont {A.~M.}\ \bibnamefont {Turner}}, \ and\
  \bibinfo {author} {\bibfnamefont {M.}~\bibnamefont {Oshikawa}},\ }\href@noop
  {} {\bibfield  {journal} {\bibinfo  {journal} {Physical Review B}\ }\textbf
  {\bibinfo {volume} {85}},\ \bibinfo {pages} {075125} (\bibinfo {year}
  {2012})}\BibitemShut {NoStop}%
\bibitem [{\citenamefont {Pollmann}\ and\ \citenamefont
  {Turner}(2012)}]{pollmann2012detection}%
  \BibitemOpen
  \bibfield  {author} {\bibinfo {author} {\bibfnamefont {F.}~\bibnamefont
  {Pollmann}}\ and\ \bibinfo {author} {\bibfnamefont {A.~M.}\ \bibnamefont
  {Turner}},\ }\href@noop {} {\bibfield  {journal} {\bibinfo  {journal}
  {Physical Review B}\ }\textbf {\bibinfo {volume} {86}},\ \bibinfo {pages}
  {125441} (\bibinfo {year} {2012})}\BibitemShut {NoStop}%
\bibitem [{\citenamefont {Chen}\ \emph {et~al.}(2011)\citenamefont {Chen},
  \citenamefont {Gu},\ and\ \citenamefont {Wen}}]{chen2011classification}%
  \BibitemOpen
  \bibfield  {author} {\bibinfo {author} {\bibfnamefont {X.}~\bibnamefont
  {Chen}}, \bibinfo {author} {\bibfnamefont {Z.-C.}\ \bibnamefont {Gu}}, \ and\
  \bibinfo {author} {\bibfnamefont {X.-G.}\ \bibnamefont {Wen}},\ }\href@noop
  {} {\bibfield  {journal} {\bibinfo  {journal} {Physical Review B}\ }\textbf
  {\bibinfo {volume} {83}},\ \bibinfo {pages} {035107} (\bibinfo {year}
  {2011})}\BibitemShut {NoStop}%
\bibitem [{\citenamefont {P{\'e}rez-Garc{\'\i}a}\ \emph
  {et~al.}(2008)\citenamefont {P{\'e}rez-Garc{\'\i}a}, \citenamefont {Wolf},
  \citenamefont {Sanz}, \citenamefont {Verstraete},\ and\ \citenamefont
  {Cirac}}]{perez2008string}%
  \BibitemOpen
  \bibfield  {author} {\bibinfo {author} {\bibfnamefont {D.}~\bibnamefont
  {P{\'e}rez-Garc{\'\i}a}}, \bibinfo {author} {\bibfnamefont {M.~M.}\
  \bibnamefont {Wolf}}, \bibinfo {author} {\bibfnamefont {M.}~\bibnamefont
  {Sanz}}, \bibinfo {author} {\bibfnamefont {F.}~\bibnamefont {Verstraete}}, \
  and\ \bibinfo {author} {\bibfnamefont {J.~I.}\ \bibnamefont {Cirac}},\
  }\href@noop {} {\bibfield  {journal} {\bibinfo  {journal} {Physical Review
  Letters}\ }\textbf {\bibinfo {volume} {100}},\ \bibinfo {pages} {167202}
  (\bibinfo {year} {2008})}\BibitemShut {NoStop}%
\bibitem [{\citenamefont {Vanderstraeten}\ \emph
  {et~al.}(2015{\natexlab{b}})\citenamefont {Vanderstraeten}, \citenamefont
  {Verstraete},\ and\ \citenamefont {Haegeman}}]{vanderstraeten2015scattering}%
  \BibitemOpen
  \bibfield  {author} {\bibinfo {author} {\bibfnamefont {L.}~\bibnamefont
  {Vanderstraeten}}, \bibinfo {author} {\bibfnamefont {F.}~\bibnamefont
  {Verstraete}}, \ and\ \bibinfo {author} {\bibfnamefont {J.}~\bibnamefont
  {Haegeman}},\ }\href@noop {} {\bibfield  {journal} {\bibinfo  {journal}
  {Physical Review B}\ }\textbf {\bibinfo {volume} {92}},\ \bibinfo {pages}
  {125136} (\bibinfo {year} {2015}{\natexlab{b}})}\BibitemShut {NoStop}%
\bibitem [{\citenamefont {Draxler}\ \emph {et~al.}(2013)\citenamefont
  {Draxler}, \citenamefont {Haegeman}, \citenamefont {Osborne}, \citenamefont
  {Stojevic}, \citenamefont {Vanderstraeten},\ and\ \citenamefont
  {Verstraete}}]{draxler2013particles}%
  \BibitemOpen
  \bibfield  {author} {\bibinfo {author} {\bibfnamefont {D.}~\bibnamefont
  {Draxler}}, \bibinfo {author} {\bibfnamefont {J.}~\bibnamefont {Haegeman}},
  \bibinfo {author} {\bibfnamefont {T.~J.}\ \bibnamefont {Osborne}}, \bibinfo
  {author} {\bibfnamefont {V.}~\bibnamefont {Stojevic}}, \bibinfo {author}
  {\bibfnamefont {L.}~\bibnamefont {Vanderstraeten}}, \ and\ \bibinfo {author}
  {\bibfnamefont {F.}~\bibnamefont {Verstraete}},\ }\href@noop {} {\bibfield
  {journal} {\bibinfo  {journal} {Physical Review Letters}\ }\textbf {\bibinfo
  {volume} {111}},\ \bibinfo {pages} {020402} (\bibinfo {year}
  {2013})}\BibitemShut {NoStop}%
\bibitem [{\citenamefont {Moudgalya}\ \emph {et~al.}(tion)\citenamefont
  {Moudgalya}, \citenamefont {Wu}, \citenamefont {Regnault},\ and\
  \citenamefont {Bernevig}}]{fqheinprep}%
  \BibitemOpen
  \bibfield  {author} {\bibinfo {author} {\bibfnamefont {S.}~\bibnamefont
  {Moudgalya}}, \bibinfo {author} {\bibfnamefont {Y.-L.}\ \bibnamefont {Wu}},
  \bibinfo {author} {\bibfnamefont {N.}~\bibnamefont {Regnault}}, \ and\
  \bibinfo {author} {\bibfnamefont {B.}~\bibnamefont {Bernevig}},\ }\href@noop
  {} {\  (\bibinfo {year} {in preparation})}\BibitemShut {NoStop}%
\bibitem [{\citenamefont {Haegeman}\ \emph {et~al.}(2014)\citenamefont
  {Haegeman}, \citenamefont {Mari{\"e}n}, \citenamefont {Osborne},\ and\
  \citenamefont {Verstraete}}]{haegeman2014geometry}%
  \BibitemOpen
  \bibfield  {author} {\bibinfo {author} {\bibfnamefont {J.}~\bibnamefont
  {Haegeman}}, \bibinfo {author} {\bibfnamefont {M.}~\bibnamefont
  {Mari{\"e}n}}, \bibinfo {author} {\bibfnamefont {T.~J.}\ \bibnamefont
  {Osborne}}, \ and\ \bibinfo {author} {\bibfnamefont {F.}~\bibnamefont
  {Verstraete}},\ }\href@noop {} {\bibfield  {journal} {\bibinfo  {journal}
  {Journal of Mathematical Physics}\ }\textbf {\bibinfo {volume} {55}},\
  \bibinfo {pages} {021902} (\bibinfo {year} {2014})}\BibitemShut {NoStop}%
\bibitem [{\citenamefont {Castro-Alvaredo}\ \emph
  {et~al.}(2018{\natexlab{a}})\citenamefont {Castro-Alvaredo}, \citenamefont
  {De~Fazio}, \citenamefont {Doyon},\ and\ \citenamefont
  {Sz{\'e}cs{\'e}nyi}}]{castro2018entanglement}%
  \BibitemOpen
  \bibfield  {author} {\bibinfo {author} {\bibfnamefont {O.~A.}\ \bibnamefont
  {Castro-Alvaredo}}, \bibinfo {author} {\bibfnamefont {C.}~\bibnamefont
  {De~Fazio}}, \bibinfo {author} {\bibfnamefont {B.}~\bibnamefont {Doyon}}, \
  and\ \bibinfo {author} {\bibfnamefont {I.~M.}\ \bibnamefont
  {Sz{\'e}cs{\'e}nyi}},\ }\href@noop {} {\bibfield  {journal} {\bibinfo
  {journal} {arXiv preprint arXiv:1805.04948}\ } (\bibinfo {year}
  {2018}{\natexlab{a}})}\BibitemShut {NoStop}%
\bibitem [{\citenamefont {Castro-Alvaredo}\ \emph
  {et~al.}(2018{\natexlab{b}})\citenamefont {Castro-Alvaredo}, \citenamefont
  {De~Fazio}, \citenamefont {Doyon},\ and\ \citenamefont
  {Sz{\'e}cs{\'e}nyi}}]{castro2018entanglement2}%
  \BibitemOpen
  \bibfield  {author} {\bibinfo {author} {\bibfnamefont {O.~A.}\ \bibnamefont
  {Castro-Alvaredo}}, \bibinfo {author} {\bibfnamefont {C.}~\bibnamefont
  {De~Fazio}}, \bibinfo {author} {\bibfnamefont {B.}~\bibnamefont {Doyon}}, \
  and\ \bibinfo {author} {\bibfnamefont {I.~M.}\ \bibnamefont
  {Sz{\'e}cs{\'e}nyi}},\ }\href@noop {} {\bibfield  {journal} {\bibinfo
  {journal} {arXiv preprint arXiv:1806.03247}\ } (\bibinfo {year}
  {2018}{\natexlab{b}})}\BibitemShut {NoStop}%
\bibitem [{\citenamefont {Alba}\ \emph {et~al.}(2009)\citenamefont {Alba},
  \citenamefont {Fagotti},\ and\ \citenamefont
  {Calabrese}}]{alba2009entanglement}%
  \BibitemOpen
  \bibfield  {author} {\bibinfo {author} {\bibfnamefont {V.}~\bibnamefont
  {Alba}}, \bibinfo {author} {\bibfnamefont {M.}~\bibnamefont {Fagotti}}, \
  and\ \bibinfo {author} {\bibfnamefont {P.}~\bibnamefont {Calabrese}},\
  }\href@noop {} {\bibfield  {journal} {\bibinfo  {journal} {Journal of
  Statistical Mechanics: Theory and Experiment}\ }\textbf {\bibinfo {volume}
  {2009}},\ \bibinfo {pages} {P10020} (\bibinfo {year} {2009})}\BibitemShut
  {NoStop}%
\bibitem [{\citenamefont {M{\"o}lter}\ \emph {et~al.}(2014)\citenamefont
  {M{\"o}lter}, \citenamefont {Barthel}, \citenamefont {Schollw{\"o}ck},\ and\
  \citenamefont {Alba}}]{molter2014bound}%
  \BibitemOpen
  \bibfield  {author} {\bibinfo {author} {\bibfnamefont {J.}~\bibnamefont
  {M{\"o}lter}}, \bibinfo {author} {\bibfnamefont {T.}~\bibnamefont {Barthel}},
  \bibinfo {author} {\bibfnamefont {U.}~\bibnamefont {Schollw{\"o}ck}}, \ and\
  \bibinfo {author} {\bibfnamefont {V.}~\bibnamefont {Alba}},\ }\href@noop {}
  {\bibfield  {journal} {\bibinfo  {journal} {Journal of Statistical Mechanics:
  Theory and Experiment}\ }\textbf {\bibinfo {volume} {2014}},\ \bibinfo
  {pages} {P10029} (\bibinfo {year} {2014})}\BibitemShut {NoStop}%
\bibitem [{\citenamefont {Sanz}\ \emph {et~al.}(2009)\citenamefont {Sanz},
  \citenamefont {Wolf}, \citenamefont {Perez-Garc{\'\i}a},\ and\ \citenamefont
  {Cirac}}]{sanz2009matrix}%
  \BibitemOpen
  \bibfield  {author} {\bibinfo {author} {\bibfnamefont {M.}~\bibnamefont
  {Sanz}}, \bibinfo {author} {\bibfnamefont {M.~M.}\ \bibnamefont {Wolf}},
  \bibinfo {author} {\bibfnamefont {D.}~\bibnamefont {Perez-Garc{\'\i}a}}, \
  and\ \bibinfo {author} {\bibfnamefont {J.~I.}\ \bibnamefont {Cirac}},\
  }\href@noop {} {\bibfield  {journal} {\bibinfo  {journal} {Physical Review
  A}\ }\textbf {\bibinfo {volume} {79}},\ \bibinfo {pages} {042308} (\bibinfo
  {year} {2009})}\BibitemShut {NoStop}%
\bibitem [{\citenamefont {Tu}\ and\ \citenamefont {Sanz}(2010)}]{tu2010exact}%
  \BibitemOpen
  \bibfield  {author} {\bibinfo {author} {\bibfnamefont {H.-H.}\ \bibnamefont
  {Tu}}\ and\ \bibinfo {author} {\bibfnamefont {M.}~\bibnamefont {Sanz}},\
  }\href@noop {} {\bibfield  {journal} {\bibinfo  {journal} {Physical Review
  B}\ }\textbf {\bibinfo {volume} {82}},\ \bibinfo {pages} {104404} (\bibinfo
  {year} {2010})}\BibitemShut {NoStop}%
\bibitem [{\citenamefont {Berry}(1989)}]{berry1989quantum}%
  \BibitemOpen
  \bibfield  {author} {\bibinfo {author} {\bibfnamefont {M.}~\bibnamefont
  {Berry}},\ }in\ \href@noop {} {\emph {\bibinfo {booktitle} {Proceedings of
  the Royal Society of London A: Mathematical, Physical and Engineering
  Sciences}}},\ Vol.\ \bibinfo {volume} {423}\ (\bibinfo {organization} {The
  Royal Society},\ \bibinfo {year} {1989})\ pp.\ \bibinfo {pages}
  {219--231}\BibitemShut {NoStop}%
\bibitem [{\citenamefont {Totsuka}\ and\ \citenamefont
  {Suzuki}(1995)}]{totsuka1995matrix}%
  \BibitemOpen
  \bibfield  {author} {\bibinfo {author} {\bibfnamefont {K.}~\bibnamefont
  {Totsuka}}\ and\ \bibinfo {author} {\bibfnamefont {M.}~\bibnamefont
  {Suzuki}},\ }\href@noop {} {\bibfield  {journal} {\bibinfo  {journal}
  {Journal of Physics: Condensed Matter}\ }\textbf {\bibinfo {volume} {7}},\
  \bibinfo {pages} {1639} (\bibinfo {year} {1995})}\BibitemShut {NoStop}%
\bibitem [{\citenamefont {Fannes}\ \emph {et~al.}(1989)\citenamefont {Fannes},
  \citenamefont {Nachtergaele},\ and\ \citenamefont
  {Werner}}]{fannes1989exact}%
  \BibitemOpen
  \bibfield  {author} {\bibinfo {author} {\bibfnamefont {M.}~\bibnamefont
  {Fannes}}, \bibinfo {author} {\bibfnamefont {B.}~\bibnamefont
  {Nachtergaele}}, \ and\ \bibinfo {author} {\bibfnamefont {R.}~\bibnamefont
  {Werner}},\ }\href@noop {} {\bibfield  {journal} {\bibinfo  {journal} {EPL
  (Europhysics Letters)}\ }\textbf {\bibinfo {volume} {10}},\ \bibinfo {pages}
  {633} (\bibinfo {year} {1989})}\BibitemShut {NoStop}%
\bibitem [{\citenamefont {Fannes}\ \emph {et~al.}(1992)\citenamefont {Fannes},
  \citenamefont {Nachtergaele},\ and\ \citenamefont
  {Werner}}]{fannes1992finitely}%
  \BibitemOpen
  \bibfield  {author} {\bibinfo {author} {\bibfnamefont {M.}~\bibnamefont
  {Fannes}}, \bibinfo {author} {\bibfnamefont {B.}~\bibnamefont
  {Nachtergaele}}, \ and\ \bibinfo {author} {\bibfnamefont {R.~F.}\
  \bibnamefont {Werner}},\ }\href@noop {} {\bibfield  {journal} {\bibinfo
  {journal} {Communications in mathematical physics}\ }\textbf {\bibinfo
  {volume} {144}},\ \bibinfo {pages} {443} (\bibinfo {year}
  {1992})}\BibitemShut {NoStop}%
\bibitem [{\citenamefont {Karimipour}\ and\ \citenamefont
  {Memarzadeh}(2008)}]{karimipour2008matrix}%
  \BibitemOpen
  \bibfield  {author} {\bibinfo {author} {\bibfnamefont {V.}~\bibnamefont
  {Karimipour}}\ and\ \bibinfo {author} {\bibfnamefont {L.}~\bibnamefont
  {Memarzadeh}},\ }\href@noop {} {\bibfield  {journal} {\bibinfo  {journal}
  {Physical Review B}\ }\textbf {\bibinfo {volume} {77}},\ \bibinfo {pages}
  {094416} (\bibinfo {year} {2008})}\BibitemShut {NoStop}%
\bibitem [{\citenamefont {Kolezhuk}\ \emph {et~al.}(1997)\citenamefont
  {Kolezhuk}, \citenamefont {Mikeska},\ and\ \citenamefont
  {Yamamoto}}]{kolezhuk1997matrix}%
  \BibitemOpen
  \bibfield  {author} {\bibinfo {author} {\bibfnamefont {A.}~\bibnamefont
  {Kolezhuk}}, \bibinfo {author} {\bibfnamefont {H.-J.}\ \bibnamefont
  {Mikeska}}, \ and\ \bibinfo {author} {\bibfnamefont {S.}~\bibnamefont
  {Yamamoto}},\ }\href@noop {} {\bibfield  {journal} {\bibinfo  {journal}
  {Physical Review B}\ }\textbf {\bibinfo {volume} {55}},\ \bibinfo {pages}
  {R3336} (\bibinfo {year} {1997})}\BibitemShut {NoStop}%
\bibitem [{\citenamefont {Fletcher}\ and\ \citenamefont
  {Sorensen}(1983)}]{fletcher1983algorithmic}%
  \BibitemOpen
  \bibfield  {author} {\bibinfo {author} {\bibfnamefont {R.}~\bibnamefont
  {Fletcher}}\ and\ \bibinfo {author} {\bibfnamefont {D.~C.}\ \bibnamefont
  {Sorensen}},\ }\href@noop {} {\bibfield  {journal} {\bibinfo  {journal} {The
  American Mathematical Monthly}\ }\textbf {\bibinfo {volume} {90}},\ \bibinfo
  {pages} {12} (\bibinfo {year} {1983})}\BibitemShut {NoStop}%
\bibitem [{\citenamefont {Bartlett}(2013)}]{bartlett2013jordan}%
  \BibitemOpen
  \bibfield  {author} {\bibinfo {author} {\bibfnamefont {P.}~\bibnamefont
  {Bartlett}},\ }\href@noop {} {\bibfield  {journal} {\bibinfo  {journal}
  {Lecture 8: The Jordan Canonical Form, UCSB Math 108B}\ } (\bibinfo {year}
  {2013})}\BibitemShut {NoStop}%
\end{thebibliography}%

\end{document}